\documentclass[letterpaper,12pt]{article}
\usepackage{amsfonts}
\usepackage{amsmath, amsthm, amssymb, astron, bbm, epsfig, setspace,multirow}
\usepackage{rotating}
\usepackage{xcolor}

\numberwithin{equation}{section}

\newtheorem{definition}{Definition}
\newtheorem{theorem}{Theorem}[section]
\newtheorem{corollary}[theorem]{Corollary}
\newtheorem{lemma}[theorem]{Lemma}

\newtheorem*{example}{Example}

\newtheorem{IDassumption}{Assumption}
           
\newtheorem{SFassumption}{Assumption}
           
\newtheorem{LLassumption}{Assumption}
           
\newtheorem{SNassumption}{Assumption}
           
\newtheorem{EXassumption}{Assumption}
           
\newtheorem{NCassumption}{Assumption}
           
\newtheorem{DXassumption}{Assumption}
           
\newtheorem{DXassumptionTwo}{Assumption}
           
\newtheorem{EVassumption}{Assumption}

\newcommand{\argmin}{\operatorname*{argmin}}
\newcommand{\plim}{\operatorname*{plim}}
\newcommand{\Tr}{{\rm Tr}}
\def\ft#1#2{{\textstyle {\frac{#1}{#2}} }}

\usepackage{tocloft}
\newtheorem{HLOneAssumption}{Assumption}

\newtheorem{HLTwoAssumption}{Assumption}

\begin{document}

\title{\bf Linear Regression for Panel with Unknown Number of Factors as Interactive Fixed Effects\footnote{This is the last working paper version of the paper published in \textit{Econometrica} \textbf{83}(4), 1543--1579, 2015; doi:10.3982/ECTA9382. The copyright to this article is held by the Econometric Society. It may be downloaded, printed and reproduced only for personal or classroom use. We thank the participants of the 2009
Cowles Summer Conference
``Handling Dependence: Temporal, Cross-sectional, and Spatial'' at Yale University,
of the 2012 North American Summer Meeting of the Econometric Society at Northwestern University,
of the 18th International Conference on Panel Data at the Banque de France,
of the 2013 North American Winter Meeting of the Econometric Society in San Diego,
of the 2014 Asia Meeting of Econometric Society in Taipei,
of the 2014 Econometric Study Group Conference in Bristol, and of the econometrics seminars in USC
and Toulouse
for many interesting comments, and we thank Dukpa Kim, Tatsushi Oka, and Alexei Onatski for helpful discussions.
We are also grateful for the comments and suggestions of James Stock,  Elie Tamer, and anonymous referees.
Moon acknowledges financial supports of the NSF via SES 0920903 and the faculty grant award of USC.
Weidner acknowledges support from the Economic and Social Research Council through the ESRC Centre for Microdata Methods and Practice grant RES-589-28-0001.
}}

\author{\setcounter{footnote}{2}
Hyungsik Roger Moon\footnote{
Department of Economics and USC Dornsife INET, University of Southern California,
Los Angeles, CA 90089-0253.
Email: {\tt moonr@usc.edu}.
Department of Economics, Yonsei University, Seoul, Korea.
}
\and Martin Weidner\footnote{Corresponding author.
 Department of Economics,
 University College London,
 Gower Street,
 London WC1E~6BT, U.K.,
 and CeMMaP.
 Email: {\tt m.weidner@ucl.ac.uk}.
}}

\date{December 2014}

\maketitle

\abstract{
\begin{center}
\vspace{3mm}
\begin{minipage}{0.75\textwidth}
\footnotesize
  In this paper we study the least squares (LS) estimator in a linear panel
  regression model with \emph{unknown} number of factors appearing as interactive fixed effects.
   Assuming that the number of factors used in estimation is larger than the true number of factors in the data
  we establish the limiting distribution of the LS estimator for the regression coefficients, as
   the number of time periods and the number of cross-sectional units jointly go to infinity.
  The main result of the paper is that under certain assumptions the
  limiting distribution of the LS estimator is independent of the number of factors
  used in the estimation, as long as this number is not underestimated.
  The important practical implication
  of this result is that for inference on the regression
  coefficients one does not necessarily need to estimate the number of
  interactive fixed effects consistently.

\end{minipage}
\vspace{3mm}
\end{center}
}

\bigskip

\noindent{\bf JEL-Classification:} C23, C33

\linespread{1.3}

\newpage

\section{Introduction}

Panel data models typically incorporate individual and time effects
to control for heterogeneity in cross-section and over time.
While often these individual and time effects enter
the model additively, they can also be interacted multiplicatively, thus giving
rise to so called interactive effects, which we also refer to as a
factor structure. The multiplicative form
captures the heterogeneity in the data more flexibly, since it allows for common
time-varying shocks (factors) to affect the cross-sectional units
with individual specific
sensitivities (factor loadings).\footnote{
The conventional additive model can be interpreted as a two factor interactive fixed effects model.}
It is this flexibility that motivated
the discussion of interactive effects in the econometrics literature, e.g.
Holtz-Eakin, Newey and Rosen~\cite*{HoltzEakin-Newey-Rosen1988},
Ahn, Lee and Schmidt \cite*{AhnLeeSchmidt2001,AhnLeeSchmidt2013},
Pesaran~\cite*{Pesaran2006}, Bai~\cite*{Bai2009,Bai2013likelihood},  Zaffaroni \cite*{Zaffaroni2009},
Moon and Weidner~\cite*{MoonWeidner2013}, and Lu and Su \cite*{LuSu2013}.

Let $N$ be the number of cross-sectional units, $T$ be the number of time periods,
$K$ be the number of regressors, and $R^0$ be the true number of interactive fixed effects.
We consider a linear regression model with observed outcomes $Y$, regressors $X_k$,
and unobserved error structure $\varepsilon$, namely
\begin{align}
   Y &= \sum_{k=1}^{K} \, \beta_{k}^{0} \, X_k \, + \, \varepsilon
        \, ,
   &
   \varepsilon &= \lambda^0 \, f^{0 \, \prime} + e  \, ,
   \label{model}
\end{align}
where $Y$, $X_k$, $\varepsilon$ and $e$ are $N\times T$ matrices,
$\lambda^0$ is an $N \times R^0$ matrix, $f^0$ is a $T \times R^0$ matrix,
and the regression parameters $\beta^0_k$ are scalars
 --- the superscript zero indicates the true value of the parameters.
We write $\beta$ for the $K$-vector of regression parameters,
and we denote the components of the different matrices
by $Y_{it}$, $X_{k,it}$, $e_{it}$, $\lambda^0_{ir}$ and $f^0_{tr}$,
where $i = 1, \ldots, N$, $t=1, \ldots, T$, and $r=1,\ldots,R^0$.
It is convenient to introduce the notation
$\beta \cdot X := \sum_{k=1}^{K} \, \beta_{k} \, X_k$.
All matrices, vectors and scalars in this paper are real valued.

We consider the interactive fixed effect specification,
i.e. we treat $\lambda^0$ and $f^0$ as nuisance parameters, which are estimated
jointly with the parameters of interest $\beta$.\footnote{%
When we refer to
interactive fixed effects we mean that both factors and factor loadings
are treated as non-random parameters.
Ahn, Lee and Schmidt \cite*{AhnLeeSchmidt2001}
take a hybrid approach in that they treat the factors as non-random,
but the factor loadings as random.
The common correlated effects estimator of
Pesaran~\cite*{Pesaran2006} was introduced in a context, where both the factor loadings
and the factors follow certain probability laws, but it exhibits many  properties
of a fixed effects estimator.}
The advantages of the fixed effects approach are for instance that it is semi-parametric,
since no assumption on the distribution of the interactive effects needs to be made,
and that the regressors can be arbitrarily correlated with the interactive effect
parameters.

We study the least squares (LS) estimator of model \eqref{model}, which
 minimizes the sum of squared residuals
to estimate the unknown parameters $\beta$, $\lambda$ and $f$.\footnote{%
The LS estimator
is sometimes called ``concentrated'' least squares estimator in the literature,
and in an earlier version of the paper we referred to it as the ``Gaussian Quasi Maximum Likelihood Estimator'',
since LS estimation is equivalent to maximizing a conditional Gaussian likelihood function.
Note also that for fixed $\beta$ the LS estimator for $\lambda$
and $f$ is simply the principal components estimator.
}
To our knowledge, this estimator was first discussed in Kiefer~\cite*{Kiefer1980}. Under an asymptotic where $N$ and $T$ grow to infinity, the asymptotic properties of the LS estimator were derived in Bai~\cite*{Bai2009} for strictly exogeneous regressors, and extended in Moon and Weidner~\cite*{MoonWeidner2013} to the case of pre-determined regressors.

An important restriction of these papers is that the number of factors $R^0$ is
assumed to be known.
However, in many empirical applications there is no consensus about the exact number of factors in the data or in the relevant economic model.
If $R^0$ is not known beforehand, then it may be estimated consistently,\footnote{See the discussion
in Bai~\cite*{Bai2009supp}, supplemental material, regarding estimation of $R^0$.}
but difficulties in obtaining reliable estimates for the number of factors
are well-documented in the literature (see, e.g., the simulation results in Onatski~\cite*{Onatski2010}, and also our empirical illustration in Section~\ref{sec:Empirical}). Furthermore, in order to use the existing inference results on  $R^0$ one still needs a good preliminary estimator for $\beta$, so that
working out the asymptotic properties of the LS-estimator for $R \geq R^0$ is still useful when taking that route.

We investigate the asymptotic properties of the LS estimator when the true number of factors $R^0$ is unknown and $R$ ($\geq R^0$) number of factors are used in the estimation.\footnote{
For $R<R^0$ the LS estimator can be inconsistent, since
then there are interactive fixed effects in  the model which can be correlated
with the regressors but are not controlled for in the estimation.
We therefore restrict attention to the case $R \geq R^0$.
}
We denote this estimator by $\widehat{\beta}_{R}$.

The main result of the paper, presented in Section~\ref{sec:main}, is that under certain assumptions
the LS estimator  $\widehat{\beta}_{R}$ has the same limiting distribution as $\widehat{\beta}_{R^0}$
for any $R \geq R^0$ under an asymptotic where both $N$ and $T$ become large, while  $R^0$  and $R$ are constant.
This implies that the LS estimator $\widehat{\beta}_{R}$ is asymptotically robust towards inclusion
of  extra interactive effects in the model, and within the LS estimation framework
there is no asymptotic efficiency loss from choosing $R$ larger than $R^0$.
The important empirical implication of our result is that the number of factors $R^0$ need not be known or estimated accurately to apply the LS estimator.

To derive this robustness result, we impose more restrictive conditions than those typically assumed with known $R^0$. These include that the errors $e_{it}$ are independent and identically (iid) normally distributed and that the regressors are composed of a ``low-rank''
strictly stationary component, a ``high-rank'' strictly stationary component, and a ``high-rank'' pre-determined
component.\footnote{
The pre-determined component of the regressors allows for linear feedback of $e_{it}$ into future realizations of $X_{k,it}$.
}
Notice that while some of these restrictions are necessary for our robustness result,
some of them (e.g.~iid normality of $e_{it}$) are imposed for technical reasons, because
in the proof we use certain results from the
theory of random matrices that are currently only available in that case (see the discussion in Section~\ref{sec:AsymptoticSummary}).
In the Monte Carlo simulations in Section~\ref{sec:MC}, we consider DGPs that violate some technical conditions to demonstrate robustness of the result.

Under less restrictive assumptions we provide intermediate results that sequentially lead to the main result in Section~\ref{sec:AsyTheory} and Appendix~\ref{sec:ConvergenceRate} and~\ref{sec:Equivalence}.
In Section~\ref{sec:consistency} we show $\sqrt{\min(N,T)}$-consistency
of the LS estimator $\widehat{\beta}_{R}$ as $N,T \rightarrow \infty$
under very mild regularity condition on $X_{it}$ and $e_{it}$, and without imposing any assumptions
on $\lambda^0$ and $f^0$ apart from $R \geq R^0$. We thus obtain consistency of the LS estimator
not only for unknown number of factors, but also for weak factors,\footnote{
See  Onatski~\cite*{Onatski2010,Onatski2012} and
Chudik, Pesaran and Tosetti~\cite*{ChudikPesaranTosetti2011}
for a discussion of ``strong'' vs. ``weak'' factors in factor models.} which is an important robustness result.

In Section~\ref{sec:expansion} we derive an asymptotic expansion of the LS profile objective function that concentrates out $f$ and $\lambda$,
for the case $R=R^0$. Given that the profile objective function is a sum of eigenvalues of a covariance matrix, its quadratic approximation is challenging because the derivatives of the eigenvalues with respect to $\beta$ are not generally known. We thus cannot use a conventional Taylor expansion,
but instead apply the perturbation theory of linear operators to derive
the approximation.

In Section~\ref{sec:AsymptoticSummary} we provide an example that satisfies the typical assumptions imposed with known $R^0$, so that $\widehat{\beta}_{R^0}$ is $\sqrt{NT}$ consistent,
but we show that $\widehat{\beta}_{R}$ with $R > R^0$ is only $\sqrt{\min(N,T)}$ consistent in that example. This shows that stronger conditions
are required to derive our main result.

In Appendix~\ref{sec:ConvergenceRate}
we show faster than $\sqrt{\min(N,T)}$-convergence of  $\widehat{\beta}_{R}$ under
assumptions that are less restrictive than those employed for the main result, in particular allowing
for either cross-sectional or time-serial correlation of the errors $e_{it}$. In Appendix~\ref{sec:Equivalence} we provide
an alternative version of our main result of asymptotic equivalence of  $\widehat{\beta}_{R^0}$
and  $\widehat{\beta}_{R}$, $R \geq R^0$, which is derived under high-level assumptions.

In Section~\ref{sec:Empirical} we follow Kim and Oka~\cite*{KimOka2014} in employing the interactive fixed effects
specification to study the effect of US divorce law reforms on divorce rates. This empirical example illustrates that
the estimates for the coefficient $\beta$ indeed become insensitive to the choice of $R$, once $R$ is chosen sufficiently
large, as expected from our theoretical results.

Section~\ref{sec:MC} contains Monte Carlo simulation results for a static panel model.
For the simulations we consider a DGP that violates the iid normality restriction of the error term. The simulation results confirm our main result of the paper even with a relatively small sample size (e.g. $N=100$, $T=10$) and non-iid-normal errors. In the supplementary appendix, we report the Monte Carlo simulation results of an AR(1) panel model. It also confirms the robustness result in large samples, but in finite samples it shows more inefficiency than the static case. In general, one should expect some finite sample inefficiency from overestimating the number of factors when the sample size is small or the number of overfitted factors is large.

A few words on notation.
The transpose of a matrix $A$ is denoted by $A'$.
For a column vectors $v$ its Euclidean norm is
defined by $\| v \| = \sqrt{v^{\prime}v}$ .
For an $m\times n$ matrix $A$
the Frobenius or Hilbert Schmidt norm is $\| A \|_{HS} = \sqrt{{\rm Tr}%
(AA^{\prime})}$, and the operator or spectral norm is $\| A \| = \max_{0 \neq v \in
\mathbb{R}^n} \, \frac{ \| A v \|} {\| v\|}$.
Furthermore, we use $P_A = A
(A^{\prime}A)^{\dagger} A'$ and $M_A = \mathbbm{1} - A (A^{\prime}A)^{\dagger} A'$,
where $\mathbbm{1}$ is the $m\times m$ identity matrix,
and $(A^{\prime}A)^{\dagger}$ denotes some generalized inverse, in case
$A$ is not of full column rank. For
square matrices $B$, $C$, we use $B>C$ (or $B\geq C$) to indicate that $B-C$ is positive (semi) definite. We use ``wpa1'' for ``with probability approaching one''.

\section{Identification of $\beta^0, \lambda^0 f^{0 \prime}$, and $R^0$}
\label{sec:model_id}

In this section we provide a set of conditions under which the regression coefficient $\beta^0$, the interactive fixed effects $\lambda^0 f^{0 \prime}$, and the number of factors $R^0$ are determined uniquely by the data.
Here, and throughout the whole paper, we treat $\lambda$ and $f$ as non-random
parameters, i.e. all stochastics in the following
are implicitly conditional on $\lambda$ and $f$.
Let $x_k={\rm vec}(X_k)$, the $NT$-vectorization of $X_k$,
and let $x=(x_1,\ldots,x_K)$, which is an
$NT \times K$ matrix.

\begin{IDassumption}[\bf Assumptions for Identification]  
   There exists a non-negative integer $R$ such that
   \label{ass:id}
   \begin{itemize}
      \item[(i)] The second moments of $X_{it}$ and
                 $e_{it}$ exist for all $i$, $t$.

      \item[(ii)]  $\mathbbm{E}(e_{it})=0$,
     $\mathbbm{E}(X_{it} e_{it})=0$,
   for all $i$, $t$.

      \item[(iii)] $\mathbbm{E}[x' (M_F \otimes M_{\lambda^0}) x]
     >0 $, for all $F \in \mathbbm{R}^{T \times R}$.

     \item[(iv)]  $R \geq R^0 := {\rm rank}(\lambda^0 f^{0 \prime})$.

   \end{itemize}
\end{IDassumption}

\begin{theorem}[\bf Identification]
   \label{th:id}
     Suppose that  the Assumptions~\ref{ass:id} are satisfied. Then, $\beta^0$, $\lambda^0 f^{0 \prime}$, and $R^0$ are identified.\footnote{%
Here, identification means that
$\beta^0$ and $\lambda^0 f^{0 \prime}$
 can be uniquely recovered from the distribution of
$(Y,X)$ conditional on those parameters.
Identification of the number of factors follows
since $R^0 = {\rm rank}( \lambda^0 f^{0 \prime} )$.
    The factor loadings
and factors $\lambda^0$ and $f^{0}$ are not separately
identified without further normalization restrictions,
but the product $\lambda^0 f^{0 \prime}$ is identified.
}
\end{theorem}

Assumption~\ref{ass:id}$(i)$ imposes existence of second moments.
Assumption~\ref{ass:id}$(ii)$ is an exogeneity condition,
which demands that $x_{it}$ and $e_{it}$ are not correlated contemporaneously,
but allows for pre-determined regressors like lagged dependent variables.
Assumption~\ref{ass:id}$(iv)$ imposes that the true number of factors
$R^0 := {\rm rank}(\lambda^0 f^{0 \prime})$ is bounded by a non-negative integer
$R$, which cannot be too large (e.g. the trivial bound $R=\min(N,T)$ is not possible),
since otherwise Assumption~\ref{ass:id}$(iii)$
cannot be satisfied.

Assumption~\ref{ass:id}$(iii)$ is a non-collinearity condition,
which demands that the regressors have significant variation across $i$ and over $t$ after projecting out all variation that
can be explained by the factor loadings $\lambda^0$ and by arbitrary
factors $F \in \mathbbm{R}^{T \times R}$.
This generalizes the within variation assumption in the conventional panel regression with time invariant individual fixed effects, which in our notation reads
$\mathbbm{E}[x' (M_{1_T} \otimes \mathbbm{1}_N) x] >0$.\footnote{%
The conventional panel regression with additive individual fixed effects and time effects requires
a non-collinearity condition of the form $\mathbbm{E}[x' (M_{1_T} \otimes M_{1_N}) x] >0$.}
This conventional fixed effect assumption rules out time-invariant regressors.
Similarly, Assumption~\ref{ass:id}$(iii)$ rules out more general ``low-rank regressors'',\footnote{
We do not consider such ``low-rank regressors'' in this paper.
Note also that Assumption A in Bai~\cite*{Bai2009} is the sample version of our Assumption~\ref{ass:id}$(iii)$.}
see our discussion of Assumption~\ref{ass:NC} below.

\section{Main Result}
\label{sec:main}

The estimator we investigate in this paper is the least squares (LS) estimator,
which for a given choice of $R$ reads\footnote{
The optimal $\widehat \Lambda_R$ and $\widehat F_R$ in \eqref{estimator} are not unique,
since the objective function is invariant under right-multiplication of $\Lambda$
with a non-degenerate $R \times R$ matrix $S$, and simultaneous right-multiplication
of $F$ with $(S^{-1})'$. However, the column spaces of $\widehat \Lambda_R$ and $\widehat F_R$
are uniquely determined.
}
\begin{align}
   \left( \widehat \beta_{R},\,\widehat \Lambda_R,\,\widehat F_R \right) &\in
     \;
  \argmin_{ \left\{ \beta \in \mathbbm{R}^K, \; \Lambda \in \mathbbm{R}^{N\times R}, \;
                F \in \mathbbm{R}^{T\times R} \right\}  } \;
    \left\| Y \, - \, \beta \cdot X \, - \, \Lambda \, F'
    \right\|^2_{HS}    \; ,
   \label{estimator}
\end{align}
where $\|.\|_{HS}$ refers to the Hilbert Schmidt norm, also called Frobenius norm.
The objective function $ \left\| Y \, - \, \beta \cdot X \, - \, \Lambda \, F'
    \right\|^2_{HS}$ is simply the sum of squared residuals.
The estimator for $\beta^0$ can equivalently be defined by minimizing
the profile objective function that concentrates out the $R$ factors and the $R$ factor loadings, namely
\begin{align}
   \widehat \beta_R &= \argmin_{\beta \in \mathbbm{R}^K} \, {\cal L}_{NT}^R(\beta) \; ,
   \label{DefLS estimator}
\end{align}
with\footnote{
The profile objective function ${\cal L}_{NT}^R(\beta)$ need not be convex in $\beta$
and can have multiple local minima.
Depending on the dimension of $\beta$
one should either perform an initial grid search or try multiple starting values
for the optimization when calculating the global minimum $\widehat \beta_R$ numerically. See also
Section~\ref{app:Numerics} of the supplementary material.
 }
\begin{align}
   {\cal L}_{NT}^R(\beta) &= \min_{\left \{ \Lambda \in \mathbbm{R}^{N\times R}, \;
                F \in \mathbbm{R}^{T\times R} \right\}  } \;
                \frac 1 {NT}  \;
    \left\| Y \, - \, \beta \cdot X \, - \, \Lambda \, F'
    \right\|^2_{HS}
  \nonumber \\
      &= \min_{F \in \mathbbm{R}^{T\times R}} \; \frac 1 {NT}  \;
          {\rm Tr}\left[ \left(Y - \beta \cdot X\right)
                         M_F
                        \left(Y - \beta \cdot X\right)' \right]
 \nonumber \\
      &= \; \frac 1 {NT}  \; \sum_{r=R+1}^{T}
            \mu_r\left[ \left(Y - \beta \cdot X \right)'
                     \left(Y - \beta \cdot X \right) \right] \; ,
   \label{LSobjective}
\end{align}
where, $\mu_r(.)$ is the $r$'th largest eigenvalue of the matrix argument.
Here, we first concentrated out $\Lambda$ by use of its own first order condition.
The resulting optimization problem for $F$ is a principal components
problem, so that the the optimal $F$ is given by the $R$ largest principal components
of the $T \times T$ matrix $\left(Y - \beta \cdot X\right)'\left(Y - \beta \cdot X\right)$. At the optimum the projector $M_F$ therefore exactly projects out
the $R$ largest eigenvalues
of this matrix, which gives rise to the final formulation of the
profile objective function as the sum over its $T-R$ smallest eigenvalues.\footnote{
This last formulation of ${\cal L}_{NT}^R(\beta)$ is very convenient
since it does not involve any explicit optimization over nuisance parameters.
Numerical calculation of eigenvalues is very fast,
so that the numerical evaluation of ${\cal L}_{NT}^R(\beta)$ is unproblematic
for moderately large values of $T$.
Since the model is symmetric under $N \leftrightarrow T$,
$\Lambda \leftrightarrow F$, $Y \leftrightarrow Y'$, $X_k \leftrightarrow X_k'$ there also exists
a dual formulation of ${\cal L}_{NT}^R(\beta)$ that involves solving an eigenvalue
problem for an $N \times N$ matrix.}
We write ${\cal L}_{NT}^0(\beta)$ for ${\cal L}_{NT}^{R^0}(\beta)$, the profile objective function obtained for the
true number of factors.
Notice that in \eqref{DefLS estimator} the parameter set for $\beta$ is the whole Euclidean space $\mathbbm{R}^K$ and we do not restrict the parameter set to be compact.

\begin{SFassumption}[\bf Strong Factor Assumption]~
\label{ass:SF}
\begin{itemize}
   \item[(i)] $0 < \plim_{N,T \rightarrow \infty}
                               \frac 1 N \, \lambda^{0\prime} \lambda^0
                             < \infty$. \quad
   \item[(ii)] $0 <  \plim_{N,T \rightarrow \infty}
                               \frac 1 T \, f^{0\prime} f^0  < \infty$.
\end{itemize}
\end{SFassumption}

\begin{NCassumption}[\bf Non-Collinearity of $X_k$]~
  \label{ass:NC}
  Consider linear combinations
  $\alpha \cdot X :=\sum_{k=1}^K \alpha_k X_k$ of the regressors $X_k$ with
  $K$-vector $\alpha$ such that $\|\alpha\|=1$.
  We assume that there exists a constant $b>0$ such that
  \begin{align*}
     \min_{\{\alpha \in \mathbb{R}^{K}, \, \|\alpha\|=1\}}   \,
     \sum_{r=R+R^0+1}^{T} \,
      \mu_{r} \left[ \frac{ (\alpha \cdot X)' (\alpha \cdot X)} {NT}\right] \;  &\geq b \; , \qquad
     \text{wpa1.}
  \end{align*}
\end{NCassumption}

\begin{LLassumption}[\bf Low Level Conditions for Main Result]  $\phantom{a}$
\label{ass:LL}
\begin{itemize}
      \item[(i)] {\bf Decomposition of Regressors:}
      $X_k = \overline X_k  + \widetilde X_k^{\rm str} +  \widetilde X_k^{\rm weak}$,
      for $k=1,\ldots,K$, where
     $\overline X_k$, $\widetilde X_k^{\rm str}$ and $\widetilde X_k^{weak}$ are $N \times T$ matrices, and
     \begin{itemize}

     \item[(i.a)] {\bf Low-Rank (strictly exogenous) Part of Regressors:} ${\rm rank}(\overline X_k)$ is bounded as $N,T \rightarrow \infty$,
     and  $\frac 1 {NT} \sum_{i=1}^N \sum_{t=1}^T \overline X_{k,it}^2  = {\cal O}_P(1)$.

    \item[(i.b)] {\bf High-Rank (strictly exogenous) Part of Regressors:}
       $\| \widetilde X_k^{\rm str} \| = {\cal O}_P(N^{3/4})$, as can be justified e.g. by Lemma~\ref{lemma:SpectralNormBound}
       in the appendix.

     \item[(i.c)] {\bf Weakly Exogenous Part of Regressors:}
     $\widetilde X_{k,it}^{\rm weak} = \sum_{\tau=1}^{t-1} \gamma_\tau e_{i,t-\tau}$,
     where the real valued coefficients $\gamma_\tau$ satisfy
     $\sum_{\tau=1}^\infty | \gamma_\tau | < \infty$.

    \item[(i.d)]  {\bf Bounded Moments:}            We assume that
                   $\mathbbm{E} \left| X_{k,it}    \right|^{2}$,
                  $\mathbbm{E} \left|(M_{\lambda^0} X_k M_{f^0})_{it}
                               \right|^{26}$,
                 $\mathbbm{E} \left|(M_{\lambda^0} X_k)_{it}
                               \right|^{8}$
                 and
                 $\mathbbm{E} \left|(X_k M_{f^0})_{it}
                               \right|^{8}$
                  are bounded uniformly over $k$, $i$, $j$, $N$ and $T$.

    \end{itemize}

      \item[(ii)] {\bf Errors are iid Normal:}
      The error matrix $e$ is independent of
                 $\lambda^0$, $f^0$,
                 $\overline X_k$, and $\widetilde X_k^{\rm str}$, $k=1,\ldots,K$,
                 and its elements $e_{it}$ are
                 independent and identically distributed as ${\cal N}(0,\sigma^2)$ across $i$ and over $t$.

     \item[(iii)] {\bf Number of Factors not Underestimated:}
     $R \geq R^0 := {\rm rank}(\lambda^0 f^{0 \prime})$.
   \end{itemize}
\end{LLassumption}

\paragraph{Remarks}

\begin{itemize}
  \item[(i)] Assumption~\ref{ass:SF} imposes that the factor $f^0$ and the factor loading $\lambda^0$ are strong. The strong factor assumption is regularly imposed in the literature on large $N$ and $T$ factor models, including Bai and Ng~\cite*{BaiNg2002},
Stock and Watson~\cite*{StockWatson2002} and Bai~\cite*{Bai2009}.

  \item[(ii)] Assumption~\ref{ass:NC} demands that there exists significant sampling variation in the regressors after concentrating out $R+R^0$ factors (or factor
  loadings). It is a sample version of the identification Assumption~\ref{ass:id}$(iii)$, and it is essentially equivalent to Assumption~A of Bai~\cite*{Bai2009},
  but avoids mentioning the unobserved loadings $\lambda^0$.\footnote{By dropping the expected value from
  Assumption~\ref{ass:id}$(iii)$ and replacing the zero lower bound by
  a positive constant
   one obtains
  $\inf_{F} \left[ x' (M_F \otimes M_{\lambda^0}) x /NT \right] \geq b >0$, wpa1,
  which is equivalent to Assumption A of Bai~\cite*{Bai2009},
  and can also be rewritten as
 $\min_{\| \alpha \|=1} \inf_{F}
 {\rm Tr}\left[ M_{\lambda^0} (\alpha \cdot X)' M_{F} (\alpha \cdot X)/NT \right] \geq b$.
  A slightly stronger version of the Assumption, which avoids mentioning
  the unobserved factor loading $\lambda^0$,
  reads
  $\min_{\| \alpha \|=1} \inf_{F} \inf_{\lambda}
 {\rm Tr}\left[ M_{\lambda} (\alpha \cdot X)' M_{F} (\alpha \cdot X)/NT \right] \geq b$,
where $F \in \mathbbm{R}^{T \times R}$ and $\lambda \in \mathbbm{R}^{N \times R^0}$,
and this slightly stronger version is equivalent to Assumption~\ref{ass:NC}.
}

  \item[(iii)] Assumption~\ref{ass:NC} is violated if there exists a linear combination $\alpha \cdot X$ of the regressors
with $\alpha \neq 0$ and ${\rm rank}(\alpha \cdot X) \leq R+R^0$, i.e. the assumption
rules out ``low-rank regressors'' like time invariant regressors
or cross-sectionally invariant regressors. These low-rank regressors
require a special treatment in the interactive fixed effect model,
see Bai~\cite*{Bai2009} and Moon and Weidner~\cite*{MoonWeidner2013},
and we do not consider them in the present paper.
If one is not interested explicitly
in their regression coefficients, then one can always
eliminate the low-rank regressors by an appropriate projection of the data, e.g.
subtraction of the time (or cross-sectional) means from the data eliminates all
time-invariant (or cross-sectionally invariant) regressors,
see Section~\ref{sec:Empirical} for an example of this.

 \item[(iv)] The norm restriction in Assumption~\ref{ass:LL}$(i.b)$ is a high level assumption.
 It is satisfied as long as $\widetilde X_{k,it}^{\rm str}$ is mean zero and
  weakly correlated across $i$ and over $t$,  for details see Appendix~\ref{app:SpectralNorm} and
  Lemma~\ref{lemma:SpectralNormBound} there.

  \item[(v)] Assumption~\ref{ass:LL}$(i)$ imposes that each regressor consists of three parts: (a) a strictly exogenous low rank component , (b) a strictly exogenous component satisfying a norm restriction, and (c) a weakly exogenous component
that follows a linear process with innovation given by the lagged error term $e_{it}$.
For example, if $X_{k,it} \, \sim \, iid \, {\cal N}( \mu_k, \sigma_k^2)$, independent of $e$, then we have
$\overline X_{k,it}=\mu_k$, $\widetilde X_{k,it}^{\rm str} \, \sim \, iid \, {\cal N}( 0, \sigma_k^2)$ and
$\widetilde X_k^{\rm weak} = 0$.
Assumption~\ref{ass:LL}$(i)$ is also satisfied for a stationary panel VAR with interactive fixed effects as in Holtz-Eakin, Newey and Rosen~\cite*{HoltzEakin-Newey-Rosen1988}. A special case of this is a dynamic panel regression with fixed effects, where $Y_{it} = \beta Y_{i,t-1} + \lambda^{0 \prime}_i f^0_t + e_{it}$, with $| \beta |<1$ and ``infinite history''. In this case, we have
$X_{it} = Y_{i,t-1} = \overline X_{it}  + \widetilde X_{it}^{\rm str} +  \widetilde X_{it}^{\rm weak}$, where
$\overline X_{it} =  \lambda^{0 \prime}_i  \sum_{\tau=1}^\infty \beta^{\tau-1} f^0_{t-\tau}$,
 $\widetilde X_{it}^{\rm str} = \sum_{\tau=t}^\infty \beta^{\tau-1} e_{i,t-\tau}$,
 and $\widetilde X_{it}^{\rm weak} = \sum_{\tau=0}^{t-1}  \beta^{\tau-1} e_{i,t-\tau}$.

  \item[(vi)] Assumption~\ref{ass:LL}$(i)$ is more restrictive than Assumption 5 in Moon and Weidner~\cite*{MoonWeidner2013}, where $R^0$ is assumed to be known. However, it is more general than the restriction on the regressors in Pesaran~\cite*{Pesaran2006}, where  -- in our notation -- the decomposition
  $X_k = \overline X_k  + \widetilde X_k^{\rm str}$ is  imposed,
  but the lower rank component $\overline X_k$ needs to satisfy further
  assumptions,
   and the weakly exogenous component   $\widetilde X_k^{\rm weak} $ is not considered.
Bai~\cite*{Bai2009} requires no such decomposition, but
imposes strict exogeneity of the regressors.

    \item[(vii)]    Among the conditions in Assumption~\ref{ass:LL},
    the iid normality condition in Assumption~\ref{ass:LL}$(ii)$ may be the most restrictive.
    In Appendix~\ref{sec:Equivalence} we provide an alternative version of Theorem~\ref{th:MAIN}
    that imposes more general high-level conditions.
    Verifying those high-level conditions requires results on the eigenvalues and eigenvectors
    of random covariance matrices, which can be verified for iid normal errors by using
    known results from the random matrix theory literature,
    see Section~\ref{sec:AsymptoticSummary} for more details.
    We believe, however, that those high-level conditions and thus our main result hold more generally,
    and we explore non-normal and serially correlated errors in our Monte Carlo simulations below.
\end{itemize}

\begin{theorem}[\bf Main Result]
\label{th:MAIN}
Let Assumption~\ref{ass:SF}, \ref{ass:NC}
and \ref{ass:LL} hold and
consider a limit $N,T \rightarrow \infty$ with $N/T \rightarrow \kappa^2$, $0<\kappa<\infty$. Then we have
\begin{align*}
    \sqrt{NT}\big(\widehat \beta_{R} - \beta^0\big)
    =\sqrt{NT}\big(\widehat \beta_{R^0} - \beta^0\big) + o_P(1).
\end{align*}
\end{theorem}

Theorem~\ref{th:MAIN} follows from Theorem~\ref{th:LimitingDistribution} and  Lemma~\ref{lemma:JustifyEV} in the appendix, whose proof is given in the supplementary material.
The theorem guarantees that the asymptotic distribution of $\widehat \beta_{R}$, $R \geq R^0$,
is identical to that of $\widehat \beta_{R^0}$ in \eqref{AsyDistribution} below.

The limiting distribution of $\sqrt{NT}\big(\widehat \beta_{R^0} - \beta^0\big)$ with known $R^0$ is available in the existing literature. According to Bai~\cite*{Bai2009} and Moon and Weidner~\cite*{MoonWeidner2013},
\begin{align}
    \sqrt{NT}\big(\widehat \beta_{R^0} - \beta^0\big) \;  \Rightarrow \;
    {\cal N}\left( -  \kappa \, \plim  W^{-1} B ,  \;
    \sigma^2  \plim  W^{-1} \right) ,
    \label{AsyDistribution}
\end{align}
where $W$ is the $K \times K$ matrix with elements
   $W_{k_1 k_2} =  \frac 1 {NT} {\rm Tr}( M_{\lambda^0} X_{k_1} M_{f^0} X_{k_2}' )$ and
 $B$ is the $K$-vector with elements
$B_k = \frac 1 N {\rm Tr}[ P_{f^0} \mathbbm{E}(e'  X_k)]$.\footnote{%
The asymptotic distribution in \eqref{AsyDistribution} can also be derived from Corollary~\ref{cor:LimitR0} below
under more general conditions than in Assumption~\ref{ass:LL} (see
Moon and Weidner~\cite*{MoonWeidner2013} for details).
Here we have used the homoscedasticity of $e_{it}$ to simplify the structure of the asymptotic variance
and bias. Bai~\cite*{Bai2009} finds further asymptotic bias
in $\widehat \beta_{R^0}$ due to heteroscedasticity and correlation in $e_{it}$,
which in our asymptotic result is ruled out by Assumption~\ref{ass:LL}$(ii)$,
but is studied in our Monte Carlo simulations below.
Moon and Weidner~\cite*{MoonWeidner2013} work out the additional asymptotic
bias in $\widehat \beta_{R^0}$ due to pre-determined regressors, which
is allowed for in Theorem~\ref{th:MAIN}.
}

The result \eqref{AsyDistribution} holds under the assumptions of Theorem~\ref{th:MAIN}
and also assuming that $ \plim  W^{-1} B$ and   $\plim  W^{-1}$ exist, where
$\plim$ refers to the probability limit as $N,T \rightarrow \infty$.
Note that Assumption~\ref{ass:NC} guarantees that $W$ is invertible asymptotically.
The asymptotic bias in \eqref{AsyDistribution} is an incidental parameter bias due to pre-determined regressors and is equal to zero for
strictly exogenous regressors (for which $\mathbbm{E}(e'  X_k)=0$); it generalizes the well-known
Nickell~\cite*{Nickell1981} bias of the within-group estimator for dynamic panel models.

Estimators for $\sigma^2$, $W$ and $B$ are given by\footnote{The first factor in $ \widehat \sigma^2$ reflects the degree of freedom correction from
estimating $\Lambda$, $F$ and $\beta$, but could simply be chosen as $1/NT$ for the purpose of consistency.
Note also that  $P_{\widehat F_R,t \tau} = {\cal O}_P(1/T)$, which explains why no $1/T$  factor is required in
the definition of $\widehat B_{R,k}$.}
\begin{align*}
     \widehat \sigma_R^2 &= \frac 1 {(N-R)(T-R)-K} \sum_{i=1}^N \sum_{t=1}^T
       \left( \widehat e_{R,it}  \right)^2 ,
    &
     \widehat W_{R,k_1 k_2} &= \frac 1 {NT}
      {\rm Tr}\left( M_{\widehat \Lambda_R} X_{k_1} M_{\widehat F_R} X_{k_2}'  \right) ,
    \\
     \widehat B_{R,k} &=
     \sum_{t=1}^T \sum_{\tau=t+1}^{t+M} P_{\widehat F_R,t \tau}  \left[ \frac 1 N  \sum_{i=1}^N  \widehat e_{R,it}  X_{k,i \tau} \right] ,
\end{align*}
where $\widehat e_{R,it}$ denotes the $(i,t)^{th}$ element of
$\widehat e_R = Y - \widehat \beta_R \cdot X- \widehat \Lambda_R \widehat F_R'$,
and $P_{\widehat F_R,t \tau}$ denotes the $(t,\tau)^{th}$ element of
$P_{\widehat F_R} = \mathbbm{1}_T - M_{\widehat F_R} =
\widehat F_R (\widehat F_R' \widehat F_R)^\dagger \widehat F_R' $,
and $M \in \{1,2,3,\ldots\}$ is a bandwidth parameter that also depends on the sample size $N,T$.
Let $\widehat W_R$ and $\widehat B_R$ be the matrix and vector with elements
$ \widehat W_{R,k_1 k_2}$ and $\widehat B_{R,k}$, respectively.

The next theorem establishes the consistency of these estimators.
Let $\lambda^{\rm red} \in \mathbbm{R}^{N \times (R-R^0)}$ and $f^{\rm red}  \in \mathbbm{R}^{T \times (R-R^0)}$ be the leading $R-R^0$ principal components obtained from the $N \times T$ matrix
$M_{\lambda^0} e M_{f^0}$, i.e. $\lambda^{\rm red}$ and $f^{\rm red}$ minimize the
objective function $\left\| M_{\lambda^0} e M_{f^0} -  \lambda^{\rm red} \, f^{\rm red \prime}  \right\|^2_{HS}$,
analogous to $\widehat \Lambda_R$ and $\widehat F_R$ defined in \eqref{estimator}.\footnote{The superscript
``red'' stands for redundant, because it turns out that $\lambda^{\rm red}$ and $f^{\rm red}$ are
asymptotically close to the $R-R^0$ redundant principal components that are estimated in~\eqref{estimator}.
}
\begin{theorem}[\bf Consistency of Bias and Variance Estimators]~
\label{th:Estimators}
\begin{itemize}
\item[(i)]
Let the conditions of Theorem~\ref{th:MAIN} hold.
Then we have $\left\| P_{\widehat F_R} - P_{[f^0,f^{\rm red}]} \right\| = o_p(1)$,
$\left\| P_{\widehat \Lambda_R} - P_{[\lambda^0,\lambda^{\rm red}]} \right\| = o_p(1)$,
$\widehat \sigma_R^2 = \sigma^2 + o_P(1)$,
and
$\widehat  W_R =  W + o_P(1)$.

\item[(ii)]
In addition, let $X_{k,\cdot t} = (X_{k,1t},...,X_{k,Nt})^{\prime}$, and assume that (1) $\gamma_\tau$ in Assumption~\ref{ass:LL}(i.c) satisfies $|\gamma_\tau| < c \tau^{-d}$
for some $c>0$ and $d>1$,
(2) $\| \lambda^0_i \|$ and $\| f^0_t \|$
are uniformly bounded over $i,t$ and $N,T$,
(3) $\max_t \| X_{k,\cdot t} \| = {\cal O}_P(\sqrt{N} \log N)$,\footnote{
The high-level assumption $\max_t \| X_{k,\cdot t} \| = {\cal O}_P(\sqrt{N} \log N)$
can be shown to be satisfied for the regressor component $\widetilde X^{\rm weak}_{k,it}$ above, and can be justified for
the other regressor components e.g. by assuming that
$\overline X_{k}$ and $\widetilde X_k^{\rm str}$ are uniformly bounded.
}
and (4) the bandwidth $M \rightarrow \infty$ such that
$M  (\log T)^2 T^{-1/6} \rightarrow 0$. Then, we have $\widehat B_R =  B + o_P(1)$.

\end{itemize}

\end{theorem}

Combining Theorems~\ref{th:MAIN} and~\ref{th:Estimators} and the asymptotic distribution in \eqref{AsyDistribution}
allows inference on $\beta$, for $R \geq R^0$. In particular, the bias corrected estimator
$\widehat \beta^{\rm BC}_{R} = \widehat \beta_{R} + \frac 1 T \widehat W_R^{-1} \widehat B_R$
satisfies\footnote%
{Instead of estimating the bias analytically one can use the result that the bias
is of order $T^{-1}$ and perform split panel bias correction as in Dhaene and Jochmans~\cite*{DhaeneJochmans2010},
which instead of the conditions of Theorem~\ref{th:Estimators}(ii) only
requires some stationary condition over time.}
\[
\sqrt{NT}\big(\widehat \beta^{\rm BC}_{R} - \beta^0\big)  \Rightarrow
    {\cal N}( 0 , \sigma^2 W^{-1} ).
\]

\paragraph{Heuristic Discussion of the Main Result}~
\\
Intuitively, the inclusion of
unnecessary factors in the LS estimation is similar to the
inclusion of irrelevant regressors in an OLS regression.
In the OLS case it is well known that if those irrelevant extra regressors are uncorrelated with the regressors
of interest, then they have no effect on the asymptotic distribution of the regression coefficients of interest.
It is therefore natural to expect that if the extra estimated factors in $\widehat F_R$ are asymptotically uncorrelated
with the regressors, then the result of Theorem~\ref{th:MAIN} should hold.
To explore this,
remember that  $\widehat F_R$ is given by the
first $R$ principal components of the matrix
$(Y-\widehat \beta_R \cdot X)' (Y-\widehat \beta_R \cdot X)$, and write
\begin{align*}
  Y-\widehat \beta_R \cdot X
   &= \lambda^0 f^{0 \prime} + e - (\widehat \beta_R - \beta^0) \cdot X.
\end{align*}
The strong factor assumption and the consistency of $\widehat \beta_R$ guarantee
 that the first
$R^0$ principal components of $(Y-\widehat \beta_R \cdot X)' (Y-\widehat \beta_R \cdot X)$
are close to $f^0$ asymptotically, i.e. the true factors are correctly picked up by
the principal component estimator.
 The additional $R-R^0$ principal  components that are
estimated for $R>R^0$ cannot pick up anymore true factors and are thus
mostly determined by the remaining term
$e - (\widehat \beta_R - \beta^0) \cdot X$. The key question for the properties of
the extra estimated factors, and thus of $\widehat \beta_R$, is therefore
whether the principal components obtained from $e - (\widehat \beta_R - \beta^0) \cdot X$
are dominated by $e$ or by $(\widehat \beta_R - \beta^0) \cdot X$.
Only if they are dominated by $e$ can we expect the extra factors in $\widehat F_R$ to be uncorrelated with $X$
and thus the result in Theorem~\ref{th:MAIN} to hold.  The result on $P_{\widehat F_R}$ in
Theorem~\ref{th:Estimators} shows that the additional estimated factors are indeed close to $f^{\rm red}$,
i.e. are mostly determined by $e$, but this result is far from obvious a priori, as the following discussion shows.

Under our assumptions we have $\| e\| = {\cal O}_P(\sqrt{N})$ and $\| X_k \| = {\cal O}_P( \sqrt{NT} )$
as $N$ and $T$ grow at the same rate.
Thus, if the convergence rate of $\widehat \beta_R$ is faster than $\sqrt{N}$, i.e.
$\| \widehat \beta_R - \beta^0 \| = o_P(\sqrt{N})$, then we have
$\| e \|  \gg \left\|  (\widehat \beta_R - \beta^0) \cdot X \right\|$ asymptotically, and we expect
the extra $\widehat F_R$ to be dominated by $e$. A crucial step in the derivation
of Theorem~\ref{th:MAIN} is therefore to show faster than $\sqrt{N}$ convergence of $\widehat \beta_R$.
Conversely, we expect counter examples to the main result to be such that the convergence rate
of the estimator $\widehat \beta_R$ is not faster than $\sqrt{N}$, and we provide such a counter example
-- which, however, violates Assumptions~\ref{ass:LL} -- in Section~\ref{sec:AsymptoticSummary} below.
Whether the intuition about ``inclusion of irrelevant regressors'' carries over to the
``inclusion of irrelevant factors'' thus crucially depends on the convergence rate of $\widehat \beta_R$.

\section{Asymptotic Theory and Discussion}
\label{sec:AsyTheory}

Here we introduce key intermediate results for the proof of the main
Theorem~\ref{th:MAIN} stated above. These intermediate results may be useful independently of the main result, e.g. Moon and Weidner~\cite*{MoonWeidner2013} and Moon, Shum, and Weidner~\cite*{MoonShumWeidner2014}
crucially use the results established in Section~\ref{sec:expansion} for the case of known $R=R^0$.
The assumptions introduced below are all implied by the low-level Assumptions~\ref{ass:LL}
above, see to Lemma~\ref{lemma:JustifyEV} in the appendix.

\subsection{Consistency of $\widehat \beta_R$}
\label{sec:consistency}

Here we present a consistency result for $\widehat \beta_R$ under an arbitrary
asymptotic $N,T \rightarrow \infty$, i.e. without the assumption that $N$ and $T$ grow at the same rate,
which is imposed everywhere else in the paper. In addition to Assumption~\ref{ass:NC} we require
the following high level assumptions to obtain the result.

\begin{SNassumption}[\bf Spectral Norm of $X_k$ and $e$]~
\label{ass:SN}
\begin{itemize}
  \item [(i)] $\|X_k\| = {\cal O}_P(\sqrt{NT})$, \quad $k=1,\ldots,K$.
  \item [(ii)] $\| e \| = {\cal O}_P(\sqrt{\max(N,T)})$.
\end{itemize}
\end{SNassumption}

\begin{EXassumption}[\bf Weak Exogeneity of $X_k$]~
  \label{ass:EX}
  $\frac 1 {\sqrt{NT}} {\rm Tr}(X_k e^{\prime}) = {\cal O}_P(1)$,
              \quad $k=1,\ldots,K$.
\end{EXassumption}

\begin{theorem}
   \label{th:consistency}
      Let Assumptions \ref{ass:SN}, \ref{ass:EX} and \ref{ass:NC} be satisfied
      and let $R \geq R^0$. For $N,T \rightarrow \infty$ we then have
   $\sqrt{\min(N,T)} \left( \widehat \beta_R - \beta^0 \right) = {\cal O}_P(1)$.
\end{theorem}

\paragraph{Remarks}

\begin{itemize}
  \item[(i)] One can justify Assumption~\ref{ass:SN}$(i)$ by use of the norm inequality
$\|X_k\| \leq \|X_k\|_{HS}$ and the fact that
$\|X_k\|^2_{HS} = \sum_{i,t} X_{k,it}^2 = {\cal O}_P(NT)$, where the last step
follows e.g. if $X_{k,it}$ has a uniformly bounded second moment.
  \item[(ii)] Assumption \ref{ass:SN}$(ii)$ is a condition on the largest
eigenvalue of the random covariance matrix $e'e$, which is
often studied in the literature on random matrix theory, e.g.
Geman \cite*{Geman1980}, Bai, Silverstein, Yin \cite*{BaiSilvYin1988},
Yin, Bai, and Krishnaiah \cite*{BaiKrishYin1988}, Silverstein \cite*{Silverstein1989}.
The results in Latala \cite*{Latala2005} show that
$\| e \| = {\cal O}_P(\sqrt{\max(N,T)})$ if $e$ has independent entries with
mean zero and uniformly bounded fourth moment. Weak dependence of the
entries $e_{it}$ across $i$ and over $t$ is also permissible, see Appendix~\ref{app:SpectralNorm}

    \item[(iii)]
    Assumption \ref{ass:EX} requires exogeneity of the regressors
$X_k$, allowing for pre-determined regressors, and some weak dependence of $X_{k,it}e_{it}$ across $i$ and over $t$.\footnote{
Note that $\frac 1 {\sqrt{NT}} {\rm Tr}(X_k e^{\prime})= \frac 1 {\sqrt{NT}} \sum_i \sum_t X_{k,it}e_{it}$.
}

  \item [(iv)] The theorem
imposes no restriction at all on $f^0$ and $\lambda^0$, apart from the
condition $R \geq {\rm rank}(\lambda^0 f^{0 \prime})$.\footnote{This is the main reason why we use a slightly different
non-collinearity Assumption~\ref{ass:NC}, which avoids mentioning $\lambda^0$,
compared to Bai~\cite*{Bai2009}.}
  In particular, the strong factor Assumption~\ref{ass:SF} is not imposed here, i.e. consistency
  of $\widehat \beta_R$ holds independently of whether the factors
  are strong, weak, or not present at all. This is an  important robustness result, which
  is new in the literature.

  \item[(v)] Under an asymptotic where $N$ and $T$ grow at the same rate, which is imposed everywhere
  else in the paper, Theorem~\ref{th:consistency} shows $\sqrt{N}$ (or equivalently $\sqrt{T}$)
consistency of the estimator $\widehat \beta_R$.  To prove the consistency, we do not use the argument of the standard consistency proof for an extremum estimator which is to apply a uniform law of large numbers to the sample objective function to find the limit function that is uniquely minimized at the true parameter. Deriving the uniform limit of the objective function ${\cal L}_{NT}^{R^0}(\beta)$ is difficult. In the proof that is available in the supplementary appendix, we find a lower bound of the objective function ${\cal L}_{NT}^{R^0}(\beta)$ that is quadratic in $\beta - \beta^0$ asymptotically and establish the desired consistency, extending
the consistency proof in Bai~\cite*{Bai2009}.

\item[(vi)] $\sqrt{N}$ consistency of $\widehat \beta_R$
 implies that the residuals $Y - \widehat \beta_R \cdot X$ will be asymptotically close to
$\lambda^0 f^{0 \prime} + e$.\footnote{In the sense that
$ \| (Y - \widehat \beta_R \cdot X) - (\lambda^0 f^{0 \prime} + e) \| = \| (\widehat \beta_R - \beta) \cdot X \| = {\cal O}_P(\sqrt{N})$.}
 This allows consistent estimation
of $R^0$ under a strong factor Assumption~\ref{ass:SF} by employing the known techniques  on factor
models without regressors (by applying, e.g.~,
Bai and Ng~\cite*{BaiNg2002} to  $Y - \widehat \beta_R \cdot X$), as also
discussed in Bai~\cite*{Bai2009supp}.\footnote{%
Bai~\cite*{Bai2009supp} does not prove the required
consistency and convergence rate of $\widehat \beta_R$, for $R > R^0$.}

\item[(vii)]
Having  a consistent estimator for $R^0$, say $\widehat{R}$, one can calculate $\widehat \beta_{\widehat{R}}$, which will be
asymptotically equal to $\widehat \beta_{R^0}$.
In practice, however, the finite sample properties of the estimator $\widehat \beta_{\widehat{R}}$ crucially depend on the finite sample properties of $\widehat{R}$. Many recent papers have documented difficulties in obtaining reliable estimates for $R^0$ at finite sample (see, e.g., the simulation results of Onatski~\cite*{Onatski2010} and Ahn and Horenstein~\cite*{AhnHorenstein2013}), and those difficulties
are also illustrated by our empirical example in Section~\ref{sec:Empirical}.

\end{itemize}

\subsection{Quadratic Approximation of ${\cal L}_{NT}^0(\beta) (:= {\cal L}_{NT}^{R^0}(\beta))$}
\label{sec:expansion}

To derive the limiting distribution of $\widehat \beta_R$,
we study the asymptotic properties of the
profile objective function ${\cal L}_{NT}^R(\beta)$ around $\beta^0$.
The expression in \eqref{LSobjective} cannot easily be discussed by
analytic means, since  no explicit formula for the
eigenvalues of a matrix is available.
In particular, a standard Taylor expansion
of ${\cal L}^R_{NT}(\beta)$ around $\beta^0$ cannot easily be derived.
Here, we consider the case of known $R=R^0$
and we perform a joint expansion
of  the corresponding profile objective function ${\cal L}^0_{NT}(\beta)$
 in the regression parameters $\beta$ and in the idiosyncratic error terms $e$. To perform this joint
expansion we apply the perturbation theory of linear operators (e.g., Kato~\cite*{Kato}).
We thereby obtain an approximate quadratic expansion of ${\cal L}_{NT}^0(\beta)$ in $\beta$,
which can be used to derive the first order asymptotic theory of
the LS estimator $\widehat \beta_{R^0}$, see Appendix~\ref{app:ExpansionDiscussion} for details.
In addition to the $K \times K$ matrix $W$ already defined in Section~\ref{sec:main} we now also define
\begin{align}
        C^{(1)}_k &=
         \frac 1 {\sqrt{NT}} \, {\rm Tr}( M_{\lambda^0} \, X_k \,
                  M_{f^0} \, e^{\prime} ) \; ,  \notag \\
        C^{(2)}_k  &=
        - \, \frac 1 {\sqrt{NT}} \, \bigg[
       {\rm Tr}\left(e M_{f^0} \, e' \, M_{\lambda^0} \, X_k \,
              f^0 \, (f^{0\prime}f^0)^{-1} \, (\lambda^{0\prime}\lambda^0)^{-1} \, \lambda^{0\prime} \right)
    \nonumber \\ & \qquad \qquad \quad
       +{\rm Tr}\left(e^{\prime}M_{\lambda^0} \, e \, M_{f^0} \, X^{\prime}_k \,
              \lambda^0 \, (\lambda^{0\prime}\lambda^0)^{-1} \, (f^{0\prime}f^0)^{-1} \, f^{0\prime} \right)
    \nonumber \\ & \qquad \qquad \quad
       +{\rm Tr}\left(e^{\prime}M_{\lambda^0} \, X_k \, M_{f^0} \, e^{\prime}
                \, \lambda^0 \, (\lambda^{0\prime}\lambda^0)^{-1} \, (f^{0\prime}f^0)^{-1} \, f^{0\prime} \right)
                        \bigg]  \; .
\end{align}
Let $C^{(1)}$ and $C^{(2)}$ be the $K$-vectors with elements $C^{(1)}_k$ and $C^{(2)}_k$, respectively.

\begin{theorem}
   \label{th:expansion}
   Let Assumptions  \ref{ass:SF} and \ref{ass:SN} be satisfied.
   Suppose that $ N,T \rightarrow \infty$ with $ N/T \rightarrow \nolinebreak[4] \kappa^2$,
   $0<\kappa<\infty$.
   Then we have
   \begin{align*}
      {\cal L}_{NT}^{0}(\beta) &= {\cal L}_{NT}^{0}(\beta^0)
                           -  \ft 2 {\sqrt{NT}}   \left(\beta-\beta^0 \right)'
                           \left( C^{(1)} + C^{(2)} \right)
                    + \left(\beta-\beta^0 \right)'  W  \left(\beta-\beta^0 \right)
                          + {\cal L}_{NT}^{0,{\rm rem}}(\beta)  ,
   \end{align*}
   where the remainder term ${\cal L}_{NT}^{0,{\rm rem}}(\beta)$
    satisfies for any sequence $c_{NT}\rightarrow 0$
   \begin{align*}
     \sup_{\{\beta :\left\| \beta -\beta^{0} \right\| \leq c_{NT}\}} \frac{
        \left| {\cal L}_{NT}^{0,{\rm rem}}(\beta) \right|  } { \left( 1 + \sqrt{NT} \, \left\| \beta -\beta^{0} \right\| \right)^2 } = o_{p}\left( \frac 1 {NT} \right) .
    \end{align*}
\end{theorem}

The bound on remainder\footnote{
The expansion in Theorem~\ref{th:expansion} contains a term that is linear  in $\beta$
and linear in $e$ ($C^{(1)}$ term), a term that is linear in $\beta$  and quadratic in $e$ ($C^{(2)}$ term),
and a term that is quadratic in $\beta$ ($W$ term). All higher order terms of the expansion
are contained in the remainder term ${\cal L}_{NT}^{0,{\rm rem}}(\beta)$.}
in  Theorem~\ref{th:expansion} is such that it has no effect on the first order asymptotic theory of
$\widehat \beta_{R^0}$, as stated in the following corollary (see also  Andrews~\cite*{Andrews1999}).

\begin{corollary}
   \label{cor:LimitR0}
   Let Assumptions \ref{ass:SF}, \ref{ass:SN}, \ref{ass:EX} and \ref{ass:NC}
   be satisfied.  In the limit $N,T \rightarrow \infty$ with
   $N/T \rightarrow \kappa^2$, $0< \kappa < \infty$, we then have
                  $\sqrt{NT}\left(\widehat \beta_{R^0} - \beta^0\right)
                   = W^{-1} \left( C^{(1)} + C^{(2)} \right) + o_P \left(1 + \| C^{(1)} \| \right)$.
   If we furthermore assume that $C^{(1)} = {\cal O}_P(1)$, then we obtain
    \[\sqrt{NT}\left(\widehat \beta_{R^0} - \beta^0\right)
                   = W^{-1} \left( C^{(1)} + C^{(2)} \right) + o_P(1)
                   = {\cal O}_P(1).\]
\end{corollary}

Note that our assumptions already guarantee
$C^{(2)}={\cal O}_P(1)$ and that $W$ is invertible with $W^{-1}={\cal O}_P(1)$, so this need not be
explicitly assumed in Corollary~\ref{cor:LimitR0}.

\paragraph{Remarks}

\begin{itemize}
  \item[(i)] More details on the expansion of  ${\cal L}_{NT}^{0}(\beta)$ are provided in
Appendix~\ref{app:ExpansionDiscussion} and the formal proofs can be found in
in Section~\ref{app:expansion1} of the supplementary appendix.

  \item[(ii)] Corollary~\ref{cor:LimitR0} allows to replicate the results in
Bai~\cite*{Bai2009}
and Moon and Weidner~\cite*{MoonWeidner2013}
on the asymptotic distribution of $\widehat \beta_{R^0}$,
including the result in formula~\eqref{AsyDistribution} above.\footnote{
Let $\rho$, $D(.)$, $D_0$, $D_Z$, $B_0$ and $C_0$ be the notation used in Assumption~A
and Theorem~3 of Bai~\cite*{Bai2009}, and let Bai's assumptions be satisfied.
Then, our $\kappa$, $W$, $C^{(1)}$ and $C^{(2)}$ satisfy
$\kappa=\rho^{-1/2}$, $W= D(f^0) \rightarrow_p D>0$, $C^{(1)} \rightarrow_d  {\cal N}(0,D_Z)$
and $W^{-1} C^{(2)} \rightarrow_p \rho^{1/2} B_0 + \rho^{-1/2} C_0$.
Corollary~\ref{cor:LimitR0} can therefore be used to replicate Theorem~3 in Bai~\cite*{Bai2009}.
For more details and extensions of this we refer to Moon and Weidner~\cite*{MoonWeidner2013}.
}
The assumptions of the corollary do not restrict the regressors
to be strictly exogenous and do not impose Assumption~\ref{ass:LL}.

  \item[(iii)] If one weakens
Assumption \ref{ass:SN}$(ii)$ to $\|e\|=o_P(N^{2/3})$,
then Theorem \ref{th:expansion} still continues to hold.
If $C^{(2)}={\cal O}_P(1)$,
then Corollary~\ref{cor:LimitR0} also holds under this weaker condition on $\|e\|$.
\end{itemize}

\subsection{Remarks on Deriving the Convergence Rate and Asymptotic Distribution of  $\widehat \beta_R$ for $R>R^0$.}
\label{sec:AsymptoticSummary}

\subsubsection*{An example that motivates stronger restrictions}

The results in Bai~\cite*{Bai2009} and
Corollary~\ref{cor:LimitR0} above show that under appropriate assumptions
the estimator
$\widehat \beta_R$ is $\sqrt{NT}$-consistent for $R=R^0$.
For $R>R^0$
we know from Theorem~\ref{th:consistency} that $\widehat \beta_R$ is $\sqrt{N}$ consistent
as $N$ and $T$ grow at the same rate, but we have not shown faster than $\sqrt{N}$ converge of
$\widehat \beta_R$ for $R>R^0$, yet, which according to the heuristic discussion at the end of
Section~\ref{sec:main} is a very important intermediate step to obtain our main result.\footnote{
One reason why  $\widehat \beta_R$ might only converge at $\sqrt{N}$ rate, but not faster, are
weak factors (both for $R>R^0$ and for $R=R^0$). A weak factor
(see e.g. Onatski~\cite*{Onatski2010,Onatski2012} and Chudik, Pesaran and Tosetti~\cite*{ChudikPesaranTosetti2011})
might not be picked up at all
or might only be estimated very inaccurately by the principal components estimator $\widehat F_R$,
in which case that factor is not properly accounted for in the LS estimation procedure. If this happens
and the weak factor is correlated with the regressors, then there is some uncorrected weak
endogeneity problem, and $\widehat \beta_R$ will only converge at $\sqrt{N}$ rate.
We do not consider the issue of weak factors any further in this paper.
}
However, one might not obtain a faster than $\sqrt{N}$ convergence rate of $\widehat \beta_R$
for $R>R^0$ without imposing further restrictions, as the following example shows.

\begin{example}
     \label{example:negative}
      Let $R^0=0$ (no true factors) and $K=1$ (one regressor).
      The true model reads $Y_{it} = \beta^0 X_{it} + e_{it}$, and we consider the following data generating process (DGP)
  \begin{align*}
   X_{it} &= a \widetilde X_{it} + \lambda_{x,i} f_{x,t}, &
   e &= \left( \mathbbm{1}_{N} + c \, \frac{\lambda _{x}\lambda _{x}^{\prime }}{N} \right)
   u \left( \mathbbm{1}_{T}+c \, \frac{f_{x}f_{x}^{\prime }}{T}\right) ,
\end{align*}
where $e$ and $u$ are $N \times T$ matrices with entries $e_{it}$ and $u_{it}$, respectively,
and $\lambda _{x}$ is an $N$-vector with entries $\lambda_{x,i}$, and $f_x$ is a $T$-vector with entries $f_{x,t}$.
Let $\widetilde X_{it}$ and $u_{it}$ be mutually independent
iid standard normally distributed random variables.
Let $\lambda_{x,i} \in {\cal B}$ and $ f_{x,t} \in {\cal B}$ be
non-random sequences with bounded range ${\cal B} \subset \mathbbm{R}$ such that
$\frac 1 N \sum_{i=1}^N  \lambda^2_{x,i} \rightarrow 1$
and $\frac 1 T \sum_{t=1}^T  f^2_{x,t}  \rightarrow 1$ asymptotically.\footnote{
    We could also allow $\lambda_x$ and $f_x$ to be random (but independent of $e$
    and $\widetilde X$) and we could let the range of ${\cal B}$ be unbounded.
    We only assume non-random $\lambda_x$ and $f_x$ to guarantee that the DGP
    satisfies Assumption~D of  Bai~\cite*{Bai2009}, namely that
    $X$ and $e$ are independent (otherwise we only have mean-independence, i.e. $\mathbbm{E}(e|X)=0$).
    Similarly, we only assume bounded ${\cal B}$ to satisfy
    the restrictions on $e_{it}$ imposed in
    Assumption~C of Bai~\cite*{Bai2009}.
}
Consider $N,T \rightarrow \infty$ such that
$N/T \rightarrow \kappa^2$, $0< \kappa < \infty$,
and let $0<a <(1/2)^{2/3} \min(\kappa^2,\kappa^{-2})$ and $c \geq    \frac{ (2+\sqrt{2})  \left( 1+\kappa \right)
 (1+ \sqrt{3} a^{-1/4} )}   { \min(1,\kappa) [1/2 - a^{3/2} \max(\kappa,\kappa^{-1})]} $.\footnote{
 The bounds on the constants $a$ and $c$ imposed here are sufficient, but not necessary
 for the result of no faster than $\sqrt{N}$ convergence of $\widehat \beta_1$.
Simulation evidence suggests that this result holds for a much larger range of $a$, $c$ values.
}
For this DGP one can show that
$\widehat \beta_1$, the LS-estimator with $R=1>R^0$, only converges at a rate of $\sqrt{N}$ to $\beta^0$, but not faster.
\end{example}

The proof of the last statement is provided in the supplementary material.
The DGP in this example satisfies
      all the assumptions imposed in Corollary~\ref{cor:LimitR0} to derive the
  limiting distribution of the LS-estimator for $R=R^0$, including $\sqrt{NT}$-consistency of
  $\widehat \beta_R$ for $R=R^0$ (=0 in this example).
    It also satisfies all the regularity conditions imposed in
   Bai~\cite*{Bai2009}.\footnote{See Section~\ref{app:CheckBai} in the supplementary material  for details.}
The aspect that is special about this DGP is
  that $\lambda_x$ and $f_x$ feature both in $X_{it}$
  and in the second moment structure
   of $e_{it}$.
   The heuristic discussion at the end of Section~\ref{sec:main} provides some intuition why this
   can be problematic, because the leading principal components obtained from only the error matrix $e$
   will have a strong sample correlation with $X_{it}$ for this DGP.

\subsubsection*{Faster than $\sqrt{N}$ convergence of $\widehat \beta_R$}

 In Appendix~\ref{sec:ConvergenceRate}, we summarize our results on faster than $\sqrt{N}$ convergence of $\widehat \beta_R$ for $R \geq R^0$.  The above example shows that this requires more restrictive
 assumptions than those imposed for the analysis of the case $R=R^0$ above,
 but the assumptions that we impose for this intermediate results are
 still significantly weaker than the Assumption~\ref{ass:LL} required for our main result above,
   in particular either cross-sectional correlation or time-serial correlation of
   $e_{it}$ are still allowed.

In that appendix we also provide one set of assumptions (Assumption~\ref{ass:DX-2})
   for faster than $\sqrt{N}$ convergence such that
     no additional conditions on $e$ are required, but where the regressors are restricted
   to essentially be  lagged dependent variables in an AR(p) model with factors.

\subsubsection*{On the role of the iid normality of $e_{it}$}

We establish the asymptotic equivalence of $\widehat \beta_R$ and $\widehat \beta_{R^0}$ in Theorem \ref{th:MAIN}
by showing that the LS objective function $\mathcal{L}_{NT}^{R}(\beta)$ can, up to a constant,
 be uniformly well approximated
 by $\mathcal{L}_{NT}^{0}(\beta)$ in shrinking neighborhoods around the true parameter. For this, we need not only the faster than $\sqrt{N}$ convergence rate of $\widehat \beta_R$, but also require the Assumption~\ref{ass:EV} in Appendix~\ref{sec:Equivalence}. This is a high-level assumption on the eigenvalues and eigenvectors of the random covariance matrices
$E E'$ and $E' E$, where $E=M_{\lambda^0} e M_{f^0}$. The assumption essentially requires
the eigenvalues of those matrices to be sufficiently separated from each other,
as well as the eigenvectors of those matrices to be
sufficiently uncorrelated with the regressors $X_k$, and with $e P_{f^0}$ and $P_{\lambda^0} e$.

We use the iid normality of $e_{it}$ to verify those high-level conditions in
Section~\ref{ass:SufficiencyLL} of the supplementary appendix.
There are three reasons why we can currently only
verify those conditions for iid normal errors:

\begin{itemize}
    \item[(i)] The random matrix theory literature studies the eigenvalues and eigenvectors of random
    covariance matrices of the form $e e'$ and $e' e$, while we have to deal with the additional projectors $M_{\lambda^0}$
    and $M_{f^0}$ in the random covariance matrices.
    These additional projections stem from integrating out the true factors and factor loadings of the model.
If the error distribution is $iid$ normal, and independent from $\lambda^0$ and $f^0$, then these projections
are unproblematic, since the distribution of $e$ is rotationally invariant from the left and right in that case,
so that the projections are mathematically equivalent to a reduction of the sample size
by $R^0$ in both panel dimensions.

    \item[(ii)]
    In the iid normal case one can furthermore use the invariance of the distribution of $e$ under orthonormal
    rotations from the left and from the right to also fully
    characterize the distribution of the eigenvectors of $E E'$ and $E E'$.\footnote{Rotational invariance implies
    that the distribution of the normalized eigenvectors
    is given by the Haar measure of a rotation group manifold.}
      The conjecture in the random matrix theory literature is that
the limiting distribution of the eigenvectors of a random covariance matrix
is ``distribution free'',
i.e. is independent of the particular distribution of $e_{it}$, see, e.g., Silverstein~\cite*{Silverstein1990} and
Bai~\cite*{bai1999review}. However, we are not currently aware of a formulation
and corresponding proof of this conjecture that is sufficient for our purposes, i.e. that would allow us to
verify our high-level Assumption~\ref{ass:EV} more generally.

    \item[(iii)]  We also require certain properties of the eigenvalues of $E E'$ and $E E'$.
    Eigenvalues are studied more intensely than eigenvectors
in the random matrix theory literature,
and it is well-known that the properly normalized
empirical distribution of the eigenvalues
(the so called empirical spectral distribution)
of an $iid$ sample covariance matrix converges
to the Mar{\v{c}}enko-Pastur-law (Mar{\v{c}}enko and Pastur~\cite*{MarcenkoPastur1967})
for asymptotics where $N$ and $T$ grow at the same rate.
This result does not require normality, and results on the limiting spectral distribution are also
known for non-iid matrices. However,
to check our high-level Assumption~\ref{ass:EV}
we also need results on the
convergence rate of the empirical spectral
distribution to its limit law, which is an ongoing research subject in the
literature, e.g. Bai~\cite*{Bai1993}, Bai, Miao and Yao \cite*{BaiMiaoYao2004},
G{\"o}tze and Tikhomirov \cite*{GotzeTikhomirov2010}, and we are currently only aware of
results on this convergence rate for the case of either iid or iid normal errors.
To verify the high-level assumption we furthermore
use a result from
Johnstone~\cite*{Johnstone2001} and Soshnikov~\cite*{Soshnikov2002}
that shows that the properly normalized few
largest eigenvalues of  $E E'$ and $E E'$ converge to the Tracy-Widom law, and to our knowledge this result
is not established for error distributions that are not iid normal.
\end{itemize}

In spite of these severe mathematical challenges,
we believe that in principle our high-level Assumption~\ref{ass:EV} could be verified for more
general error distributions, implying that our main result of asymptotic equivalence of $\widehat \beta_R$
and $\widehat \beta_{R^0}$ holds more generally. This is also supported by our Monte Carlo
simulations, where we explore non-independent and non-normal error distributions.

\section{Empirical Illustration}
\label{sec:Empirical}

As an illustrative empirical example, we estimate the dynamic effects of
unilateral divorce law reforms on the state-wise divorce rates in the US. The impact of the
divorce law reform has been studied by many researches (e.g., Allen~\cite*{Allen1992},
Peters~\cite*{Peters1986,Peters1992}, Gray~\cite*{Gray1998}, Friedberg~\cite*{Friedberg1998},
Wolfers~\cite*{Wolfers2006}, and Kim and Oka~\cite*{KimOka2014}). In this section we revisit this topic,
extending Wolfers~\cite*{Wolfers2006} and Kim and Oka~\cite*{KimOka2014} by controlling
for interactive fixed effects and also a lagged dependent variable.

Let $Y_{it}$ denote the number of divorces per 1000 people in state $i$
at time $t$, and let $D_{i}$ denote the year in which state $i$ introduced the unilateral
divorce law, i.e. before year $D_i$ state $i$ had a consent divorce law, while from $D_i$ onwards state $i$
had  a unilateral ``no-fault'' divorce law, which loweres the barrier for divorce.
The goal is to estimate the dynamic effects of this law change on the divorce rate. The empirical model we estimate is
\begin{equation}
Y_{it}=\beta_{0} \, Y_{i,t-1} + \sum_{k=1}^{8} \beta _{k}  X_{k,it}
+\alpha_{i}+ \gamma_i \, t + \delta_i \, t^{2} + \mu_t +\lambda _{i}^{\prime }f_{t}+e_{it},
\label{m:empirical.ex}
\end{equation}%
where we follow Wolfers~\cite*{Wolfers2006} in defining the regressors as bi-annual dummies:
\begin{eqnarray*}
X_{k,it} &=&\mathbbm{1}\{D_{i}+2(k-1)\leq t\leq D_{i}+2k-1\},
\quad \text{for} \; \; k=1,...,7, \\
X_{8,it} &=&\mathbbm{1}\{D_{i}+2(k-1)\leq t\}.
\end{eqnarray*}%
The dummy variable and quadratic trend specification $\alpha_{i}+ \gamma_i \, t + \delta_i \, t^{2} + \mu_t$
is also used in Friedberg~\cite*{Friedberg1998} and Wolfers~\cite*{Wolfers2006}. The additional
interactive fixed effects $\lambda _{i}^{\prime }f_{t}$ were added in Kim and Oka~\cite*{KimOka2014}
to control for additional unobserved heterogeneity in the divorce rate, e.g. due to social, cultural
or demographic factors. We extend the specification further by
adding a lagged dependent variable $Y_{i,t-1}$ to control for state dependence of the divorce rate,
but we also report results without $Y_{i,t-1}$ below.
We use the dataset of Kim and Oka~\cite*{KimOka2014},\footnote{%
The data is available from {\tt http://qed.econ.queensu.ca/jae/2014-v29.2/kim-oka/}}
which is a balanced panel of $N=48$ states over $T=33$ years,
leaving $T=32$ time periods if the lagged dependent variable is included.

For estimation we first eliminate $\alpha_{i}$, $\gamma_i$, $\delta_i$ and $\mu_t$ from the model by projecting the
outcome variable and all regressors accordingly, e.g. $\widetilde Y = M_{1_N} Y M_{(1_T, {\bf t}, {\bf t}^2)}$,
where $1_N$ and $1_T$ are $N$- and $T$-vectors, respectively, with all entries equal to one,
and ${\bf t}$ and ${\bf t}^2$ are $T$-vectors with entries $t$ and $t^2$, respectively.
The model after projection reads
$\widetilde Y_{it}=\beta_{0} \, \widetilde Y_{i,t-1} + \sum_{k=1}^{8} \beta _{k}  \widetilde X_{k,it}
  +\widetilde \lambda _{i}^{\prime } \widetilde f_{t} + \widetilde e_{it}$, which is exactly the model we have studied
  so far in this paper.\footnote{To construct $\widetilde Y_{i,t-1}$ we first apply the lag-operator and then apply the
  projections $M_{1_N}$ and $M_{(1_T, {\bf t}, {\bf t}^2)}$.}
   We use the LS estimator described above to estimate this model.
The projection reduces the effective sample size to $N=48-1=47$ and $T=32-3=29$, which
should be accounted for when calculating standard errors, e.g. in the formula for
$\widehat \sigma_R^2$ above (degree of freedom correction). Our theoretical results are still applicable.\footnote{
If $e_{it}$ is iid normal, then $\widetilde e_{it}$ is not, but one can apply appropriate orthogonal rotations in $N$- and $T$-space such that $\widetilde e_{it}$ becomes iid normal again, although with sample size reduced to $N=47$ and $T=29$.
The rotation has no effect on the LS estimator, i.e. it does not matter whether we work in the original or the rotated
frame.
}

We need to decide on a number of factors $R$ when implementing the LS estimator. As already mentioned in
the last remark  in Section~\ref{sec:consistency} above, we can can apply known techniques from the literature
on factor models without regressors to obtain a consistent estimator of $R^0$. To do so we choose a maximum
number of factors of $R_{\max}=9$ to obtain the preliminary estimate $\widehat \beta_{R_{\max}}$
and then calculate the residuals
$\widehat u_{it} = \widetilde Y_{it} - \widehat \beta_{R_{\max},0} \, \widetilde Y_{i,t-1}
  - \sum_{k=1}^{8} \widehat \beta _{R_{\max},k}  \widetilde X_{k,it}$.
We then apply the IC, PC and BIC3 criteria of Bai and Ng~\cite*{BaiNg2002},\footnote{%
Following Onatski~\cite*{Onatski2010} and Ahn and Horenstein~\cite*{AhnHorenstein2013}
we report only BIC3 among the AIC and BIC criteria of Bai and Ng~\cite*{BaiNg2002}.
}
the criterion described in Onatski~\cite*{Onatski2010}, and the ER and GR criteria of
Ahn and Horenstein~\cite*{AhnHorenstein2013} to $\widehat u$.\footnote{%
To include $R=0$ as a possible outcome for the Ahn and Horenstein (2013) criterion, we use the mock
eigenvalue used in their simulations.}
Most of these criteria also require specification
of $R_{\max}$, and we continue to use $R_{\max} = 9$.
The corresponding estimation results for $R$ are presented in Table~\ref{table:facor.number.est}.
In addition, we also report the log scree plot, i.e. the sorted eigenvalues of $\widehat u' \widehat u$
in Figure 1.

The log scree plot already shows that it is not obvious how to decompose the eigenvalue spectrum
into a few larger eigenvalues stemming from factors and the remaining smaller eigenvalues stemming from
the idiosyncratic error term.\footnote{The first largest eigenvalue is 2.2 times larger than the second eigenvalue,
 the second is 1.6 times larger than the third, the third is 1.9 times larger than fourth. So the largest view eigenvalues
 are larger than the remaining ones, and the strong factor assumption might not be completely inappropriate here.
 However, deciding on a cutoff between factor and non-factor eigenvalues is difficult.}
 This problem is also reflected in the very different estimates for $R$ that one obtains from the various criteria.
 It might appear that IC1, IC3, PC1, PC2 and PC3 all agree on $\widehat R=9$, but this is simply $\widehat R=R_{\max}$,
 and if we choose $R_{\max}=10$, then all these criteria deliver $\widehat R=10$, so this should not be considered
 a reliable estimate.

\begin{table}[tb!]
\begin{minipage}{\textwidth}

\begin{minipage}[b]{0.47\textwidth}
\begin{center}
\begin{tabular}{c@{\;}c|c@{\;}c|c@{\;}c}
Criterion: & $\widehat{R}$ & Criterion: & $\widehat{R}$ & Criterion: & $\widehat{R}$
\\
\hline \hline
IC1: & $9$ & PC1: & $9$ & Onatski: & 1 \\
IC2: & $7$ & PC2: & $9$ & ER: & 1 \\
IC3: & $9$ & PC3: & $9$ & GR: & 3 \\
BIC3: & 6 &  &  &  &
\end{tabular}
\end{center}
\caption{\label{table:facor.number.est} \footnotesize
Estimated number of factors in the residuals $\widehat u$,
using different criteria for estimation and $R_{\max}=9$. The  IC, PC and BIC criteria
are described in Bai and Ng~(2002), the ER and GR criteria are from Ahn and Horenstein (2013),
and we also use the criterion of Onatski~(2010).
}

\end{minipage}
\hfill
\begin{minipage}[b]{0.47\textwidth}
\begin{center}
\includegraphics[width=0.9\textwidth]{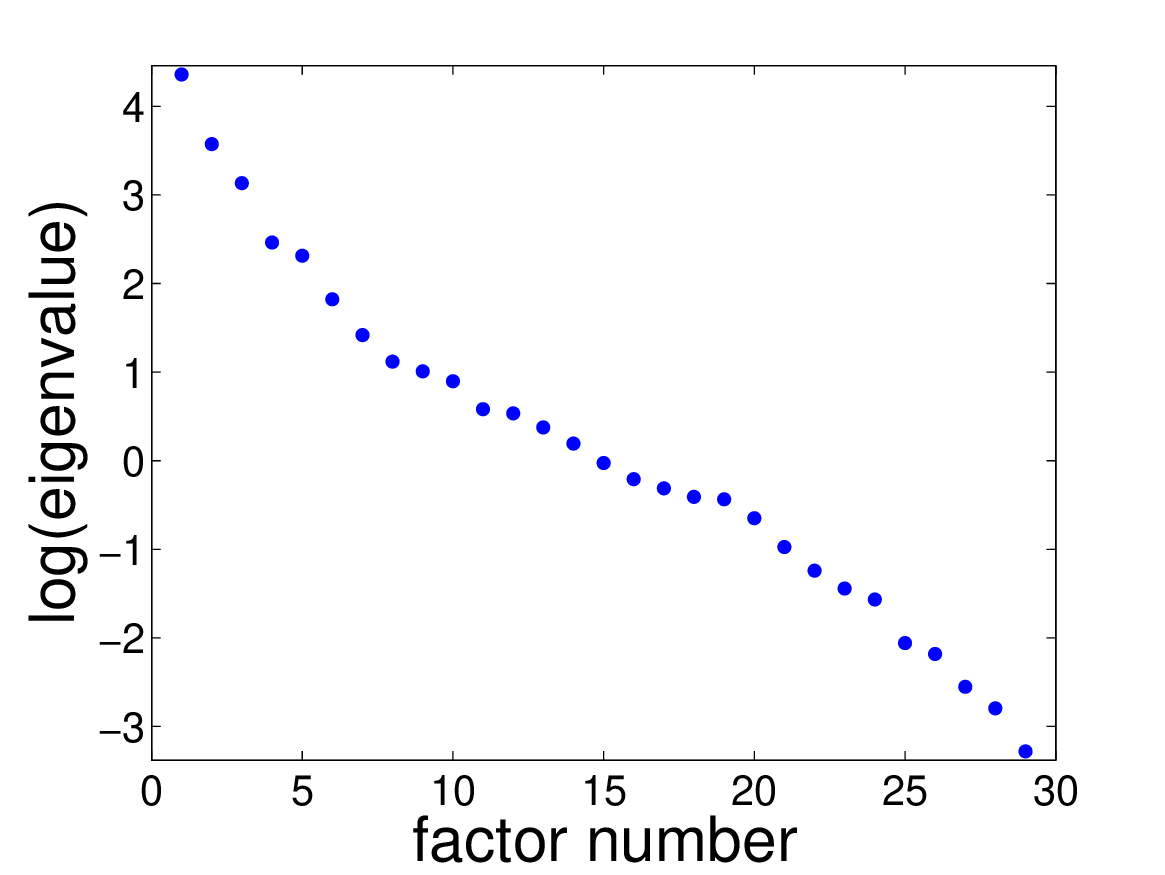}
\renewcommand{\tablename}{Figure}
\end{center}
\vspace{-0.3cm}
Figure~1: {\footnotesize
Log scree plot.
The natural logarithm of the sorted eigenvalues (corresponding to the principal components, or factors)
of $\widehat u' \widehat u$ are plotted.
}
\renewcommand{\tablename}{Table}

\end{minipage}
\end{minipage}
\end{table}

\begin{table}[tb!]
\centering
\includegraphics[width=\textwidth]{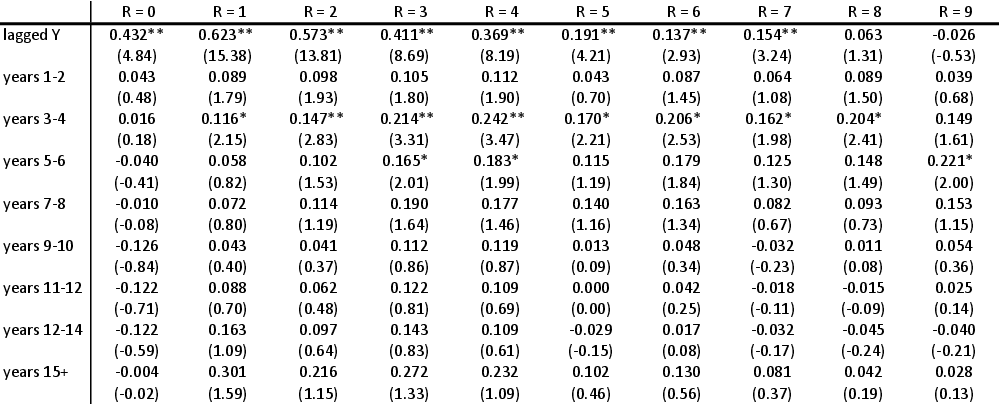}
\caption{\label{tab:EstimateBetaLag}\footnotesize
Dynamic effects of divorce law reform. We report
bias corrected LS-estimates for the regression coefficients in model~\eqref{m:empirical.ex}. Each column corresponds to
a different number of factors $R \in \{0,1,\ldots,9\}$ used in the estimation. t-values are reported in parenthesis.}
\end{table}

\begin{table}[htb!]
\centering
\includegraphics[width=\textwidth]{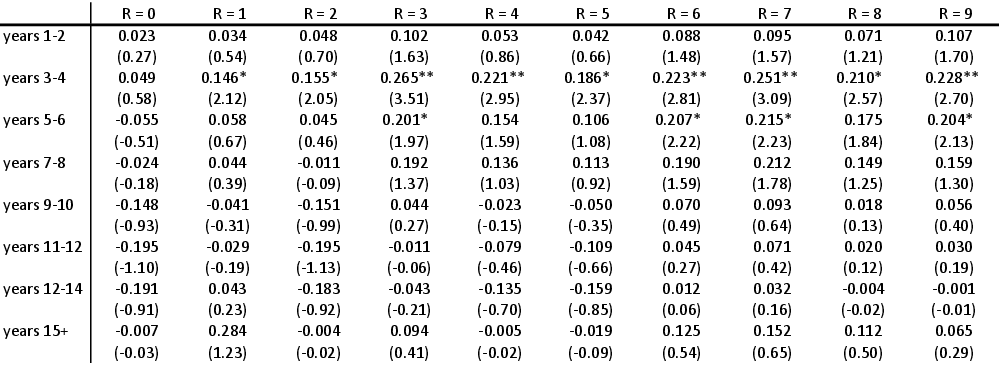}
\caption{\label{tab:EstimateBetaNoLag}\footnotesize
Same as Table~\ref{tab:EstimateBetaLag}, but without including the lagged dependent variable into the model.}
\end{table}

On the other hand, our asymptotic theory suggests, that the exact choice of $R$ in the estimation of
$\widehat \beta_R$ should not matter too much, as long as $R$ is chosen large enough to cover all relevant factors.
Table~\ref{tab:EstimateBetaLag} contains the estimation results for the bias corrected $\widehat \beta_R$ for $R \in \{0,1,\ldots,9\}$.
Table~\ref{tab:EstimateBetaNoLag} contains estimates if the lagged dependent variable is not included into the model.\footnote{
The result for $R=7$  in Table~\ref{tab:EstimateBetaNoLag} should be equal to column (6) in Table III
of Kim and Oka~\cite*{KimOka2014}. The discrepancy is explained
by a coding error in their bias computation.
Note also that the result for $R=0$ in Table~\ref{tab:EstimateBetaNoLag} does not match the one in Wolfers~\cite*{Wolfers2006},
because he uses WLS with state population weights, while we use OLS for simplicity.
Kim and Oka~\cite*{KimOka2014} estimate both WLS and OLS and find that the difference between the resulting
estimates becomes insignificant, once a sufficient number of interactive fixed effects is controlled for.}
For all reported estimates we perform bias correction and standard error estimation as described in
Bai~\cite*{Bai2009} and Moon and Weidner~\cite*{MoonWeidner2013}.\footnote{We correct for the biases
due to heterscedasticity in both panel dimensions worked out in Bai~\cite*{Bai2009},
 as well as for the dynamic bias worked out in Moon and Weidner~\cite*{MoonWeidner2013}.
 For the latter we use the formula for $ \widehat B_{R,k}$ above, with bandwidth $M=2$.
For the standard error estimation we allow for heterscedasticity in both panel dimensions, also
following Bai~\cite*{Bai2009} and Moon and Weidner~\cite*{MoonWeidner2013}.
The bias and standard error formulas in those paper assume $R=R^0$ known, but we strongly expect that those formulas
are robust towards $R>R^0$, as partly justified by Theorem~\ref{th:Estimators} above.
For the model without lagged dependent variable we also allow for serial correlation in $e_{it}$ when estimating the
bias and standard deviation of $\widehat \beta_R$.
}

When ignoring the lagged dependent variable coefficient,
one finds that in both Table~\ref{tab:EstimateBetaLag} and Table~\ref{tab:EstimateBetaNoLag}
the estimation results for $\widehat \beta_R$ and the corresponding t-values
are quite sensitive to changes in $R$ for very small values of $R$, but become
much more stable as $R$ increases, and actually do not change too much anymore from roughly $R=2$ onwards.
These findings are very well in line with our asymptotic theory, and the dynamic effect of divorce law reform
that we find are also similar to the findings in Wolfers~\cite*{Wolfers2006}
and Kim and Oka~\cite*{KimOka2014}. The effect of the law reform on the divorce rates initially increases over time,
is certainly significant in year 3-4 after the reform, and declines and becomes insignificant afterwards.\footnote{
The magnitude of the estimates is smaller than those in Wolfers~\cite*{Wolfers2006}, i.e. controlling
for unobserved factors reduced the effect size, as already pointed out by Kim and Oka~\cite*{KimOka2014}.}

In contrast,
the estimated coefficient on the lagged dependent variable in Table~\ref{tab:EstimateBetaLag} is quite large and
highly significant for small values of $R$, but decreases steadily with $R$, until it gets close to zero and
insignificant for $R \geq 8$. A plausible interpretation of this finding is that the model that includes the lagged dependent
variable is misspecified, and that the estimated value of $\beta_0$ for small values of $R$ does not correspond
to a true state dependence of $Y_{it}$, but simply reflects the time-serial correlation of the error process
being picked up by the autoregressive model. 
According to this interpretation,
once we include more
and more factors into the model we control for more and more serial dependence of the unobserved error term,
thus uncovering the true insignificance of $\beta_0$ in the estimates for $R \geq 8$.

This empirical example shows that instead of relying on a single estimate $\widehat R$ for the number of factors
and reporting the corresponding $\widehat \beta_{\widehat R}$ it can be very informative to calculate
$\widehat \beta_{R}$ for multiple values of $R$. Whether the estimated coefficients become stable for sufficiently
large $R$ values, as our asymptotic theory suggests, is a useful robustness check for the model. When reporting the final results, then, it is better,
within a reasonable range, to choose
an $R$ that is too large than one that is too small.

We also perform a Monte Carlo
simulation that is tailored towards the empirical application. 
For this we use the static model without lagged dependent variable.
To generate $Y_{it}$ in equation \eqref{m:empirical.ex} with $\beta_0=0$
we use the observed regressors $X_{k,it}$, as described above, and 
as true parameters we use the
$\beta$ (bias corrected), $\alpha_i$, $\gamma_i$, $\delta_i$, $\mu_t$, $\lambda_i$ and $f_t$ obtained from the estimation with $R=4$ 
(i.e. $\beta_k$ as reported in the $R=4$ column of Table~\ref{tab:EstimateBetaNoLag}). We generate $e_{it}$ 
from an MA(1) model with $t(5)$ distributed
innovations. Note that this error distribution violates the 
assumption~\ref{ass:LL}$(ii)$.

In this ``empirical Monte Carlo'' we have 
$N=48$, $T=33$ and true number of factors $R^0=4$. 
We find that the
bias corrected estimates for $\beta_k$, $k=1,\ldots,8$, 
are essentially unbiased when $R \geq R^0$ factors are used in the estimation, but for $R<R^0$ the coefficient estimates are often biased.
For $\beta_k$, $k \geq 3$, there are only small changes in the
 standard deviation of the estimator between $R=4$ and $R=9$,
  but for $\beta_k$, $k=1,2$, we observe standard deviation inflation
of up to $25 \%$ between $R=4$ and $R=9$.
Given the relatively small sample size 
the difference between $R=9$ and $R^0=4$ is relatively large, and some
finite sample inefficiency is not too surprising. 
The detailed results are available in the supplementary appendix.

\section{Monte Carlo Simulations}
\label{sec:MC}

In addition to the ``empirical Monte Carlo'' discussed above
we now investigate the finite sample properties of $\widehat \beta_{R}$ 
and $\widehat \beta_{R}^{\rm BC}$  further. In the simulations in this section we use 
a generated regressor $X_{it}$ that is correlated with the interactive fixed effects.
The serial correlation of the error term $e_{it}$ together with the 
data generating process (DGP) for $X_{it}$, $\lambda_i$ and $f_t$ are such that
the naive LS estimator has an asymptotic bias. This allows to verify
whether the bias is essentially unchanged for $R>R^0$ and whether
bias correction works well for $R>R^0$ in finite sample. We also study
various combinations of $N$ and $T$.

The model is a static panel model with one regressor ($K=1$), two factors ($R^0=2$),
and the following DGP:
\begin{align}
   Y_{it} &= \beta^0 X_{it} + \sum_{r=1}^2 \lambda_{ir} f_{tr} + e_{it} ,
   \nonumber \\
   X_{it} &= 1 + \widetilde X_{it}
               + \sum_{r=1}^2 (\lambda_{ir}+\chi_{ir}) (f_{tr} + f_{t-1,r} ) ,
   \nonumber \\
   e_{it} &= \frac 1 {\sqrt{2}} (v_{it} +v_{i,t-1} )  .
   \label{DGP-Static}
\end{align}
The random variables
$\widetilde X_{it}$, $\lambda_{ir}$, $f_{tr}$, $\chi_{ir}$ and $v_{it}$
are mutually independent; with $\widetilde X_{it}$ and $f_{tr} \, \sim \, iid  \, {\cal N}(0,1)$;
$\lambda_{ir}$ and $\chi_{ir} \, \sim iid \, {\cal N}(1,1)$; and
$v_{it} \, \sim \, iid \, t(5)$, i.e.~$v_{it}$ has a Student's t-distribution  with 5 degrees of freedom.

Note that this model satisfies Assumptions~\ref{ass:SF},~\ref{ass:NC}, and ~\ref{ass:LL}(i), but not ~\ref{ass:LL}(ii).
The error term $e_{it}$ is \emph{not} distributed as $iid$ normal. The time series of $e_{it}$ follows an MA(1)
process with innovations distributed as $t(5)$.

We choose $\beta^0=1$, and
use $10,000$ repetitions in our simulation.
The true number of factors is chosen to be $R^0=2$.
For each draw of $Y$ and $X$
we compute the LS estimator $\widehat \beta_R$ according to equation \eqref{estimator}
for different values of $R$, namely $R \in \{0,1,2,3,4,5\}$.

\begin{table}[tb!]
\centering
\includegraphics[width=14cm]{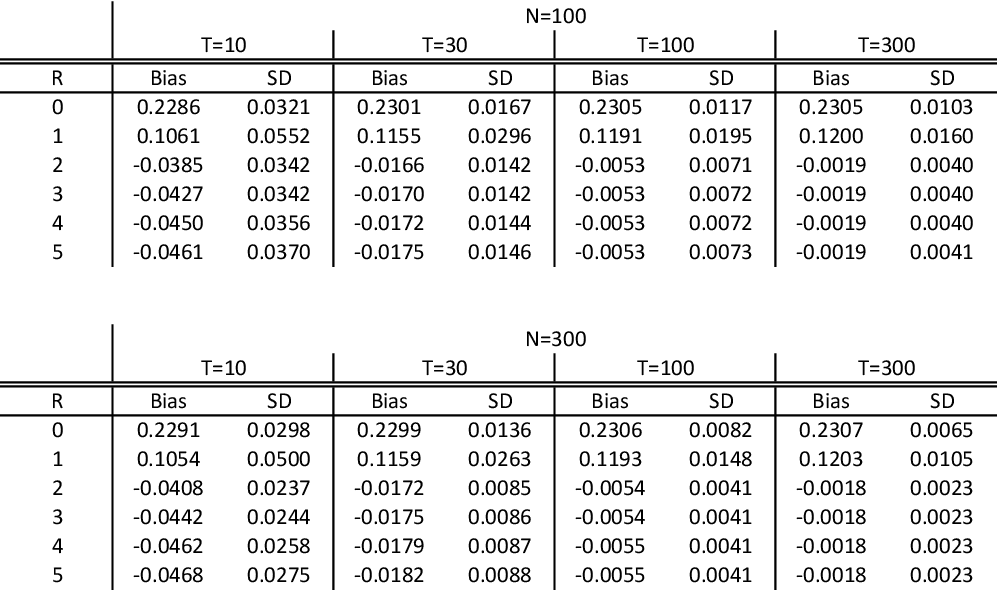}
\caption{\label{tab:MC-Static-1}\footnotesize
For different combinations of sample sizes $N$ and $T$ we report
the bias and standard deviation of the estimator $\widehat \beta_R$, for $R=0,1,\ldots,5$,  based on
simulations with $10,000$ repetition of design \eqref{DGP-Static}, where the true number of
factors is $R^0=2$.}
\end{table}

Table~\ref{tab:MC-Static-1} reports bias and standard deviation of the estimator $\widehat \beta_R$
for different combinations of $R$, $N$ and $T$. For $R<R^0=2$ the model is misspecified and $\widehat \beta_R$
turns out to be severely biased.
There is also bias in $\widehat \beta_R$ for $R \geq R^0$, due to time-serial correlation of $e_{it}$.
This bias was worked out in Bai~\cite*{Bai2009}, and bias correction is also discussed there.

\begin{table}[tb!]
\centering
\includegraphics[width=14cm]{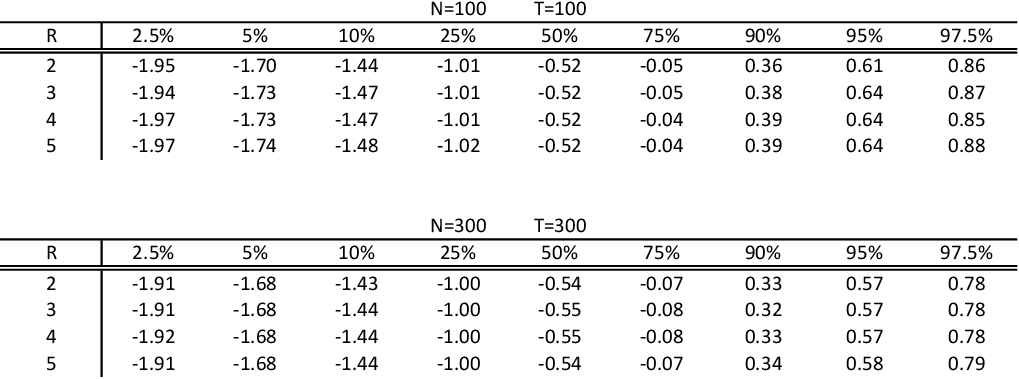}
\caption{\label{tab:MC-Static-2}\footnotesize
Quantiles of the distribution of $\sqrt{NT}( \widehat \beta_R - \beta^0 )$
are reported for
$N=T=100$ and $N=T=300$, with $R=2,3,4,5$,
 based on
simulations with $10,000$ repetition of design \eqref{DGP-Static}, where the true number of
factors is $R^0=2$.
 }
\end{table}

Table~\ref{tab:MC-Static-2} reports various quantiles of
the distribution of $\sqrt{NT}( \widehat \beta_R - \beta^0 )$
for $N=T=100$ and $N=T=300$, and different values of $R \geq R^0$.
From these tables, we see that as $N,T$ increases the distribution of $\widehat \beta_R$ gets closer to that of $\widehat\beta_{R^0}$.

\begin{table}[tb!]
\centering
\includegraphics[width=14cm]{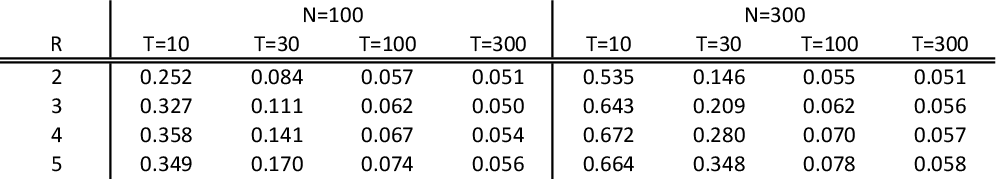}
\caption{\label{tab:MC-Static-3}\footnotesize
The empirical size of a t-test with $5 \%$ nominal size is reported for different combinations of $N$, $T$ and $R$,
based on $10,000$ repetition of design \eqref{DGP-Static}. A bias corrected estimator $\widehat \beta^{\rm BC}_R$ is used to calculate the
test statistics, and we allow for heteroscedasticity and time-serial correlation when estimating bias and standard
deviation. Results for $R=0,1$ are not reported since those have
size=1 due to misspecification.  }
\end{table}

Table~\ref{tab:MC-Static-3} reports the size of a t-test with nominal size equal to $5 \%$ for $R \geq R^0$.
We use the results in Bai~\cite*{Bai2009}
to correct for the leading $1/N$ (not actually present in our DGP)
and $1/T$ (present in our DGP) biases in $\widehat \beta_R$ before calculating the t-test statistics, allowing for heteroscedsticity in both panel dimensions and for time-serial correlation when estimating the bias and standard
deviation of $\widehat \beta_R$. The finite sample size distortions are mostly due to residual bias after bias correction,
but also partly due to some finite sample downward bias in the standard error estimates. The size distortions increase
with $R$, but for all values of $R \geq R^0$ in Table~\ref{tab:MC-Static-3} the size distortions decrease rapidly as $T$ increases.

Monte Carlo Simulation results for an AR(1) model with factors can be found in Section~\ref{app:MC}
of the supplementary material. Those additional simulations show that the finite sample
properties (e.g. for $T=30$) of $\widehat \beta_{R^0}$ and $\widehat \beta_{R}$, $R>R^0$,
can be quite different, but those differences vanish as $T$ becomes large, as predicted
by our asymptotic theory.
In general, we always expect some finite sample inefficiency from overestimating the number
of factors.

\section{Conclusions}
\label{sec:conclusions}

We show that under certain assumptions
the limiting distribution of the LS estimator of a linear panel regression with interactive fixed effects
does not change when we include redundant factors in the estimation. The implication of this is
that one can use an upper bound of the number of factors $R$ in the estimation without
asymptotic efficiency loss. However, some finite sample efficiency loss from overestimating $R$
is likely, so that $R$ should not be chosen too large in actual applications.
We  impose $iid$ normality of the regression errors to derive the asymptotic result, because we require certain
results on the eigenvalues and eigenvectors of random covariance matrices that are only known in that case.
We expect that progress in the literature on large dimensional
random covariance matrices will allow verification of our high-level assumptions
under more general error distributions,
and our Monte Carlo simulations suggest that the result also holds for non-normal
and correlated errors.
We also provide multiple intermediate asymptotic results under more general conditions.

\begin{appendix}
\section{Appendix}

\footnotesize

\subsection{Spectral Norm of Random Matrices}
\label{app:SpectralNorm}

Consider an $N \times T$ matrix $u$ whose entries $u_{it}$ have uniformly bounded second moments.
Then we have $\| u \| \leq \| u \|_{HS} = \sqrt{\sum_{i,t} u_{it}^2} = {\cal O}_P(\sqrt{NT})$.
However, in Assumption~\ref{ass:LL}$(i.b)$ and Assumption~\ref{ass:DX-1}$(i)$
and  Assumption~\ref{ass:DX-2}$(i)$
we impose $\| \widetilde X_k^{\rm str} \| = {\cal O}_P(N^{3/4})$
and $\| \widetilde X_k \| = {\cal O}_P(N^{3/4})$, respectively,
as $N$ and $T$ grow at the same rate,
and in Assumption~\ref{ass:SN}$(ii)$ we impose
$\| e \| = {\cal O}_P(\sqrt{\max(N,T)})$ under an arbitrary asymptotic $N,T \rightarrow \infty$.
Those smaller asymptotic rates for the spectral norms of $\widetilde X_k^{\rm str}$, $\widetilde X_k$
and $e$ can be justified by firstly assuming that the entries of these matrices are mean zero and
have certain bounded moments,
and secondly imposing weak cross-sectional and time-serial correlation.
The purpose of this appendix section is to provide some examples of matrix distributions that
make the last statement more precise. We consider the
$N \times T$ matrix $u$, which can represent either $e$,  $\widetilde X_k^{\rm str}$ or $\widetilde X_k$.

{\bf Example 1:} If we assume that $\mathbbm{E} u_{it} =0$, that $\mathbbm{E} u^4_{it} $ is uniformly bounded,
and that the $u_{it}$ are independently distributed across $i$ and over $t$,
then the results in Latala \cite*{Latala2005} show that  $\| u \| = {\cal O}_P(\sqrt{\max(N,T)})$.

{\bf Example 2:} Onatski~\cite*{Onatski2013} provides the following example, which allows for both cross-sectional
and time-serial dependence: Let $\epsilon$ be an $N \times T$ matrix with
mean zero, independent entries that have uniformly bounded fourth moment,
 let $\epsilon_t$ denote the columns of $\epsilon$, and also define past $\epsilon_t$, $t \leq 0$,
satisfying the same distributional assumptions.
Let $u_{t} = \sum_{j=0}^m \Psi_{N,j} \epsilon_{t-j}$, where $m$ is a fixed integer, and
$\Psi_{N,j}$ are $N \times N$ matrices such that $\max_j \| \Psi_{N,j} \|$ is uniformly bounded.
Then, the $N \times T$ matrix $u$ with columns $u_t$ satisfies $\| u \| = {\cal O}_P(\sqrt{\max(N,T)})$.

More examples of matrix distributions that satisfy $\| u \| = {\cal O}_P(\sqrt{\max(N,T)})$
are discussed in Onatski~\cite*{Onatski2013} and Moon and Weidner~\cite*{MoonWeidner2013}.
Theorem 5.48 and Remark 5.49 in Vershynin~\cite*{Vershynin2010} can also be used to
obtain a slightly weaker bound on $\| u \|$ under very general correlation of $u$ in one of its dimensions.

Note that the random matrix theory literature often only discusses limits where $N$ and $T$
grow at the same rate and shows $\| u \| = {\cal O}_P(\sqrt{N})$ under that asymptotic. Those results
can easily be extended to more general asymptotics with $N,T \rightarrow \infty$ by considering
$u$ as a submatrix of a $\max(N,T) \times \max(N,T)$ matrix $u^{\rm big}$, and
using that $\|u\| \leq \|u^{\rm big}\|$.

{\bf Example 3:} The following Lemma provides a justification for the bounds on $\| \widetilde X_k^{\rm str} \|$
and $\| \widetilde X_k \|$, allowing for a quite general type of correlation in both panel dimensions.

\begin{lemma}
    \label{lemma:SpectralNormBound}
    Let $u$ be an $N \times T$ matrix with entries $u_{it}$.
    Let
    $\Sigma_{ij} = \frac 1 T \sum_{t=1}^T \mathbbm{E}( u_{it} u_{jt} )$,
    and let $\Sigma$ be the $N \times N$ matrix with entries $\Sigma_{ij}$.
      Let
     $\eta_{ij}
        = \frac 1 {\sqrt{T}} \sum_{t=1}^T\left[ u_{it} u_{jt} - \mathbbm{E}(u_{it} u_{jt}) \right]$,
     $\Psi_{ij} = \frac 1 N \sum_{k=1}^N \mathbbm{E}( \eta_{ik} \eta_{jk} )$,
     and
     $\chi_{ij} = \frac 1 {\sqrt{N}} \sum_{k=1}^N
            \left[ \eta_{ik} \eta_{jk} - \mathbbm{E}( \eta_{ik} \eta_{jk} ) \right]$.
    Consider $N,T \rightarrow \infty$
    such that $N/T$ converges to a finite positive constant, and
     assume that
    \begin{itemize}
         \item[(i)]  $\| \Sigma  \| = {\cal O}(1)$.

          \item[(ii)]
             $\frac 1 {N^2} \sum_{i,j=1}^N \mathbbm{E}( \eta^2_{ij} ) = {\cal O}(1)$.

         \item[(iii)]
            $\frac 1 N \sum_{i,j=1}^N \Psi^2_{ij} = {\cal O}(1)$.

         \item[(iv)]
            $\frac 1 {N^2} \sum_{i,j=1}^N  \mathbbm{E}( \chi^2_{ij} ) = {\cal O}(1)$.
    \end{itemize}
    Then we have $\|u\| = {\cal O}_P(N^{5/8}) $.
\end{lemma}

The Lemma does not impose $\mathbbm{E} u_{it} = 0$ explicitly, but justification of assumption $(i)$ in the lemma
usually requires $\mathbbm{E} u_{it} = 0$.
The assumptions $(ii)$, $(iii)$ and $(iv)$ in the lemma can e.g. be justified by assuming
appropriate mixing conditions in both panel dimensions, see e.g. Cox and Kim~\cite*{CoxKim1995}
for the time-series case.

As pointed out above, our results in Section~\ref{sec:expansion} can be obtained
under the weaker condition $\|e\|=o_P(N^{2/3})$, and Lemma~\ref{lemma:SpectralNormBound}
can also be applied with $u=e$ then. In that case, the assumptions in Lemma~\ref{lemma:SpectralNormBound}
are not the same, but are similar to those imposed  in Bai~\cite*{Bai2009}.

\subsection{Expansion of Objective Function when $R=R^0$}
\label{app:ExpansionDiscussion}

Here we provide a heuristic derivation of the expansion of ${\cal L}^0_{NT}(\beta)$
in Theorem~\ref{th:expansion}.
We expand the profile objective function ${\cal L}^0_{NT}(\beta)$ simultaneously
in $\beta$ and in the spectral norm of $e$.
Let the $K+1$ expansion parameters be defined by $\epsilon_0 = \|e\|/\sqrt{NT}$ and
$\epsilon_k = \beta^0_{k} - \beta_{k}$, $k=1,\ldots,K$,
and define the $N\times T$ matrix $X_0=(\sqrt{NT}/\|e\|) e$. With these definitions we obtain
\begin{align}
  \frac 1 {\sqrt{NT}} \left( Y - \beta \cdot X \right)
    &\; =\;  \frac 1 {\sqrt{NT}}
        \left[ \lambda^0 f^{0\prime}
      + (\beta^0-\beta) \cdot X + e \right]
\; = \;  \frac{\lambda^0 f^{0\prime}} {\sqrt{NT}} + \sum_{k=0}^K \,
    \epsilon_k \, \frac{X_k}{\sqrt{NT}} \; .
  \label{ExpIdea}
\end{align}
According to equation \eqref{LSobjective} the profile objective function
${\cal L}^0_{NT}(\beta)$ can be written as the sum over the $T-R^0$ smallest eigenvalues of the matrix in \eqref{ExpIdea} multiplied by its transposed.
We consider
$\sum_{k=0}^K \, \epsilon_k \, X_k / \sqrt{NT}$ as a small perturbation
of the unperturbed matrix $\lambda^0 f^{0\prime}/ {\sqrt{NT}}$,
and thus expand ${\cal L}^0_{NT}(\beta)$ in the perturbation parameters
$\epsilon=(\epsilon_0,\ldots,\epsilon_K)$ around $\epsilon=0$, namely
\begin{align}
  {\cal L}^0_{NT} \left( \beta \right) &= \frac{1} {NT} \, \sum_{g=0}^\infty \,
  \sum_{k_1,\ldots,k_g=0}^K \,  \epsilon_{k_1} \, \epsilon_{k_2} \, \ldots \, \epsilon_{k_g}
  \, L^{(g)}\left(\lambda^0,\, f^0,\, X_{k_1},\, X_{k_2}, \ldots ,X_{k_g}\right) \; ,
  \label{ExpIdea2}
\end{align}
where $L^{(g)}=L^{(g)}\left(\lambda^0,\, f^0,\, X_{k_1},\, X_{k_2}, \ldots ,X_{k_g}\right)$
are the expansion coefficients.

The unperturbed matrix $\lambda^0 f^{0\prime}/ {\sqrt{NT}}$ has rank $R^0$,
so that the $T-R^0$ smallest
eigenvalues of the unperturbed $T\times T$ matrix $f^0 \lambda^{0\prime} \lambda^0 f^{0\prime}/NT$ are all zero, i.e. ${\cal L}^0_{NT} \left( \beta \right)=0$
for $\epsilon=0$ and thus $L^{(0)}\left(\lambda^0,\, f^0 \right) = 0$.
Due to Assumption~\ref{ass:SF} the $R^0$ non-zero eigenvalues of the
unperturbed
$T\times T$ matrix $f^0 \lambda^{0\prime} \lambda^0 f^{0\prime}/NT$ converge to positive constants as $N,T \rightarrow \infty$.
This means that the ``separating distance'' of the $T-R^0$ zero-eigenvalues of the unperturbed $T\times T$ matrix $f^0 \lambda^{0\prime} \lambda^0 f^{0\prime}/NT$ converges to
a positive constant, i.e. the next largest eigenvalue is well separated.
This is exactly the technical condition under which
the perturbation theory of linear operators guarantees that the above
expansion of ${\cal L}^0_{NT}$ in $\epsilon$ exists and is convergent as long as the spectral norm of the perturbation
$\sum_{k=0}^K \, \epsilon_k \, X_k / \sqrt{NT}$
is smaller than a particular convergence radius $r_0(\lambda^0,f^0)$,
which is closely related to the separating distance of the zero-eigenvalues.
For details on that see Kato~\cite*{Kato}
and Section~\ref{app:expansion1} of the supplementary appendix, where we define $r_0(\lambda^0,f^0)$ and show that
it converges to a positive constant as $N,T \rightarrow \infty$.
Note that for the expansion \eqref{ExpIdea2} it is crucial that we have
$R=R^0$, since the perturbation theory of linear operators
describes the perturbation of the sum of {\it all} zero-eigenvalues of the unperturbed
matrix $f^0 \lambda^{0\prime} \lambda^0 f^{0\prime}/NT$. For $R>R^0$ the sum in
${\cal L}^R_{NT} \left( \beta \right)$ leaves out the $R-R^0$ largest of
these perturbed zero-eigenvalues, which results in a much more complicated
mathematical problem, since the structure and ranking among these perturbed zero-eigenvalues needs to be discussed.

The above expansion of ${\cal L}^0_{NT} \left( \beta \right)$ is applicable whenever the operator norm
of the perturbation matrix $\sum_{k=0}^K \, \epsilon_k \, X_k / \sqrt{NT}$ is smaller than
$r_0(\lambda^0,f^0)$. Since our assumptions guarantee that
$\| X_k  / \sqrt{NT} \| = {\cal O}_P(1)$, for $k=0,\ldots,K$, and
$\epsilon_0 = {\cal O}_P(\min(N,T)^{-1/2}) = o_P(1)$, we have
$\left\| \sum_{k=0}^K \, \epsilon_k \, X_k / \sqrt{NT} \right\|=
{\cal O}_P(\|\beta-\beta^0\|) + o_P(1)$, i.e.
the above expansion is always applicable asymptotically within a shrinking
neighborhood of $\beta^0$ --- which is sufficient since we already know that
$\widehat \beta_R$ is consistent for $R \geq R^0$.

In addition, to guaranteeing converge of the series expansion, the
perturbation theory of linear operators
also provides explicit formulas for the expansion coefficients $L^{(g)}$,
namely for $g=1,2,3$ we have
$L^{(1)}\left(\lambda^0, f^0, X_{k}\right)=0$,
$L^{(2)}\left(\lambda^0, f^0, X_{k_1}, X_{k_2}\right)
 ={\rm Tr}(M_{\lambda^0} X_{k_1} M_{f^0} X_{k_2}')$,
$L^{(3)}\left(\lambda^0, f^0, X_{k_1},X_{k_2},
X_{k_3} \right) = -  \frac{1}{3} [{\rm Tr}\left(
M_{\lambda^0} X_{k_1} M_f X^{\prime}_{k_2} \lambda^0
(\lambda^{0\prime}\lambda^0)^{-1} (f^{0\prime}f^0)^{-1}  f^{0\prime}
X^{\prime}_{k_3}\right) + \ldots]$, where the
dots refer to 5 additional terms obtained from the first one
by permutation of $k_1$, $k_2$ and $k_3$, so that the expression
becomes totally symmetric in these indices.
A general expression for the coefficients for all orders in $g$ is given in Lemma~\ref{lemma:expansion} in the appendix.
One can show that for $g \geq 3$ the coefficients $L^{(g)}$ are bounded as follows
\begin{align}
  \frac {1} {NT} \left| L^{(g)}\left(\lambda^0,\, f^0,\, X_{k_1},\,
                   X_{k_2}, \ldots ,X_{k_g}\right) \right|
    \; \leq \; a_{NT} \; (b_{NT})^{g} \;
    \frac{\|X_{k_1}\|} {\sqrt{NT}} \, \frac{\|X_{k_2}\|} {%
\sqrt{NT}} \, \ldots \, \frac{\|X_{k_g}\|} {\sqrt{NT}}  \; ,
   \label{ExpBound}
\end{align}
where $a_{NT}$ and $b_{NT}$ are functions of $\lambda^0$ and $f^0$ that converge to finite positive constants
in probability.
This bound on the coefficients $L^{(g)}$ allows us to derive a bound on the remainder term, when the profile
objective expansion is truncated at a particular order.
The expansion can be applied under more general asymptotics, but here
we only consider
the limit $N,T \rightarrow \infty$ with $N/T \rightarrow \kappa^2$, $0<\kappa<\infty$,
i.e. $N$ and $T$ grow at the same rate.
Then, apart from the constant ${\cal L}_{NT}^{0}(\beta^0)$, the relevant
coefficients of the expansion, which are not treated as part of the remainder term
turn out to be
$W_{k_1 k_2} = \frac 1 {NT}
             L^{(2)}\left(\lambda^0,\, f^0,\, X_{k_1},\, X_{k_2}\right)$,
$C^{(1)}_k = \frac 1 {\sqrt{NT}}
          L^{(2)}\left(\lambda^0,\, f^0,\, X_{k},\, e\right) =
         \frac 1 {\sqrt{NT}} \, {\rm Tr}( M_{\lambda^0} \, X_k \,
                  M_{f^0} \, e^{\prime} )$,
and
$C^{(2)}_k = \frac 3 {2 \, \sqrt{NT}}
          L^{(3)}\left(\lambda^0,\, f^0,\, X_{k},\,  e, \, e\right)$,
which corresponds exactly to the definitions  in the main text.
From the expansion \eqref{ExpIdea2} and the bound \eqref{ExpBound} we obtain
Theorem~\ref{th:expansion}. For a more rigorous derivation we refer to
Section~\ref{app:expansion1} in the supplementary appendix.

\subsection{$N^{3/4}$-Convergence Rate of $\widehat \beta_R$ for $R>R^0$}
\label{sec:ConvergenceRate}

The discussion at the end of Section~\ref{sec:main}
reveals that showing faster than $\sqrt{N}$ convergence of $\widehat \beta_R$ is
a very important step on the way to the main result. For purely technical reasons
we show $N^{3/4}$-convergence first, but it will often be the case that if
$\widehat \beta_R$ is $N^{3/4}$-consistent, then it is also $\sqrt{NT}$-consistent
as $N$ and $T$ grow at the same rate. We require one of the following two alternative assumptions.

\begin{DXassumption}[\bf Decomposition of $X_k$ and Distribution of $e$, Version 1]
\label{ass:DX-1} $\phantom{a}$
\begin{itemize}
    \item[(i)]  For $k=1,\ldots,K$ we have
      $X_k = \overline X_k + \widetilde X_k$, where
       ${\rm rank}(\overline X_k)$ is bounded as $N,T \rightarrow \infty$,
       and  $\| \overline X_k\| = {\cal O}_P(\sqrt{NT})$, and $\| \widetilde X_k \| = {\cal O}_P(N^{3/4})$.

    \item[(ii)]
    Let  $u$ be an $N \times T$ matrix whose
   elements are distributed as i.i.d. ${\cal N}(0,1)$,
   independent of $\lambda^0$, $f^0$ and $\overline X_k$, $k=1,\ldots,K$,
   and let one of the following hold
   \begin{itemize}
       \item[(a)] either: $e= \Sigma^{1/2} \, u$, where $\Sigma$ is an $N \times N$ covariance matrix, independent of $u$,
       which satisfies $\| \Sigma \| = {\cal O}_P(1)$. In that case, define
        $g$ to be an $N \times Q$ matrix, independent of $u$,
        for some $Q \leq \sum_{k=1}^K {\rm rank}(\overline X_k)$,
        such that $g' g = \mathbbm{1}_Q$
        and ${\rm span}(M_{\lambda^0} \overline X_k )  \subset {\rm span}(g)$ for all $k=1,\ldots,K$.\footnote{
       The column space of $g$ thus contains the column space of all $M_{\lambda^0} \overline X_k$.
       $g' g = \mathbbm{1}_Q$ is just a normalization.}

       \item[(b)] or: $e = u \, \Sigma^{1/2}$, where $\Sigma$ is a $T \times T$ covariance  matrix, independent of $u$,
       which satisfies $\| \Sigma \| = {\cal O}_P(1)$. In that case, define
        $g$ to be a $T \times Q$ matrix, independent of $u$,
        for some $Q \leq \sum_{k=1}^K {\rm rank}(\overline X_k)$,
        such that $g' g = \mathbbm{1}_Q$
        and ${\rm span}(M_{f^0} \overline X'_k )  \subset {\rm span}(g)$ for all $k=1,\ldots,K$.
   \end{itemize}
   In addition, we
   assume that there exists a (potentially random) integer sequence $n=n_{NT}>0$ with
        $1/n = {\cal O}_P(1/N)$ such that
        $\mu_n(\Sigma) \geq \| g' \Sigma g \|$.
    Finally, assume that either $R \geq Q$ or that
        $g' \Sigma g = \|g' \Sigma g \| \mathbbm{1}_Q + {\cal O}_P(N^{-1/2})$.
\end{itemize}
\end{DXassumption}

\begin{DXassumptionTwo}[\bf Decomposition of $X_k$ and Distribution of $e$, Version 2]
\label{ass:DX-2} $\phantom{a}$
\begin{itemize}
    \item[(i)]  For $k=1,\ldots,K$ we have
      $X_k = \overline X_k + \widetilde X_k$, such that $M_{\lambda^0} \overline X_k  M_{f^0} = 0$,
       and  $\| \overline X_k\| = {\cal O}_P(\sqrt{NT})$, and $\| \widetilde X_k \| = {\cal O}_P(N^{3/4})$.

    \item[(ii)] $\| e \| = {\cal O}_P(\sqrt{\max(N,T)})$. (same as Assumption~\ref{ass:SN}(ii))
\end{itemize}

\end{DXassumptionTwo}

\begin{theorem}
    \label{th:ConvergenceRate2}
    Let $R>R^0$. Let Assumptions~\ref{ass:SF}, \ref{ass:NC} and
  \ref{ass:EX} hold, and let either Assumption~\ref{ass:DX-1} or~\ref{ass:DX-2}
  be satisfied. Consider
     $N,T \rightarrow \infty$ with
   $N/T \rightarrow \kappa^2$, $0<\kappa<\infty$.
   Then we have  $N^{3/4} \left(\widehat \beta_{R} - \beta^0\right)
                   = {\cal O}_P(1)$.
\end{theorem}

\paragraph{Remarks}

\begin{itemize}
\item[(i)] Assumption~\ref{ass:SN} is not explicitly imposed in
Theorem~\ref{th:ConvergenceRate2}, because it is already implied by both
 Assumption~\ref{ass:DX-1} and~\ref{ass:DX-2}, see also Lemma~\ref{lemma:JustifyEV} below.

\item[(ii)] The restrictions that Assumption~\ref{ass:DX-1} imposes on $X_k$ are
   weaker than those imposed in Assumption~\ref{ass:LL} above. The regressors
   are decomposed into a low-rank strictly exogenous part $\overline X_k$
   and a term $\widetilde X_k$, which can be both strictly or weakly exogenous.
     The spectral norm bound $\| \widetilde X_k \| = {\cal O}_P(N^{3/4})$ is satisfied
     as long as $\widetilde X_{k,it}$ is mean zero and weakly correlated across $i$ and over $t$,
     see Appendix~\ref{app:SpectralNorm}.
     We can always write $\overline X_k = \ell h'$ for some appropriate
     $\ell \in \mathbbm{R}^{N \times {\rm rank}(\overline X_k) }$
  and $h \in \mathbbm{R}^{T \times {\rm rank}(\overline X_k)}$.
     Thus, the decomposition $X_k = \overline X_k + \widetilde X_k = \ell h' + \widetilde X_k $
     essentially imposes an approximate factor structure on $X_k$, with factor part  $\overline X_k$
     and idiosyncratic part $\widetilde X_k$.
     In addition to those conditions we need sufficient variation in $X_k$,
     as formalized by the non-collinarity Assumption~\ref{ass:NC}.

\item[(iii)] The restrictions that Assumption~\ref{ass:DX-1} imposes on $e$ are also
   weaker than those imposed in Assumption~\ref{ass:LL} above.  Normality is imposed, but either
   cross-sectional correlation and heteroscedasticity (case (a)) or time-serial correlation
   and heteroscedasticity (case (b)), described by $\Sigma$, are still allowed.
   The condition $\| \Sigma \| = {\cal O}_P(1)$ requires the correlation of $e_{it}$ to be weak.\footnote{
   A sufficient condition for $\| \Sigma \| = {\cal O}_P(1)$ is, for example, $\max_i \sum_j |\Sigma_{ij}| = {\cal O}_P(1)$,
    formulated here for case (a). Note that $\Sigma$ is symmetric.
}

\item[(iv)]
  The additional restrictions on $\Sigma$ in Assumption~\ref{ass:DX-1}
  rule out the type of correlation of the low-rank regressor
   part $\overline X_k$ with the second moment structure of $e_{it}$ that was
   the key feature of the counter example in Section~\ref{sec:AsymptoticSummary} above.\footnote{
    However, in the
   example in Section~\ref{sec:AsymptoticSummary}
   we have both time-serial and cross-sectional correlation in $e_{it}$, one of which is already ruled out by
    Assumption~\ref{ass:DX-1}.
}
    Firstly,
    the condition $\mu_n(\Sigma) \geq \| g' \Sigma g \|$ guarantees that the eigenvectors corresponding
    to the largest few eigenvectors of $\Sigma$ (the eigenvectors $\nu_r$ of $\Sigma$ when normalized satisfy
    $\mu_r(\Sigma) = \nu_r' \Sigma \nu_r$) are not strongly correlated with $g$ (and thus with $\overline X_k$).
    Secondly, the condition $g' \Sigma g = \|g' \Sigma g \| \mathbbm{1}_Q + {\cal O}_P(N^{-1/2})$
    guarantees  that $\Sigma$ behaves almost as an identity matrix when projected with $g$, thus not possessing
    special structure in the ``direction  of  $\overline X_k$''.
        Both of these assumption are obviously satisfied when
    $\Sigma$ is proportional to the identity matrix.

\item[(v)]
    Instead of Assumption~\ref{ass:DX-1} we can also impose Assumption~\ref{ass:DX-2}
    to obtain $N^{3/4}$-consistency in Theorem~\ref{th:ConvergenceRate2}.
    The Assumption on $e$ imposed in Assumption~\ref{ass:DX-2} is the same as in
    Assumption~\ref{ass:SN}, and as already discussed above, this assumption is quite weak
    (see also Appendix~\ref{app:SpectralNorm}). However,
    Assumption~\ref{ass:DX-2} imposes a much stronger assumption on the regressors
    by requiring that $M_{\lambda^0} \overline X_k M_{f^0} = 0$.
    This condition
  implies that $\overline X_k = \lambda^0 h' + \ell f^{0 \prime}$ for some $\ell \in \mathbbm{R}^{N \times R^0}$
  and $h \in \mathbbm{R}^{T \times R^0}$, i.e. the factor structure of the regressors is severely restricted.
  The AR(1) model discussed in Remark (v) of Section~\ref{sec:main} does satisfy
  $M_{\lambda^0} \overline X_k = 0$,
  and the same is true for a stationary AR(p) model
  without additional regressors, i.e. for such AR(p) models with factors
  we obtain $N^{3/4}$-consistency of $\widehat \beta_{R}$ without imposing
  strong assumptions (like normality) of $e_{it}$.
  Assumption~\ref{ass:DX-2}$(i)$ is furthermore satisfied if $ \overline X_k = 0$, i.e. if the regressors
  $X_k = \widetilde X_k$ satisfy $\| X_k \| = {\cal O}_P(N^{3/4})$, which is true for  zero mean
  weakly correlated  processes (see Appendix~\ref{app:SpectralNorm}).

   \item[(vi)]
   Theorem~\ref{th:ConvergenceRate} in the supplementary material
   provides an alternative $N^{3/4}$-consistency result, in which
   Assumptions~\ref{ass:DX-1} and ~\ref{ass:DX-2} are replaced by a high-level
   condition, which is more general, but not easy
   to verify in terms of low-level assumptions.
\end{itemize}

\subsection{Asymptotic Equivalence of $\widehat \beta_{R^0}$ and $\widehat \beta_R$ for $R>R^0$}
\label{sec:Equivalence}

Here, we provide high level conditions on the singular values and
singular vectors of the error matrix (or equivalently on the eigenvalues
and eigenvectors of the corresponding random covariance matrix).
Under those assumptions we then establish the main result of the paper
that $\widehat \beta_{R^0}$ and $\widehat \beta_R$ with $R>R^0$ are asymptotically
equivalent, that is, $\sqrt{NT}(\widehat \beta_R - \widehat \beta_{R^0}) = o_P(1)$.

\begin{EVassumption}{\bf (Eigenvalues and Eigenvectors of Random Cov.~Matrix)}
   \label{ass:EV}
   Let the singular value decomposition of $M_{\lambda^0} e M_{f^0}$
   be given by
   $M_{\lambda^0} e M_{f^0} = \sum_{r=1}^{Q} \sqrt{\rho_r} \, v_r \, w'_r$,
   where $Q = \min(N,T)-R^0$, and
   $\sqrt{\rho_r}$ are the singular values, and $v_r$ and $w_r$ are
   normalized $N$- and $T$-vectors, respectively.\footnote{%
   Thus, $w_r$ is the normalized eigenvector corresponding to the
    eigenvalue $\rho_r$ of $M_{f^0} e' M_{\lambda^0} e M_{f^0}$, while
   $v_r$ is the normalized eigenvector corresponding to the
    eigenvalue $\rho_r$ of $M_{\lambda^0} e M_{f^0} e' M_{\lambda^0}$.
    We use a convention were
    eigenvalues with non-trivial multiplicity appear multiple times in the list of eigenvalues $\rho_r$,
    but under standard distributional assumptions on $e$ all eigenvalues are simple with probability one anyways.
    }
  Let $\rho_1 \geq \rho_2 \geq \ldots \geq \rho_{Q} \geq 0$. We assume that
  there exists a constant $c>0$ and a series of integers $q_{NT}>R-R^0$
  with $q_{NT} = o(N^{1/4})$ such that as $N,T \rightarrow \infty$ we have
  \begin{itemize}
       \item[(i)] $\displaystyle \frac{ \rho_{R-R^0} } N > c$, wpa1.

       \item[(ii)]   $\displaystyle \frac 1 {q_{NT}} \, \sum_{r=q_{NT}}^{Q} \frac 1 {\rho_{R-R^0} - \rho_r}
              = {\cal O}_P(1)$.

       \item[(iii)] $\displaystyle \max_{r} \| v_r'  e  P_{f^0} \| = o_P\left(  N^{1/4}  \, q_{NT}^{-1} \right)$, \; \;
                      $\displaystyle \max_{r} \| w_r'  e'  P_{\lambda^0} \| = o_P\left(  N^{1/4}  \, q_{NT}^{-1} \right)$, \\[5pt]
                      $\displaystyle   \max_{r} \| v_r'  X_k  P_{f^0} \| = o_P\left(  N  \, q_{NT}^{-1} \right)$, \; \; \;
                      $\displaystyle  \max_{r} \| w_r'  X'_k  P_{\lambda^0} \| = o_P\left(  N  \, q_{NT}^{-1} \right)$, \\[5pt]
                      $\displaystyle \max_{r,s,k} |v_r'  X_k  w_s| = o_P\left(  N^{1/4}  \, q_{NT}^{-1} \right)$,  \; \;
       where $r,s = 1,\ldots, Q$, and $k=1,\ldots,K$.
  \end{itemize}
\end{EVassumption}

\begin{theorem}
     \label{th:LimitingDistribution}
   Let $R>R^0$. Let Assumptions \ref{ass:SF}, \ref{ass:NC}, \ref{ass:EX}, and \ref{ass:EV} hold,
   and let either Assumption~\ref{ass:DX-1} or~\ref{ass:DX-2} hold,
   and assume
   that $C^{(1)} = {\cal O}_P(1)$.
   In the limit $N,T \rightarrow \infty$ with
   $N/T \rightarrow \kappa^2$, $0<\kappa<\infty$, we then have
   \begin{align*}
       \sqrt{NT}\left(\widehat \beta_{R} - \beta^0\right)
                   =  \sqrt{NT}\left(\widehat \beta_{R^0} - \beta^0\right) + o_P(1)
                   = {\cal O}_P(1).
   \end{align*}
\end{theorem}

\paragraph{Remarks}

\begin{itemize}
  \item[(i)]
  Theorem~\ref{th:LimitingDistribution} also holds if we replace
  the Assumptions~\ref{ass:EX}, \ref{ass:DX-1}, \ref{ass:DX-2}
  by any other condition that guarantees that
  Assumption~\ref{ass:SN} holds
  and that  $N^{3/4} \left(\widehat \beta_{R} - \beta^0\right)
                   = {\cal O}_P(1)$.

  \item[(ii)] Consider Assumption~\ref{ass:EV}$(iii)$.
Since $v_r$ and $w_r$ are the normalized singular vectors of
$M_{\lambda^0} e M_{f^0}$ we expect them
to be essentially uncorrelated with $X_k$ and $e  P_{f^0}$,
and therefore we expect $v_r'  X_k  w_s = {\cal O}_P(1)$,
$\| v_r'  e  P_{f^0} \| = {\cal O}_P(1)$,
$\| w_r'  e'  P_{\lambda^0} \| = {\cal O}_P(1)$.
We also expect
$\| v_r'  X_k  P_{f^0} \| = {\cal O}_P(\sqrt{T})$
and $\| w_r'  X'_k  P_{\lambda^0} \|  = {\cal O}_P(\sqrt{N})$,
which is different to the analogous expressions with $e$, since $X_k$
may be strongly correlated with $f^0$ and $\lambda^0$.
The key to making this discussion rigorous
is a good knowledge
of the properties of the eigenvectors $v_r$ and $w_r$.
If the entries $e_{it}$ are $iid$ normal, then the distribution of $v_r$ and $w_r$
can be characterized as follows:
      Let $\widetilde v$ be an $N$-vector with $iid {\cal N}(0,1)$
    entries and let $\widetilde w$ be an $T$-vector with $iid {\cal N}(0,1)$ entries. Then we have
    $v_r  \operatorname*{=}_d \| M_{\lambda^0} \widetilde v \|^{-1} M_{\lambda^0} \widetilde v$
    and $w_r \operatorname*{=}_d \| M_{f^0} \widetilde w \|^{-1}  M_{f^0} \widetilde w$,
    see also Lemma~\ref{lemma:iidEV} in the supplementary material.
    Here,  $\operatorname*{=}_d$ refers to ``equal in distribution''.
    Thus, if $R^0=0$, then $v_r$ and $w_r$ are distributed as $iid {\cal N}(0,1)$ vectors, normalized
    to satisfy $\| v_r \| = \| w_r \| =1$. This follows from the rotational invariance of the distribution of
    $e$ when $e_{it}$ is  $iid$ normally distributed. Using this characterization of $v_r$ and $w_r$
    one can formally show that Assumption~\ref{ass:EV}$(iii)$ holds, see Lemma~\ref{lemma:JustifyEV} below.
   The conjecture in the random matrix theory literature is that
the limiting distribution of the eigenvectors of a random covariance matrix
is ``distribution free'',
i.e. is independent of the particular distribution of $e_{it}$
(see, e.g., Silverstein~\cite*{Silverstein1990},
Bai~\cite*{bai1999review}). However, we are not aware of a formulation
and corresponding proof of this conjecture that is sufficient for our purposes, which is one reason why
we have to impose $iid$ normality of $e_{it}$.

  \item[(iii)]    Assumption~\ref{ass:EV}$(ii)$ imposes a condition on the eigenvalues
$\rho_r$ of the random covariance matrix $M_{f^0} e' M_{\lambda^0} e M_{f^0}$.
Eigenvalues are studied more intensely than eigenvectors
in the random matrix theory literature,
and it is well-known that the properly normalized
empirical distribution of the eigenvalues
(the so called empirical spectral distribution)
of an $iid$ sample covariance matrix converges
to the Mar{\v{c}}enko-Pastur-law (Mar{\v{c}}enko and Pastur~\cite*{MarcenkoPastur1967})
for asymptotics where $N$ and $T$ grow at the same rate.
This means that the sum over the function of the eigenvalues $\rho_s$ in Assumption~\ref{ass:EV}$(ii)$
can be approximated by an integral over the Mar{\v{c}}enko-Pastur
limiting spectral distribution.
To bound the asymptotic error of this approximation one needs to know the
convergence rate of the empirical spectral
distribution to its limit law, which is an ongoing research subject in the
literature, e.g. Bai~\cite*{Bai1993}, Bai, Miao and Yao \cite*{BaiMiaoYao2004},
G{\"o}tze and Tikhomirov \cite*{GotzeTikhomirov2010}. This literature usually considers
either $iid$ or $iid$ normal distributions of $e_{it}$.

\item[(iv)]
For random covariance
matrices from $iid$ normal errors, it is known from
Johnstone~\cite*{Johnstone2001} and Soshnikov~\cite*{Soshnikov2002}
that the properly normalized few
largest eigenvalues converge to the Tracy-Widom law.\footnote{To our knowledge this result
is not established for error distributions that are not normal.
 Soshnikov~\cite*{Soshnikov2002} has a result under non-normality
 but only for asymptotics with $N/T \rightarrow 1$.}
 This result can be used to verify Assumption~\ref{ass:EV}$(i)$ in the case of $iid$ normal $e_{it}$.

  \item[(v)] Details on how to derive
Theorem~\ref{th:LimitingDistribution}
are given in Section~\ref{app:MoreEquivalence} of the supplementary material.

\end{itemize}

The following Lemma provides the connection between Theorem~\ref{th:LimitingDistribution}
and our main result Theorem~\ref{th:MAIN}. The proof is given in the supplementary material.

\begin{lemma}
    \label{lemma:JustifyEV}
   Let Assumption~\ref{ass:LL} hold, let $R^0 = {\rm rank}(\lambda^0) = {\rm rank}(f^0)$, and
consider a limit $N,T \rightarrow \infty$ with $N/T \rightarrow \kappa^2$, $0<\kappa<\infty$.
Then Assumptions~\ref{ass:SN}, \ref{ass:EX}, \ref{ass:DX-1}
 and \ref{ass:EV} are satisfied, and we have $C^{(1)} = {\cal O}_P(1)$.
\end{lemma}

\end{appendix}

\clearpage
\setcounter{section}{0}
\setcounter{equation}{0}
\setcounter{theorem}{0}
\setcounter{assumption}{0}
\setcounter{definition}{0}
\setcounter{table}{0}
\setcounter{figure}{0}
\renewcommand{\thesection}{S.\arabic{section}}
\renewcommand{\theequation}{S.\arabic{equation}}
\renewcommand{\thetheorem}{S.\arabic{theorem}}
\renewcommand{\thetable}{S.\arabic{table}}
\renewcommand{\thefigure}{S.\arabic{figure}}
\begin{center}
{\Large\bfseries Supplementary Material}
\end{center}
\bigskip

\vspace{-1cm}

\pagebreak
\section{Proofs for Main Text Results}
\label{app:Proofs-Main}

\begin{proof}[\bf \underline{Proof of Theorem~\ref{th:id} (Identifictaion)}]
     Let
     $Q(\beta,\Lambda,F) =
     \mathbbm{E}\left(\left\| Y \, - \, \beta \cdot X \, - \, \Lambda \, F'
    \right\|^2_{HS}\right)$.
    Existence of $Q(\beta,\Lambda,F)$ is guaranteed by Assumption~\ref{ass:id}$(i)$.
    The statement of the theorem follows if we can show that  $Q(\beta,\Lambda,F)$ is uniquely minimized at
    $\beta=\beta^0$ and $\Lambda F' = \lambda^0 f^{0 \prime}$. 
   We have
   \begin{align}
        Q(\beta,\Lambda,F)
         &= \mathbbm{E} \; \Tr \left[ \left( Y \, - \, \beta \cdot X \, - \, \Lambda \, F' \right)
                      \left( Y \, - \, \beta \cdot X \, - \, \Lambda \, F' \right)' \right]
      \nonumber \\
         &=  \mathbbm{E} \; \Tr \left[
          \left(  \lambda^0 f^{0 \prime} -  \Lambda F' -  (\beta-\beta^0) \cdot X  + e \right)
           \left(  \lambda^0 f^{0 \prime} -  \Lambda F' -  (\beta-\beta^0) \cdot X  + e \right)'  \right]
       \nonumber \\
          &=
          \underbrace{ \mathbbm{E} \; \Tr \left[
          \left(  \lambda^0 f^{0 \prime} -  \Lambda F' -  (\beta-\beta^0) \cdot X  \right)
           \left(  \lambda^0 f^{0 \prime} -  \Lambda F' -  (\beta-\beta^0) \cdot X  \right)' \right]
          }_{\displaystyle \equiv Q^*(\beta,\Lambda,F) }
              +  \mathbbm{E} \; \Tr \left( e e' \right).
    \end{align}
    Here, we used the model, and
    in the last step we employed Assumption~\ref{ass:id}$(ii)$.
    Next, we derive a lower bound
    on $Q^*(\beta,\Lambda,F)$. We have
    \begin{align}
        Q^*(\beta,\Lambda,F)
        &\geq
        \mathbbm{E} \; \Tr \left[ 
          \left(  \lambda^0 f^{0 \prime} -  \Lambda F' -  (\beta-\beta^0) \cdot X  \right)
          M_{F}
           \left(  \lambda^0 f^{0 \prime} -  \Lambda F' -  (\beta-\beta^0) \cdot X  \right)'   \right]
      \nonumber \\
        &=
        \mathbbm{E} \; \Tr \left[ 
          M_{F}
           \left(  \lambda^0 f^{0 \prime} -  \Lambda F' -  (\beta-\beta^0) \cdot X  \right)' 
    \left(  \lambda^0 f^{0 \prime} -  \Lambda F' -  (\beta-\beta^0) \cdot X  \right)       
          M_{F} 
             \right]
      \nonumber \\
        &\geq
        \mathbbm{E} \; \Tr \left[ 
          M_{F}
           \left(  \lambda^0 f^{0 \prime} -  \Lambda F' -  (\beta-\beta^0) \cdot X  \right)'  M_{\lambda^0}
    \left(  \lambda^0 f^{0 \prime} -  \Lambda F' -  (\beta-\beta^0) \cdot X  \right)       
          M_{F} 
             \right]
      \nonumber \\
            &=
        \mathbbm{E} \; \Tr \left[  M_F
          \left(   (\beta-\beta^0) \cdot X  \right)'
          M_{\lambda^0}
           \left(    (\beta-\beta^0) \cdot X  \right) \right]
      \nonumber \\
            &= (\beta-\beta^0)'
            \left\{ \mathbbm{E}[x' (M_{F} \otimes M_{\lambda^0}) x] \right\}
               (\beta-\beta^0).
    \end{align}
    From this and
    Assumption~\ref{ass:id}$(iii)$
    we conclude that $Q^*(\beta,\Lambda,F)>0$ for all $\beta \neq \beta^0$.
    On the other hand, we have $Q^*(\beta^0,\lambda^0,f^0)=0$.
    Thus, every minimum of $Q^*(\beta,\Lambda,F)$ satisfies $\beta=\beta^0$.
    Furthermore,
    at $\beta=\beta^0$ we have
   $Q^*(\beta^0,\Lambda,F) = \|  \lambda^0 f^{0 \prime} -  \Lambda F'   \|_{HS}^2$,
   which is zero if and only if $\Lambda F' = \lambda^0 f^{0 \prime}$. The minima
   of $Q^*(\beta,\Lambda,F)$ therefore satisfy $\beta=\beta^0$ and
   $\Lambda F' = \lambda^0 f^{0 \prime}$. Since $Q^*(\beta,\Lambda,F)$
   and $Q(\beta,\Lambda,F)$ only differ by a constant the same result holds for
   $Q(\beta,\Lambda,F)$.
   Notice that the result that the optimal
$\Lambda$ and $F$ satisfy $\Lambda F' = \lambda^0 f^{0 \prime}$ implies that ${\rm rank}( \Lambda F')=R^0$,
i.e. the true number of factors $R^0$ is also identified.
\end{proof}

\begin{proof}[\bf \underline{Proof of Theorem~\ref{th:MAIN} (Main Result)}]
     Follows from Theorem~\ref{th:LimitingDistribution}
     and  Lemma~\ref{lemma:JustifyEV}.
\end{proof}

\begin{proof}[\bf  \underline{Proof of Theorem~\ref{th:Estimators}
(Consistency of Bias and Variance Estimators)}]~\\
   See Section~\ref{app:Estimators} below.
\end{proof}

\begin{proof}[\bf \underline{Proof of Theorem \ref{th:consistency} (Consistency)}]
   We first establish a lower bound on ${\cal L}^R_{NT}(\beta)$.
   Let $\Delta \beta = \beta - \beta^0$.
   Consider the definition of
   ${\cal L}^R_{NT}(\beta)$ in equation \eqref{LSobjective} and plug in the model
    $Y=  \beta \cdot X + \lambda^{0} f^{0\prime} + e$.
    We have
    \begin{align}
       {\cal L}^R_{NT}(\beta) &=  \min_{\left \{ \Lambda \in \mathbbm{R}^{N\times R}, \;
                F \in \mathbbm{R}^{T\times R} \right\}  }   \frac 1 {NT}  {\rm Tr}
            \left[ \left( \Delta \beta \cdot X + e + \lambda^{0} f^{0\prime} - \Lambda F'  \right)  
                   \left( \Delta \beta \cdot X + e + \lambda^{0} f^{0\prime} - \Lambda F' \right)' \right] 
    \nonumber \\
        &\geq            \min_{\left \{ \tilde \Lambda \in \mathbbm{R}^{N\times (R+R^0)}, \;
                \tilde F \in \mathbbm{R}^{T\times (R+R^0)} \right\}  }  \frac 1 {NT}  {\rm Tr}
            \left[ \left( \Delta \beta \cdot X + e   - \tilde \Lambda \tilde F'  \right)  
                   \left( \Delta \beta \cdot X + e   - \tilde \Lambda \tilde F' \right)' \right]    
    \nonumber \\
         &=
        \frac 1 {NT}   \min_{ \tilde F \in \mathbbm{R}^{T\times (R+R^0)} }
         {\rm Tr}
            \left[ \left( \Delta \beta \cdot X + e \right) M_{\widetilde F}
                   \left( \Delta \beta \cdot X + e \right)' \right]  
    \nonumber \\
         &=
   \frac 1 {NT} \min_{ \tilde F \in \mathbbm{R}^{T\times (R+R^0)} }  \Bigg\{ 
         {\rm Tr}
            \left[ \left( \Delta \beta \cdot X  \right) M_{\widetilde F}
                   \left(\Delta  \beta \cdot X   \right)' \right]   +  {\rm Tr} \left( e  e' \right)
                   -  {\rm Tr} \left( e P_{\widetilde F}  e' \right)
        \nonumber \\ & \qquad \qquad        \qquad   \qquad \qquad        \qquad      \qquad        \qquad    
           + 2       {\rm Tr}
            \left[ \left( \Delta \beta \cdot X  \right)  e' \right]      
           -  2   {\rm Tr}
            \left[ \left( \Delta \beta \cdot X  \right) P_{\tilde F}  e' \right]   
           \Bigg\}
    \nonumber \\
         &\geq
   \frac 1 {NT}  \Bigg\{  
     \sum_{r=R+R^0+1}^T  \mu_r\left[  (\Delta \beta \cdot X)' (\Delta \beta \cdot X)  \right]   
     +  {\rm Tr} \left( e  e' \right) -  2 (R+R^0)   \|e\|^2 
        \nonumber \\ & \qquad \qquad        \qquad  \qquad        \qquad        \qquad        \qquad  
           + 2       {\rm Tr}
            \left[ \left( \Delta \beta \cdot X  \right)  e' \right]      
           -  2 (R+R^0)   \|e\|   \|    \Delta  \beta \cdot X \| 
           \Bigg\}
         \nonumber \\
           &\geq b \, \| \Delta \beta \|^2 + \frac 1 {NT} \, {\rm Tr} \left( e  e' \right)
           +  {\cal O}_P \left( \frac 1 {\min(N,T)} \right)
                + {\cal O}_P\left( \frac {\| \Delta \beta \|}
                             {\sqrt{\min(N,T)}} \right)     .
       \label{ConBnd1}
    \end{align}
   Here,  we applied the inequality $\left| {\rm Tr}(A) \right| \leq {\rm rank}(A) \|A\|$ with $A 
   = \left( \Delta \beta \cdot X \right)  P_{\widetilde F} e'$
   and also with $A = e P_{\widetilde F}  e' $.
   We also used that 
   $ \min_{ \tilde F }  
    {\rm Tr}
            \left[ \left( \Delta \beta \cdot X  \right) M_{\widetilde F}
                   \left(\Delta  \beta \cdot X   \right)' \right]
           =  \sum_{r=R+R^0+1}^T  \mu_r\left[  (\Delta \beta \cdot X)' (\Delta \beta \cdot X)  \right]   $,
    which follows by the same logic as equation \eqref{LSobjective}  in the main text.     
     In the last step of \eqref{ConBnd1} we applied Assumptions~\ref{ass:SN}, \ref{ass:EX} and \ref{ass:NC}.
     
    Next, we establish an upper bound on ${\cal L}^R_{NT}(\beta^0)$. Since $R \geq R^0$ 
    we can choose $\Lambda \, F' = \lambda^{0} f^{0\prime}$ in the minimization problem
    in the first line of equation \eqref{LSobjective}, and therefore
    \begin{align}
       {\cal L}^R_{NT}(\beta^0)
        \, &=  \min_{\left \{ \Lambda \in \mathbbm{R}^{N\times R}, \;
                F \in \mathbbm{R}^{T\times R} \right\}  } \;
                \frac 1 {NT}  \; \left\|  e + \lambda^{0} f^{0\prime}  - \, \Lambda \, F'
    \right\|^2_{HS}
        \nonumber \\  
          &\leq  \frac 1 {NT} \left\| e
    \right\|^2_{HS} =
         \frac 1 {NT} \, {\rm Tr} \left( e  e' \right) .
       \label{ConBnd2}
    \end{align}
    Since we could choose $\beta=\beta^0$ in the minimization of $\beta$, the optimal
    $\widehat \beta_R$ needs to satisfy ${\cal L}^R_{NT}(\widehat \beta_R) \leq {\cal L}^R_{NT}(\beta^0)$. 
    Together with \eqref{ConBnd1} and \eqref{ConBnd2} this gives
    \begin{align}
         b \, \| \widehat \beta_R - \beta^0 \|^2
                          +  {\cal O}_P\left( \frac {\| \widehat \beta_R - \beta^0 \|} {\sqrt{\min(N,T)}} \right)
                          +  {\cal O}_P \left( \frac 1 {\min(N,T)} \right)
                          \, \leq  \, 0 \; .
    \end{align}
    From this it follows that
    $\| \widehat \beta_R - \beta^0 \|  =
         {\cal O}_P \left( \min(N,T)^{-1/2} \right)$,
    which is what we wanted to show.
\end{proof}

\begin{proof}[\bf \underline{Proof of Theorem \ref{th:expansion} (Quadratic Approximation of ${\cal L}_{NT}^0(\beta)$)}]   See Section~\ref{app:expansion1} below.
\end{proof}

\begin{proof}[\bf \underline{Proof of Corollary~\ref{cor:LimitR0} (Asymptotic Characterization of $\widehat \beta_{R^0}$)}] $\phantom{a}$ \\
    Define $\gamma \equiv W^{-1} \left( C^{(1)} + C^{(2)} \right)/\sqrt{NT}$.
    Applying Theorem~\ref{th:expansion} we obtain
    \begin{align}
     {\cal L}_{NT}^{0}\left( \widehat \beta_{R^0} \right) &=
      {\cal L}_{NT}^{0}(\beta^0)
      +   \left( \widehat \beta_{R^0}-\beta^0-\gamma \right)' W
      \left( \widehat \beta_{R^0}-\beta^0 - \gamma \right)
      - \gamma' W \gamma
       + {\cal L}_{NT}^{0,{\rm rem}}( \widehat \beta_{R^0} ),
     \nonumber \\
     {\cal L}_{NT}^{0}\left(\beta^0+\gamma\right) &=
      {\cal L}_{NT}^{0}(\beta^0) - \gamma' W \gamma + {\cal L}_{NT}^{0,{\rm rem}}(\beta^0+\gamma) .
    \end{align}
    The first equation above is obtained by completing the square and using the definition of $\gamma$,
    while the second equation is just a special case of the first.
    Applying the above to
      the inequality
   ${\cal L}_{NT}^{0}(\widehat \beta_{R^0}) \leq
{\cal L}_{NT}^{0}\left(\beta^0+\gamma\right)$ gives
   \begin{align}
      \left( \widehat \beta_{R^0}-\beta^0-\gamma \right)' W
      \left( \widehat \beta_{R^0}-\beta^0 - \gamma \right)
        &\leq {\cal L}_{NT}^{0,{\rm rem}}(\beta^0+\gamma)
             -{\cal L}_{NT}^{0,{\rm rem}}(\widehat \beta_{R^0}).
        \label{BoundSquare}
   \end{align}
   We have
   $W \geq \mu_K(W) \mathbbm{1}_K$,
   and using Assumption~\ref{ass:NC} we find for $R=R^0$ that    
    \begin{align}
      \mu_K(W)  &= 
      \min_{\{\alpha \in \mathbbm{R}^K, \|\alpha\|=1\}}
      \alpha' W \alpha
  \nonumber \\
      &=  \min_{\{\alpha \in \mathbbm{R}^K, \|\alpha\|=1\}}
       \frac 1 {NT} {\rm Tr}\left(  M_{\lambda^0} (\alpha \cdot X)
                                     M_{f^0} (\alpha \cdot X)'  \right)
  \nonumber \\
      &=   \min_{\{\alpha \in \mathbbm{R}^K, \|\alpha\|=1\}}
       \frac 1 {NT} {\rm Tr}\left(  M_{f^0} (\alpha \cdot X)' M_{\lambda^0} (\alpha \cdot X)
                                     M_{f^0} \right)
     \nonumber \\
       & \geq   \frac 1 {NT} \sum_{r=2R^0+1}^T
                \mu_r \left[ (\alpha \cdot X)'  (\alpha \cdot X) \right]
         \geq b \; , \qquad \text{wpa1,}
   \end{align}
   and therefore $W^{-1} \leq \mathbbm{1}_K/b$ wpa1. Using Assumption~\ref{ass:SN}
   we find
   \begin{align}
      |C^{(1)}_k| &\leq \left| \frac 1 {\sqrt{NT}} {\rm Tr}\left( X_k e' \right)
                        \right|
                    + \frac {2 R^0} {\sqrt{NT}} \|X_k\| \|e\|
                = {\cal O}_P \left(\sqrt{\max(N,T)} \right)   ,
      \nonumber \\
      |C^{(2)}_k| &\leq  \frac {9 R^0} {2 \sqrt{NT}}
                     \|e\|^2 \|X_k\|
                     \left\|  \lambda^0 \, (\lambda^{0\prime}\lambda^0)^{-1} \, (f^{0\prime}f^0)^{-1} \, f^{0\prime}  \right\|
                      = {\cal O}_P \left(1 \right) ,
   \end{align}
   and therefore we have
  $\gamma = {\cal O}_P[ (1+ \|C^{(1)}\| )/\sqrt{NT} ]
  = o_P(1)$.  We also know $\| \widehat \beta_{R^0} - \beta^0 \| = o_P(1)$
  from Theorem~\ref{th:consistency}. Thus, the bound
     on the remainder in Theorem~\ref{th:expansion} becomes applicable and we have
     \begin{align}
         & {\cal L}_{NT}^{0,{\rm rem}}(\beta^0+\gamma)
             -{\cal L}_{NT}^{0,{\rm rem}}(\widehat \beta_{R^0}) \leq
             o_P\left( \frac 1 {NT} \right)
             \left[  \left( 1 + \sqrt{NT}  \gamma \right)^2
               + \left( 1 + \sqrt{NT} \| \widehat \beta_{R^0} - \beta^0 \| \right)^2
             \right]
          \nonumber \\    &  \qquad \qquad
             = o_P\left( \frac 1 {NT} \right)
                \left\{   {\cal O}_P \left[ (1+ \|C^{(1)}\| )^2 \right]
               + \left( 1 + \sqrt{NT} \| \widehat \beta_{R^0} - \beta^0 \| \right)^2
             \right\} .
     \end{align}
     Applying this, and \eqref{BoundSquare}, and $W^{-1} \leq 1/b$, and the inequality $\sqrt{(x+y)} \leq \sqrt{x} + \sqrt{y}$,
     which holds for all non-negative real number $x$, $y$, we find that
     \begin{align}
         \sqrt{NT} \left\|  \widehat \beta_{R^0}-\beta^0-\gamma \right\|
         &\leq
          o_P \left( 1+ \|C^{(1)}\|  \right)
               + o_P \left( 1 + \sqrt{NT} \| \widehat \beta_{R^0} - \beta^0 \| \right) .
     \end{align}
     Since $\gamma = {\cal O}_P[ (1+ \|C^{(1)}\| )/\sqrt{NT} ] $ it follows
     from this that $\sqrt{NT} \| \widehat \beta_{R^0}-\beta^0 \| = {\cal O}_P \left( 1+ \|C^{(1)}\|  \right)$,
     and therefore
     \begin{align}
         \sqrt{NT} \left\|  \widehat \beta_{R^0}-\beta^0-\gamma \right\|
         &\leq
          o_P \left( 1+ \|C^{(1)}\|  \right) ,
     \end{align}
     which is what we wanted to show.
\end{proof}

\begin{proof}[\bf \underline{Proof of Example in Section~\ref{sec:AsymptoticSummary} (Counter Example for $\sqrt{NT}$ Convergence Rate)}]~

Consider the DGP and asymptotic as described in the example in Section~\ref{sec:AsymptoticSummary}.
Let ${\cal L}_{NT}^1(\beta)$ be the profile objective function for $R=1$, defined in \eqref{LSobjective}.
We want to show that for any sequence $\Delta_{NT}>0$ with $\Delta_{NT} = o(N^{-1/2})$ we have
\begin{align}
     \min_{\beta \in \mathbbm{R}}  {\cal L}_{NT}^1(\beta)
     &< \min_{\beta \in [\beta^0- \Delta_{NT} , \beta^0 + \Delta_{NT}] }  {\cal L}_{NT}^1(\beta) ,
     \qquad \text{wpa1.}
   \label{ResultExample}  
\end{align}
This implies that $\left\| \widehat \beta_1 - \beta^0 \right\|$ cannot converge to zero at a faster than
$\sqrt{N}$ rate.

What is left to do is to proof \eqref{ResultExample}.
We decompose $Y - \beta \cdot X = e - (\beta-\beta^0)  X = e_1(\beta) + e_2(\beta)$, where
\begin{align}
     e_1(\beta) &= \frac c N \lambda _{x} (M_{f_{x}} u' \lambda _{x} )'
          + \frac c T
          (M_{\lambda_{x}} u f_{x}) f_x'
        -  (\beta-\beta^0)   \lambda_x f_x' ,
   \nonumber \\
     e_2(\beta) &=   \widetilde u + \frac{ c^2 (\lambda _{x}^{\prime } u  f_{x})}{NT} \lambda _{x} f_{x}^{\prime }
         + \frac c N  \lambda _{x} \lambda _{x}' u  P_{f_{x}}
         + \frac c T  P_{\lambda_{x}} u f_{x}f_{x}^{\prime }
      ,
\end{align}
with $\widetilde u = u - a (\beta-\beta^0) \widetilde X$.
Since $\|  \lambda _{x} \| = {\cal O}(\sqrt{N})$,
$\|  f_x \|  = {\cal O}(\sqrt{T})$,
and $\lambda _{x}^{\prime } u  f_{x} =  {\cal O}_P(\sqrt{NT})$
we have $\| e_2(\beta) - \widetilde u \| = o_P(N)$.
The matrix $\widetilde u$ has iid normal entries with
mean zero and variance $1 + a^2 (\beta-\beta^0)^2$.
According to Geman~\cite*{Geman1980} we thus have
$\| \widetilde u \|^2 = (1 + a^2 (\beta-\beta^0)^2) (\sqrt{N} + \sqrt{T})^2 + o_P(N)$.
Thus, as $N,T \rightarrow \infty$ at the same rate we have
\begin{align}
     \|  e_2 (\beta)  \|^2 \leq
      (1 + a^2 (\beta-\beta^0)^2) ( \sqrt{N} + \sqrt{T} )^2
      + o_P(N).
   \label{ResE2neg}
\end{align}
Note that ${\rm rank}(e_1(\beta))=2$, which implies that $e_1$ can be written as
$e_1 = A \widetilde e_1 B'$, where $A$ is an $N \times 2$ matrix satisfying $A'A = \mathbbm{1}_2$,
$B$ is a $T \times 2$ matrix satisfying $B'B = \mathbbm{1}_2$, and $\widetilde e_1$ is a $2 \times 2$
matrix, namely
\begin{align}
    \widetilde e_1 &= \left( \begin{array}{cc} 
                         (\beta-\beta^0)     \left\|      \lambda_x f_x'  \right\| &  \left\|   \frac c N \lambda _{x} (M_{f_{x}} u' \lambda _{x} )'  \right\| \\
                              \left\|   \frac c T
          (M_{\lambda_{x}} u f_{x}) f_x'  \right\| &  0
                       \end{array}  \right) .
\end{align}
Using this characterization of $e_1$ as well as
$\|  \lambda _{x} \|^2 = N + o(N)$,
$\|  f_x \|^2 = T + o(T)$,
$\| M_{f_{x}} u' \lambda _{x} \|^2 = NT + o_P(NT)$, and
$\| M_{\lambda_{x}} u f_{x} \|^2 = NT + o_P(NT)$,
we  find
\begin{align}
   &  \| e_1(\beta) \|^2
  \nonumber \\
   &= \mu_1\left[ e_1(\beta)' e_1(\beta) \right] = \mu_1\left[ \widetilde e_1(\beta)' \widetilde e_1(\beta) \right]
  \nonumber \\
     &=   \mu_1 \left[ \left( \begin{array}{c@{\qquad}c}
    \| f_x \|^2
    \left(   \frac{c^2 \| M_{\lambda_{x}} u f_{x} \|^2 } {T^2}
         + \| \lambda_x \|^2 (\beta - \beta^0)^2
    \right)  &  \frac{ c \| \lambda_x \|^2 \| f_x \| \| M_{f_{x}} u' \lambda _{x} \| (\beta-\beta^0)} N  \\
      \frac{ c \| \lambda_x \|^2 \| f_x \| \| M_{f_{x}} u' \lambda _{x} \| (\beta-\beta^0) } N
      & \frac{c^2 \| \lambda_x \|^2  \| M_{f_{x}} u' \lambda _{x} \|^2 } {N^2}
        \end{array} \right) \right]
  \nonumber \\
     &=   \mu_1 \left[ \left( \begin{array}{cc}
     c^2 N + N T (\beta-\beta^0)^2  &  c T \sqrt{N} (\beta-\beta^0) \\
      c T \sqrt{N} (\beta-\beta^0) & c^2 T
        \end{array} \right) \right]  + o_P\left[ \left(\sqrt{N} + \sqrt{NT} \| \beta-\beta^0)\| \right)^2  \right]
  \nonumber \\
      &= \frac 1 2 \left( c^2 N + c^2 T + N T (\beta - \beta^0)^2
         + \sqrt{ \left[ c^2 N + c^2 T + N T (\beta - \beta^0)^2 \right]^2 - 4 c^4 N T
            }
          \right) 
           \nonumber \\
       & \qquad + o_P\left[ \left(\sqrt{N} + \sqrt{NT} \| \beta-\beta^0)\| \right)^2  \right] .
   \label{ResE1neg}
\end{align}
The objective function for $R=1$ reads
\begin{align}
   {\cal L}_{NT}^1(\beta)
      &= {\cal L}_{NT}^0(\beta)
         -  \mu_1\left[ \left(Y - \beta \cdot X \right)'
                     \left(Y - \beta \cdot X \right) \right]
 \nonumber \\
     &= {\rm Tr}\left[  \left(Y - \beta \cdot X \right)'
                     \left(Y - \beta \cdot X \right) \right]
        -  \mu_1\left[ \left(Y - \beta \cdot X \right)'
                     \left(Y - \beta \cdot X \right) \right]
 \nonumber \\
    &= {\rm Tr}( e' e) + 2 (\beta- \beta^0)  {\rm Tr}( X' e) + (\beta- \beta^0)^2 {\rm Tr}( X' X)
 \nonumber \\   & \qquad
       -  \mu_1\left[ \left(e_1(\beta) + e_2(\beta) \right)'
                    \left(e_1(\beta) + e_2(\beta) \right) \right]
 \nonumber \\
    &= {\rm Tr}( e' e)  + (\beta- \beta^0)^2 ( NT + a^2 NT  )
       + {\cal O}_P( \sqrt{NT} \| \beta- \beta^0 \|  )
       + o_P( NT \| \beta- \beta^0 \|^2  )
 \nonumber \\   & \qquad
       -  \mu_1\left[ \left(e_1(\beta) + e_2(\beta) \right)'
                    \left(e_1(\beta) + e_2(\beta) \right) \right] .
\end{align}
We have
\begin{align}
   & \left|  \mu_1\left[ \left(e_1(\beta) + e_2(\beta) \right)'
                    \left(e_1(\beta) + e_2(\beta) \right) \right]
                - \mu_1\left[ e_1(\beta)' e_1(\beta) \right]
   \right|
 \nonumber \\
    &\leq \left\|    \left(e_1(\beta) + e_2(\beta) \right)'
                    \left(e_1(\beta) + e_2(\beta) \right)
                 -   e_1(\beta)' e_1(\beta) \right\|
 \nonumber \\
    &\leq  2 \| e_1(\beta)  \| \, \| e_2(\beta)  \|
       +  \| e_2(\beta) \|^2 ,
\end{align}
and therefore
\begin{align}
    &  \left|  {\cal L}_{NT}^1(\beta) -
     {\rm Tr}( e' e)  - (\beta- \beta^0)^2 ( NT + a^2 NT  )
       + \| e_1(\beta) \|^2  \right|
  \nonumber \\
     & \qquad \leq
     2 \| e_1(\beta)  \| \, \| e_2(\beta)  \|
       +  \| e_2(\beta) \|^2
       +
       {\cal O}_P( \sqrt{NT} \| \beta- \beta^0 \|  )
       + o_P( NT \| \beta- \beta^0 \|^2  ) .
   \label{LobjBoundsNeg}
\end{align}
Using this inequality together with the results on $\| e_1(\beta)  \|$
and $\| e_2(\beta) \|$ above one can show that (for details see below)
\begin{align}
   \min_{\beta \in [\beta^0- \Delta_{NT} , \beta^0 + \Delta_{NT}] }  {\cal L}_{NT}^1(\beta)
      &\geq    {\rm Tr}( e' e) -  T \,
      \underbrace{\left[ c \max(1,\kappa) + 1 + \kappa  \right]^2}_{\equiv f_1(\kappa,a,c)}   + o_P(N) ,
   \label{BoundL1neg}
\end{align}
and for $\widetilde \beta_{NT} = \beta^0 +  c {(a NT)^{-1/4}} $ we have (again, for details see below)
\begin{align}
   &  {\cal L}_{NT}^1(\widetilde \beta_{NT})
    \nonumber \\
      & \leq   {\rm Tr}( e' e)  - \underbrace{ \left[
       c^2 g(a,\kappa) - c^2 a^{-1/2}  ( 1 + a^2 ) \kappa
       - 2  c  ( 1+ \kappa )  \sqrt{g(a,\kappa)}
       - ( 1+ \kappa )^2 \right]}_{\equiv f_2(\kappa,a,c)} T + o_P(N) ,
   \label{BoundL2neg}
\end{align}
where
\begin{align}
    g(a,\kappa) &= \frac 1 2 \left(1 + \kappa^2 + \frac{\kappa} { \sqrt{a}}
         + \sqrt{ \left(1 + \kappa^2 + \frac{\kappa} { \sqrt{a}} \right)^2 -  4 \kappa^2
            }
          \right).
\end{align}
For $0<a <(1/2)^{2/3} \min(\kappa^2,\kappa^{-2})$ and $c \geq    \frac{ (2+\sqrt{2})  \left( 1+\kappa \right)
 (1+ \sqrt{3} a^{-1/4} )}   { \min(1,\kappa) [1/2 - a^{3/2} \max(\kappa,\kappa^{-1})]} $
  one can show that
$f_1(\kappa,a,c) < f_2(\kappa,a,c)$ (for details on this below).
Thus, for these values of $a$ and $c$ we can conclude that wpa1
\begin{align}
    \min_{\beta \in [\beta^0- \Delta_{NT} , \beta^0 + \Delta_{NT}] }  {\cal L}_{NT}^1(\beta) >  {\cal L}_{NT}^1(\widetilde \beta_{NT})
    \geq  \min_{\beta \in \mathbbm{R}}  {\cal L}_{NT}^1(\beta)  .
\end{align}
This is what we wanted to show. In the following we provide more details regarding how to obtain
\eqref{BoundL1neg} and \eqref{BoundL2neg} and $f_1(\kappa,a,c) < f_2(\kappa,a,c)$.

 \# Derivation of \eqref{BoundL1neg}:
Remember $\Delta_{NT} = o(N^{-1/2})$.
Thus, for any $\beta \in [\beta^0- \Delta_{NT} , \beta^0 + \Delta_{NT}]$
we find from \eqref{ResE2neg}, \eqref{ResE1neg}, and \eqref{LobjBoundsNeg} that
\begin{align}
    \| e_1(\beta) \| ^2 &= c^2 \max( N, T ) + o_P(N) = c^2 \max( 1, \kappa^2 ) T  + o_P(N) ,
    \nonumber \\
    \| e_2(\beta) \|^2 &=  ( \sqrt{N} + \sqrt{T} )^2
      + o_P(N) =  ( 1+ \kappa )^2 T
      + o_P(N)  ,
    \nonumber \\
   {\cal L}_{NT}^1(\beta)
   &\geq
     {\rm Tr}( e' e)
       - \| e_1(\beta) \|^2      -   2 \| e_1(\beta)  \| \, \| e_2(\beta)  \|
    -  \| e_2(\beta) \|^2
       + o_P( N   )
  \nonumber \\
     &=    {\rm Tr}( e' e)
      -   \left( \| e_1(\beta) \| + \| e_2(\beta)  \| \right)^2 + o_P( N   )
   \nonumber \\
      &=  {\rm Tr}( e' e) -  T \, \left[ c \max(1,\kappa) + 1 + \kappa  \right]^2   + o_P(N)
      .
\end{align}

 \# Derivation of \eqref{BoundL2neg}:
We defined $\widetilde \beta_{NT} = \beta^0 +  c {(a NT)^{-1/4}} $.
From \eqref{ResE2neg} we find
$\| e_2(\beta) \|^2  =  ( 1+ \kappa )^2 T
      + o_P(N)$ as before.
Furthermore,  we find from \eqref{ResE1neg} that
$\left\| e_1(\widetilde \beta_{NT}) \right\|^2
     = c^2 \, T \, g(a,\kappa) + o_P(N)$.
Equation \eqref{LobjBoundsNeg} thus gives
\begin{align}
   &  {\cal L}_{NT}^1(\widetilde \beta_{NT})
 \nonumber \\
        &\leq {\rm Tr}( e' e)  + c^2 a^{-1/2}  ( 1 + a^2 ) \kappa T
       -  \left\| e_1(\widetilde \beta_{NT}) \right\|^2
       + 2  \left\| e_1(\widetilde \beta_{NT}) \right\|  \left\| e_2(\widetilde \beta_{NT}) \right\|
       +   \left\| e_2(\widetilde \beta_{NT}) \right\|^2 + o_P(N)
 \nonumber \\
        &= {\rm Tr}( e' e)  + \left[ c^2 a^{-1/2}  ( 1 + a^2 ) \kappa
       -  c^2 g(a,\kappa)
       + 2  c  ( 1+ \kappa )  \sqrt{g(a,\kappa)}
       + ( 1+ \kappa )^2 \right] T + o_P(N) .
\end{align}

 \# Show that $f_1(\kappa,a,c) < f_2(\kappa,a,c)$:
Recall
\begin{eqnarray*}
f_{1}\left( \kappa ,a,c\right) &=&\left( \max \left\{ 1,\kappa \right\}
c+1+\kappa \right) ^{2}=\max \left\{ 1,\kappa ^{2}\right\} c^{2}+2\max
\left\{ 1,\kappa \right\} \left( 1+\kappa \right) c+\left( 1+\kappa \right)
^{2}, \\
f_{2}\left( \kappa ,a,c\right) &=&\left( g\left( a,\kappa \right) -\frac{%
1+a^{2}}{\sqrt{a}}\kappa \right) c^{2}-2\left( 1+\kappa \right) \sqrt{%
g\left( a,\kappa \right) }c-\left( 1+\kappa \right) ^{2}.
\end{eqnarray*}
Note that $f_{2}\left( \kappa
,a,c\right) -f_{1}\left( \kappa ,a,c\right) $ is a quadratic polynomial in $c$, namely
\begin{equation}
f_{2}\left( \kappa ,a,c\right) -f_{1}\left( \kappa ,a,c\right) =h_{1}\left(
a,\kappa \right) c^{2}-2h_{2}\left( a,\kappa \right) c-h_{3}\left( \kappa
\right) ,  \label{quadratic.equation}
\end{equation}%
where
\begin{eqnarray*}
h_{1}\left( a,\kappa \right)  &=&g\left( a,\kappa \right) -\frac{1+a^{2}}{%
\sqrt{a}}\kappa -\max \left\{ 1,\kappa ^{2}\right\} , \\
h_{2}\left( a,\kappa \right)  &=&\left( 1+\kappa \right) \sqrt{g\left(
a,\kappa \right) }+\max \left\{ 1,\kappa \right\} \left( 1+\kappa \right) >0 ,
\\
h_{3}\left( \kappa \right)  &=&2\left( 1+\kappa \right) ^{2}>0.
\end{eqnarray*}%
We first want to show that $h_{1}\left( a,\kappa \right) >0$.
By assumption we have $a=\epsilon ^{2}\min \left\{ \kappa ^{2},\kappa ^{-2}\right\}$ with
$0<\epsilon \leq \left( 1/2\right) ^{1/3}$.
Suppose that $\kappa \geq 1$, i.e. $a=\frac{\epsilon ^{2}}{\kappa
^{2}}.$ Then, we have%
\begin{align*}
&h_{1}\left( a,\kappa \right)   \\
&=g\left( a,\kappa \right) -\frac{1+a^{2}}{\sqrt{a}}\kappa -\kappa ^{2} \\
&=\frac{1}{2}\left( 1+\left( 1+\frac{1}{\epsilon }\right) \kappa ^{2}+\sqrt{%
\left( 1+\left( 1+\frac{1}{\epsilon }\right) \kappa ^{2}\right) ^{2}-4\kappa
^{2}}\right) -\frac{1}{\epsilon }\kappa ^{2}-\frac{\epsilon ^{3}}{\kappa ^{2}%
}-\kappa ^{2} \\
&=\frac{1}{2}-\frac{\epsilon ^{3}}{\kappa ^{2}}+\left\{ \frac{1}{2}\left( 1+%
\frac{1}{\epsilon }\right) +\frac{1}{2}\sqrt{\left( \frac{1}{\kappa ^{2}}%
+\left( 1+\frac{1}{\epsilon }\right) \right) ^{2}-\frac{4}{\kappa ^{2}}}%
-\left( 1+\frac{1}{\epsilon }\right) \right\} \kappa ^{2} \\
&=\frac{1}{2}-\frac{\epsilon ^{3}}{\kappa ^{2}}+\left\{ \frac{1}{2}\left( 1+%
\frac{1}{\epsilon }\right) +\frac{1}{2}\sqrt{\left( 1+\frac{1}{\epsilon }%
\right) ^{2}+\left( \frac{2}{\epsilon }-2\right) \frac{1}{\kappa ^{2}}+\frac{%
1}{\kappa ^{4}}}-\left( 1+\frac{1}{\epsilon }\right) \right\} \kappa ^{2} \\
&>\frac{1}{2}-\frac{\epsilon ^{3}}{\kappa ^{2}}\geq \frac{1}{2}-\epsilon
^{3}\geq 0,
\end{align*}%
where the first strict inequality holds since
\[
\sqrt{\left( 1+\frac{1}{\epsilon }\right) ^{2}+\left( \frac{2}{\epsilon }%
-2\right) \frac{1}{\kappa ^{2}}+\frac{1}{\kappa ^{4}}}>\sqrt{\left( 1+\frac{1%
}{\epsilon }\right) ^{2}}=1+\frac{1}{\epsilon }.
\]
Analogously
one can show that $h_{1}\left( a,\kappa \right) > \kappa^2(1/2 - \epsilon^3 \kappa^2)>0$
for $\kappa <1$.
Since $h_{1}\left( a,\kappa \right)>0$  and $h_{3}\left( \kappa \right) >0$,
the quadratic equation $h_{1}\left(
a,\kappa \right) c^{2}-2h_{2}\left( a,\kappa \right) c-h_{3}\left( \kappa
\right) =0$ has two
real roots, the larger of which reads
\begin{eqnarray*}
c_{\rm bnd}\left( a,\kappa \right)  &=&\frac{h_{2}\left( a,\kappa \right) +\sqrt{%
h_{2}\left( a,\kappa \right) ^{2}+h_{1}\left( a,\kappa \right) h_{3}\left(
\kappa \right) }}{h_{1}\left( a,\kappa \right) },
\end{eqnarray*}
and we have $f_2(\kappa,a,c) - f_1(\kappa,a,c) > 0 $
if $c>c_{\rm bnd}\left( a,\kappa \right)$. Since $\sqrt{x+y} \leq \sqrt{x}+\sqrt{y}$ for all
positive numbers $x,y$, and
  $ h_{1}\left( a,\kappa \right) h_{3}\left(   \kappa \right) \leq 2 h_{2}\left( a,\kappa \right)^2$ we have
\begin{align*}
   c_{\rm bnd}\left( a,\kappa \right)  & \leq
      \frac{2 h_{2}\left( a,\kappa \right) +
      \sqrt{ h_{1}\left( a,\kappa \right) h_{3}\left(
   \kappa \right) }}{h_{1}\left( a,\kappa \right) }
   \leq  (2+\sqrt{2})  \frac{ h_{2}\left( a,\kappa \right)  } {h_{1}\left( a,\kappa \right)}
   .
\end{align*}
Above we have already shown the lower bound
$h_{1}\left( a,\kappa \right) >   \min(1,\kappa^2) [1/2 - \epsilon^3 \min(\kappa^2,\kappa^{-2})]
 =   \min(1,\kappa^2) [1/2 - a^{3/2} \max(\kappa,\kappa^{-1})]$.
 In addition, we have $g(a,\kappa) < 3 \max(1,\kappa^2)/\sqrt{a}$
 and therefore
 $h_{2}\left( a,\kappa \right)  <
 \max (1,\kappa) \left( 1+\kappa \right)
 (1+ \sqrt{3} a^{-1/4} )$. Thus,
 \begin{align}
       c_{\rm bnd}\left( a,\kappa \right)  & <
         \frac{ (2+\sqrt{2})  \left( 1+\kappa \right)
 (1+ \sqrt{3} a^{-1/4} )}   { \min(1,\kappa) [1/2 - a^{3/2} \max(\kappa,\kappa^{-1})]} .
 \end{align}
Our assumptions guarantee that $c$ is larger or equal to the rhs of the last inequality,
i.e. also $c>c_{\rm bnd}\left( a,\kappa \right)$ and $f_2(\kappa,a,c) - f_1(\kappa,a,c) > 0 $.
\end{proof}

\begin{proof}[\bf \underline{Proof of Theorem~\ref{th:ConvergenceRate2} ($N^{3/4}$ Convergence Rate of $\widehat \beta_{R^0}$)}]
    The result follows from Theorem~\ref{th:ConvergenceRate}
     and   Lemma~\ref{lemma:JustifyHL1} below.
\end{proof}

\begin{proof}[\bf \underline{Proof of Theorem~\ref{th:LimitingDistribution} (Asymptotic Equivalence of $\widehat \beta_{R^0}$ and $\widehat \beta_R$, $R>R^0$)}]
   The result follows  from Corollary~\ref{cor:LimitRgen2}
     and   Lemmas~\ref{lemma:JustifyHL1} and \ref{lemma:JustifyHL2} below.
\end{proof}

\begin{proof}[\bf \underline{Proof of Lemma~\ref{lemma:JustifyEV} (Justification of Main Text High-Level Assumptions)}]
    See Section~\ref{ass:SufficiencyLL} below.
\end{proof}

\begin{proof}[\bf \underline{Proof of Lemma~\ref{lemma:SpectralNormBound} (Spectral Norm Bound for Random Matrices)}]
   Let $\Sigma$, $\eta$, $\Psi$, $\chi$ be the $N \times N$ matrices with entries $\Sigma_{ij}$,
     $\eta_{ij}$, $\Psi_{ij}$ and $\chi_{ij}$, respectively.
     Assumption $(ii)$ of the Lemma guarantees that
    \begin{align}
          \mathbbm{E} \| \eta \|^2_{HS}  &=  \sum_{i,j=1}^n \mathbbm{E}( \eta^2_{ij})  = {\cal O}(N^2),
     \end{align}
      from which we conclude that $\|\eta\|_{HS} = {\cal O}_P(N)$.
      Analogously, we find that assumption $(iv)$ of the Lemma implies $\|\chi\|_{HS} = {\cal O}_P(N)$.
     Furthermore, assumption $(iii)$ of the Lemma guarantees that
     $ \| \Psi \|_{HS}  = \sqrt{ \sum_{i,j=1}^n \Psi^2_{ij} } = {\cal O}_P(N^{1/2})$.
     Since $\eta^2 =  N \Psi + N^{1/2} \chi$ we thus have
     \begin{align}
          \| \eta^2 \|_{HS} &=  \| N \Psi + N^{1/2} \chi \|_{HS}
              \leq  N \| \Psi \|_{HS} + N ^{1/2} \| \chi \|_{HS}
              = {\cal O}_P(N^{3/2}).
     \end{align}  
     Since $\Sigma$ is a symmetric positive definite matrix we have
     $ \| \Sigma \| =  \mu_1(\Sigma)$,
     i.e. by assumption $(i)$ of the Lemma we have $\|\Sigma\| = {\cal O}(1)$.

     Using the above results on $\|\eta\|_{HS}$, $\| \eta^2 \|_{HS}$
     and $\| \Sigma \|$, and the fact that $e e' = T \Sigma+ T^{1/2} \eta$, we obtain
     \begin{align}
          \| e \|^4 &= \| (e e')^2 \| \leq \| (e e')^2  \|_{HS} = \| (T \Sigma+ T^{1/2} \eta)^2 \|_{HS}
          \nonumber \\
          &\leq T^2 \| \Sigma^2 \|_{HS} + 2 T^{3/2} \| \Sigma \eta \|_{HS} + T \| \eta^2 \|_{HS}
          \nonumber \\
          &\leq T^{2} N^{1/2} \| \Sigma \|^2 + 2 T^{3/2} \| \Sigma \| \| \eta \|_{HS} 
             +  T \| \eta^2 \|_{HS} 
          \nonumber \\
           &= {\cal O}_P\left( T^{2} N^{1/2} +  T^{3/2} N  +  T N^{3/2} \right) = {\cal O}_P(N^{5/2}), 
     \end{align}
     where in the
     second to last line we applied the general matrix norm inequalities $\| A\|_{HS} \leq {\rm rank}(A) \| A\|$
     and $\| C D\|_{HS} \leq \|C\| \|D\|_{HS}$ with $A=\Sigma^2$, $C=\Sigma$ and $D=\eta$.  
     We thus conclude that $\|e\| = {\cal O}_P(N^{5/8}) $.   
\end{proof}

\section{Details for Quadratic Approximation of ${\cal L}_{NT}^0(\beta)$}
\label{app:expansion1}

The following extends the discussion in Section~\ref{sec:expansion}
and Appendix~\ref{app:ExpansionDiscussion} of the main paper.
Using the perturbation theory of linear operators
we provide an asymptotic expansion of the least squares objective function
${\cal L}_{NT}^0(\beta)$ when $R=R^0$. 
Lemma~\ref{lemma:expansion} is the key result of this section, which is afterwards used
to show Theorem~\ref{th:expansion}.
The proofs for the intermediate results of this section
are provided in Section~\ref{app:Proofs-Supp} below.

This section is only concerned with $R=R^0$, in which case we write $\widehat \lambda := \widehat \Lambda_{R^0}$
and $\widehat f = \widehat F_{R^0}$. It is also convenient to define $\widehat \lambda(\beta)$ and
$\widehat f(\beta)$ as the minimizers of the LS objective for different values of $\beta$.
We have $\widehat \lambda = \widehat \lambda(\widehat \beta_{R^0})$
and $\widehat f = \widehat f(\widehat \beta_{R^0})$.
Finally, we define $M_{\widehat \lambda}(\beta) := M_{\widehat \lambda(\beta)}$
and  $M_{\widehat f}(\beta) := M_{\widehat f(\beta)}$,
and the residuals $\widehat e(\beta) := Y - \beta \cdot X - \widehat \lambda(\beta) \widehat f'(\beta)$.

\subsection{General Expansion Result and Proof of Theorem \ref{th:expansion} }

\begin{definition}
For the $N\times R^0$ matrix $\lambda^0$ and the $T\times R^0$ matrix $f^0$ we
define
\begin{align}
d_{\max}(\lambda^0,f^0) &= \frac{1} {\sqrt{NT}} \, \left\| \lambda^0
f^{0\prime} \right\| =  \sqrt{ \frac{1} {NT} \, \mu%
_1(\lambda^{0\prime} f^0 f^{0\prime} \lambda^0) } \; ,  \notag \\
d_{\min}(\lambda^0,f^0) &=  \, \sqrt{ \frac{1} {NT} \, \mu%
_{R^0} (\lambda^{0\prime} f^0 f^{0\prime} \lambda^0) } \; ,
\end{align}
\textit{i.e.} $d_{\max}(\lambda^0,f^0)$ and $d_{\min}(\lambda^0,f^0)$ are the square
roots of the maximal and the minimal eigenvalue of $\lambda^{0\prime} f^0
f^{0\prime} \lambda^0/NT$. Furthermore, the convergence radius $%
r_0(\lambda^0, f^0)$ is defined by
\begin{align}
r_0(\lambda^0, f^0) &= \left( \frac{4 d_{\max}(\lambda^0,f^0)} {%
d_{\min}^2(\lambda^0,f^0)} + \frac{1} {2 d_{\max}(\lambda^0,f^0)}
\right)^{-1} \; .  \label{Defr0}
\end{align}
\end{definition}

\begin{lemma}
\label{lemma:expansion} If the following condition is satisfies
\begin{align}
\sum_{k=1}^{K}\left| \beta_{k} - \beta^0_{k} \right| \frac{\| X_{k} \|} {%
\sqrt{NT}} + \frac{\|e\|} {\sqrt{NT}} \, &< \, r_0(\lambda^0,f^0) \; ,
\label{exp_bound}
\end{align}
then

\begin{itemize}
\item[(i)] the profile least squares objective function can be written as a power
series in the $K+1$ parameters $\epsilon_0 = \|e\|/\sqrt{NT}$ and $%
\epsilon_k = \beta^0_{k} - \beta_{k}$, namely
\begin{align*}
{\cal L}^0_{NT}\left( \beta \right) &= \frac{1} {NT} \, \sum_{g=2}^\infty \,
\sum_{k_1=0}^K \, \sum_{k_2=0}^K \, \ldots \sum_{k_g=0}^K \,
\epsilon_{k_1} \, \epsilon_{k_2} \, \ldots \, \epsilon_{k_g}
\, L^{(g)}\left(\lambda^0,\, f^0,\, X_{k_1},\, X_{k_2}, \ldots
,X_{k_g}\right) ,
\end{align*}
where the expansion coefficients are given by\footnote{
Here we use the round bracket notation $(k_1,k_2,\ldots,k_g)$ for total symmetrization
of these indices, e.g.
$\widetilde L^{(2)}\left(\lambda^0,\, f^0,\, X_{(k_1},\, X_{k_2)}\right)
 = \frac 1 2 
 \left[ \widetilde L^{(2)}\left(\lambda^0,\, f^0,\, X_{k_1},\, X_{k_2}\right)
   + \widetilde L^{(2)}\left(\lambda^0,\, f^0,\, X_{k_2},\, X_{k_1}\right) \right]$.
}
\begin{align*}
L^{(g)}&\left(\lambda^0,\, f^0,\, X_{k_1},\, X_{k_2}, \ldots
,X_{k_g}\right) = \widetilde L^{(g)}\left(\lambda^0,\, f^0,\,
X_{(k_1},\, X_{k_2}, \ldots ,X_{k_g)}\right)
 \nonumber \\
& \qquad  = \frac 1 {g!} \left[ \widetilde L^{(g)}\left(\lambda^0,\, f^0,\,X_{k_1},\,
X_{k_2}, \ldots ,X_{k_g}\right) + \text{all permutations of $%
k_1,\ldots,k_g$} \right] \; ,
\end{align*}
\textit{i.e.} $L^{(g)}$ is obtained by total symmetrization of the last $g$ arguments
of~\footnote{
  One finds $\widetilde L^{(1)} \left(\lambda^0,\, f^0,\, X_{k_1},\, X_{k_2},
  \ldots ,X_{k_g}\right) = 0$, which is why the sum in the power series of ${\cal L}^0_{NT}$
  starts from $g=2$ instead of $g=1$. For $g=2$ and $g=3$ we have
  \begin{align*}
L^{(2)}\left(\lambda^0,\, f^0,\, X_{k_1},\, X_{k_2} \right) &=
{\rm Tr}\left(M_{\lambda^0} \, X_{k_1} \, M_{f^0} \,
X_{k_2}\right) \; ,  \notag \\
L^{(3)}\left(\lambda^0,\, f^0,\, X_{k_1},\, X_{k_2},\,
X_{k_3} \right) &= - \, \frac{1}{3} \, \bigg[ {\rm Tr}\left(
M_{\lambda^0} \, X_{k_1} \, M_f \, X^{\prime}_{k_2} \, \lambda^0
\, (\lambda^{0\prime}\lambda^0)^{-1} \, (f^{0\prime}f^0)^{-1} \, f^{0\prime}
\, X^{\prime}_{k_3}\right)  \notag \\
& \qquad \qquad \qquad + \text{5 permutations of $k_1 \ldots k_3$} %
\bigg]  .
\end{align*}

}
\begin{align*}
\widetilde L^{(g)} & \left(\lambda^0,\, f^0,\, X_{k_1},\, X_{k_2},
\ldots ,X_{k_g}\right)  \notag \\
&   = \, \sum_{p=1}^g \, (-1)^{p+1} \, \sum_{%
\begin{minipage}{3cm}\center\scriptsize $\nu_1+\ldots+\nu_P=g$ \\
$m_1+\ldots+m_{p+1}=p-1$\\$2 \geq \nu_j \geq 1$ ,\, $m_j \geq
0$\end{minipage}} {\rm Tr}\left( S^{(m_1)} \, {\cal T}^{(\nu_1)}_{k_1
\ldots} \, S^{(m_2)} \, \ldots \, S^{(m_P)} \, {\cal T}^{(\nu_P)}_{\ldots k_g}
\, S^{(m_{p+1})} \right) \; ,
\end{align*}
with
\begin{align*}
S^{(0)} &= - M_{\lambda^0} \; , & S^{(m)} &= \left[ \lambda^0
(\lambda^{0\prime}\lambda^0)^{-1} (f^{0\prime}f^0)^{-1}
(\lambda^{0\prime}\lambda^0)^{-1} \lambda^{0\prime} \right]^{m} \; , \;
\text{for $m\geq 1$},  \notag \\
{\cal T}^{(1)}_{k} &= \lambda^0 \, f^{0\prime} \, X_k^{\prime}+ X_k
\, f^0 \, \lambda^{0\prime} \; , & {\cal T}^{(2)}_{k_1k_2} &=
X_{k_1} \, X_{k_2}^{\prime}\; , \qquad \text{for $%
k,k_1,k_2 = 0\ldots K$} \; , \notag \\
X_0 &= \frac{\sqrt{NT}} {\|e\|} \; e \; , & X_k &= X_k \; , \qquad
\text{for $k =  1\ldots K$} \; .
\end{align*}

\item[(ii)] the projector $M_{\widehat\lambda}(\beta)$ can be written as a power
series in the same parameters $\epsilon_k$ ($k=0,\ldots,K$),
namely
\begin{align*}
M_{\widehat\lambda} \left( \beta \right) &= \sum_{g=0}^\infty \,
\sum_{k_1=0}^K \, \sum_{k_2=0}^K \, \ldots \sum_{k_g=0}^K \,
\epsilon_{k_1} \, \epsilon_{k_2} \, \ldots \, \epsilon_{k_g}
\, M^{(g)}\left(\lambda^0,\, f^0,\, X_{k_1},\, X_{k_2}, \ldots
,X_{k_g}\right) \; ,
\end{align*}
where the expansion coefficients are given by $M^{(0)}(\lambda^0,\,
f^0)=M_{\lambda^0}$, and for $g\geq1$
\begin{align*}
M^{(g)}&\left(\lambda^0,\, f^0,\, X_{k_1},\, X_{k_2}, \ldots
,X_{k_g}\right) = \widetilde M^{(g)}\left(\lambda^0,\, f^0,\,
X_{(k_1},\, X_{k_2}, \ldots ,X_{k_g)}\right)  \notag \\
& \qquad = \frac 1 {g!} \left[ \widetilde M^{(g)}\left(X_{k_1},\,
X_{k_2}, \ldots ,X_{k_g}\right) + \text{all permutations of $%
k_1,\ldots,k_g$} \right] \; ,
\end{align*}
\textit{i.e.} $M^{(g)}$ is obtained by total symmetrization of the last $g$ arguments
of
\begin{align*}
\widetilde M^{(g)} & \left(\lambda^0,\, f^0,\, X_{k_1},\, X_{k_2},
\ldots ,X_{k_g}\right)  \\
&  = \, \sum_{p=1}^g \, (-1)^{p+1} \, \sum_{%
\begin{minipage}{3cm}\center\scriptsize $\nu_1+\ldots+\nu_P=g$ \\
$m_1+\ldots+m_{p+1}=p$\\$2 \geq \nu_j \geq 1$ ,\, $m_j \geq 0$\end{minipage}%
} S^{(m_1)} \, {\cal T}^{(\nu_1)}_{k_1 \ldots} \, S^{(m_2)} \, \ldots \,
S^{(m_P)} \, {\cal T}^{(\nu_P)}_{\ldots k_g} \, S^{(m_{p+1})} \; ,
\end{align*}
where $S^{(m)}$, ${\cal T}^{(1)}_{k}$, ${\cal T}^{(2)}_{k_1k_2}$, and
$X_k$ are given above.

\item[(iii)] For $g \geq 3$ the coefficients $L^{(g)}$ in the series expansion of $%
{\cal L}^0_{NT}(\beta)$ are bounded as follows
\begin{align*}
& \frac {1} {NT} \left| L^{(g)}\left(\lambda^0,\, f^0,\, X_{k_1},\,
X_{k_2}, \ldots ,X_{k_g}\right) \right|  \notag \\
& \qquad \qquad \qquad \leq \frac{R \, g \, d_{\min}^2(\lambda^0,f^0)} {2}
\, \left(\frac{16 \; d_{\max}(\lambda^0,f^0)} {d_{\min}^2(\lambda^0,f^0)}
\right)^{g} \frac{\|X_{k_1}\|} {\sqrt{NT}} \, \frac{\|X_{k_2}\|} {%
\sqrt{NT}} \, \ldots \, \frac{\|X_{k_g}\|} {\sqrt{NT}}  \, .
\end{align*}
Under the stronger condition
\begin{align}
\sum_{k=1}^{K}\left| \beta^0_{k} - \beta_{k} \right| \frac{\| X_{k} \|} {%
\sqrt{NT}} + \frac{\|e\|} {\sqrt{NT}} \, &< \, \frac{d_{\min}^2(%
\lambda^0,f^0)} {16 \; d_{\max}(\lambda^0,f^0)} \; ,
\label{StrongConRad}
\end{align}
we therefore have the following bound on the remainder when the series
expansion for ${\cal L}^0_{NT}(\beta)$ is truncated at order $G \geq 2$:
\begin{align*}
\bigg| {\cal L}^0_{NT}\left( \beta \right) - \frac 1 {NT} \, \sum_{g=2}^G \,
\sum_{k_1=0}^K \, \ldots \, \sum_{k_g=0}^K \, \epsilon_{k_1}
\, \ldots \, \epsilon_{k_g} \, L^{(g)} & \left(\lambda^0,\, f^0,\,
X_{k_1},\, X_{k_2}, \ldots ,X_{k_g}\right) \bigg|  \\
& \quad \leq \frac{R \, (G+1)\, \alpha^{G+1} \, d_{\min}^2(\lambda^0,f^0)} {%
2 (1-\alpha)^2} \; ,
\end{align*}
where
\begin{align*}
\alpha \, &= \frac{16 \; d_{\max}(\lambda^0,f^0)} {d_{\min}^2(\lambda^0,f^0)}
\, \left( \sum_{k=1}^{K}\left| \beta^0_{k} - \beta_{k} \right| \frac{\|
X_{k} \|} {\sqrt{NT}} + \frac{\|e\|} {\sqrt{NT}} \right) \; < 1 \; .
\end{align*}

\item[(iv)]

The operator norm of the coefficient $M^{(g)}$ in the series expansion of
$M_{\widehat\lambda} \left( \beta \right)$ is bounded as follows, for $g\geq 1$
\begin{align*}
   \left\| M^{(g)} \left(\lambda^0, f^0, X_{k_1}, X_{k_2}, \ldots,X_{k_g}\right) \right\|
        \leq \frac{g} {2}
             \left(\frac{16 \; d_{\max}(\lambda^0,f^0)} {d_{\min}^2(\lambda^0,f^0)} \right)^{g}
               \frac{\|X_{k_1}\|} {\sqrt{NT}} \,
               \frac{\|X_{k_2}\|} {\sqrt{NT}} \, \ldots \,
               \frac{\|X_{k_g}\|} {\sqrt{NT}} \; .
\end{align*}
Under the condition \eqref{StrongConRad} we therefore have the following bound on operator norm of
the remainder of the series expansion of $M_{\widehat\lambda} \left( \beta \right)$, for $G\geq 0$
\begin{align*}
        \bigg\| M_{\widehat\lambda} \left( \beta \right) \, - \, \sum_{g=0}^G \, \sum_{k_1=0}^K \,
                                       \ldots \, \sum_{k_g=0}^K \,
                                     \epsilon_{k_1} \, \ldots \, \epsilon_{k_g} \,
         & M^{(g)}\left(\lambda^0,\, f^0,\, X_{k_1},\, X_{k_2}, \ldots ,X_{k_g}\right)  \bigg\|
      \nonumber \\
         & \qquad \qquad\qquad\qquad\qquad \leq  \frac{(G+1)\, \alpha^{G+1}} {2 (1-\alpha)^2} \; .
\end{align*}

\end{itemize}
\end{lemma}

\begin{proof}[\bf \underline{Proof of Theorem \ref{th:expansion} (Quadratic Approximation of ${\cal L}_{NT}^0(\beta)$)}]
The $R^0$ non-zero eigenvalues of the matrix
$\lambda^{0\prime} f^0 f^{0\prime} \lambda^0/NT$
are identical to the eigenvalues of the $R^0 \times R^0$ matrix
$(f^0 f^{0\prime}/T)^{-1/2}
(\lambda^0 \lambda^{0\prime}/N) (f^0 f^{0\prime}/T)^{-1/2}$,
and Assumption~\ref{ass:SF} guarantees that these eigenvalues,
including $d_{\max}(\lambda^0,f^0)$ and $d_{\min}(\lambda^0,f^0)$
converge to positive constants in probability. Therefore,
also $r_0(\lambda^0, f^0)$ converges to a positive constant in probability.

Assumptions \ref{ass:SF} and \ref{ass:SN}
furthermore imply that in the limit
$N,T \rightarrow \infty$ with $N/T \rightarrow \kappa^2$, $0<\kappa<\infty$, we have
\begin{align}
\frac{\| \lambda^0 \|} {\sqrt{N}} &= \mathcal{O}_P(1) \; , & \frac{\| f^0 \|%
} {\sqrt{T}} &= \mathcal{O}_P(1) \; , &
\left\| \left(\frac{ \lambda^{0\prime} \lambda^0 } {N} \right)^{-1} \right\|
&= \mathcal{O}_P(1) \; , & \left\| \left(\frac{ f^{0\prime} f^0} {T}
\right)^{-1} \right\| &= \mathcal{O}_P(1) \; , &   \notag \\
\frac{\| X_k \|} {\sqrt{NT}} &= \mathcal{O}_P(1) \; , & \frac{\| e \|} {%
\sqrt{NT}} &= {\cal O}_P\left( N^{-1/2} \right) \; . &
\end{align}
Thus, for $\left\| \beta -\beta ^{0}\right\| \leq c_{NT}$,  $c_{NT}=o(1)$,
we have $\alpha \rightarrow 0$ as $N,T \rightarrow \infty$, i.e. the
condition \eqref{StrongConRad} in part (iii) of Lemma~\ref{lemma:expansion} is asymptotically satisfied, and by applying the Lemma we find
\begin{align}
\frac 1 {NT} (\epsilon_0)^{g-r}  L^{(g)} \left(\lambda^0, f^0,
X_{k_1},\ldots, X_{k_r}, X_0, \ldots, X_{0}\right)
 &= \mathcal{O}_P\left( \left( \frac{\|e\|}{\sqrt{NT}} \right)^{g-r} \right)
  = \mathcal{O}_P\left( N^{-\, \frac{g-r} 2} \right) ,
  \label{boundExpTerms}
\end{align}
where we used $\epsilon_0 X_0 = e$ and the linearity of
$L^{(g)}
\left(\lambda^0,\, f^0,\, X_{k_1},\, X_{k_2}, \ldots
,X_{k_g}\right)$
in the arguments $X_k$. Truncating the expansion of
${\cal L}^0_{NT}(\beta)$ at order $G=3$ and applying the corresponding
result in Lemma~\ref{lemma:expansion}(iii) we obtain
\begin{align}
 {\cal L}^0_{NT}(\beta) =& \frac 1 {NT}
\sum_{k_1,k_2=0}^K \epsilon_{k_1}
 \epsilon_{k_2}  L^{(2)} \left(\lambda^0, f^0,
X_{k_1}, X_{k_2} \right)
 \nonumber \\
  & \qquad + \frac 1 {NT}
\sum_{k_1,k_2,k_3=0}^K \epsilon_{k_1}
 \epsilon_{k_2}  \epsilon_{k_3}  L^{(3)} \left(\lambda^0, f^0,
X_{k_1}, X_{k_2}, X_{k_3}\right)
   + \mathcal{O}_P\left( \alpha^{4} \right)
   \nonumber \\
      =& {\cal L}_{NT}^{0}(\beta^0)
                          \, - \, \frac 2 {\sqrt{NT}}  \, \left(\beta-\beta^0 \right)'
                          \, \left( C^{(1)} + C^{(2)} \right)
                \nonumber \\ & \qquad\qquad\quad
                    + \left(\beta-\beta^0 \right)' \, W \, \left(\beta-\beta^0 \right)
                          + {\cal L}_{NT}^{0,{\rm rem}}(\beta) \; ,
\end{align}
where, using \eqref{boundExpTerms} we find
\begin{align}
   {\cal L}_{NT}^{0,{\rm rem}}(\beta)
   =&  \frac 3 {NT}
\sum_{k_1,k_2=1}^K \epsilon_{k_1}
 \epsilon_{k_2}  \epsilon_{0}  L^{(3)} \left(\lambda^0, f^0,
X_{k_1}, X_{k_2}, X_{0}\right)
 \nonumber \\ & \quad
  +  \frac 1 {NT}
\sum_{k_1,k_2,k_3=1}^K \epsilon_{k_1}
 \epsilon_{k_2}  \epsilon_{k_3}  L^{(3)} \left(\lambda^0, f^0,
X_{k_1}, X_{k_2}, X_{k_3}\right)
 \nonumber \\ & \quad \quad
   + \mathcal{O}_P\left[ \left( \sum_{k=1}^{K}\left| \beta^0_{k} - \beta_{k}
\right| \frac{\| X_{k} \|} {\sqrt{NT}} + \frac{\|e\|} {\sqrt{NT}}
\right)^{4} \right]
   - \mathcal{O}_P\left[ \left(  \frac{\|e\|} {\sqrt{NT}}
\right)^{4} \right]
  \nonumber \\
    =& {\cal O}_P\left( \|\beta-\beta^0\|^2 N^{-1/2} \right)
      +{\cal O}_P\left( \|\beta-\beta^0\|^3 \right)
      +{\cal O}_P\left( \|\beta-\beta^0\| N^{-3/2} \right)
    \nonumber \\ & \;
      +{\cal O}_P\left( \|\beta-\beta^0\|^2 N^{-1} \right)
      +{\cal O}_P\left( \|\beta-\beta^0\|^3 N^{-1/2} \right)
      +{\cal O}_P\left( \|\beta-\beta^0\|^4 \right) \; .
  \label{bound_remainder}
\end{align}
Here
$\mathcal{O}_P\left[ \left(  \frac{\|e\|} {\sqrt{NT}}
\right)^{4} \right]$ is not just some term of that order, but exactly
the term of that order contained in
${\cal O}_P(\alpha^4) =
\mathcal{O}_P\left[ \left( \sum_{k=1}^{K}\left| \beta^0_{k} - \beta_{k}
\right| \frac{\| X_{k} \|} {\sqrt{NT}} + \frac{\|e\|} {\sqrt{NT}}
\right)^{4} \right]$. This term is not present in ${\cal L}_{NT}^{0,{\rm rem}}(\beta)$
since it is already contained in ${\cal L}_{NT}^{0}(\beta^0)$.\footnote{
  Alternatively, we could have truncated the expansion at order $G=4$.
  Then, the term $\mathcal{O}_P\left[ \left(  \frac{\|e\|} {\sqrt{NT}}
\right)^{4} \right]$ would be more explicit, namely it would equal
  $\frac 1 {NT}
  \epsilon_{0}^4 L^{(4)} \left(\lambda^0,f^0,X_{0},X_{0},X_{0},X_{0}\right)$,
  which is clearly contained in ${\cal L}_{NT}^{0}(\beta^0)$.
}
Equation \eqref{bound_remainder} shows that the remainder satisfies the bound
stated in the theorem, which concludes the proof.
\end{proof}

\subsection{Expansion of Other Quantities}

\begin{lemma}
   \label{lemma:ProExp}
   Define the pseudo-inverses
    $(\lambda^0 f^{0 \prime} )^\dagger \equiv
   f^0 \, (f^{0\prime}f^0)^{-1} \, (\lambda^{0\prime}\lambda^0)^{-1} \lambda^{0\prime}$
   and $(f^0 \lambda^{0 \prime} )^\dagger \equiv
   \lambda^0 \, (\lambda^{0\prime}\lambda^0)^{-1} \, (f^{0\prime}f^0)^{-1} \, f^{0\prime}$.
   Under the assumptions of Theorem \ref{th:expansion} we have
   \begin{align*}
     M_{\widehat \lambda}(\beta) &= M_{\lambda^0} + M_{\widehat \lambda,e}^{(1)}
                                     + M_{\widehat \lambda,e}^{(2)}
                                     - \sum_{k=1}^K \left( \beta_k - \beta^0_k \right)  M_{\widehat \lambda,X,k}^{(1)}
                                     + M_{\widehat \lambda}^{({\rm rem})}(\beta) \; ,
    \nonumber \\
     M_{\widehat f}(\beta) &= M_{f^0} + M_{\widehat f,e}^{(1)}
                                     + M_{\widehat f,e}^{(2)}
                                     - \sum_{k=1}^K \left( \beta_k - \beta^0_k \right)  M_{\widehat f,X,k}^{(1)}
                                     + M_{\widehat f}^{({\rm rem})}(\beta) \; ,
   \end{align*}
  where the expansion coefficients in the expansion
  of $M_{\widehat \lambda}(\beta)$ are $N\times N$ matrices given by
  \begin{align*}
     M^{(1)}_{\widehat \lambda,e} &= - \, M_{\lambda^0} \, e \,
      (\lambda^0 f^{0 \prime} )^\dagger
       \, - \, (f^0 \lambda^{0 \prime} )^\dagger
              \, e' \, M_{\lambda^0}  \; ,
    \nonumber \\
   M^{(1)}_{\widehat \lambda,X,k} &= - \, M_{\lambda^0} \, X_k  \,
      (\lambda^0 f^{0 \prime} )^\dagger
       \, - \, (f^0 \lambda^{0 \prime} )^\dagger
              \,  X'_k  \, M_{\lambda^0}  \; ,
    \nonumber \\
   M^{(2)}_{\widehat \lambda,e} &=
        M_{\lambda^0} \, e \, (\lambda^0 f^{0 \prime} )^\dagger
             \, e \, (\lambda^0 f^{0 \prime} )^\dagger
       +(f^0 \lambda^{0 \prime} )^\dagger
             \, e' \, (f^0 \lambda^{0 \prime} )^\dagger
             \, e' \, M_{\lambda^0}
       \nonumber \\ & \qquad
          - M_{\lambda^0} \, e \, M_{f^0} \, e' \,
(f^0 \lambda^{0 \prime} )^\dagger \, (\lambda^0 f^{0 \prime} )^\dagger
-(f^0 \lambda^{0 \prime} )^\dagger \, (\lambda^0 f^{0 \prime} )^\dagger
   \, e \, M_{f^0} \, e' \, M_{\lambda^0}
       \nonumber \\ & \qquad
          - M_{\lambda^0} \, e \,
      (\lambda^0 f^{0 \prime} )^\dagger \, (f^0 \lambda^{0 \prime} )^\dagger
         \, e' \, M_{\lambda^0}
        + (f^0 \lambda^{0 \prime} )^\dagger
            \, e' \, M_{\lambda^0} \, e \,
              (\lambda^0 f^{0 \prime} )^\dagger \, ,
  \end{align*}
  and analogously we have $T\times T$ matrices
  \begin{align*}
   M^{(1)}_{\widehat f,e} &= \, - \, M_{f^0} \, e' \,
      (f^0 \lambda^{0 \prime} )^\dagger
       \, - \, (\lambda^0 f^{0 \prime} )^\dagger
              \, e \, M_{f^0}   \; ,
    \nonumber \\
  M^{(1)}_{\widehat f,X,k} &=
       \, - \, M_{f^0} \,  X'_k  \,
      (f^0 \lambda^{0 \prime} )^\dagger
       \, - \, (\lambda^0 f^{0 \prime} )^\dagger
              \,  X_k \,M_{f^0}  \; ,
    \nonumber \\
  M^{(2)}_{\widehat f,e} &=
       M_{f^0} \, e' \, (f^0 \lambda^{0 \prime} )^\dagger
             \, e' \, (f^0 \lambda^{0 \prime} )^\dagger
       +(\lambda^0 f^{0 \prime} )^\dagger
             \, e \, (\lambda^0 f^{0 \prime} )^\dagger
             \, e \, M_{f^0}
       \nonumber \\ & \qquad
          - M_{f^0} \, e' \, M_{\lambda^0} \, e \,
(\lambda^0 f^{0 \prime} )^\dagger \, (f^0 \lambda^{0 \prime} )^\dagger
-(\lambda^0 f^{0 \prime} )^\dagger \, (f^0 \lambda^{0 \prime} )^\dagger
   \, e' \, M_{\lambda^0} \, e \, M_{f^0}
       \nonumber \\ & \qquad
          - M_{f^0} \, e' \,
      (f^0 \lambda^{0 \prime} )^\dagger \, (\lambda^0 f^{0 \prime} )^\dagger
         \, e \, M_{f^0}
        + (\lambda^0 f^{0 \prime} )^\dagger
            \, e \, M_{f^0} \, e' \,
              (f^0 \lambda^{0 \prime} )^\dagger \, .
  \end{align*}
  Finally, the remainder terms of the expansions satisfy
  for any sequence
$c_{NT}\rightarrow 0$
\begin{align*}
     \sup_{\{\beta :\left\| \beta -\beta^{0} \right\| \leq c_{NT}\}}
        \frac{\left\| M_{\widehat \lambda}^{({\rm rem})}(\beta) \right\|}
        { \|\beta - \beta^0\|^2 + N^{-1/2}  \|\beta - \beta^0\| \,  \,  + N^{-3/2} }
            &= {\cal O}_P\left(1\right) \, ,
     \nonumber \\
     \sup_{\{\beta :\left\| \beta -\beta^{0} \right\| \leq c_{NT}\}}
        \frac{\left\| M_{\widehat f}^{({\rm rem})}(\beta) \right\|}
        { \|\beta - \beta^0\|^2 + N^{-1/2} \, \|\beta - \beta^0\| \,  \,  + N^{-3/2}}
            &= {\cal O}_P\left(1\right) \, .
\end{align*}
\end{lemma}

\begin{lemma}
   \label{lemma:ResExp}
   Let 
    $(\lambda^0 f^{0 \prime} )^\dagger$
    and  $(f^0 \lambda^{0 \prime} )^\dagger$
    as defined in Lemma~\ref{lemma:ProExp} above.
   Under the assumptions of Theorem \ref{th:expansion} we have
   \begin{align*}
      \widehat e(\beta) &=  M_{\lambda^0} \, e \, M_{f^0}
            + \widehat e^{(1)}_e
             + \widehat e^{(2)}_e
             - \sum_{k=1}^K \left( \beta_k - \beta^0_k \right)
             \left( \widehat e^{(1)}_{X,k} + \widehat e^{(2)}_{X,k} \right)
            + \widehat e^{({\rm rem})}(\beta) \; ,
   \end{align*}
  where the $N \times T$ matrix valued expansion coefficients read
  \begin{align*}
   \widehat e^{(1)}_{X,k} &= M_{\lambda^0} \, X_k \, M_{f^0} \; ,
             \nonumber \\
\widehat e^{(2)}_{X,k} &= - M_{\lambda^0}  X_k  M_{f^0}  e'
   (f^0 \lambda^{0 \prime} )^\dagger
           - M_{\lambda^0}  e  M_{f^0}  X_k'
   (f^0 \lambda^{0 \prime} )^\dagger
  - (f^0 \lambda^{0 \prime} )^\dagger
    X_k'  M_{\lambda^0}  e  M_{f^0}
             \nonumber \\ & \quad
  - (f^0 \lambda^{0 \prime} )^\dagger
    e'  M_{\lambda^0}  X_k  M_{f^0}
  - M_{\lambda^0}  X_k
      (\lambda^0 f^{0 \prime} )^\dagger
               e  M_{f^0}
  - M_{\lambda^0}  e
      (\lambda^0 f^{0 \prime} )^\dagger
              \, X_k \, M_{f^0} \; ,
             \nonumber \\
\widehat e^{(1)}_e &= - M_{\lambda^0} \, e \, M_{f^0} \, e' \,
   (f^0 \lambda^{0 \prime} )^\dagger
  - (f^0 \lambda^{0 \prime} )^\dagger
   \, e' \, M_{\lambda^0} \, e \, M_{f^0}
  - M_{\lambda^0} \, e \,
      (\lambda^0 f^{0 \prime} )^\dagger
              \, e \, M_{f^0} \; ,
             \nonumber \\
\widehat e^{(2)}_e &=
         M_{\lambda^0} e  M_{f^0} \, e' \, (f^0 \lambda^{0 \prime} )^\dagger
             \, e' \, (f^0 \lambda^{0 \prime} )^\dagger
          - M_{\lambda^0} e M_{f^0} \, e' \, M_{\lambda^0} \, e \,
(\lambda^0 f^{0 \prime} )^\dagger \, (f^0 \lambda^{0 \prime} )^\dagger
       \nonumber \\ & \quad
          - M_{\lambda^0} e M_{f^0} \, e' \,
      (f^0 \lambda^{0 \prime} )^\dagger \, (\lambda^0 f^{0 \prime} )^\dagger
         \, e \, M_{f^0}
       + \, M_{\lambda^0} \, e \,
      (\lambda^0 f^{0 \prime} )^\dagger e
          M_{f^0} \, e' \,
      (f^0 \lambda^{0 \prime} )^\dagger
       \nonumber \\ & \quad
         + \, (f^0 \lambda^{0 \prime} )^\dagger
              \, e' \, M_{\lambda^0}  e M_{f^0} \, e' \,
      (f^0 \lambda^{0 \prime} )^\dagger
       + \, M_{\lambda^0} \, e \,
      (\lambda^0 f^{0 \prime} )^\dagger e
         (\lambda^0 f^{0 \prime} )^\dagger
              \, e \, M_{f^0}
       \nonumber \\ & \quad
         + \, (f^0 \lambda^{0 \prime} )^\dagger
              \, e' \, M_{\lambda^0}  e (\lambda^0 f^{0 \prime} )^\dagger
              \, e \, M_{f^0}
       +(f^0 \lambda^{0 \prime} )^\dagger
             \, e' \, (f^0 \lambda^{0 \prime} )^\dagger
             \, e' \, M_{\lambda^0} e M_{f^0}
       \nonumber \\ & \quad
-(f^0 \lambda^{0 \prime} )^\dagger \, (\lambda^0 f^{0 \prime} )^\dagger
   \, e \, M_{f^0} \, e' \, M_{\lambda^0} e M_{f^0}
          - M_{\lambda^0} \, e \,
      (\lambda^0 f^{0 \prime} )^\dagger \, (f^0 \lambda^{0 \prime} )^\dagger
         \, e' \, M_{\lambda^0} e M_{f^0} \; ,
  \end{align*}
  and the remainder term satisfies
  for any sequence
$c_{NT}\rightarrow 0$
\begin{align*}
     \sup_{\{\beta :\left\| \beta -\beta^{0} \right\| \leq c_{NT}\}}
       \frac{ \left\| \widehat e^{({\rm rem})}(\beta) \right\| }
         { N \|\beta - \beta^0\|^2 + \|\beta - \beta^0\| \,  \,  + N^{-1}}
           &= {\cal O}_P\left(1\right)  \; .
\end{align*}
\end{lemma}

\section{Details for $N^{3/4}$-Convergence Rate of $\widehat \beta_R$}
\label{app:MoreConvergenceRate}

This section extends the discussion in Section~\ref{sec:ConvergenceRate} of the main paper.
We provide the high-level Assumption~\ref{ass:HL1} under which 
$N^{3/4} \left(\widehat \beta_{R} - \beta^0\right)
                   = {\cal O}_P(1)$ 
can be shown, see Theorem~\ref{th:ConvergenceRate} below.      
Lemma~\ref{lemma:JustifyHL1} then provides the connection between
our main text assumptions and Assumption~\ref{ass:HL1}. 
The proofs are provided in Section~\ref{app:Proofs-Supp} below.
Combining Theorem~\ref{th:ConvergenceRate} and
Lemma~\ref{lemma:JustifyHL1} yields 
Theorem~\ref{th:ConvergenceRate2} in the main text.

We first note that equation \eqref{LSobjective} implies that
\begin{align}
   {\cal L}_{NT}^R(\beta)
      &= {\cal L}_{NT}^0(\beta)
         - \frac 1 {NT}  \; \sum_{r=R^0+1}^{R}
            \mu_r\left[ \left(Y - \beta \cdot X \right)'
                     \left(Y - \beta \cdot X \right) \right]
   \nonumber\\
       &= {\cal L}_{NT}^0(\beta)
      - \frac 1 {NT}  \; \sum_{r=1}^{R-R^0}
            \mu_r\left[ \widehat e'(\beta) \widehat e(\beta) \right] .
   \label{LRL0}
\end{align}
The extra term $\frac 1 {NT}  \sum_{r=R^0+1}^{R} \mu_r\left[ \left(Y - \beta \cdot X \right)'
\left(Y - \beta \cdot X \right) \right]$ is due to overfitting on the extra factors.
In the second line
of \eqref{LRL0} we used that $\widehat e'(\beta) \widehat e(\beta)$ is the residual of
$\left(Y - \beta \cdot X \right)' \left(Y - \beta \cdot X \right)$
after subtracting the first $R^0$ principal components, which implies
that the eigenvalues of these two matrices are the same,
except from the $R^0$ largest ones which are missing in
$\widehat e'(\beta) \widehat e(\beta)$.
The decomposition in equation \eqref{LRL0} together with the expansion result for  $\widehat e(\beta)$
in Lemma~\ref{lemma:ResExp} give rise to the following Lemma.

\begin{lemma}
   \label{lemma:expansion2a}
   Under Assumption~\ref{ass:SF} and \ref{ass:SN}
   and for $R>R^0$ we have
   $$ {\cal L}_{NT}^R(\beta)
             = {\cal L}_{NT}^0(\beta)
      - \frac 1 {NT}  \; \sum_{r=1}^{R-R^0}
            \mu_r\left[ A(\beta) \right]
                          + {\cal L}_{NT}^{R,{\rm rem},1}(\beta),$$
          where
          $A(\beta) = M_{f^0}\left[ e  -  \Delta \beta \cdot X \right]'
                  M_{\lambda^0}  \left[ e  - \Delta \beta \cdot X \right]  M_{f^0}$,
           with $\Delta \beta = \beta - \beta^0$,       
          and for any constant $c>0$ we have
          \begin{align*}
     \sup_{\{\beta : \sqrt{N} \left\| \beta -\beta^{0} \right\| \leq c\}} \frac{
       \left| {\cal L}_{NT}^{R,{\rm rem},1}(\beta) \right| } { \sqrt{N} + \sqrt{NT} \, \left\| \beta -\beta^{0} \right\| } = {\cal O}_{p}\left( \frac 1 {NT} \right) .
          \end{align*}
\end{lemma}

The following high-level assumption guarantees that the 
$\beta$-dependence of $\frac 1 {NT} \sum_{r=1}^{R-R^0}
            \mu_r\left[ A(\beta) \right]$  is small, so that apart
from a constant the approximate quadratic expansions of ${\cal L}_{NT}^R(\beta)$
and ${\cal L}_{NT}^0(\beta)$ around $\beta^0$ are identical.

\begin{HLOneAssumption}[\bf First High-Level Assumption on Matrix Spectra]
   \label{ass:HL1}
   Let $\Delta \beta = \beta - \beta^0$ and
   \begin{align*}
      d(\beta) &= \sum_{r=1}^{R-R^0}  \bigg\{  \mu_r\left[
      M_{f^0} \left( e  -  \Delta \beta \cdot X \right)'
                  M_{\lambda^0}  \left( e  -  \Delta \beta \cdot X \right)  M_{f^0} \right]
        \\ & \qquad \qquad \quad
                  -\mu_r\left[ M_{f^0} e' M_{\lambda^0} e M_{f^0} \right] -
                  \mu_r\left[ M_{f^0} \left( \Delta \beta \cdot X \right)'
                  M_{\lambda^0}  \left(  \Delta \beta \cdot X \right)  M_{f^0} \right]
        \bigg\} .
   \end{align*}
   For all constants $c>0$ we assume that
   \begin{align*}
     \sup_{\{\beta : \sqrt{N} \left\| \beta -\beta^{0} \right\| \leq c\}}
      \frac { \max[ d(\beta),0]  }
       {\sqrt{N} + N^{5/4} \| \beta-\beta^0\|
                + N^2 \| \beta-\beta^0\|^2 / \log(N) }
         = {\cal O}_P \left( 1 \right).
   \end{align*}
\end{HLOneAssumption}

Combining Lemma~\ref{lemma:expansion2a} with this high-level assumption
yields the following theorem.

\begin{theorem}
   \label{th:ConvergenceRate}
   Let $R>R^0$, let Assumptions \ref{ass:SF}, \ref{ass:SN}, \ref{ass:NC}, \ref{ass:EX}
    and
   \ref{ass:HL1} be satisfied and furthermore assume
   that $C^{(1)} = {\cal O}_P(N^{1/4})$.
   In the limit $N,T \rightarrow \infty$ with
   $N/T \rightarrow \kappa^2$, $0<\kappa<\infty$, we then have
 $N^{3/4} \left(\widehat \beta_{R} - \beta^0\right)
                   = {\cal O}_P(1)$.
\end{theorem}

The theorem follows from the inequality
${\cal L}^R_{NT}(\widehat \beta_{R}) \leq {\cal L}^R_{NT}(\beta^0)$
by applying Lemma~\ref{lemma:expansion2a}, Assumption~\ref{ass:HL1},
and our expansion of ${\cal L}^0_{NT}(\beta)$. The detailed proof is given below.

\subsection{Justification of Assumption~\ref{ass:HL1}}

We first present two technical Lemmas, which are used to show Lemma~\ref{lemma:JustifyHL1} below.

\begin{lemma}
   \label{lemma:help}
   Let $g$ be an $N \times Q$ matrix
   and $h$ be a $T \times Q$ matrix such that
   $g' g = h' h = \mathbbm{1}_Q$.
   Let $U$ be an $N \times T$ matrix and $C$ a $Q \times Q$ matrix.
   Assume that ${\rm rank}[(U' g, h)] = 2Q$. Let\footnote{Note that  $ \Delta_{\max}=0$
   if $R \geq Q$,
   and that $ \Delta_{\max} \leq \mu_1( g' U U' g ) - \mu_Q( g' U U' g ) $
   for $R<Q$. }
   \begin{align*}
        \Delta_{\max} &=
    \max_{r \in \{1,2,\ldots,\min(R,Q)\}} \left[ \mu_r( g' U U' g ) -  \mu_{r+Q-\min(Q,R)}( g' U U' g )  \right].
   \end{align*}
   We then have
   \begin{align*}
      & \sum_{r=1}^R
      \mu_r \left[ \left( U + g C  h' \right)' \left( U + g C  h' \right) \right]
    \nonumber \\ & \qquad
        \leq  \sum_{r=1}^{R}
      \mu_r \left( U'U + \| g' U U' g \| \, P_{\left( M_{U' g} h \right)}
         +  \Delta_{\max}  P_{(U' g)}   \right)
    \nonumber \\ & \qquad \qquad \qquad
     + \sum_{r=1}^{\min(Q,R)}
      \mu_r \left(  C C' + g' U h C' + C h' U' g \right) .
   \end{align*}
\end{lemma}

\begin{lemma}
    \label{lemma:LowRankAdd}
    Let $e$ be an $N \times T$ matrix, whose columns $e_t$, $t = 1,\ldots T$, are distributed as
    $e_t \sim \, iid \, {\cal N}(0,\Sigma)$, where $\Sigma$ is a symmetric
    positive definite non-random $N \times N$ matrix
    with eigenvalues $\mu_1(\Sigma)$, \ldots, $\mu_N(\Sigma)$. Let $A$ be a symmetric positive definite
    non-random $T \times T$ matrix with ${\rm rank}(A)=Q$. Let $n$ be the number of eigenvalues of $\Sigma$
    that is larger or equal than $\|A\|/T$, i.e. $n \leq N$ is the largest integer such that
     $\mu_n(\Sigma) \geq \|A\|/T$. Consider an asymptotic where $N,T,n \rightarrow \infty$ jointly,
     while $Q$ and $R$ are constant positive integers.  We then have
     \begin{align*}
         \sum_{r=1}^R \mu_r \left( e' e + A \right)
         -          \sum_{r=1}^R \mu_r \left( e' e \right)
         =  {\cal O}_P \left(   \sqrt{ (N+T)T/n }  \right)   .
     \end{align*}
\end{lemma}

The following Lemma  connects Lemma~\ref{th:ConvergenceRate} to the main text.

\begin{lemma}
    \label{lemma:JustifyHL1}
    Let $R>R^0$ and let Assumptions~\ref{ass:SF} hold.
    Let either Assumption~\ref{ass:DX-1} or~\ref{ass:DX-2}  be satisfied.
    Consider
     $N,T \rightarrow \infty$ with
   $N/T \rightarrow \kappa^2$, $0<\kappa<\infty$.
   Then Assumptions~\ref{ass:SN} and \ref{ass:HL1} are satisfied.
   If, in addition, Assumption~\ref{ass:EX} holds, then we have $C^{(1)} = {\cal O}_P(N^{1/4})$.
\end{lemma}
 
Combining Theorem~\ref{th:ConvergenceRate} and Lemma~\ref{lemma:JustifyHL1}
we obtain Theorem~\ref{th:ConvergenceRate2} in the main text.

\section{Details for Asymptotic Equivalence of $\widehat \beta_{R^0}$ and $\widehat \beta_R$}
\label{app:MoreEquivalence}

This section extends the discussion of Section~\ref{sec:Equivalence} in the main paper.
By applying the expansion of $\widehat e(\beta)$ in equation~\eqref{LRL0}
to the expression for ${\cal L}_{NT}^R(\beta)$ one obtains the following.
\begin{lemma}
   \label{th:expansion2b}
   Under Assumption~\ref{ass:SF} and \ref{ass:SN}
   and for $R>R^0$ we have
   $$ {\cal L}_{NT}^R(\beta)
             = {\cal L}_{NT}^0(\beta)
      - \frac 1 {NT}  \; \sum_{r=1}^{R-R^0}
            \mu_r\left[ B(\beta) + B'(\beta) \right]
                          + {\cal L}_{NT}^{R,{\rm rem}}(\beta),$$
          where
          \begin{align*}
             B(\beta) &= \ft 1 2 M_{f^0}\left[ e  -  (\beta-\beta^0) \cdot X \right]'
                  M_{\lambda^0}  \left[ e  -  (\beta-\beta^0) \cdot X \right]  M_{f^0}
             \nonumber \\ & \quad
             - M_{f^0}  e'  M_{\lambda^0}
             e  M_{f^0}  e'
        \lambda^0(\lambda^{0\prime}\lambda^0)^{-1}(f^{0\prime}f^0)^{-1}  f^{0\prime}
             \nonumber \\ & \quad
               + M_{f^0}  \left[ (\beta-\beta^0) \cdot X -e \right]'  M_{\lambda^0}   e
     f^0 (f^{0\prime}f^0)^{-1}  (\lambda^{0\prime}\lambda^0)^{-1}\lambda^{0\prime}e M_{f^0}
             \nonumber \\ & \quad
               + M_{f^0}  e'  M_{\lambda^0}   \left[ (\beta-\beta^0) \cdot X \right]
     f^0 (f^{0\prime}f^0)^{-1}  (\lambda^{0\prime}\lambda^0)^{-1}\lambda^{0\prime}e M_{f^0}
             \nonumber \\ & \quad
               + M_{f^0}  e'  M_{\lambda^0}   e
     f^0 (f^{0\prime}f^0)^{-1}  (\lambda^{0\prime}\lambda^0)^{-1}\lambda^{0\prime}
     \left[ (\beta-\beta^0) \cdot X \right] M_{f^0}
      \nonumber \\ & \quad
        + B^{(eeee)}
        + M_{f^0} B^{(\rm rem,1)}(\beta) P_{f^0}
        + P_{f^0} B^{(\rm rem,2)} P_{f^0} ,
          \end{align*}
          and
          \begin{align*}
             B^{(eeee)} &=
      - M_{f^0}  e' M_{\lambda^0} e M_{f^0}  e'
      \lambda^0  (\lambda^{0\prime}\lambda^0)^{-1}  (f^{0\prime}f^0)^{-1}  (\lambda^{0\prime}\lambda^0)^{-1}  \lambda^{0\prime}
          e  M_{f^0}
   \nonumber \\ & \quad
       + M_{f^0}  e' M_{\lambda^0}  e
      f^0  (f^{0\prime}f^0)^{-1}  (\lambda^{0\prime}\lambda^0)^{-1} \lambda^{0\prime} e
         f^0  (f^{0\prime}f^0)^{-1}  (\lambda^{0\prime}\lambda^0)^{-1}  \lambda^{0\prime}
               e  M_{f^0}
       \nonumber \\ & \quad
          - \ft 1 2 M_{f^0}  e' M_{\lambda^0}  e
      f^0  (f^{0\prime}f^0)^{-1}  (\lambda^{0\prime}\lambda^0)^{-1}  (f^{0\prime}f^0)^{-1}  f^{0\prime}
          e'  M_{\lambda^0} e M_{f^0}
       \nonumber \\ & \quad
         + \ft 1 2
          M_{f^0}  e'
      \lambda^0  (\lambda^{0\prime}\lambda^0)^{-1}  (f^{0\prime}f^0)^{-1}  f^{0\prime}
               e'
         M_{\lambda^0}  e
      f^0  (f^{0\prime}f^0)^{-1}  (\lambda^{0\prime}\lambda^0)^{-1}  \lambda^{0\prime}
               e  M_{f^0}  \; .
          \end{align*}
          Here, $B^{(\rm rem,1)}(\beta)$
          and $B^{(\rm rem,2)}$ are $T \times T$ matrices,
          $B^{(\rm rem,2)}$ is independent of $\beta$
          and satisfies
          $\| B^{(\rm rem,2)}\| = {\cal O}_P(1)$,
          and for any constant $c>0$
          \begin{align*}
     \sup_{\{\beta : \sqrt{N} \left\| \beta -\beta^{0} \right\| \leq c\}} \frac{
       \| B^{(\rm rem,1)}(\beta) \|  } { 1 + \sqrt{NT} \, \left\| \beta -\beta^{0} \right\|  } &= {\cal O}_{p}\left( 1 \right) ,
       \nonumber \\
              \sup_{\{\beta : \sqrt{N} \left\| \beta -\beta^{0} \right\| \leq c\}} \frac{
        \left| {\cal L}_{NT}^{R,{\rm rem}}(\beta) \right| }
        { (1 + \sqrt{NT} \, \left\| \beta -\beta^{0} \right\|)^2 } &= o_{p}\left( \frac 1 {NT} \right)  .
          \end{align*}
\end{lemma}
Here, the remainder term ${\cal L}_{NT}^{R,{\rm rem}}(\beta)$ stems from terms in
$\widehat e'(\beta) \widehat e(\beta)$ whose spectral norm is smaller than
 $o_P(1)$ within a $\sqrt{N}$ shrinking
neighborhood of $\beta$  after
dividing by  $\left( 1 + \sqrt{NT} \, \left\| \beta -\beta^{0} \right\| \right)^2$.
Using Weyl's inequality those terms can be separated from the eigenvalues
$ \mu_r\left[ \widehat e'(\beta) \widehat e(\beta) \right]$.
The expression for $B(\beta)$ looks complicated, in particular the
terms in $B^{(eeee)}$. Note however, that $B^{(eeee)}$ is $\beta$-independent
and satisfies $\| B^{(eeee)} \| = {\cal O}_P(1)$ under our assumptions,
so that it is relatively easy to deal with these terms.
Note furthermore that the structure of $B(\beta)$
is closely related to the expansion of ${\cal L}_{NT}^{0}(\beta)$,
since by definition we have ${\cal L}_{NT}^{0}(\beta) =
(NT)^{-1} {\rm Tr}(\widehat e'(\beta) \widehat e(\beta))$, which
can be approximated by $(NT)^{-1} {\rm Tr}(B(\beta)+B'(\beta))$.
Plugging the definition of $B(\beta)$
into $(NT)^{-1} {\rm Tr}(B(\beta)+B'(\beta))$ one indeed
recovers the terms of the approximated Hessian and score
provided by Theorem~\ref{th:expansion}, which is a convenient
consistency check. We do not give explicit formulas for
$B^{(\rm rem,1)}(\beta)$ and $B^{(\rm rem,2)}$, because those terms enter
$B(\beta)$ projected by $P_{f^0}$, which makes them orthogonal
to the leading term in $B(\beta)+B'(\beta)$, so that they can only have limited influence
on the eigenvalues of $B(\beta)+B'(\beta)$. The bounds on
the norms of $B^{(\rm rem,1)}(\beta)$ and $B^{(\rm rem,2)}$
provided in the lemma  are sufficient for all conclusions on the properties of
$\mu_r\left[ B(\beta) + B'(\beta) \right]$ below. The proof of the
lemma  can be found in the section~\ref{app:Proofs-Supp} below. The lemma motivates
the following high-level assumption.

\begin{HLTwoAssumption}[\bf Second High-Level Assumption on Matrix Spectra]
   \label{ass:HL2}
   For all constants $c>0$
   \begin{align*}
      \sup_{\{\beta : N^{3/4} \left\| \beta -\beta^{0} \right\| \leq c\}}
      \frac { \left| \sum_{r=1}^{R-R^0}  \left\{  \mu_r\left[ B(\beta) + B'(\beta) \right] -\mu_r\left[ B(\beta^0) + B'(\beta^0)\right]
        \right\} \right| }
       {(1+\sqrt{NT} \| \beta-\beta^0\|)^2}
          = o_P(1),
   \end{align*}
   where $B(\beta)$ was defined in Lemma~\ref{th:expansion2b}.
\end{HLTwoAssumption}

Combining Lemma~\ref{lemma:expansion2a},
Assumption~\ref{ass:HL2}, and Theorem~\ref{th:expansion}, we find that
the profile objective function for $R>R^0$ can be written as
\begin{align*}
      {\cal L}_{NT}^{R}(\beta) &= {\cal L}_{NT}^{R}(\beta^0)
                           -  \frac 2 {\sqrt{NT}}   \left(\beta-\beta^0 \right)'
                           \left( C^{(1)} + C^{(2)} \right)
                    + \left(\beta-\beta^0 \right)'  W  \left(\beta-\beta^0 \right)
                          + {\cal L}_{NT}^{R,{\rm rem,2}}(\beta)  ,
\end{align*}
with a remainder term that satisfies for all constants $c>0$
\begin{align*}
     \sup_{\{\beta : N^{3/4} \left\| \beta -\beta^{0} \right\| \leq c\}} \frac{
        \left| {\cal L}_{NT}^{R,{\rm rem,2}}(\beta) \right|  } { \left( 1 + \sqrt{NT} \, \left\| \beta -\beta^{0} \right\| \right)^2 } = o_{p}\left( \frac 1 {NT} \right) .
\end{align*}
This result, together with our $N^{3/4}$-consistency result for $\widehat \beta_{R}$, gives
rise to the following corollary.
\begin{corollary}
   \label{cor:LimitRgen2}
   Let $R>R^0$, let Assumptions \ref{ass:SF}, \ref{ass:SN}, \ref{ass:NC}, \ref{ass:EX}, \ref{ass:HL1}
   and \ref{ass:HL2} be satisfied and furthermore assume
   that $C^{(1)} = {\cal O}_P(1)$.
   In the limit $N,T \rightarrow \infty$ with
   $N/T \rightarrow \kappa^2$, $0<\kappa<\infty$, we then have
   \begin{align*}
       \sqrt{NT}\left(\widehat \beta_{R} - \beta^0\right)
                   = W^{-1} \left( C^{(1)} + C^{(2)} \right) + o_P(1)
                   = {\cal O}_P(1).
   \end{align*}
\end{corollary}
The proof of Corollary~\ref{cor:LimitRgen2} is analogous to that
of Corollary~\ref{cor:LimitR0}. The combination of both corollaries shows
that our main result holds under high-level assumptions,
i.e. the limiting distributions of $\widehat \beta_{R}$
and $\widehat \beta_{R^0}$ are indeed identical.

\subsection{Justification of Assumption~\ref{ass:HL2}}

The following is a technical lemma, which is crucially used in the proof of
Lemma~\ref{lemma:JustifyHL2} below.

\begin{lemma}
   \label{Lemma:EVbound1}
   Let A and B be symmetric $n\times n$ matrices, and let $A$ be positive semi-definite. Let
   $\mu_1(A) \geq \mu_2(A) \geq \ldots \geq \mu_n(A) \geq 0$ be the sorted eigenvalues
   of $A$, and let $\nu_1$, $\nu_2, \ldots , \nu_n$ be the corresponding eigenvectors
   that are orthogonal and normalized such that $\|\nu_i\| = 1$ for $i=1,\ldots,n$.
   Let $b=\max_{i,j=1,\ldots,n} |\nu_i' B \nu_j|$.
   Let $r$ and $q$ be positive integers with $r<q\leq n$,
   and let $\sum_{i=q}^n b \, (\mu_r(A)-\mu_i(A))^{-1}  \leq 1$ be satisfied.
   Then we have
   \begin{align*}
      \left| \mu_r(A+B) -  \mu_r(A) \right|
          \, &\leq \, \frac{(q-1) \, b}{1 - \sum_{i=q}^n \frac b {\mu_r(A)-\mu_i(A)}}
   \end{align*}
\end{lemma}

The following Lemma provides conditions under which Assumption~\ref{ass:HL2}
is satisfied. It crucially connects the current section with the main text.

\begin{lemma}
   \label{lemma:JustifyHL2}
   Let Assumptions~\ref{ass:SF}, \ref{ass:SN} and \ref{ass:EV} hold,
   let $R>R^0$ and consider
   a limit $N,T \rightarrow \infty$ with $N/T \rightarrow \kappa^2$,
   $0<\kappa<\infty$.    Then, for all constants $c>0$ and $r=1,\ldots,R-R^0$ we have
   \begin{align*}
      \sup_{\{\beta : N^{3/4} \left\| \beta -\beta^{0} \right\| \leq c\}}
     \left| \mu_r\left( B(\beta) + B'(\beta) \right) - \rho_r
         \right|
          = o_P(1),
   \end{align*}
   which implies that Assumption~\ref{ass:HL2} is satisified. 
\end{lemma}

\subsection{Sufficiency of Low-Level Assumptions in Main Text}
\label{ass:SufficiencyLL}

The following Lemma summarizes some properties of the singular value vectors $v_r$ and $w_r$
of $M_{f^0} e M_{\lambda^0}$ for the case where $e_{it}$ is $iid$ normally distributed. Those properties
are used in the proof of the main text  Lemma~\ref{lemma:JustifyEV} below.

\begin{lemma}
    \label{lemma:iidEV}
    Let Assumption~\ref{ass:LL} hold 
    and let $v_r$ and $w_r$ be defined as in Assumption~\ref{ass:EV}. Then the following holds.
    
    \begin{itemize}
        \item[(i)]
      Let $\widetilde v$ be an $N$-vector with $iid {\cal N}(0,1)$
    entries; let $\widetilde w$ be a $T$-vector, independent of $\widetilde v$, also 
    with $iid {\cal N}(0,1)$ entries; and let $\widetilde v$ and $\widetilde w$ be independent
    of $\lambda^0$, $f^0$, $\overline X_k$, and $\widetilde X^{\rm str}_k$ and $e P_{f^0}$. Then,  
    for all $r,s=1,\ldots,\min(N,T)-R^0$ we have
    \begin{align*}
   \left( \begin{array}{c} v_r \\ w_s \end{array} \right)
     \operatorname*{=}_d
       \left( \begin{array}{c}  \| M_{\lambda^0} \widetilde v \|^{-1} M_{\lambda^0} \widetilde v
         \\ 
        \| M_{f^0} \widetilde w \|^{-1}  M_{f^0} \widetilde w  \end{array} \right) ,
    \end{align*}
    where $\operatorname*{=}_d$ refers to equally distributed. Furthermore, the 
    squares of $\| M_{\lambda^0} \widetilde v \|^{-1}$
and $\| M_{f^0} \widetilde w\|^{-1}$
have inverse chi-square distributions with $N-R^0$ and $T-R^0$ degrees
of freedom, respectively, which implies that 
for every $\xi>0$ there exists
a constant $c>0$ such that we have
\begin{align*}
 \mathbbm{E}\left( \sqrt{N} \|  M_{\lambda^0} \widetilde v \|^{-1}  \right)^{\xi}  &< c,
     &
     \mathbbm{E}\left(  \sqrt{T} \| M_{f^0} \widetilde w \|^{-1}   \right)^{\xi}   &< c,
\end{align*}
for all $N  > 4\xi + R^0$ and $T > 4\xi + R^0$.

       \item[(ii)]
         There exists $\varepsilon \in [0,1/12)$ such that as $N,T$ become large we have
         \begin{align*}
                    \max_{r,s,\tau} \left| \sum_{t=\tau+1}^{T}  w_{r,t} w_{s,t-\tau} \right| = {\cal O}_P( T^{-1/2+\varepsilon} ) ,
          \end{align*}
         where $r,s=1,\ldots,\min(N,T)-R^0$ and $\tau=1,2,\ldots,T-1$.   
       
       \item[(iii)] The matrices $P_{\lambda^0} e P_{f^0}$,  $P_{\lambda^0} e M_{f^0}$, $M_{\lambda^0} e P_{f^0}$
         and $M_{\lambda^0} e M_{f^0}$ are all mutually independent, and
         their entries have uniformly bounded moments of arbitrary order.
       
   \end{itemize}    
\end{lemma}

The proof of Lemma~\ref{lemma:iidEV} is given in section~\ref{ass:ProofsInterEqu} below.

\begin{proof}[\bf \underline{Proof of Lemma~\ref{lemma:JustifyEV} (Justification of Main Text High-Level Assumptions)}]
   
      \# First, we show that Assumptions~\ref{ass:SN}, \ref{ass:EX} and \ref{ass:DX-1} 
  are satisfied, and that $C^{(1)} = {\cal O}_P(1)$:  
  
  Since $e_{it}$ is iid ${\cal N}(0,\sigma^2)$ we have $\| e \| = {\cal O}_P(\sqrt{N})$
  as $N,T$ grow at the same rate, see e.g. Geman~\cite*{Geman1980}.
  This also implies that
  $\| \widetilde X_k^{\rm weak} \| \leq \sum_{\tau=1}^\infty | \gamma_\tau | \|e\| = {\cal O}_P(\sqrt{N})$.  
  Assumption~\ref{ass:LL} therefore implies that
  Assumption~\ref{ass:DX-1} holds with $\widetilde X_k = \widetilde X_k^{\rm str} +  \widetilde X_k^{\rm weak}$
  and $\Sigma = \sigma^2 \mathbbm{1}$. Note that for this $\Sigma$ we have
  $g' \Sigma g = \sigma^2 \mathbbm{1}_Q = \|g' \Sigma g \| \mathbbm{1}_Q$
  and
    $\mu_n(\Sigma) = \sigma^2 = \|g' \Sigma g \|$ for all $n$.
  Assumption~\ref{ass:DX-1} also implies that Assumption~\ref{ass:SN} holds, as also
  noted in Lemma~\ref{lemma:JustifyHL1}.
  
   Since we assume that $\mathbbm{E} \left| X_{k,it}    \right|^{2}$ is uniformly bounded
   we have    $\mathbbm{E} \frac 1 {NT} {\rm Tr}( X_k'   X_k ) 
   = \frac 1 {NT} \sum_{i,t} \mathbbm{E}  X_{k,it}^{2} = {\cal O}(1)$
   and therefore ${\rm Tr}( X_k'   X_k )  = {\cal O}_P(NT)$.   
   We also have
   $\mathbbm{E}\left[ {\rm Tr}(  X_k e')^2 | X_k  \right]
      = \sigma^2 {\rm Tr}( X_k'   X_k )  =  {\cal O}_P(NT)$, and therefore
    $\frac 1 {\sqrt{NT}} {\rm Tr}(X_k \, e^{\prime} )  = {\cal O}_P(1)$,
    i.e. Assumption~\ref{ass:EX} holds. 
    By replacing $X_k$ with $M_{\lambda^0} X_k M_{f^0}$ in the previous argument
    we also find that $C^{(1)} = {\cal O}_P(1)$.

  \# Assumption \ref{ass:EV}$(i)$ holds for any $c < c_{\max} = \lim_{N,T \rightarrow \infty}
   \left( \sqrt{N} + \sqrt{T} \right)^2 / N$, because 
   from Theorem~1 in Soshnikov~\cite*{Soshnikov2002} we know that
   $\rho_{R-R^0}/N - c_{\max} = {\cal O}_P( N^{-2/3} )$. Some more details are given below.
        
  \# We now show that Assumption \ref{ass:EV}$(ii)$ holds with $q_{NT}= \log(N) N^{1/6}$.
Without loss of generality, we set $\sigma=1$ in this part of the proof.
We want to show that  $q_{NT}= \log(N) N^{1/6}$ also satisfies
\begin{align*}
\frac 1 {q_{NT} (T-R^0)} \, \sum_{r=q_{NT}}^{Q} \frac 1 {\mu_{R-R^0} - \mu_r}
              &= {\cal O}_P(1),
\end{align*}
where $\mu_r \equiv \rho_r/(T-R^0)$. Note that it is not important whether the sum
runs to $Q=N-R^0$ or $Q=T-R^0$, since the contributions of small
eigenvalues between $r=N-R^0$ and $r=T-R^0$ are of order ${\cal O}_P(1)$ anyways. Without loss of generality let
$\lim_{N,T \rightarrow \infty} N/T = \kappa^2 \leq 1$
in the rest of this proof (the proof for $\kappa \geq 1$
is analogous, since all arguments are symmetric under interchange of $N$ and $T$).
Let
$\mu_{NT} = \left[(N-R^0)^{1/2} + (T-R^0)^{1/2}\right]^2$,
$\sigma_{NT} = \left[(N-R^0)^{1/2} + (T-R^0)^{1/2}\right]
                   \left[(N-R^0)^{-1/2} + (T-R^0)^{-1/2}\right]^{1/3}$,
$\overline x = \lim_{N,T \rightarrow \infty} \mu_{NT}/(T-R^0)=(1+\kappa)^2$, and
$\underline x = (1-\kappa)^2$.
From Theorem~1 in Soshnikov~\cite*{Soshnikov2002} we know that
 the joint distribution of
$\sigma_{NT}^{-1} (\rho_1-\mu_{NT},\rho_2-\mu_{NT},\ldots,\rho_{R^0+1}-\mu_{NT})$ converges to the
Tracy-Widom law, i.e. to the limiting distribution of the first $R^0+1$ eigenvalues
of the Gaussian Orthogonal Ensemble.
Note that $\sigma_{NT}$ is of order $N^{1/3}$,
and that the Tracy-Widom law is a continuous distribution, so that the result of
Soshnikov implies that
\begin{align}
    \overline x - \mu_{R-R^0} &= {\cal O}_P(N^{-2/3}) \; ,
    &
   (\mu_{R-R^0} - \mu_{R-R^0+1})^{-1} = {\cal O}_P\left( N^{2/3} \right)\; .
   \label{FewLargestEV}
\end{align}
The empirical distribution of the $\mu_r$ is defined as
$F_{NT}(x) = Q^{-1} \sum_{r=1}^Q 1(\mu_r \leq x)$, where
$1(.)$ is the indicator function. This empirical
distribution converges to the Marchenko-Pastur
limiting spectral distribution $F_{\rm LSD}(x)$, which has domain
$[\underline x,\overline x]$, and whose
density
$f_{\rm LSD}(x) = dF_{\rm LSD}(x)/dx$ is given by
\begin{align}
   f_{\rm LSD}(x) &= \frac 1 {2 \pi \kappa^2 x}
                     \sqrt{ (\overline x - x) (x- \underline x) } \; .
\end{align}
An upper bound for $f_{\rm LSD}(x)$ is given by
$\frac 1 {2 \pi \kappa^2 \underline x} \sqrt{ (\overline x - x)
                        (\overline x - \underline x) }$,
and by integrating that upper bound we obtain
\begin{align}
   1-F_{\rm LSD}(x) \, &\leq  \, a \, (\overline x - x)^{3/2} \; ,
    &
   a &= \frac 2 {3 \pi \kappa^{3/2} \underline x} \; .
\end{align}
From Theorem~1.2 in G\"otze and Tikhomirov (2007) we know that
\begin{align}
   \sup_{x} \left| F_{NT}(x) - F_{LSD}(x) \right| &= {\cal O}_P(N^{-1/2}) \; .
\end{align}
Let $c_{1,NT} = \left\lceil 2 N^{1/2+\epsilon} \right\rceil$ and
$c_{2,NT} = \left\lceil 2 N^{3/4} \right\rceil$,
where $\left\lceil a \right\rceil$ is the smallest integer larger or equal
to $a$. Plugging in $x=\mu_{c_{1,NT}}$ into the result of
G\"otze and Tikhomirov, and using $F_{NT}(\mu_r) = 1 - (r-1)/N$, we find
\begin{align}
   a \left(\overline x - \mu_{c_{1,NT}} \right)^{3/2}
    &\geq 1-F_{\rm LSD}(\mu_{c_{1,NT}})
    = \frac{ c_{1,NT}-1 } N + {\cal O}_P(N^{-1/2})
  \nonumber \\
    &\geq N^{-1/2+\epsilon} , \qquad \text{wpa1.}
\end{align}
Using this and \eqref{FewLargestEV} we obtain
$(\mu_{R-R^0} - \mu_{c_1})^{-1} = {\cal O}_P\left( N^{1/3-2/3 \epsilon} \right)$.
Analogously one can show that
$(\mu_{R-R^0} - \mu_{c_2})^{-1} = {\cal O}_P\left( N^{1/6} \right)$.
In the following we just write $q$, $c_1$ and $c_2$
for $q_{NT}$, $c_{1,NT}$ and $c_{2,NT}$.
Combining the above results we find
\begin{align*}
   \frac 1 {q\, n} \, \sum_{r=q}^{Q} \frac 1 {\mu_{R-R^0} - \mu_r}
              &=
  \frac 1 {q \, n} \, \sum_{r=q}^{c_1-1} \frac 1 {\mu_{R-R^0} - \mu_r}
 +\frac 1 {q \, n} \, \sum_{r=c_1}^{c_2-1} \frac 1 {\mu_{R-R^0} - \mu_r}
 +\frac 1 {q \, n} \, \sum_{r=c_2}^{Q} \frac 1 {\mu_{R-R^0} - \mu_r}
   \nonumber \\
  &\leq   \frac {c_1} {q n (\mu_{R-R^0} - \mu_{R-R^0+1})}
        + \frac {c_2} {q n (\mu_{R-R^0} - \mu_{c_1})}
        + \frac {Q} {q n (\mu_{R-R^0} - \mu_{c_2})}
   \nonumber \\
  &= {\cal O}_P(1) + {\cal O}_P(N^{-1/12-5/3 \epsilon}) + {\cal O}_P(N^{-\epsilon})
   = {\cal O}_P(1) \; .
\end{align*}
This is what we wanted to show.

   \# We now show that Assumption \ref{ass:EV}$(iii)$ holds with $q_{NT}= \log(N) N^{1/6}$.
   Define
\begin{align}
      d^{(1)}_{NT} &= \max_{r,s,k} |v_r'  X_k  w_s| , &
      d^{(2)}_{NT} &= \max_{r} \| v_r'  e  P_{f^0} \| , &
      d^{(3)}_{NT} &= \max_{r} \| w_r'  e'  P_{\lambda^0} \| ,
    \nonumber \\
      d^{(4)}_{NT} &= N^{-3/4} \max_{r} \| v_r'  X_k  P_{f^0} \| , &
      d^{(5)}_{NT} &= N^{-3/4} \max_{r} \| w_r'  X'_k  P_{\lambda^0} \| .
      \label{DefDNT}
\end{align}
Furthermore, define
$d_{NT} = \max\left( 1, \; d_{NT}^{(1)}, \;
                           d_{NT}^{(2)}, \;
                           d_{NT}^{(3)}, \;
                           d_{NT}^{(4)}, \;
                           d_{NT}^{(5)} \right)$.
Then, Assumption~\ref{ass:EV}$(iii)$ can be summarized as $d_{NT} q_{NT} = o_P(N^{1/4})$,
i.e. given our choice of $q_{NT}$ we need to show that $d_{NT} = o_P(N^{1/12}/\log(N))$.
 
We decompose $X_k = X^{(A)}_k + X^{(B)}_k$, where
 $X^{(A)}_k = \overline X_k  + \widetilde X_k^{(str)}
        +   \sum_{\tau=1}^{t-1} \gamma_\tau [e P_{f^0}]_{i,t-\tau}$ 
and   $X^{(B)}_k =  \sum_{\tau=1}^{t-1} \gamma_\tau [e M_{f^0}]_{i,t-\tau}$.
Note that $X^{(A)}_k$ contains the strictly exogenous part of the regressor $X_k$,
but also contains the part of $\widetilde X_{k,it}^{\rm weak}$, which is independent
of $M_{\lambda^0} e M_{f^0}$, i.e. $X^{(A)}_k$ is independent of $\rho_r$, $v_r$ and $w_r$.
$X^{(B)}_k$ is the part of 
the weakly exogenous part $\widetilde X_{k,it}^{\rm weak}$
that is not independent of $\rho_r$, $v_r$ and $w_r$.
 Following the decomposition $X_k = X^{(A)}_k + X^{(B)}_k$ we also introduce the
corresponding decomposition of $d^{(1)}_{NT}= d^{(A,1)}_{NT} + d^{(B,1)}_{NT}$, where
$d^{(A,1)}_{NT} = \max_{r,s,k} |v_r'  X^{(A)}_k  w_s|$
and $d^{(B,1)}_{NT} = \max_{r,s,k} |v_r'  \widetilde X_k^{(B)}  w_s|$,
and analogously we define $d^{(A,4)}_{NT}$, $d^{(B,4)}_{NT}$, $d^{(A,5)}_{NT}$, 
and $d^{(B,5)}_{NT}$. Note that $d^{(1)}_{NT} \leq d^{(A,1)}_{NT} + d^{(B,1)}_{NT}$
and analogously for $d^{(4)}_{NT}$ and $d^{(5)}_{NT}$.

Using Lemma~\ref{lemma:iidEV} and Holder's inequality we have for sufficiently large $N,T$
\begin{align*}
    & \mathbbm{E}\left[ \left|  v_r'  X^{(A)}_k  w_s 
     \right|^{25}  \Bigg|  X^{(A)}_k \right]
   \nonumber \\
     &= \mathbbm{E}\left[ \left| \frac{\widetilde v' M_{\lambda^0}  X^{(A)}_k  M_{f^0} \widetilde w} 
         {  \| M_{\lambda^0} \widetilde v \|   \| M_{f^0} \widetilde w \|   }
     \right|^{25}  \Bigg|  X^{(A)}_k \right]
   \nonumber \\
      &=   \mathbbm{E}\left[ \left(
        \sqrt{N} \| M_{\lambda^0} \widetilde v \|^{-1}  \sqrt{T} \| M_{f^0} \widetilde w \|^{-1}
        \left| \frac 1 {\sqrt{NT}}
            \sum_{i,t}   \widetilde v_i       \widetilde w_t  [M_{\lambda^0}  X^{(A)}_k  M_{f^0}]_{it}
         \right|   
     \right)^{25} \Bigg|  X^{(A)}_k \right]
   \nonumber \\
     &\leq  \left\{  \mathbbm{E}\left( \sqrt{N} \|  M_{\lambda^0} \widetilde v \|^{-1}  \right)^{\xi} \right\}^{25/\xi}
                \left\{ \mathbbm{E}\left(  \sqrt{T} \| M_{f^0} \widetilde w \|^{-1}   \right)^{\xi}  \right\}^{25/\xi}
      \nonumber \\ & \qquad     \qquad       
                \left\{  \mathbbm{E}\left[  
        \left( \frac 1 {\sqrt{NT}}
            \sum_{i,t}   \widetilde v_i       \widetilde w_t  [M_{\lambda^0}  X^{(A)}_k  M_{f^0}]_{it}
         \right)^{26} \Bigg|  X^{(A)}_k \right]  \right\}^{25/26} 
    \nonumber \\
      &\leq C    \left[ \frac 1 {NT} \sum_{it} \left( [M_{\lambda^0}  X^{(A)}_k  M_{f^0}]_{it} \right)^2 \right]^{13} ,
\end{align*}
where $\xi$ satisfied $2/\xi + 1/26 = 1/25$, and $C$ is a global constant.
Here, as everywhere else in the paper and supplementary material, we implicitly also condition on $\lambda^0$ and $f^0$.
Since we assume that 
$\mathbbm{E} \left|(M_{\lambda^0} X_k M_{f^0})_{it}
                               \right|^{26}$
                     and therefore also          
$\mathbbm{E} \left|(M_{\lambda^0} X^{(A)}_k M_{f^0})_{it}
                               \right|^{26}$
is uniformly bounded we thus obtain that
$\mathbbm{E}\left[ \left|  v_r'  X^{(A)}_k  w_s   \right|^{25}   \right]$
is also uniformly bounded. We thus conclude that
\begin{align}
     \mathbbm{E} \left( \max_{r,s} |v_r'  X^{(A)}_k  w_s| \right)^{25}
     &= \mathbbm{E} \left( \max_{r,s} |v_r'  X^{(A)}_k  w_s|^{25} \right)
      \leq \mathbbm{E} \left( \sum_{r,s} |v_r'  X^{(A)}_k  w_s|^{25} \right) 
      = {\cal O}\left( N^2 \right) ,
\end{align}
which implies that $d^{(A,1)}_{NT} = {\cal O}_P\left( N^{2/25} \right) = o_P(N^{1/12}/\log(N))$.

We have
\begin{align}
   d^{(A,4)}_{NT} = N^{-3/4} \max_{r,k} \|v_r' X^{(A)}_k P_{f^0}\|
    &\leq  N^{-3/4} \max_{r,k} \|v_r' X^{(A)}_k\|
  \nonumber \\
    &\leq  N^{-3/4} \sqrt{T} \max_{r,t,k} |v_r'  X^{(A)}_{k,\cdot t} |
  \nonumber \\
    &\leq N^{-3/4} \sqrt{T} \max_{r,t,k} \sum_{i=1}^N v_{r,i} X^{(A)}_{k,it},
\end{align}
where $t=1,\ldots,T$, and we applied the inequality
$\|z\| \leq \sqrt{T} \max_t z_t$, which holds for
all $T$-vectors $z$. The remaining treatment of $d^{(4)}_{NT}$
is analogous to that of $d^{(1)}_{NT}$. Using Lemma~\ref{lemma:iidEV} and the assumption that
$(M_{\lambda^0} X_k)_{it}$ and thus also
$(M_{\lambda^0} X^{(A)}_k)_{it}$ has uniformly bounded
$8$'th moment one can show
that $\mathbbm{E}\left[ \left|  \sum_{i=1}^N v_{r,i} X_{k,it}   \right|^{7}   \right]$
is also uniformly bounded, which then implies
$d^{(A,4)}_{NT} = {\cal O}_P\left( N^{-3/4} \sqrt{T}  N^{2/7} \right) 
= {\cal O}_P\left( N^{1/28} \right) = o_P(N^{1/12}/\log(N))$.
Analogously one obtains $d^{(A,5)}_{NT}  = o_P(N^{1/12}/\log(N))$.

Since $[M_{\lambda^0} \widetilde X_k^{(B)}]_{it}
 = \sum_{\tau=1}^{t-1} \gamma_\tau [M_{\lambda^0} e M_{f^0}]_{i,t-\tau} = 
 \sum_{r=1}^{Q} \sqrt{\rho_r} \, v_{r,i}  \sum_{\tau=1}^{t-1} \gamma_\tau w'_{r,t-\tau} $
we find
\begin{align}
    d^{(B,1)}_{NT} &= \max_{r,s,k} |v_r'  \widetilde X_k^{(B)}  w_s|
    \nonumber \\ 
       &= \max_{r,s,k} |v_r'  M_{\lambda^0} \widetilde X_k^{(B)}    w_s|
    \nonumber \\ 
       &= \max_{r,s,k} \left|    \sqrt{\rho_r} 
       \sum_{t=1}^T  \sum_{\tau=1}^{t-1} \gamma_\tau w'_{r,t-\tau}    w_{s,t} 
           \right|
    \nonumber \\ 
       &\leq    \sqrt{\rho_1}  \max_{r,s,k} \left| 
       \sum_{t=1}^T  \sum_{\tau=1}^{t-1} \gamma_\tau w'_{r,t-\tau}    w_{s,t} 
           \right|
    \nonumber \\ 
       &\leq   \| e \|  \max_{r,s,k} 
         \left( \sum_{\tau=1}^{T-1} | \gamma_\tau | \right)
         \left(   \max_{r,s,\tau} \left| \sum_{t=\tau+1}^{T}  w_{r,t} w_{s,t-\tau} \right| \right)
         = {\cal O}_P\left( N^{\epsilon} \right)  = o_P(N^{1/12}/\log(N)),
\end{align}
where we used that $v'_r v_r = 1$ and $v'_r  v_s = 0$ for $r \neq s$,
and we also employed Lemma~\ref{lemma:iidEV} in the last step, which guarantees that $\epsilon<1/12$.
We thus have shown that $d^{(1)}_{NT}= o_P(N^{1/12}/\log(N))$.

We have $\| X^{(B)}_k \| \leq  \sum_{\tau=1}^{t-1} |\gamma_\tau| \| e \| = {\cal O}_P(\sqrt{N})$
and therefore
\begin{align}
    d^{(B,4)}_{NT} &= N^{-3/4} \max_{r} \| v_r'  X_k^{(B)}  P_{f^0} \|
           \leq N^{-3/4} \| X_k^{(B)} \|  = {\cal O}_P\left( N^{-1/4} \right),
\end{align}
and therefore $ d^{(4)}_{NT}= o_P(N^{1/12}/\log(N))$.
Analogously we obtain $d^{(B,5)}_{NT}  = {\cal O}_P\left( N^{-1/4} \right)$
and thus $ d^{(5)}_{NT}= o_P(N^{1/12}/\log(N))$.

Let $\widetilde f$ be an $N \times R^0$ matrix such that $P_{f^0}=P_{\widetilde f}$,
i.e. the column spaces of $f^0$ and $\widetilde f$ are identical, and
$\widetilde f' \widetilde f =\mathbbm{1}_{R^0}$. Then we have
$\|v_r' e P_{f^0}\| = \|v_r' e \widetilde f'\|$. Note that $e \widetilde f'$
is an $N\times R^0$ matrix with iid normal entries,
independently distributed of $v_r$ for all $r=1,\ldots,Q$.
Together with the distributional characterization of $v_r$ in 
Lemma~\ref{lemma:iidEV} it is then easy to show that
$\max_{r} \|v_r' e P_{f^0}\| = {\cal O}_P(N^\delta)$
for any $\delta>0$, and the same is true for
$\max_{r} \|w_r' e' P_{\lambda^0}\|$, i.e.
we have $d^{(2)}_{NT} = o_P(N^{1/12}/\log(N))$
and $d^{(3)}_{NT} = o_P(N^{1/12}/\log(N))$.
We have thus shown that Assumption \ref{ass:EV}$(iii)$ holds.
\end{proof}

\section{Estimated Factors, Loadings, Variance, and Bias \\ (Proof of Theorem~\ref{th:Estimators})}
\label{app:Estimators}

The goal of this section if to prove Theorem~\ref{th:Estimators} in the main text.
The Lemmas~\ref{lemma:EVEC} and~\ref{lemma:HelpConsistencyBias} are intermediate results that are used in the
proof of the main theorem below.

\begin{lemma}
   \label{lemma:EVEC}
   Let $A$, $B$, $\nu_i$, $i=1,\ldots,n$, $b$ be defined as in Lemma~\ref{Lemma:EVbound1}
   (but we do not require the assumption of Lemma~\ref{Lemma:EVbound1} here). Assume that
   for all $r \in \{1,2,\ldots,R+1\}$ we have
   $ \left| \mu_r(A+B) -  \mu_r(A) \right| \leq c_1$
   and  for all $r \in \{1,2,\ldots,R\}$ we have $\mu_r(A) -  \mu_{r+1}(A)   \geq c_2$
   for positive constants $c_1$ and $c_2$.\footnote{
   The inequality  $\left| \mu_r(A+B) -  \mu_r(A) \right| \leq c_1$ can be justified by 
   applying Lemma~\ref{Lemma:EVbound1}. 
   }
   Furthermore, let $\widetilde \nu_i$,  $i=1,\ldots,n$ be the eigenvector of $A+B$ corresponding to the 
   eigenvalue $\mu_i(A+B)$, normalized such that $\|\widetilde \nu_i\| = 1$. Then for  $r \in \{1,2,\ldots,R\}$  we have
   \begin{align*}
       \| \widetilde \nu_r - \nu_r \|^2 \leq   \frac{2 (4^r - 1)(b +  c_1) }  {3c_2} .
   \end{align*}
\end{lemma}

\begin{proof}[\bf Proof of Lemma~\ref{lemma:EVEC}]
     Since $A$ and $B$ are symmetric $\{\nu_i\}$ and $\{\widetilde \nu_i\}$ are orthogonal bases of $\mathbbm{R}^n$.
     With $\omega_{ri} := \nu_r' \widetilde  \nu_i \in [-1,1]$ we thus have
     $\nu_r = \sum_{i=1}^n  \omega_{ri}  \widetilde  \nu_i$.
     Note that  $\sum_{i=1}^n \omega_{ri}^2 = 1$,
     and also $\sum_{r=1}^n \omega_{ri}^2 = 1$.
     Let $q \in \{1,\ldots,R\}$. We have
     \begin{align*}
            \left|  \sum_{r=1}^q  \nu_r' (A+B)  \nu_r  -  \sum_{r=1}^q \mu_r(A) \right| 
            = \left| \sum_{r=1}^q  \nu_r' B  \nu_r \right| \leq q b ,
     \end{align*}
     and 
     \begin{align*}
        \sum_{r=1}^q  \nu_r' (A+B)  \nu_r = 
        \sum_{r=1}^q \sum_{i,j=1}^n   \omega_{ri}  \omega_{rj}    \widetilde \nu_i' (A+B)  \widetilde \nu_j
         = \sum_{r=1}^q \sum_{i=1}^n   \omega_{ri}^2 \mu_i(A+B) .
     \end{align*}
     Therefore
     \begin{align*} 
          \sum_{r=1}^q \sum_{i=1}^n   \omega_{ri}^2 \mu_i(A+B) - \sum_{r=1}^q \mu_r(A)    \geq   - qb   .
     \end{align*}
     Using $\mu_{q+1}(A+B) \geq \mu_i(A+B)$ for $q+1 \leq i$ we obtain
     \begin{align*} 
          \sum_{r=1}^q \sum_{i=1}^q   \omega_{ri}^2 \mu_i(A+B) 
          +\mu_{q+1}(A+B)  \sum_{r=1}^q \sum_{i=q+1}^n   \omega_{ri}^2   - \sum_{r=1}^q \mu_r(A)    \geq   - qb   ,
     \end{align*}
     and  using $| \mu_i(A+B) - \mu_i(A) | \leq c_1$ for $i \in \{1,\ldots,q+1\}$ 
     and $\sum_{i=1}^n  \omega_{ri}^2  = 1$ we find
     \begin{align*} 
          \sum_{r=1}^q \sum_{i=1}^q   \omega_{ri}^2 \mu_i(A) 
          +\mu_{q+1}(A)  \sum_{r=1}^q 
          \underbrace{ \sum_{i=q+1}^n   \omega_{ri}^2 }_{= 1 - \sum_{i=1}^q \omega_{ri}^2}   - 
          \underbrace{ \sum_{r=1}^q \mu_r(A) }_{=\sum_{i=1}^q \mu_i(A) }    \geq  
           - qb  - q c_1 .
     \end{align*}
     The last inequality can be rewritten as
     \begin{align*}
       \sum_{i=1}^q
        \left[ \mu_i(A) - \mu_{q+1}(A)  \right]  \left[  1 - \sum_{r=1}^q  \omega_{ri}^2     \right]
           \leq   q (b + c_1) .
     \end{align*}
     Using $\mu_i(A) - \mu_{q+1}(A) \geq c_2$ for $i \in \{1,\ldots,q\}$ we obtain
     \begin{align}
           q - \sum_{i=1}^q \sum_{r=1}^q   \omega_{ri}^2   \leq 
           \frac{ q (b +  c_1)  } {c_2}.
         \label{ResOmegaBound}  
     \end{align}
     For $q=1$ we find $1- \omega_{11}^2 \leq \frac{b+c_1} {c_2}$, which also implies that
     $\omega_{12}^2 \leq 1- \omega_{11}^2 \leq \frac{b+c_1} {c_2}$
     and $\omega_{21}^2 \leq 1- \omega_{11}^2 \leq \frac{b+c_1} {c_2}$. Using these results and
     \eqref{ResOmegaBound} for $q=2$ gives
     $1- \omega_{22}^2 \leq   \frac{ 2 (b +  c_1)  } {c_2} 
      +   \left(  1- \omega_{11}^2  \right) + \omega_{12}^2 + \omega_{21}^2 \leq  \frac{5 (b +  c_1)  } {c_2}$. 
      By continuing to apply \eqref{ResOmegaBound} recursively for increasing $q$ one obtains
      for $r \in \{1,\ldots,R\}$ that
      \begin{align*}
           1- \omega_{rr}^2 \leq  \frac{(4^r - 1)} 3 \frac{(b +  c_1)  } {c_2} ,
      \end{align*}
    and therefore
     \begin{align*}
          \| \widetilde \nu_r - \nu_r \|^2 
          =  ( \widetilde \nu_r - \nu_r )' (\widetilde \nu_r - \nu_r)
          = 2 \left( 1 -   \omega_{rr}^2 \right) \leq
           \frac{2 (4^r - 1)(b +  c_1) }  {3c_2}.
     \end{align*}
\end{proof}

It is convenient to introduce some additional notation.
We decompose $\widehat F_R = ( \widehat F^0_{R}, \widehat F_{R}^{\rm red})$,
     where the $T \times R^0$ matrix $\widehat F^0_{R}$ contains the eigenvectors corresponding to the
     $R^0$ largest eigenvalues of the $T\times T$ matrx $\left( Y - \widehat \beta_R \cdot X\right)' \left( Y - \widehat \beta_R \cdot X\right)$,
     while the $T \times (R-R^0)$ matrix $\widehat F_{R}^{\rm red}$ contains the eigenvalues corresponding to the
     next $R-R^0$ largest eigenvalues of this $T\times T$ matrix. Note that
     $\widehat F^0_{R}=\widehat f(\widehat \beta_R)$. By applying Lemma~\ref{lemma:ProExp}
     we find that $P_{\widehat F^0_{R}} = \mathbbm{1}_T - M_{\widehat F^0_{R}}$
     satisfies $\left\| P_{\widehat F^0_{R}} - P_{f^0} \right\| = {\cal O}_P(1/\sqrt{N}) $ under our assumptions.
     See also Bai~\cite*{Bai2009}
     and Moon and Weidner~\cite*{MoonWeidner2013} for further details on the estimated factors for $R=R^0$.
     The new difficulty in this section is to work out the asymptotic behavior of the redundantly estimated factors
     $\widehat F_{R}^{\rm red}$.     
     Note that $\widehat F_{R}^{\rm red}$ are the leading principal components 
     (eigenvalues corresponding to $R-R^0$ largest eigenvalues)
     of $\widehat e'( \widehat \beta_R ) \widehat e( \widehat \beta_R )$,
     with $\widehat e(\beta)$ defined at the beginning of Section~\ref{app:expansion1}
     (residuals after subtracting only $R^0$ principal component).
     
   Ananlogous to the decomposition $\widehat F_R = ( \widehat F^0_{R}, \widehat F_{R}^{\rm red})$
   we also introduce the decomposition $\widehat \Lambda_R = ( \widehat \Lambda^0_{R}, \widehat \Lambda_{R}^{\rm red})$  for the estimated factor loadings.

   Following Moon and Weidner~\cite*{MoonWeidner2013} we define the truncation kernel
   function $\Gamma: \mathbb{R} \rightarrow \mathbb{R}$ by $\Gamma(x)=1$ for $|x|\leq 1$, and $\Gamma(x)=0$
otherwise. Furthermore, for a $T\times T$ matrix $B$ with elements $B_{ts}$ we define
the right-sided  Kernel truncation of $B$ as the $T \times T$ matrix $B^{\rm truncR}$
          with elements $B^{\rm truncR}_{ts} = \Gamma\left( \frac{s-t} M \right) B_{ts}$ for $t<s$,
          and $B^{\rm truncR}_{ts}=0$ otherwise.
       Note that    $B^{\rm truncR}$ depends on the bandwidth parameter $M$, but this dependence is suppressed 
       in the notation.  
        With this definition we have
\begin{align*}
    \widehat B_{R,k} &=   
     \sum_{t=1}^T \sum_{s=t+1}^{t+M} P_{\widehat F_R,t s}  
     \left[ \frac 1 N  \sum_{i=1}^N  \widehat e_{R,it}  X_{k,i s} \right]
     = \frac 1 N {\rm Tr}\left[ P_{\widehat F_R} \left(  \widehat e_R' X_k \right)^{\rm truncR} \right] 
\end{align*}

\begin{lemma}
     \label{lemma:HelpConsistencyBias}
    Under the assumptions of Theorem~\ref{th:Estimators} we have
    \begin{itemize}
        \item[(i)] $N^{-1} \left\| \mathbbm{E}(e' X_k) -   \left(  \widehat e_R' X_k \right)^{\rm truncR} \right\| = o_P(1)$,
        \item[(ii)] $N^{-1} \left\| \mathbbm{E}(e' X_k) \right\| = {\cal O}(1)$,
        \item[(iii)] $N^{-1} {\rm Tr} \left[ P_{f^{\rm red}}  \mathbbm{E}(e' X_k)  \right] = o_P(1)$ .
     \end{itemize}
\end{lemma}

\begin{proof}[\bf Proof of Lemma~\ref{lemma:HelpConsistencyBias}]
   \# \underline{Part (i):}
   For $R=R^0$ statement (i) is identical to Lemma~S.10.5(i) in  Moon and Weidner~\cite*{MoonWeidner2013},
   and the proof there also applies to $R>R^0$, with only one additional issue left to work out: Namely, we have
   $\widehat e_R = \widehat e(\widehat \beta_R) - \widehat \Lambda^{{\rm red} \prime}_R \widehat F^{\rm red}_R $,
   i.e. we have to account for the fact that $R-R^0$ redundant principal components are subtracted from 
   the residuals $\widehat e(\widehat \beta_R)$ that were introduced in Section~\ref{app:expansion1} based on
   the correct number of factors $R^0$.
   The fact that $\widehat \beta_R$ instead of
   $\widehat \beta_{R^0}$ is used to define the residuals makes no difference in the proof, since both are
   $\sqrt{NT}$ consistent under our assumptions.   
   In addition to the proof of Lemma~S.10.5 already provided in Moon and Weidner~\cite*{MoonWeidner2013}
   we therefore also need to show that
   \begin{align*}
        N^{-1} \left\|  \left(  \widehat F^{{\rm red} \prime}_R \widehat \Lambda^{\rm red}_R X_k \right)^{\rm truncR}
                   \right\| = o_P(1).
   \end{align*}  
   We have 
   \begin{align*}
         \widehat \Lambda^{{\rm red} \prime}_R \widehat F^{{\rm red}}_R 
         &= \left[\widehat e(\widehat \beta_R) \right] P_{\widehat F^{\rm red}_R} ,
       &
         \lambda^{{\rm red} \prime} f^{{\rm red}}   
         = M_{\lambda^0} e M_{f^0} P_{f^{\rm red}} 
         = M_{\lambda^0} e P_{f^{\rm red}} .
   \end{align*}
  The current lemma is only used for the proof of the last part of Theorem~\ref{th:Estimators}
  (i.e. $\widehat B_R =  B + o_P(1)$).
   The proof of Theorem~\ref{th:Estimators} below starts by showing
   $\left\|  P_{\widehat F_{R}^{\rm red}} - P_{f^{\rm red}}  \right\|  = {\cal O}_P( N^{-1/6} \log N  )$,
   and we will already make use of this result here.
   Applying Lemma~\ref{lemma:ResExp} we furthermore find that
   $\left\| \widehat e(\widehat \beta_R) - M_{\lambda^0} e M_{f^0} \right\| = {\cal O}_P(1)$. Using this we find that
   \begin{align*}
         \left\|   \widehat \Lambda^{{\rm red} \prime}_R \widehat F^{{\rm red}}_R 
             -  \lambda^{{\rm red} \prime} f^{{\rm red}}  \right\|
         &= \left\|     \left[\widehat e(\widehat \beta_R) \right] P_{\widehat F^{\rm red}_R}
                       - M_{\lambda^0} e M_{f^0} P_{f^{\rm red}}  \right\|
      \\
        &= \left\| \left(    \left[\widehat e(\widehat \beta_R) \right] - M_{\lambda^0} e M_{f^0} \right)
          P_{\widehat F^{\rm red}_R}
               + M_{\lambda^0} e M_{f^0} \left( P_{\widehat F^{\rm red}_R} - P_{f^{\rm red}} \right)
            \right\|        
     \\
       & \leq \left\|   \widehat e(\widehat \beta_R) - M_{\lambda^0} e M_{f^0} \right\|
              +  \left\|   e   \right\|      
              \left\|  P_{\widehat F^{\rm red}_R} - P_{f^{\rm red}} \right\|  
    \\
      &= {\cal O}_P(1) + {\cal O}_P(N^{1/2 - 1/6} \log N)
         =    {\cal O}_P(N^{1/3} \log N)          .        
   \end{align*}
   Let $C=  \widehat \Lambda^{{\rm red} \prime}_R \widehat F^{{\rm red}}_R 
             -  \lambda^{{\rm red} \prime} f^{{\rm red}}$.
    We have shown $\|C\|= {\cal O}_P(N^{1/3} \log N)  $.                    
   For $t=1,\ldots,T$ let
  $C_t$ and $X_{k,t}$ be the $t$'th column of the $N \times T$ matrices $C$ and $X_k$,
  so that $C_t' X_{k,s}$ is the element at position $(t,s)$ in the $T \times T$ matrix $C' X_k $.
  Remember also that we assume $\max_t \left\| X_{k,t} \right\|= {\cal O}_P( \log N \sqrt{N})$.
  Using Lemma S.8.3 in Moon and Weidner~\cite*{MoonWeidner2013}
  we have
   \begin{align*}
        N^{-1} \left\|  \left( C' X_k \right)^{\rm truncR}
                   \right\| 
         &\leq \frac M N \max_{t,s} \left| C_t' X_{k,s} \right|      
           \leq \frac M N \max_{t,s} \left\|  C_t \right\| \left\|    X_{k,s} \right\|
           \leq \frac M N \left\| C \right\| \max_t \left\| X_{k,t} \right\|
       \\
        &=  {\cal O}_P\left( \frac{M  N^{1/3+1/2} (\log N)^2 } {N} \right)     
         =  {\cal O}_P\left( M N^{-1/6} (\log N)^2 \right)  = o_P(1).
   \end{align*}  
   We have
   \begin{align*}
        N^{-1} \left\|  \left(  \widehat F^{{\rm red} \prime}_R \widehat \Lambda^{\rm red}_R X_k \right)^{\rm truncR}
                   \right\| 
         &\leq  
          N^{-1} \left\|  \left(  f^{{\rm red} \prime} \lambda^{\rm red} X_k \right)^{\rm truncR}
                   \right\| 
                   +
        N^{-1} \left\|  \left( C' X_k \right)^{\rm truncR}
                   \right\| 
        \\ &=N^{-1} \left\|  \left(  f^{{\rm red} \prime} \lambda^{\rm red} X_k \right)^{\rm truncR}
                   \right\|  + o_P(1).  
   \end{align*}  
   Thus, what is left to show is that
   $N^{-1} \left\|  \left(  f^{{\rm red} \prime} \lambda^{\rm red} X_k \right)^{\rm truncR} \right\|= o_P(1)$.
   In the notation of Assumption~\ref{ass:EV} we have 
   $f^{{\rm red} \prime} \lambda^{\rm red} = \sum_{r=1}^{R-R^0} \sqrt{\rho_r} w_r v'_r$.
   We have $\sqrt{\rho_r} \leq \|e\| = {\cal O}_P(\sqrt{N})$.
   The distribution of the unit vectors $w_r$ and $w_r$ is characterized by Lemma~\ref{lemma:iidEV},
   from which it is easy to show that
    $  \max_{t} \left| w_{r,t} \right | = {\cal O}_P(N^{-1/2+1/8})$.
   Again using Lemma S.8.3 in Moon and Weidner~\cite*{MoonWeidner2013} we thus find
   \begin{align*}
        N^{-1} \left\|  \left( f^{{\rm red} \prime} \lambda^{\rm red} X_k \right)^{\rm truncR}
                   \right\| 
       &=           N^{-1} \left\|  \left(  \sum_{r=1}^{R-R^0} \sqrt{\rho_r} w_r v'_r X_k \right)^{\rm truncR}
                   \right\|        
        \leq     \sum_{r=1}^{R-R^0} \frac{ \sqrt{\rho_r}  } N \left\|  \left(   w_r v'_r X_k \right)^{\rm truncR}
                   \right\|        
     \\              
         &\leq   \sum_{r=1}^{R-R^0} \frac{ M \sqrt{\rho_r}  } N 
         \left( \max_{t} \left| w_{r,t} \right |  \right)
         \left( \max_t \left| v'_r X_{k,t} \right| \right)
   \\
      &\leq    \sum_{r=1}^{R-R^0} \frac{ M \sqrt{\rho_r}  } N 
         \left( \max_{t} \left| w_{r,t} \right |  \right)
         \left( \max_t  \left\| X_{k,t}  \right\| \right) = o_P(1).
   \end{align*}     
   
   \# \underline{Part (ii):}   
   We have $N^{-1} \mathbbm{E}(e_t' X_{k,s}) = N^{-1} \mathbbm{E}(e_t' \widetilde X^{\rm weak}_{k,s})
    = \sigma^2 \gamma_{s-t}$ for $s>t$, and $=0$ otherwise. 
   Therefore
   $N^{-1} \left\| \mathbbm{E}(e' X_k) \right\| \leq 
   \sqrt{N^{-1} \left\| \mathbbm{E}(e' X_k) \right\|_{1} N^{-1}  \left\| \mathbbm{E}(e' X_k) \right\|_{\infty} }
   \leq \sigma^2 \sum_{t=1}^\infty | \gamma_t | = {\cal O}(1)$.
   
   \# \underline{Part (iii):} 
    In the notation of Assumption~\ref{ass:EV}  we have
    $P_{f^{\rm red}} = P_{(w_1,\ldots,w_{R-R^0})} = \sum_{r=1}^{R-R^0} w_r w_r'$.
    We thus have
    \begin{align*}
        N^{-1} {\rm Tr} \left[ P_{f^{\rm red}}  \mathbbm{E}(e' X_k)  \right] 
        &=  \sum_{r=1}^{R-R^0} \sum_{t=1}^T \sum_{s=t+1}^T  \gamma_{s-t}  w_{r,s} w_{r,t} .
    \end{align*}
    Again, using the distributional characterization of $w_r$ in Lemma~\ref{lemma:iidEV}
    it is easy to show that this term is $o_P(1)$.
\end{proof}

\begin{proof}[\bf Proof of Theorem~\ref{th:Estimators}
(Consistency of Bias and Variance Estimators)]~\\
   \# \underline{Consistency for factors:} 
          We want to apply Lemma~\ref{lemma:EVEC} with
     $A = M_{f^0} e' M_{\lambda^0} e M_{f^0}$
     and $A+B = \widehat e'( \widehat \beta_R ) \widehat e( \widehat \beta_R )$,
     i.e. in the notation of Lemma~\ref{lemma:EVEC} we have
     $f^{\rm red} = (\nu_1,\nu_2,\ldots,\nu_{R-R^0}) H_1$
     and $\widehat F_{R}^{\rm red} = (\widetilde \nu_1,\widetilde \nu_2,\ldots,\widetilde \nu_{R-R^0}) H_2$,
     for some invertible $(R-R^0)\times (R-R^0)$ matrices $H_1$ and $H_2$.
     Note that $\mu_r(A)=\rho_r$ by definition of $\rho_r$ in the main text,
     and by applying Lemma~\ref{lemma:ResExp} and the definition of $B(\beta)$ (which is different from $B$ here)
     in Lemma~\ref{th:expansion2b} we find
     $\mu_r(A+B) = \mu_r(B(\beta) + B'(\beta)) + {\cal O}_P(1)$,
     so that by also applying Lemma~\ref{lemma:JustifyHL2} we find $| \mu_r(A+B)- \mu_r(A) | = {\cal O}_P(1)$.
     Note that Lemma~\ref{lemma:JustifyHL2} only states this for
     $r=1,\ldots,R-R^0$, but it can also be shown for $r=R-R^0+1$ by following the proof in Lemma~\ref{lemma:JustifyHL2},
     which is what we require here, since we want to apply Lemma~\ref{lemma:EVEC} with the $R$ in the Lemma
     equal to $R-R^0$ here. The assumption   $ \left| \mu_r(A+B) -  \mu_r(A) \right| \leq c_1$
     is therefore satisfied with $c_1 = {\cal O}_P(1)$.     
     Following the steps in the proof of  Lemma~\ref{th:expansion2b}
     and Lemma~\ref{lemma:JustifyHL2} we also find that $b=\max_{i,j=1,\ldots,n} |\nu_i' B \nu_j| = {\cal O}_P(1)$
     here. From Johnstone~\cite*{Johnstone2001} and Soshnikov~\cite*{Soshnikov2002} 
     it is known that $(\mu_1(A),\ldots,\mu_{R-R^0+1}(A))$ properly shifted and rescaled (by $N^{1/3}$)
      are jointly asymptotically distributed according to the Tracy-Widom law, which is a continuous distribution,
      from which we can conclude that $\mu_r(A) -  \mu_{r+1}(A)   \geq c_2$ holds for $c_2 = o_P(N^{1/3})$,
      e.g. $c_2 = N^{1/3} / \log(N)^2$.
      By applying Lemma~\ref{lemma:EVEC} we thus obtain
      $  \| \widetilde \nu_r - \nu_r \|^2  = {\cal O}_P( N^{-1/3} (\log N)^2  )$ for $r \in \{1,\ldots,R-R^0\}$.
      We thus also have
      \begin{align*}
          \left\|  P_{\widehat F_{R}^{\rm red}} - P_{f^{\rm red}}  \right\|  = 
          \left\|  P_{(\nu_1,\ldots,\nu_{R-R^0})} - P_{(\widetilde \nu_1,\ldots,\widetilde \nu_{R-R^0})}  \right\|  
          = {\cal O}_P( N^{-1/6} \log N  ) .
      \end{align*}
      Together with the above result $\left\| P_{\widehat F^0_{R}} - P_{f^0} \right\| = {\cal O}_P(1/\sqrt{N}) $ 
      we thus have
      \begin{align*}
          \left\| P_{\widehat F_R} - P_{[f^0,f^{\rm red}]} \right\| = {\cal O}_P( N^{-1/6} \log N  ) = o_P(1),
      \end{align*}
      which is the first statement of Theorem~\ref{th:Estimators}. 
     Note that       $f^0$ and $f^{\rm red}$ are orthogonal (i.e. we have $f^{0 \prime} f^{\rm red} = 0$),
     so that $P_{[f^0,f^{\rm red}]} = P_{f^0} + P_{f^{\rm red}}$.
      
    \# \underline{Consistency for factors loadings:} The problem is symmetric under exchange of $N \leftrightarrow T$,
    so analogous to the proof for the factors one also finds $\left\| P_{\widehat \Lambda_R} - P_{[\lambda^0,\lambda^{\rm red}]} \right\| = o_P(1)$.
    
    \# \underline{Consistency of $\widehat \sigma_R^2$:} 
    Using the definition of  ${\cal L}^R_{NT}(\beta) $, Lemma~\ref{lemma:expansion2a}, 
    Theorem~\ref{th:expansion}, the definition of ${\cal L}^0_{NT}( \beta^0)$,
    and the WLLN we obtain
    \begin{align*}
   \frac 1 {NT} \sum_{i=1}^N \sum_{t=1}^T 
       \left( \widehat e_{R,it}  \right)^2  
  &= {\cal L}^R_{NT}( \widehat \beta_R)  
         = {\cal L}^0_{NT}( \widehat \beta_R) + {\cal O}_P\left( \frac{\|e\|^2}{NT} + \frac{\sqrt{N}} {NT}  \right) 
       =  {\cal L}^0_{NT}( \beta^0) + {\cal O}_P(1/N)
     \\
       &= \frac 1 {NT} {\rm Tr}( M_{\widehat \lambda} e   M_{\widehat f} e' ) + {\cal O}_P(1/N)
       =  \frac 1 {NT} {\rm Tr}( ee' )  + {\cal O}_P\left( \frac{\|e\|^2}{NT} + 1/N \right)  
    \\
      &= \frac 1 {NT} \sum_{i=1}^N \sum_{t=1}^T e_{it}^2 + {\cal O}_P(1/N)  
        = \frac 1 {NT} \sum_{i=1}^N \sum_{t=1}^T \mathbbm{E}(e_{it}^2) + o_P(1)  
     \\   
      & = \sigma^2 + o_P(1). 
    \end{align*}
    We thus also have $\widehat \sigma_R^2 = \sigma^2 + o_P(1)$. 
    
    \# \underline{Consistency of $\widehat W_R$:} Using 
    that $M_{\widehat F_R} - M_{[f^0,f^{\rm red}]} = P_{[f^0,f^{\rm red}]} - P_{\widehat F_R}$
    and 
    $M_{\widehat \Lambda_R} - M_{[\lambda^0,\lambda^{\rm red}]}
      =P_{[\lambda^0,\lambda^{\rm red}]} -  P_{\widehat \Lambda_R}$
     are low rank matrices (rank $\leq 2R$) and satisfy the spectral norm bounds     
     $\left\| M_{\widehat F_R} - M_{[f^0,f^{\rm red}]} \right\|  
    = o_P(1)$
    and
$\left\| M_{\widehat \Lambda_R} - M_{[\lambda^0,\lambda^{\rm red}]} \right\| = o_P(1)$ we obtain
    \begin{align*}
        \widehat W_{R,k_1 k_2}  &= \frac 1 {NT}
      {\rm Tr}\left( M_{\widehat \Lambda_R} X_{k_1} M_{\widehat F_R} X_{k_2}'  \right)
    \\  
        &= \frac 1 {NT} {\rm Tr}\left( M_{[\lambda^0,\lambda^{\rm red}]} X_{k_1} M_{[f^0,f^{\rm red}]} X_{k_2}'  \right)
           + o_P\left( \frac{( \| X_{k_1} \| +  \| X_{k_2} \| )^2} {NT} \right)
    \\
       &=    \frac 1 {NT} {\rm Tr}\left( M_{[\lambda^0,\lambda^{\rm red}]} X_{k_1} M_{[f^0,f^{\rm red}]} X_{k_2}'  \right)
          + o_P(1)
    \\
       &=    \underbrace{ \frac 1 {NT} {\rm Tr}\left( M_{\lambda^0} X_{k_1} M_{f^0} X_{k_2}'  \right)  }_{=W_{k_1 k_2}}      
           -  \frac 1 {NT} {\rm Tr}\left( P_{\lambda^{\rm red}} X_{k_1} M_{f^0} X_{k_2}'  \right) 
        \\ & \qquad \qquad     
           -  \frac 1 {NT} {\rm Tr}\left( M_{\lambda^0} X_{k_1} P_{f^{\rm red}}  X_{k_2}'  \right)   
           +  \frac 1 {NT} {\rm Tr}\left( P_{\lambda^{\rm red}} X_{k_1} P_{f^{\rm red}} X_{k_2}'  \right)   ,
    \end{align*}
    where in the last step we used that 
    $M_{[\lambda^0,\lambda^{\rm red}]} = M_{\lambda^0} - P_{\lambda^{\rm red}}$
    and $M_{[f^0,f^{\rm red}]} = M_{f^0} -   P_{f^{\rm red}}$.
    Remember that $P_{f^{\rm red}} = P_{(\nu_1,\nu_2,\ldots,\nu_{R-R^0})}$ in the notation used here,
    and $P_{f^{\rm red}} = P_{(w_1,w_2,\ldots,w_{R-R^0})}$ in the notation of Assumption~\ref{ass:EV}.
    Using the characterization of the distribution of $w_s$ given in Lemma~\ref{lemma:iidEV},
    and the methods used in the proof of Lemma~\ref{lemma:JustifyHL2} we obtain
    $\left\| X_{k} P_{f^{\rm red}} \right\| = o_P(\sqrt{NT})$,
    and analogously $\left\| P_{\lambda^{\rm red}}  X_{k} \right\| = o_P(\sqrt{NT})$.
    We thus obtain
    $$\left| \frac 1 {NT} {\rm Tr}\left( P_{\lambda^{\rm red}} X_{k_1} M_{f^0} X_{k_2}'  \right) 
    \right| \leq \frac{R}{NT} \left\| P_{\lambda^{\rm red}} X_{k_1}  \right\| \|X_{k_2}\| = o_P(1),$$
    and analogously we find $ \frac 1 {NT} {\rm Tr}\left( M_{\lambda^0} X_{k_1} P_{f^{\rm red}}  X_{k_2}'  \right)   = o_P(1)$
    and $ \frac 1 {NT} {\rm Tr}\left( P_{\lambda^{\rm red}} X_{k_1} P_{f^{\rm red}} X_{k_2}'  \right) = o_P(1)$.
    Combining this with the above result for 
    $\widehat W_{R,k_1 k_2} $ gives $ \widehat W_R = W + o_P(1)$.
   
   \# \underline{Consistency $\widehat B_R$:}
      Remember that $B_k = \frac 1 N {\rm Tr}[ P_{f^0} \mathbbm{E}(e'  X_k)]$
      and 
    $\widehat B_{R,k}  =\frac 1 N {\rm Tr}\left[ P_{\widehat F_R} \left(  \widehat e_R' X_k \right)^{\rm truncR} \right]$.   
   Using Lemma~\ref{lemma:HelpConsistencyBias} we find
   \begin{align*}
        \widehat B_{R,k}
        &=  \frac 1 N {\rm Tr}\left[ P_{\widehat F_R} \mathbbm{E}(e' X_k)  \right]  + o_P(1)
        && \text{(using part (i) of the lemma)}
      \\
        &=   \frac 1 N {\rm Tr}\left[ P_{[f^0,f^{\rm red}]} \mathbbm{E}(e' X_k)  \right]  + o_P(1)  
        && \text{(using part (ii) and $\left\| P_{\widehat F_R} - P_{[f^0,f^{\rm red}]} \right\|  = o_P(1)$)}
      \\
       &=  B_k +  \frac 1 N {\rm Tr}\left[ P_{f^{\rm red}} \mathbbm{E}(e' X_k)  \right]  + o_P(1) 
       && \text{(using $P_{[f^0,f^{\rm red}]} = P_{f^0} + P_{f^{\rm red}}$)}
     \\
       &=   B_k + o_P(1) ,
        && \text{(using part (iii))} 
   \end{align*}
   which is the desired consistency result for $\widehat B_R$.
\end{proof}

\section{Proofs for Intermediate Results}
\label{app:Proofs-Supp}
\subsection{Proofs for Expansions of ${\cal L}_{NT}^0(\beta)$,
$M_{\widehat \lambda}(\beta)$, $M_{\widehat f}(\beta)$ and $\widehat e(\beta)$
}

\begin{proof}[\bf Proof of Lemma \ref{lemma:expansion}]$\phantom{1}$

   \begin{itemize}
      \item[(i,ii)]
      We apply perturbation theory in Kato \cite*{Kato}. The unperturbed operator is
      ${\cal T}^{(0)}=\lambda^0 f^{0 \prime} f^0 \lambda^{0\prime}$,
      the perturbed operator is
      ${\cal T}={\cal T}^{(0)}+{\cal T}^{(1)}+{\cal T}^{(2)}$
      (\textit{i.e.} the parameter $\kappa$ that appears in Kato is set to 1),
      where ${\cal T}^{(1)} = \sum_{k=0}^K \epsilon_k X_k f^{0} \lambda^{0\prime} +
                       \lambda^0 f^{0\prime} \sum_{k=0}^K \epsilon_k X'_k$,
      and ${\cal T}^{(2)} = \sum_{k_1=0}^K \sum_{k_2=0}^K \epsilon_{k_1} \epsilon_{k_2}
                   X_{k_1} X'_{k_2}$.
      The matrices ${\cal T}$ and ${\cal T}^{0}$ are real and symmetric (which implies that they are normal operators),
      and positive semi-definite.
      We know that ${\cal T}^{(0)}$ has an eigenvalue $0$ with multiplicity $N-R^0$,
      and the separating distance of this eigenvalue is $d = N T d_{\min}^2(\lambda^0,f^0)$.
      The bound \eqref{exp_bound} guarantees that
      \begin{align}
         \| {\cal T}^{(1)}+{\cal T}^{(2)} \| &\leq \frac {NT} 2 \, d_{\min}^2(\lambda^0,f^0) \; .
      \end{align}
      By Weyl's inequality we therefore find that the $N-R^0$ smallest eigenvalues of
      ${\cal T}$ (also counting
      multiplicity) are all smaller than $\frac{NT} 2 d_{\min}^2(\lambda^0,f^0)$, and they ``originate''
      from the zero-eigenvalue of ${\cal T}^{(0)}$, with the power series expansion for ${\cal L}^0_{NT}(\beta)$
      given in (2.22) and (2.18) at p.77/78 of Kato, and the expansion of $M_{\widehat \lambda}$
      given in (2.3) and (2.12) at p.75,76 of Kato. We still need to justify the convergence radius of
      this series. Since we set the complex parameter $\kappa$ in Kato to 1, we need to show that the convergence
      radius ($r_0$ in Kato's notation) is at least 1.
      The condition (3.7) in Kato p.89 reads
      $\| {\cal T}^{(n)} \| \leq a c^{n-1}$, $n=1,2,\ldots$, and it is satisfied for
      $a=2 \sqrt{N T} d_{\max}(\lambda^0,f^0) \sum_{k=0}^K |\epsilon_k| \|X_k\|$
      and $c=\sum_{k=0}^K |\epsilon_k| \|X_k\| / \sqrt{N T} / 2 / d_{\max}(\lambda^0,f^0)$.
      According to equation (3.51) in Kato p.95, we therefore find that the power series for
      ${\cal L}^0_{NT}(\beta)$ and $M_{\widehat \lambda}$ are convergent ($r_0 \geq 1$ in his notation) if
      $1 \leq \left( \frac{2 a} {d} + c \right)^{-1}$,
      and this becomes exactly our condition \eqref{exp_bound}.

      When ${\cal L}^0_{NT}\left( \beta \right)$ is approximated up to order $G \in \mathbb{N}$, Kato's equation
      (3.6) at p.89 gives the following bound on the remainder
      \begin{align}
        \bigg| {\cal L}^0_{NT}\left( \beta \right) - \frac 1 {NT} \, \sum_{g=2}^G \, \sum_{k_1=0}^K \,
                                       \ldots \, \sum_{k_g=0}^K \,
                                     \epsilon_{k_1} \, \ldots \, \epsilon_{k_g} \,
                       L^{(g)} &\left(\lambda^0,\, f^0,\, X_{k_1},\, X_{k_2}, \ldots ,X_{k_g}\right) \bigg|
        \nonumber \\ &
          \leq  \frac{(N-R^0) \gamma^{G+1} \, d_{\min}^2(\lambda^0,f^0) }
                     { 4 (1-\gamma) } \; ,
        \label{bound_rem_1}
     \end{align}
     where
     \begin{align}
        \gamma \, &= \, \frac{ \sum_{k=1}^{K}\left| \beta^0_{k} - \beta_{k} \right|
            \frac{\| X_{k} \|} {\sqrt{NT}}  + \frac{\|e\|} {\sqrt{NT}} }
                     {r_0(\lambda^0,f^0)}   \,< \, 1 \; .
     \end{align}
     This bound again shows convergence of the series expansion, since
     $\gamma^{G+1} \rightarrow 0$ as $G\rightarrow \infty$. Unfortunately, for our purposes
     this is not a good bound since it still involves the factor $N-R^0$ (in Kato this factor is hidden since
     his $\widehat \lambda(\kappa)$ is the average of the eigenvalues, not the sum),
     but as we show below this can be avoided.
  \item[(iii,iv)]
     We have $\|S^{(m)}\| = \left( NT d_{\min}^2(\lambda^0,f^0) \right)^{-m}$,
     $\| {\cal T}^{(1)}_{k} \| \leq 2 \sqrt{N T} d_{\max}(\lambda^0,f^0) \|X_k\|$,
     and $\| {\cal T}^{(2)}_{k_1k_2} \| \leq \|X_{k_1}\| \|X_{k_2} \|$.
     Therefore
     \begin{align}
        &\left\| S^{(m_1)} \, {\cal T}^{(\nu_1)}_{k_1 \ldots} \, S^{(m_2)} \, \ldots
         \, S^{(m_P)} \, {\cal T}^{(\nu_P)}_{\ldots k_g} \, S^{(m_{p+1})} \right\|
        \nonumber \\ & \qquad
         \leq \left( NT d_{\min}^2(\lambda^0,f^0) \right)^{-\sum m_j}
              \left( 2 \sqrt{N T} d_{\max}(\lambda^0,f^0) \right)^{2 p-\sum \nu_j}
              \|X_{k_1}\| \|X_{k_2} \| \ldots \| X_{k_g} \| \; .
       \label{BoundSTS}
     \end{align}
     We have
     \begin{align}
        \sum_{\begin{minipage}{3cm}\center\scriptsize
                            $\nu_1+\ldots+\nu_P=g$ \\
                            $2 \geq \nu_j \geq 1$ \end{minipage}} 1
        &\leq 2^p \; ,
        \nonumber \\
        \sum_{\begin{minipage}{3cm}\center\scriptsize
                            $m_1+\ldots+m_{p+1}=p-1$\\$m_j \geq 0$\end{minipage}} 1
         & \leq
        \sum_{\begin{minipage}{3cm}\center\scriptsize
                            $m_1+\ldots+m_{p+1}=p$\\$m_j \geq 0$\end{minipage}} 1
         = \frac{(2p)!} {(p!)^2} \leq 4^p        \; .
        \label{CombSumBound}
     \end{align}
     Using this we find\footnote{The sum over p only starts from $\lceil g/2 \rceil$, the smallest integer
     larger or equal $g/2$, because $\nu_1+\ldots+\nu_P=g$ can not be satisfied for smaller $p$, since
     $\nu_j \leq 2$.}
     \begin{align}
       & \left\| M^{(g)} \left(\lambda^0,\, f^0,\, X_{k_1},\, X_{k_2}, \ldots,X_{k_g}\right) \right\|
       \nonumber \\ & \qquad \quad
         \leq \left( 2 \sqrt{N T} d_{\max}(\lambda^0,f^0) \right)^{-g}
              \|X_{k_1}\| \|X_{k_2} \| \ldots \| X_{k_g} \|
     \sum_{p=\lceil g/2 \rceil}^g
      \left( \frac{ 32 \, d_{\max}^2(\lambda^0,f^0) } {d_{\min}^2(\lambda^0,f^0)} \right)^{p}
       \nonumber \\ & \qquad \quad
         \leq \frac{g} {2} \,
             \left(\frac{16 \; d_{\max}(\lambda^0,f^0)} {d_{\min}^2(\lambda^0,f^0)} \right)^{g}
               \frac{\|X_{k_1}\|} {\sqrt{NT}} \,
               \frac{\|X_{k_2}\|} {\sqrt{NT}} \, \ldots \,
               \frac{\|X_{k_g}\|} {\sqrt{NT}} \; .
        \label{BoundMq}
     \end{align}
     For $g\geq 3$ there always appears at least one factor $S^{(m)}$, $m\geq 1$, inside the trace of the terms
     that contribute to $L^{(g)}$, and we have ${\rm rank}(S^{(m)})=R^0$ for $m\geq 1$.
     Using ${\rm Tr}(A)\leq {\rm rank}(A) \|A\|$, and the equations  \eqref{BoundSTS} and
     \eqref{CombSumBound}, we therefore find\footnote{The calculation for the bound of $L^{(g)}$ is almost
     identical to the one for $M^{(g)}$. But now there appears an additional factor $R^0$ from the rank,
     and since $\sum m_j = p-1$ (not $p$ as before), there is also an additional factor $NT d_{\min}^2(\lambda^0,f^0)$.}
     for $g \geq 3$
     \begin{align}
        & \frac {1} {NT}
        \left| L^{(g)}\left(\lambda^0,\, f^0,\, X_{k_1},\, X_{k_2}, \ldots ,X_{k_g}\right) \right|
        \nonumber \\ & \qquad
         \leq R^0 \, d_{\min}^2(\lambda^0,f^0) \,
              \left(2 \sqrt{NT} d_{\max}(\lambda^0,f^0) \right)^{-g}
           \nonumber \\ & \qquad\qquad\qquad\qquad
              \|X_{k_1}\| \|X_{k_2} \| \ldots \| X_{k_g} \|
              \sum_{p=\lceil g/2 \rceil}^g
                         \left(\frac{32 \; d_{\max}^2(\lambda^0,f^0)} {d_{\min}^2(\lambda^0,f^0)} \right)^{p}
        \nonumber \\ & \qquad
         \leq \frac{R^0 \, g \, d_{\min}^2(\lambda^0,f^0)} {2} \,
             \left(\frac{16 \; d_{\max}(\lambda^0,f^0)} {d_{\min}^2(\lambda^0,f^0)} \right)^{g}
               \frac{\|X_{k_1}\|} {\sqrt{NT}} \,
               \frac{\|X_{k_2}\|} {\sqrt{NT}} \, \ldots \,
               \frac{\|X_{k_g}\|} {\sqrt{NT}} \; .
     \end{align}
     This implies for $g \geq 3$
     \begin{align}
        \frac{1} {NT} & \left| \sum_{k_1=0}^K \,
                      \sum_{k_2=0}^K \, \ldots
                      \sum_{k_g=0}^K \,
                      \epsilon_{k_1} \, \epsilon_{k_2} \, \ldots \, \epsilon_{k_g} \,
                    L^{(g)}\left(\lambda^0,\, f^0,\, X_{k_1},\, X_{k_2}, \ldots ,X_{k_g}\right) \right|
          \nonumber \\ & \qquad \qquad\qquad
           \leq \frac{R^0 \, g \, d_{\min}^2(\lambda^0,f^0)} {2} \,
             \left(\frac{16 \; d_{\max}(\lambda^0,f^0)} {d_{\min}^2(\lambda^0,f^0)} \right)^{g}
             \left( \sum_{k=0}^K
               \frac{\|\epsilon_k \, X_{k}\|} {\sqrt{NT}}  \right)^g .
     \end{align}
     Therefore for $G \geq 2$ we have
     \begin{align}
        & \bigg| {\cal L}^0_{NT}\left( \beta \right) - \frac 1 {NT} \, \sum_{g=2}^G \, \sum_{k_1=0}^K \,
                                       \ldots \, \sum_{k_g=0}^K \,
                                     \epsilon_{k_1} \, \ldots \, \epsilon_{k_g} \,
                       L^{(g)} \left(\lambda^0,\, f^0,\, X_{k_1},\, X_{k_2}, \ldots ,X_{k_g}\right) \bigg|
        \nonumber \\ & \qquad
         = \frac{1} {NT} \, \sum_{g=G+1}^\infty \, \sum_{k_1=0}^K \,
                      \sum_{k_2=0}^K \, \ldots
                      \sum_{k_g=0}^K \,
                      \epsilon_{k_1} \, \epsilon_{k_2} \, \ldots \, \epsilon_{k_g} \,
                    L^{(g)}\left(\lambda^0,\, f^0,\, X_{k_1},\, X_{k_2}, \ldots ,X_{k_g}\right)
        \nonumber \\ & \qquad
        \leq \sum_{g=G+1}^\infty
            \frac{R^0 \, g \, \alpha^g \, d_{\min}^2(\lambda^0,f^0)} {2} \,
        \nonumber \\ & \qquad
        \leq  \frac{R^0 \, (G+1)\, \alpha^{G+1} \, d_{\min}^2(\lambda^0,f^0)} {2 (1-\alpha)^2} \; ,
        \label{BoundLNT}
     \end{align}
     where
     \begin{align}
        \alpha \, &= \frac{16 \; d_{\max}(\lambda^0,f^0)} {d_{\min}^2(\lambda^0,f^0)}
                     \sum_{k=0}^K \frac{\|\epsilon_k \, X_{k}\|} {\sqrt{NT}}
           \nonumber \\
                  &= \frac{16 \; d_{\max}(\lambda^0,f^0)} {d_{\min}^2(\lambda^0,f^0)}
                      \, \left( \sum_{k=1}^{K}\left| \beta^0_{k} - \beta_{k} \right|
            \frac{\| X_{k} \|} {\sqrt{NT}}  + \frac{\|e\|} {\sqrt{NT}} \right) \; < 1 \; .
     \end{align}
     Using the same argument we can start from equation \eqref{BoundMq} to obtain the bound%
     for the remainder of the series expansion for $M_{\widehat\lambda} \left( \beta \right)$.

     Note that compared to the bound \eqref{bound_rem_1} on the remainder, the new bound \eqref{BoundLNT}
     only shows convergence of the power series within the
     the smaller convergence radius
     $\frac{d_{\min}^2(\lambda^0,f^0)} {16 \; d_{\max}(\lambda^0,f^0)}  < r_0(\lambda^0,f^0)$.
     However, the factor $N-R^0$ does not appear in this new bound, which is crucial for our approximations.
  \end{itemize}
  \vspace{-15pt}
\end{proof}

\begin{proof}[\bf Proof of Lemma~\ref{lemma:ProExp}]
   The general expansion of $M_{\widehat \lambda}(\beta)$ is given in
   Lemma \ref{lemma:expansion}. The present Lemma just makes this
   expansion explicit for the first few orders.
   The bound on the remainder $M_{\widehat \lambda}^{({\rm rem})}(\beta)$
   is obtained from the bound \eqref{BoundMq} by the same logic as in the
   proof of Theorem~\ref{th:expansion}.
   The analogous result for $M_{\widehat f}(\beta)$ is obtained by applying the symmetry
   $N \leftrightarrow T$, $\lambda \leftrightarrow f$, $e \leftrightarrow e'$,
   $X_k \leftrightarrow X_k'$.
\end{proof}

\begin{proof}[\bf Proof of Lemma~\ref{lemma:ResExp}]
   The general expansion of $M_{\widehat \lambda}(\beta)$ is given in
   Lemma \ref{lemma:expansion}, and the analogous
   expansion for $M_{\widehat f}(\beta)$ is obtained by applying the symmetry
   $N \leftrightarrow T$, $\lambda \leftrightarrow f$, $e \leftrightarrow e'$,
   $X_k \leftrightarrow X_k'$. Lemma~\ref{lemma:ProExp} above
   provides a more explicit version of these projector expansions.
   For the residuals $\widehat e(\beta)$ we have
   \begin{align}
     \widehat e(\beta) &= M_{\widehat \lambda}(\beta) \, \left( Y - \beta \cdot X \right)
         \, M_{\widehat f}(\beta)
        = M_{\widehat \lambda}(\beta) \, \left[ e -
         \left( \beta - \beta^0 \right) \cdot X
                                 +  \lambda^0 f^{0\prime} \right]
          \, M_{\widehat f}(\beta) \; ,
   \end{align}
   and plugging in the expansions of $M_{\widehat \lambda}(\beta)$
   and $M_{\widehat f}(\beta)$ it is straightforward
   to derive the expansion of $\widehat e(\beta)$ from this,
   including the bound on the remainder.
\end{proof}

\subsection{Proofs for $N^{3/4}$ Convergence Rate Result}

\begin{proof}[\bf Proof of Lemma~\ref{lemma:expansion2a}]
   The result  follows from
   Lemma~\ref{th:expansion2b}
   by applying Weyl's inequality,
   because the terms in $B(\beta)+B'(\beta)$ in addition to $A(\beta)$
   all have a spectral norm of order ${\cal O}_P(\sqrt{N})$
   for $\sqrt{N} \|\beta - \beta^0 \| \leq c$. 
\end{proof}

\begin{proof}[\bf Proof of Theorem~\ref{th:ConvergenceRate}]
   From Theorem \ref{th:consistency} we know
   that $\sqrt{N} ( \widehat \beta_R - \beta_0 )={\cal O}_P(1)$,
   so that the bounds in Lemma~\ref{lemma:expansion2a}
   and Assumption~\ref{ass:HL1} are applicable.
   Since $\widehat \beta_R$ minimizes ${\cal L}^R_{NT}(\beta)$
   it must in particular satisfy
   ${\cal L}^R_{NT}(\widehat \beta_R) \leq {\cal L}^R_{NT}(\beta^0)$.
   Applying this, Lemma~\ref{lemma:expansion2a}, and  Assumption~\ref{ass:HL1}
   we obtain
   \begin{align}
        0 &\geq {\cal L}^R_{NT}(\widehat \beta_R) - {\cal L}^R_{NT}(\beta^0)
     \nonumber \\
          &=    {\cal L}^0_{NT}(\widehat \beta_R) - {\cal L}^0_{NT}(\beta^0)
             - \frac 1 {NT}  \; \sum_{r=1}^{R-R^0}
            \left[ \mu_r\left( A(\widehat \beta_R)  \right) -  \mu_r\left( A(\beta^0)  \right)  \right]
       \nonumber \\ & \qquad
                   + \frac 1 {NT} {\cal O}_P\left[ \sqrt{N}+ \sqrt{NT} \| \widehat \beta_R-\beta^0\|
                \right]
     \nonumber \\
          &\geq    {\cal L}^0_{NT}(\widehat \beta_R) - {\cal L}^0_{NT}(\beta^0)
             - \frac 1 {NT}  \; \sum_{r=1}^{R-R^0}
               \mu_r\left[ M_{f^0} \left( \Delta \beta \cdot X \right)'
                  M_{\lambda^0}  \left(  \Delta \beta \cdot X \right)  M_{f^0} \right]
      \nonumber \\ & \qquad
             +\frac 1 {NT} {\cal O}_P\left[ \sqrt{N}+ N^{5/4} \| \widehat \beta_R-\beta^0\|
                + N^2 \| \widehat \beta_R-\beta^0\| / \log(N)
                     \right] .
   \end{align}
   Applying Theorem \ref{th:expansion} then gives
   \begin{align}
     & \left(\widehat \beta_R-\beta^0 \right)'  W  \left(\widehat \beta_R-\beta^0 \right)
       -  \frac 2 {\sqrt{NT}}   \left(\widehat \beta_R-\beta^0 \right)'
                           \left( C^{(1)} + C^{(2)} \right)
    \nonumber \\
     & \quad  \leq \frac 1 {NT}
     \bigg\{ \sum_{r=1}^{R-R^0}   \mu_r\left[ M_{f^0} \left( \Delta \beta \cdot X \right)'
                  M_{\lambda^0}  \left(  \Delta \beta \cdot X \right)  M_{f^0} \right]
  \nonumber \\ & \qquad      \qquad   \quad
         + {\cal O}_P\left[ \sqrt{N}+ N^{5/4} \| \widehat \beta_R-\beta^0\|
                + N^2 \| \widehat \beta_R-\beta^0\| / \log(N)
                     \right]  \bigg\}.
   \end{align}
   Our assumptions guarantee $C^{(2)} = {\cal O}_P(1)$,
   and we explicitly assume $C^{(1)} = {\cal O}_P(N^{1/4})$.
   Furthermore, Assumption~\ref{ass:NC} guarantees that
   \begin{align}
      (\Delta \beta)'  W  (\Delta \beta)
       - \frac 1 {NT}\sum_{r=1}^{R-R^0}
               \mu_r\left[ M_{f^0} \left( \Delta \beta \cdot X \right)'
                  M_{\lambda^0}  \left(  \Delta \beta \cdot X \right)  M_{f^0} \right]
                  \geq b \| \Delta \beta \|^2 ,
   \end{align}
   which we apply for $\Delta \beta =  \widehat \beta_R-\beta^0  $.   
   Thus, we obtain
   \begin{align}
      b \left( N^{3/4} \| \widehat \beta_R-\beta^0 \| \right)^2
        \leq {\cal O}_P\left( 1 \right)
              + {\cal O}_P\left( N^{3/4} \| \widehat \beta_R-\beta^0\| \right)
              + o_P\left[ \left( N^{3/4} \| \widehat \beta_R-\beta^0 \| \right)^2
                       \right],
   \end{align}
   from which we can conclude that $N^{3/4} \| \widehat \beta_R-\beta^0 \| = {\cal O}_P(1)$,
   which proves the lemma.
\end{proof}

\begin{proof}[\bf Proof of Lemma~\ref{lemma:help}]
    Note that $P_g = g g'$ and $P_h = h h'$. We decompose
   \begin{align}
      \left( U + g C  h' \right)' \left( U + g C  h' \right)
       &= A_1+A_2(C) \; ,
   \end{align}
   where  
   \begin{align}
      A_1 &\equiv U' U +  \| g' U U' g \| \, P_{\left( M_{U' g} h \right)} +  \Delta_{\max}  P_{(U' g)}  \; ,
    \nonumber \\
      A_2(C) &\equiv \left( U + g C  h' \right)' P_g \left( U + g C  h' \right)
             - U' P_g U  -  \| g' U U' g \| \, P_{\left( M_{U' g} h \right)} - \Delta_{\max}  P_{(U' g)}   .
   \end{align}
   By Weyl's  inequality we then have
   \begin{align}
       \sum_{r=1}^R
       \mu_r \left[ \left( U + g C  h' \right)' \left( U + g C  h' \right) \right]
        \leq
         \sum_{r=1}^R \mu_r(A_1)
        +\sum_{r=1}^R \mu_r\left[ A_2(C)) \right] .
     \label{p_inequ_1}
   \end{align}
   We have $A_2(C) = P_{(h, \, U' g)} A_2(C) P_{(h, \, U' g)}$,
   i.e. $A_2(C)$ has $T-2 Q$ zero-eigenvalues
   and only $2 Q$ non-zero eigenvalues.  Let
   $\widetilde h =  (h, U' g) [(h, U' g)'(h, U' g)]^{-1/2}$, which is a $T \times 2Q$ matrix that satisfies
   $\widetilde h' \widetilde h = \mathbbm{1}_{2Q}$ and $\widetilde h \widetilde h' = P_{(h, \, U' g)}$. We then
   have
   \begin{align}
         \sum_{r=1}^R \mu_r\left[ A_2(C)) \right]
         &=  \sum_{r=1}^{\min(R,2Q)} \mu_r\left[ \widetilde h' A_2(C)) \widetilde h\right] ,
     \label{p_inequ_2}
   \end{align}
   and
   \begin{align}
      &\sum_{r=1}^{\min(R,2Q)} \mu_r\left[ \widetilde h' A_2(C)) \widetilde h\right]
     \nonumber \\
        &\leq   \sum_{r=1}^{\min(R,2Q)} \mu_r\left[  \widetilde h'  \left( U + g C  h' \right)' P_g \left( U + g C  h' \right)
         \widetilde h \right]
           \nonumber \\ &\qquad
               +  \sum_{r=1}^{\min(R,2Q)}
                \mu_r\left[  \widetilde h'  \left( - U' P_g U  - \| g' U U' g \| \, P_{\left( M_{U' g} h \right)}
                -  \Delta_{\max}  P_{(U' g)}  \right)
                 \widetilde h  \right]
      \nonumber \\
         &=   \sum_{r=1}^{\min(R,Q)} \mu_r\left[ g' \left( U + g C  h' \right)\left( U + g C  h' \right)' g \right]
           \nonumber \\ &\qquad
               -  \sum_{r=2Q-\min(R,2Q)+1}^{2Q}
                \mu_r\left[  \widetilde h'  \left( U' P_g U  + \| g' U U' g \| \, P_{\left( M_{U' g} h \right)}
                +  \Delta_{\max}  P_{(U' g)}  \right)
                 \widetilde h  \right] .
   \end{align}
   Here, in the first step we again used Weyl's  inequality,
   and in the second step we used that
   the $Q$ non-zero eigenvalues of
   $\widetilde h' \left( U + g C  h' \right)' g g' \left( U + g C  h' \right) \widetilde h$ are identical
   to the eigenvalues of $g' \left( U + g C  h' \right)\left( U + g C  h' \right)' g$,
   and that the eigenvalues of a matrix
     are equal to minus the eigenvalues of the negative of the matrix (but interchanging
     the ordering of the eigenvalues).

   The eigenvalues of     $ \widetilde h'  \left( U' P_g U  + \| g' U U' g \| \, P_{\left( M_{U' g} h \right)}
                +  \Delta_{\max}  P_{(U' g)}  \right)
                 \widetilde h $
   are given by $Q$ eigenvalues equal to
    $ \| g' U U' g \| $ (stemming from         $ \| g' U U' g \| \, P_{\left( M_{U' g} h \right)}$),
    while the remaining $Q$ eigenvalues
    are given by  $\mu_r(U' P_g U) + \Delta_{\max}$, $r=1,\ldots,Q$,
    and satisfy
    $\mu_{r+R-\min(Q,R)}(U' P_g U) + \Delta_{\max} \geq \mu_{r} (U' P_g U)$,
    for $r \in \{1,2,\ldots,\min(R,Q)\}$ (by the definition of $\Delta_{\max}$).
    Therefore we have
    \begin{align}
         \sum_{r=2Q-\min(R,2Q)+1}^{2Q}
                \mu_r\left[  \widetilde h'  \left( U' P_g U  + \| g' U U' g \| \, P_{\left( M_{U' g} h \right)}
                +  \Delta_{\max}  P_{(U' g)}  \right)
                 \widetilde h  \right]
           \geq    \sum_{r=1}^{\min(R,Q)}  \mu_{r} (U' P_g U)   .
    \end{align}
    We can thus conclude that
     \begin{align}
      &\sum_{r=1}^{\min(R,2Q)} \mu_r\left[ \widetilde h' A_2(C)) \widetilde h\right]
        \nonumber \\
         &\leq   \sum_{r=1}^{\min(R,Q)} \mu_r\left[ g' \left( U + g C  h' \right)\left( U + g C  h' \right)' g \right]
                         -   \sum_{r=1}^{\min(R,Q)}  \mu_{r} (g' U U' g)
      \nonumber \\
        & \leq
        \sum_{r=1}^{\min(Q,R)} \mu_r\left( C C' + g' U h C' + C h' U' g \right) .
      \label{p_inequ_3}
   \end{align}
   Combining \eqref{p_inequ_1}, \eqref{p_inequ_2} and \eqref{p_inequ_3} gives the statement
   of the lemma.
\end{proof}

\begin{proof}[\bf Proof of Lemma~\ref{lemma:LowRankAdd}]
     Let $h$ be a $T \times Q$ matrix whose span equals the span of $A$, i.e. $P_h A= A$,
     and that satisfies $h' h = \mathbbm{1}_Q$, and let $\rho = \|A\|/T$.
     Then $A \leq T  \rho P_h  $, which implies
      $\sum_{r=1}^R \mu_r \left( e' e + A \right) \leq \sum_{r=1}^R \mu_r \left( e' e +  T  \rho P_h \right)$.

     The distribution of $e$ is  invariant under orthogonal  transformations $e \mapsto e O$, where
     $O$ is an arbitrary orthogonal $T \times T$ matrix, i.e.~$O O' = \mathbbm{1}_T$. The distribution of the
     eigenvalues of $e' e + T  \rho P_h$ therefore does not depend on $h$ at all, but only on $\rho$ and $\Sigma$.
     We can therefore choose $h$ arbitrarily, even as a random matrix (but independent from $e$).
     Let $u$ be a $Q \times T$ matrix that is independent of $e$,
     and whose columns $u_t$, $t = 1,\ldots T$, are distributed as
     $u_t \sim  iid  {\cal N}(0,\rho \mathbbm{1}_Q)$.
     We choose $h$ such that the span of $h$ equals the span of $u'$, i.e. $u P_h = P_h$.
     Since we consider an asymptotic where $Q$ is finite, while $T \rightarrow \infty$
     it is easy to verify that $\| T  \rho P_h - u' u \| = {\cal O}_P(\sqrt{T})$, which implies
     $\sum_{r=1}^R \mu_r \left( e' e +  T  \rho P_h \right)
      = \sum_{r=1}^R \mu_r \left( e' e + u' u\right) +  {\cal O}_P(\sqrt{T})$.

     Let $U = (e', u')'$ and $E = (e' , 0_{T \times Q} )'$,
     which are $(N+Q) \times T$ matrices. The non-zero eigenvalues of
     the $T \times T$ matrices $U' U = e' e + u' u$
     and $E' E = e' e$
      are equal
     to the non-zero eigenvalues of the $(N+Q) \times (N+Q)$ matrices $U U'$
     and $E E'$, respectively. Let $v$ be the
     $(N+Q) \times R$ matrix whose columns equal to the normalized eigenvectors that correspond
     to the $R$ largest eigenvalues of $U U'$. We then have
     \begin{align}
          \sum_{r=1}^R \mu_r \left( e' e + u' u\right)
          &=  \sum_{r=1}^R \mu_r \left( U U' \right) = {\rm Tr}\left( v' U U' v \right) ,
        \nonumber \\
          \sum_{r=1}^R \mu_r \left( e' e \right)
          &=  \sum_{r=1}^R \mu_r \left( E E' \right)
          \geq {\rm Tr}\left( v' E E' v \right) ,
     \end{align}
     where the last inequality follows from the maximization property of the eigenvalues of $E E'$.
     Decompose $v= (v'_1, v'_2)'$ into the $N \times R$ matrix  $v_1$ and the $Q \times R$ matrix $v_2$.
     We then have
     \begin{align}
          \sum_{r=1}^R \mu_r \left( e' e + u' u\right)
          -   \sum_{r=1}^R \mu_r \left( e' e \right)
          &\leq  {\rm Tr}\left( v' U U' v \right) - {\rm Tr}\left( v' E E' v \right)
        \nonumber \\
           &={\rm Tr}\left[ v' \left( \begin{array}{cc} 0_{N \times N} & e u' \\ u e' & u u' \end{array} \right) v \right]
        \nonumber \\
           &=  2 {\rm Tr}\left( v_1'  e u' v_2 \right)
                + {\rm Tr}\left( v_2'  u u'  v_2 \right)
        \nonumber \\
           &\leq  2 R \| v_1'  e u' v_2 \| + R \| v_2'  u u'  v_2  \|
        \nonumber \\
           &\leq  2 R \| e\| \|u\| \|v_2 \| + R \| u \|^2   \| v_2  \|^2 ,
   \end{align}
    where we used that for any square matrix $B$ we have ${\rm Tr}(B) \leq {\rm rank}(B) \|B\|$,
    and also that $\| v_1 \| \leq 1$.
    We have $\|e\| = {\cal O}_P(\sqrt{\max(N,T)}) = {\cal O}_P(\sqrt{N+T})$, $\|u\| = {\cal O}_P(\sqrt{T})$
    and, as will be shown below,
    $\|v_2\| = {\cal O}_P(1/\sqrt{n})$. Therefore
    \begin{align}
          \sum_{r=1}^R \mu_r \left( e' e + u' u\right)
          -   \sum_{r=1}^R \mu_r \left( e' e \right)
           = {\cal O}_P(   \sqrt{ (N+T)T/n }  ) .
     \end{align}
     Combining the above results we find
     \begin{align}
         \sum_{r=1}^R \mu_r \left( e' e + A \right)
         &  \leq \sum_{r=1}^R \mu_r \left( e' e +  T  \rho P_h \right)
      \nonumber \\
          & \leq   \sum_{r=1}^R \mu_r \left( e' e + u' u\right) +  {\cal O}_P(\sqrt{T})
      \nonumber \\
         & \leq          \sum_{r=1}^R \mu_r \left( e' e \right)
         +  {\cal O}_P \left(   \sqrt{ (N+T)T/n }  \right)  +  {\cal O}_P(\sqrt{T})
      \nonumber \\
         & \leq          \sum_{r=1}^R \mu_r \left( e' e \right)
         +  {\cal O}_P \left(   \sqrt{ (N+T)T/n }  \right)   ,
     \end{align}
     where in the last step we used that $N/n \geq 1$. The last statement is what we wanted to show.
     However, we still have to justify that $\|v_2\| = {\cal O}_P(1/\sqrt{n})$. For this we first note that increasing
     the eigenvalues of $\Sigma$ can only decrease $\| v_2 \|$. Without loss of generality we can therefore
     consider the case where all the $n$ eigenvalues of $\Sigma$ that are smaller than $\rho$ are increased to
     be exactly equal to $\rho$. In that case the distribution of $U$ is symmetric under left-multiplication
     with orthogonal $O(n+Q)$ matrices, which only act on the the $(n+Q)$-dimensional eigenspace of
     the $(N+Q) \times (N+Q)$ covariance matrix of $U$ corresponding to eigenvector $\rho$.  Since the distribution of
     $U$ has this symmetry, the same needs to be true for the distribution of the eigenvectors $v$
     of $UU'$. Since $Q$ is finite, while $n \rightarrow \infty$
     this implies that $\|v_2\| = {\cal O}_P(1/\sqrt{n})$.
\end{proof}

\begin{proof}[\bf Proof of Lemma~\ref{lemma:JustifyHL1}, Part 1] Here, we consider the case where
Assumption~\ref{ass:DX-1} holds, and show that Lemma~\ref{lemma:JustifyHL1} holds in that case.

\#  We want to show that $C^{(1)} = {\cal O}_P(N^{1/4})$.
    By definition of $C^{(1)}$ and Assumption~\ref{ass:EX} we have
    \begin{align}
       C^{(1)}_k &=
         \frac 1 {\sqrt{NT}} \, {\rm Tr}( M_{\lambda^0} \, X_k \,
                  M_{f^0} \, e^{\prime} )
       \nonumber \\
              &=   \frac 1 {\sqrt{NT}}  {\rm Tr}(   X_k   e^{\prime} )
              -    \frac 1 {\sqrt{NT}} {\rm Tr}( P_{\lambda^0} \, X_k \, e^{\prime} )
              +   \frac 1 {\sqrt{NT}} \, {\rm Tr}( P_{\lambda^0} \, X_k \,
                  P_{f^0} \, e^{\prime} )
       \nonumber \\
             &= {\cal O}_P(1)  -    \frac 1 {\sqrt{NT}} {\rm Tr}( P_{\lambda^0} \, X_k \, e^{\prime} )
              +   \frac 1 {\sqrt{NT}} \, {\rm Tr}( P_{\lambda^0} \, X_k \,
                  P_{f^0} \, e^{\prime} ) .
     \end{align}
     Since $\| \widetilde X_k \| = {\cal O}_P( N^{3/4})$ we have
     \begin{align}
           \left|  \frac 1 {\sqrt{NT}} {\rm Tr}( P_{\lambda^0} \,  \widetilde X_k \, e^{\prime} ) \right|
           \leq  \frac R {\sqrt{NT}}  \|  \widetilde X_k \| \|e\| = {\cal O}_P(N^{1/4}),
     \end{align}
     i.e. $ \frac 1 {\sqrt{NT}} {\rm Tr}( P_{\lambda^0} \, \widetilde X_k \, e^{\prime} )  = {\cal O}_P(N^{1/4})$.
     Analogously we obtain $  \frac 1 {\sqrt{NT}} \, {\rm Tr}( P_{\lambda^0} \, \widetilde X_k \,
                  P_{f^0} \, e^{\prime} ) =  {\cal O}_P(N^{1/4})$.
      Regarding the  $\overline X_k$  contribution to $C^{(1)}_k$, consider  $e= \Sigma^{1/2} u$, i.e. case $(a)$ of Assumption~\ref{ass:DX-1} (the proof for case $(b)$
      is analogous).
     Using our assumptions on the
     distribution of $e$ and $\overline X_k$
     we have $\mathbbm{E}\left[ {\rm Tr}(P_{\lambda^0} \, \overline X_k e')^2 | X_k, \lambda^0, \Sigma \right]
      = {\rm Tr}(   \overline X_k' P_{\lambda^0} \Sigma P_{\lambda^0}  \overline X_k )
       \leq {\rm rank}( \overline X_k) \|  \overline X_k \|^2 \|\Sigma\| =  {\cal O}_P(NT)$, and therefore
    $\frac 1 {\sqrt{NT}} {\rm Tr}( P_{\lambda^0} \, \overline X_k \, e^{\prime} )  = {\cal O}_P(1)$.
     Analogously we find
     $  \frac 1 {\sqrt{NT}} \, {\rm Tr}( P_{\lambda^0} \, \overline X_k \,
                  P_{f^0} \, e^{\prime} ) = {\cal O}_P(1)$.
    Combining the above results gives $C^{(1)} = {\cal O}_P(N^{1/4})$. 

  \# We want to show that Assumption~\ref{ass:SN} holds. We have
     $\|X_k\| \leq \| \overline X_k \| +  \| \widetilde X_k \|
     = {\cal O}_P(\sqrt{NT}) + {\cal O}_P(N^{3/4}) = {\cal O}_P(\sqrt{NT})$,
   i.e. Assumption~\ref{ass:SN}$(i)$ is satisfied.
  In the following we assume that $e= \Sigma^{1/2} u$, i.e. case $(a)$ of Assumption~\ref{ass:DX-1}.
The proof for case $(b)$ follows by symmetry of the problem ($N \leftrightarrow T$).
 We have $\| e \| = \| \Sigma \|^{1/2} \| u \| = {\cal O}_P(1) \| u\|$, since we assume that
   $\| \Sigma \| = {\cal O}_P(1)$. Thus, we are left to show $ \| u \| = {\cal O}_P(\sqrt{\max(N,T)})$.
   Lemma~\ref{lemma:JustifyHL1} assumes $N/T \rightarrow \kappa^2$, but it turns out that this 
   assumption is not necessary to show $ \| u \| = {\cal O}_P(\sqrt{\max(N,T)})$, i.e. for the moment
   consider an arbitrary limit $N,T \rightarrow \infty$.
      By assumption, the errors $u_{it}$ are iid ${\cal N}(0,1)$.
      Since an arbitrary limit $N,T \rightarrow \infty$ is not considered very often
      in Random Matrix Theory,
      we define the $\max(N,T) \times \max(N,T)$
     matrix $u^{\rm big}$, which contains $u$ as a submatrix, and whose remaining
     elements are also iid ${\cal N}(0,1)$ and independent of $u$. We then have
     $\|u\| \leq \|u^{\rm big}\| = {\cal O}_P(\sqrt{\max(N,T)})$, where the
     last step is due to Geman~\cite*{Geman1980}.

\# Finally, we show that Assumption~\ref{ass:HL1} holds. 
Consider case (a) of Assumption~\ref{ass:DX-1}$(ii)$ in the following.
Using the decomposition  $X_k = \overline X_k + \widetilde X_k$ we have
\begin{align}
     & \sum_{r=1}^{R-R^0}  \left\{  \mu_r\left[
      M_{f^0} \left( e  -  \Delta \beta \cdot X \right)'
                  M_{\lambda^0}  \left( e  -  \Delta \beta \cdot X \right)  M_{f^0} \right]
         -   \mu_r\left[ M_{f^0} \left(  \Delta \beta \cdot X \right)'
                  M_{\lambda^0}  \left(   \Delta \beta \cdot X \right)  M_{f^0} \right]       \right\}  
   \nonumber \\
     &=    \sum_{r=1}^{R-R^0}   \left\{ \mu_r\left[
      M_{f^0} \left( e  -  \Delta \beta \cdot \overline X \right)'
                  M_{\lambda^0}  \left( e  -  \Delta \beta \cdot \overline X \right)  M_{f^0} \right]
          -   \mu_r\left[ M_{f^0} \left(  \Delta \beta \cdot \overline X \right)'
                  M_{\lambda^0}  \left(   \Delta \beta \cdot \overline X \right)  M_{f^0} \right]       \right\}           
         \nonumber \\ & \qquad     \qquad \qquad      
             + {\cal O}_P( \|e\|  \|\widetilde X_k\|  \|     \Delta \beta \| )
             + {\cal O}_P(  \|\widetilde X_k\|   \| X_k\|   \|   \Delta \beta \|^2 )
   \nonumber \\
     &=    \sum_{r=1}^{R-R^0}   \left\{ \mu_r\left[
      M_{f^0} \left( e  -  \Delta \beta \cdot \overline X \right)'
                  M_{\lambda^0}  \left( e  -  \Delta \beta \cdot \overline X \right)  M_{f^0} \right]
          -   \mu_r\left[ M_{f^0} \left(  \Delta \beta \cdot \overline X \right)'
                  M_{\lambda^0}  \left(   \Delta \beta \cdot \overline X \right)  M_{f^0} \right]       \right\}           
         \nonumber \\ & \qquad     \qquad \qquad      
             + {\cal O}_P( N^{5/4}  \|     \Delta \beta \| )
             + {\cal O}_P(  N^{7/4}    \|   \Delta \beta \|^2 ).
    \label{HL1derive}         
\end{align}
We now apply Lemma~\ref{lemma:help} with
$U=  M_{\lambda^0} e M_{f^0}$ and
 $g C h' =  - M_{\lambda^0} (\Delta \beta \cdot \overline X) M_{f^0} $,
 where $g$ and $h$ are define in Assumption~\ref{ass:DX-1},
 and $C = g' (\Delta \beta \cdot \overline X) h$.
We obtain
\begin{align}
      & \sum_{r=1}^{R-R^0}  \bigg\{  \mu_r\left[
      M_{f^0} \left( e  -  \Delta \beta \cdot \overline X \right)'
                  M_{\lambda^0}  \left( e  -  \Delta \beta \cdot \overline X \right)  M_{f^0} \right]
    \nonumber \\ & \qquad
        \leq  \sum_{r=1}^{R-R^0}
      \mu_r \left(  M_{f^0} e' M_{\lambda^0} e M_{f^0}+
      \| g' e M_{f^0} e' g \| \, P_{\left( M_{[M_{f^0} e' g]} h \right)}
         +  \Delta_{\max}  P_{(M_{f^0} e' g)}   \right)
    \nonumber \\ & \qquad \qquad \qquad
     + \sum_{r=1}^{\min(Q,R-R^0)}
      \mu_r\left[ M_{f^0} \left( \Delta \beta \cdot \overline X \right)'
                  M_{\lambda^0}  \left(  \Delta \beta \cdot \overline X \right)  M_{f^0} \right]
     + {\cal O}_P\left( \| g'  e   h \| \| \overline X\| \| \Delta \beta\| \right)
    \nonumber \\ & \qquad
        \leq  \sum_{r=1}^{R-R^0}
      \mu_r \left(  M_{f^0} e' M_{\lambda^0} e M_{f^0}+
      \| g' e M_{f^0} e' g \| \, P_h
         \right)
    \nonumber \\ & \qquad \qquad \qquad
     + \sum_{r=1}^{R-R^0}
      \mu_r\left[ M_{f^0} \left( \Delta \beta \cdot \overline X \right)'
                  M_{\lambda^0}  \left(  \Delta \beta \cdot \overline X \right)  M_{f^0} \right]
     + {\cal O}_P\left( \sqrt{NT} \| \Delta \beta\| \right)
     + {\cal O}_P(\sqrt{N})
    \nonumber \\ & \qquad
        \leq  \sum_{r=1}^{R-R^0}
      \mu_r \left(  M_{f^0} e' M_{\lambda^0} e M_{f^0}+
      T \| g' \Sigma g \| \, P_h
         \right)
    \nonumber \\ & \qquad \qquad \qquad
     + \sum_{r=1}^{R-R^0}
      \mu_r\left[ M_{f^0} \left( \Delta \beta \cdot \overline X \right)'
                  M_{\lambda^0}  \left(  \Delta \beta \cdot \overline X \right)  M_{f^0} \right]
     + {\cal O}_P\left( \sqrt{NT} \| \Delta \beta\| \right)
     + {\cal O}_P(\sqrt{N}) ,
     \label{HL1lemmaInterm}
\end{align}
where we used that under our assumptions we have 
\begin{itemize}
   \item[(i)]  $\displaystyle \| g'   e   h \| ={\cal O}_P(1)$,
 
   \item[(ii)] $\displaystyle g' e M_{f^0} e' g = T g' \Sigma g +  {\cal O}_P(\sqrt{N})$,
 
    \item[(iii)]  $\displaystyle \Delta_{\max} \equiv    \max_{r \in \{1,2,\ldots,\min(R,Q)\}} \left[ \mu_r( g'  e M_{f^0} e'  g ) -  \mu_{r+Q-\min(Q,R)}( g'  e M_{f^0} e'  g )  \right] = {\cal O}_P(\sqrt{N})$,
 
    \item[(iv)]  $\displaystyle \left\| P_{\left( M_{[M_{f^0} e' g]} h \right)}
     - P_h \right\|  = {\cal O}_P( N^{-1/2} )$.
\end{itemize}

Statement (i) above holds, because $g'   e   h = g' \Sigma^{1/2} u h$ is a projection of $u$
to a $Q \times Q$ submatrix, with $g' \Sigma^{1/2}$ and $h$ independent of $u$, and
$\| g' \Sigma^{1/2} \| = {\cal O}_P(1)$ and $\|h\|=1$. 

Statement (ii) holds, because we can calculate
the expectation and variance of $g' e M_{f^0} e' g = g' \Sigma^{1/2} uu' \Sigma^{1/2} g$
conditional on $\Sigma^{1/2} g$ to show that
$g' \Sigma^{1/2} uu' \Sigma^{1/2} g =  g' \Sigma^{1/2} \mathbbm{E}(uu') \Sigma^{1/2} g   + {\cal O}_P(\sqrt{N})$,
with $\mathbbm{E}(uu') = T \mathbbm{1}_N$.

Statement (iii) holds, because assume that either $R \geq Q$, in which case $\Delta_{\max}=0$,
or we assume  $g' \Sigma g =  \|g' \Sigma g \| \mathbbm{1}_{Q}  + {\cal O}_P(N^{-1/2})$,
so that $g' e M_{f^0} e' g = T \|g' \Sigma g \| \mathbbm{1}_{Q} +  {\cal O}_P(\sqrt{N}) $,
where $T \|g' \Sigma g \| \mathbbm{1}_{Q}$ gives no contribution to $\Delta_{\max}$.
 
We now apply Lemma~\ref{lemma:LowRankAdd} with ``$e$'' in the Lemma
equal to $M_{\lambda^0} e M_{f^0}$,
``$\Sigma$'' in the Lemma equal to $M_{\lambda^0} \Sigma M_{\lambda^0}$, and
$A =  T \| g' \Sigma g \| \, P_h$.
We have $\mu_{n-R^0}(M_{\lambda^0} \Sigma M_{\lambda^0}) \geq 
\mu_{n}( \Sigma)$
and $g' \Sigma g = g' M_{\lambda^0} \Sigma M_{\lambda^0} g$.
We therefore choose ``$n$'' in the Lemma equal to $n-R^0$
when applying Lemma~\ref{lemma:LowRankAdd},
and our assumption $\mu_n(\Sigma) \geq \| g' \Sigma g\|$ with
$1/n= {\cal O}_P(1/N)$ is now used.
When employing Lemma~\ref{lemma:LowRankAdd} here
we also use that rotational invariance of $e' M_{\lambda^0} e = u' \Sigma^{1/2} M_{\lambda^0} \Sigma^{1/2} u$
allows us to treat $M_{f^0} e' M_{\lambda^0} e M_{f^0}$ as an $(N-R^0) \times (N-R^0)$ matrix,
which requires that $u$ is $iid$ normally distributed.  
By Lemma~\ref{lemma:LowRankAdd} we then have
\begin{align}
    & \sum_{r=1}^{R-R^0}
      \mu_r \left(  M_{f^0} e' M_{\lambda^0} e M_{f^0}+
      T \| g' \Sigma g \| \, P_h
         \right)
   \nonumber \\    
      &=   \sum_{r=1}^{R-R^0}
      \mu_r \left(  M_{f^0} e' M_{\lambda^0} e M_{f^0}          \right)
       + {\cal O}_P \left( \sqrt{(N+T-2R^0) (T-R^0)/ (n-R^0)} \right)
   \nonumber \\    
      &=   \sum_{r=1}^{R-R^0}
      \mu_r \left(  M_{f^0} e' M_{\lambda^0} e M_{f^0}          \right)
       + {\cal O}_P(\sqrt{N}).
\end{align}
Combining this with \eqref{HL1derive} and \eqref{HL1lemmaInterm} gives Assumption~\ref{ass:HL1}.
\end{proof}

\begin{proof}[\bf Proof of Lemma~\ref{lemma:JustifyHL1}, Part 2] Here, we consider the case where
Assumption~\ref{ass:DX-2} holds, and show that Lemma~\ref{lemma:JustifyHL1} holds in that case.

    Using the assumption $M_{\lambda^0} \overline X_k M_{f^0} = 0$  
     simplifies the calculation in \eqref{HL1derive}, namely
\begin{align}
     & \sum_{r=1}^{R-R^0}    \mu_r\left[
      M_{f^0} \left( e  -  \Delta \beta \cdot X \right)'
                  M_{\lambda^0}  \left( e  -  \Delta \beta \cdot X \right)  M_{f^0} \right]
   \nonumber \\
     &=    \sum_{r=1}^{R-R^0}    \mu_r\left[
      M_{f^0} \left( e  -  \Delta \beta \cdot \overline X \right)'
                  M_{\lambda^0}  \left( e  -  \Delta \beta \cdot \overline X \right)  M_{f^0} \right]
             + {\cal O}_P( \|e\|  \|\widetilde X_k\|  \|     \Delta \beta \| )
             + {\cal O}_P(  \|\widetilde X_k\|^2  \|   \Delta \beta \|^2 )
   \nonumber \\
     &=    \sum_{r=1}^{R-R^0}   \mu_r\left[
      M_{f^0} e'
                  M_{\lambda^0} e  M_{f^0} \right]
             + {\cal O}_P( N^{5/4}  \|     \Delta \beta \| )
             + {\cal O}_P(  N^{3/2}    \|   \Delta \beta \|^2 ),
\end{align}   
and analogously we obtain
$\sum_{r=1}^{R-R^0}    \mu_r\left[
      M_{f^0} \left(   \Delta \beta \cdot X \right)'
                  M_{\lambda^0}  \left(   \Delta \beta \cdot X \right)  M_{f^0} \right]
  = {\cal O}_P(  N^{3/2}    \|   \Delta \beta \|^2 )$.
We therefore have
\begin{align}
    d(\beta) =  {\cal O}_P( N^{5/4}  \|     \Delta \beta \| )
             + {\cal O}_P(  N^{3/2}    \|   \Delta \beta \|^2 ) ,
\end{align}  
which implies that Assumption~\ref{ass:HL1} holds. The result for $C^{(1)}$ follows
because with $M_{\lambda^0} \overline X_k M_{f^0} = 0$   we find
\begin{align}
       C^{(1)}_k &=
         \frac 1 {\sqrt{NT}} \, {\rm Tr}( M_{\lambda^0} \, X_k \,
                  M_{f^0} \, e^{\prime} )
       \nonumber \\
              &=   \frac 1 {\sqrt{NT}}  {\rm Tr}(   X_k   e^{\prime} )
                    + {\cal O}_P( \|e\|  \|\widetilde X_k\| / \sqrt{NT}  )
       \nonumber \\
             &={\cal O}_P(1)  + {\cal O}_P(N^{1/4})    .
     \end{align}
Finally, Assumption~\ref{ass:SN} holds obviously under Assumption~\ref{ass:DX-2}.     
\end{proof}

\subsection{Proofs for Details on Asymptotic Equivalence}
\label{ass:ProofsInterEqu}

\begin{proof}[\bf Proof of Lemma~\ref{th:expansion2b}]
      Applying the expansion of $\widehat e(\beta)$
   in Lemma \ref{lemma:ResExp} together with
   $\| M_{\lambda^0} e M_{f^0} \| = {\cal O}_P(\sqrt{N})$,
   $\| \widehat e_e^{(1)} \| = {\cal O}_P(1)$,
   $\| \widehat e_e^{(2)} \| = {\cal O}_P(N^{-1/2})$,
   $\| \widehat e_k^{(1)} \| = {\cal O}_P(N)$
   $\| \widehat e_k^{(2)} \| = {\cal O}_P(\sqrt{N})$
   and the bound on $\|\widehat e^{(\rm rem)} \|$ given in the Lemma
   we obtain
   \begin{align}
      \widehat e'(\beta) \widehat e(\beta) &=
          B(\beta) + B'(\beta) + T^{(\rm rem)}(\beta) \; ,
   \end{align}
   where the terms $B^{(\rm rem,1)}(\beta)$
   and $B^{(\rm rem,2)}$ in $B(\beta)$ are given by
   \begin{align}
      B^{(\rm rem,1)}(\beta) &=  M_{f^0}  [(\beta-\beta^0 \cdot X)]'  M_{\lambda^0}
             e  M_{f^0}  e'
        \lambda^0(\lambda^{0\prime}\lambda^0)^{-1}(f^{0\prime}f^0)^{-1}  f^{0\prime}
            \nonumber \\ & \quad
                  + M_{f^0}  e'  M_{\lambda^0}
             [(\beta-\beta^0 \cdot X)]  M_{f^0}  e'
        \lambda^0(\lambda^{0\prime}\lambda^0)^{-1}(f^{0\prime}f^0)^{-1}  f^{0\prime}
            \nonumber \\ & \quad
                  + M_{f^0}  e'  M_{\lambda^0}
             e  M_{f^0}  [(\beta-\beta^0 \cdot X)]'
        \lambda^0(\lambda^{0\prime}\lambda^0)^{-1}(f^{0\prime}f^0)^{-1}  f^{0\prime}
            \nonumber \\ & \quad
                     +  M_{f^0}
                         \left( M_{f^0} e' M_{\lambda^0} \widehat e^{(2)}_e
                              + \widehat e^{(1) \prime}_e \widehat e^{(2)}_e
                              + \widehat e^{(2) \prime}_e M_{\lambda^0} e' M_{f^0}
                         \right) P_{f^0} \; ,
     \nonumber \\
      B^{(\rm rem,2)} &= \ft 1 2 P_{f^0}
                         \left( M_{f^0} e' M_{\lambda^0} \widehat e^{(2)}_e
                              + \widehat e^{(1) \prime}_e \widehat e^{(2)}_e
                              + \widehat e^{(2) \prime}_e M_{\lambda^0} e' M_{f^0}
                         \right) P_{f^0}
                \nonumber \\
                       &=
        f^0(f^{0\prime}f^0)^{-1}(\lambda^{0\prime}\lambda^0)^{-1}\lambda^{0\prime}
                       e  M_{f^0}  e'
                       M_{\lambda^0}  e  M_{f^0}  e'
   \lambda^0(\lambda^{0\prime}\lambda^0)^{-1}(f^{0\prime}f^0)^{-1}f^{0\prime}  ,
   \end{align}
   and for $\sqrt{N} \|\beta - \beta^0 \| \leq c$
   (which implies $\|\widehat e(\beta) \| = {\cal O}_P(\sqrt{N})$) we have
   \begin{align}
      \| T^{(\rm rem)}(\beta) \| = {\cal O}_P(N^{-1/2})
                                  + \| \beta-\beta^0\| {\cal O}_P(N^{1/2})
                                   + \| \beta-\beta^0\|^2 {\cal O}_P(N^{3/2}) \; .
   \end{align}
   which holds uniformly over $\beta$.
   Note also that
   \begin{align}
      B^{(eeee)}+B^{(eeee) \prime}
         &= M_{f^0}
                         \left( M_{f^0} e' M_{\lambda^0} \widehat e^{(2)}_e
                              + \widehat e^{(1) \prime}_e \widehat e^{(2)}_e
                              + \widehat e^{(2) \prime}_e M_{\lambda^0} e' M_{f^0}
                         \right) M_{f^0}.
   \end{align}
   Thus, we have $\|B^{(\rm rem,2)}\| = {\cal O}_P(1)$,
   and for $\sqrt{N} \|\beta - \beta^0 \| \leq c$
   we have
   $\|B^{(\rm rem,1)}(\beta)\| = {\cal O}_P(1) + \|\beta-\beta^0\| {\cal O}_P(N)$,
   and by Weyl's inequality
   \begin{align}
      \mu_r \left[ \widehat e'(\beta) \widehat e(\beta) \right]
      &= \mu_r \left[ B(\beta) + B'(\beta) \right]
         + o_P\left[ \left( 1 + \|\beta-\beta^0\| \right)^2 \right] \; ,
   \end{align}
   again uniformly over $\beta$. This proves the lemma.
\end{proof}

\begin{proof}[\bf Proof of Corollary \ref{cor:LimitRgen2}]
   From Theorem~\ref{th:ConvergenceRate} we know that 
   $N^{3/4} \| \widehat \beta_R-\beta^0 \| = {\cal O}_P(1)$,
   so that the bound in Assumption~\ref{ass:HL2} becomes applicable.
   Let
   $\gamma \equiv W^{-1} \left( C^{(1)} + C^{(2)} \right)/\sqrt{NT}
    = {\cal O}_P(1/\sqrt{NT})$,
   as in the proof proof of Corollary \ref{cor:LimitR0}.
   Since $\widehat \beta_R$ minimizes ${\cal L}^R_{NT}(\beta)$
   it must in particular satisfy
   ${\cal L}_{NT}^{R}(\widehat \beta_{R}) \leq
    {\cal L}_{NT}^{R}\left(\beta^0+\gamma\right)$.
   Using
   Lemma~\ref{th:expansion2b}
   and Assumption~\ref{ass:HL2} it follows that
   \begin{align}
      {\cal L}_{NT}^{0}(\widehat \beta_{R}) \leq
    {\cal L}_{NT}^{0}\left(\beta^0+\gamma\right)
    + \frac 1 {NT} \;
       o_P\left[ \left( 1 + \sqrt{NT} \|\widehat \beta_{R}-\beta^0 \|^2 \right)^2 \right] .
   \end{align}
   The rest of the proof is analogous to the proof of Corrollary~\ref{cor:LimitR0}.
\end{proof}

\begin{proof}[\bf Proof of Lemma~\ref{Lemma:EVbound1}]
   For the eigenvalues of $A+B$  we have
   \begin{align}
      \mu_r(A+B) &= \min_{\Gamma}\,\max_{\left\{\gamma: \, \|\gamma\|=1, \, P_\Gamma \gamma = 0 \right\}}
                            \, \gamma' (A+B) \gamma \; ,
   \end{align}
   where $\Gamma$
   is a $n\times (r-1)$ matrix with full rank $r-1$,
   and $\gamma$ is a $n\times 1$ vector. In the following
   we only consider
   those $\gamma$ that lie in the span of the first $r$ eigenvectors $A$, i.e.
   $\gamma=\sum_{i=1}^r c_i \nu_i$. The condition $\|\gamma\|=1$ implies $\sum_{i=1}^r c_i^2 = 1$.
   The column space of $\Gamma$ is $(r-1)$-dimensional. Therefore, for a given $\gamma=\sum_{i=1}^r c_i \nu_i$
   there always exists a $\Gamma$ such that the conditions $\|\gamma\|=1$ and $P_\Gamma \gamma = 0$
   {\it uniquely} determine $\gamma$ up to the sign.
   We thus have
   \begin{align}
      \mu_r(A+B) &\geq \min_{\Gamma} \,
             \max_{\left\{\gamma: \, \gamma=\sum_{i=1}^r c_i \nu_i , \, \|\gamma\|=1, \, P_\Gamma \gamma = 0 \right\}}
                     \, \gamma' (A+B) \gamma
         \nonumber \\
         &= \min_{\left\{\gamma: \, \gamma=\sum_{i=1}^r c_i \nu_i , \, \|\gamma\|=1 \right\}} \, \gamma' (A+B) \gamma
         \nonumber \\
         &\geq  \min_{\left\{(c_1,\ldots,c_r): \, \sum_{i=1}^r c_i^2 = 1 \right\}}
                 \left[ \sum_{i=1}^r \, c_i^2 \, \mu_i(A) - b \left(\sum_{i=1}^r |c_i| \right)^2 \right]
         \nonumber \\
         &\geq  \mu_r(A) - r \, b
         \nonumber \\
         &\geq \mu_{r}(A) -   \frac{(q-1) \, b}{1 - \sum_{i=q}^n \frac b {\mu_r(A)-\mu_i(A)}} \; ,
   \end{align}
   where we used that $q-1 \geq r$ and that the additional fraction we multiplied with is larger than one.
   This is the lower bound for $\mu_r(A+B)$ that we wanted to show. We now want to derive the upper bound.
   Let $\widetilde A$, $\widetilde B$ and $\bar B$ be $(n-r+1)\times(n-r+1)$ matrices defined by
   $\widetilde A_{ij}=\nu_{i+r-1}' A \nu_{j+r-1}$, $\widetilde B_{ij}=\nu_{i+r-1}' B \nu_{j+r-1}$,
   and $\bar B_{ij}=b$, where $i,j=1,\ldots,n-r+1$.
   We can choose $\Gamma=(\nu_1,\nu_2,\ldots,\nu_{r-1})$ in the above minimization problem, in which case
   $\gamma$ is restricted to the span of $\nu_r, \, \nu_{r+1}, \ldots, \nu_n$. Therefore
   \begin{align}
      \mu_r(A+B) &\leq \max_{\{ \widetilde \gamma : \, \|\widetilde \gamma\|=1 \}} \,
                      \widetilde \gamma' (\widetilde A+\widetilde B) \widetilde \gamma
             \nonumber \\
                 &= \mu_1( \widetilde A+\widetilde B ) \; ,
   \end{align}
   where $\widetilde \gamma$ is a $(n-r+1)$-dimensional vector, whose components are denoted
   $\widetilde \gamma_i$, $i=1,\ldots,n-r+1$, in the following. Note that $\widetilde A$ is a diagonal
   matrix with entries $\mu_{i+r-1}(A)$, $i=1,\ldots,n-r+1$. Therefore
   \begin{align}
      \mu_r(A+B) &\leq \max_{\{ \widetilde \gamma : \, \|\widetilde \gamma\|=1 \}} \,
                      \left[    \sum_{i=1}^{n+r-1} (\widetilde \gamma_i)^2 \mu_{i+r-1}(A)
                       + \sum_{i,j=1}^{n+r-1} \widetilde \gamma_i \, \widetilde \gamma_j \, \widetilde B_{ij} \right]
      \nonumber \\
                 &\leq \max_{\{ \widetilde \gamma : \, \|\widetilde \gamma\|=1 \}} \,
                     \left[   \sum_{i=1}^{n+r-1} (\widetilde \gamma_i)^2 \mu_{i+r-1}(A)
                       + b \sum_{i,j=1}^{n+r-1} |\widetilde \gamma_i| \, |\widetilde \gamma_j| \right]
      \nonumber \\
                 &= \max_{\{ \widetilde \gamma : \, \|\widetilde \gamma\|=1 \}} \,
                     \left[    \sum_{i=1}^{n+r-1} (\widetilde \gamma_i)^2 \mu_{i+r-1}(A)
                       + \sum_{i,j=1}^{n+r-1} \widetilde \gamma_i \, \widetilde \gamma_j \, \bar B_{ij} \right]
      \nonumber \\
                 &= \mu_1(\widetilde A + \bar B) \; .
   \end{align}
   In the last maximization problem the maximum is always attained at a point with
   $\widetilde \gamma_i\geq 0$, which is why we could omit the absolute values around $\widetilde \gamma_i$.

   The eigenvalue $\widetilde \mu \equiv \mu_1(\widetilde A + \bar B)$
   is a solution of the characteristic polynomial of $\widetilde A + \bar B$
   which can be written as
   \begin{align}
      1 &= \sum_{i=r}^n \frac b {\widetilde \mu-\mu_i(A)} \; ,
   \end{align}
   where $\mu_i(A)=\mu_{i-r+1}(\widetilde A)$ are the eigenvalues of $\widetilde A$. In addition we have
   $\widetilde \mu = \mu_1(\widetilde A+\bar B) > \mu_1(\widetilde A) = \mu_r(A)$, because $\bar B$ is positive semi-definite
   (which gives $\geq$) and the eigenvectors of $\widetilde A$ do not agree with those of $\bar B$ (which gives $\neq$).
   From the characteristic polynomial we therefore find
   \begin{align}
      1 &= \sum_{i=r}^{q-1} \frac b {\widetilde \mu-\mu_i(A)} + \sum_{i=q}^n \frac b {\widetilde \mu-\mu_i(A)}
     \nonumber \\
        &\leq \frac {b(q-1)} {\widetilde \mu-\mu_r(A)} + \sum_{i=q}^n \frac b {\mu_r(A)-\mu_i(A)} .
   \end{align}
   Since we assume $1 \geq \sum_{i=q}^n \frac b {\mu_r(A)-\mu_i(A)}$,
   this gives an upper bound on $\widetilde \mu$, and since $\mu_r(A+B)\leq \widetilde \mu$ the same
   bound holds for $\mu_r(A+B)$, namely
   \begin{align}
      \mu_r(A+B)  &\leq \mu_r(A) + \frac{(q-1) \, b}{1 - \sum_{i=q}^n \frac b {\mu_r(A)-\mu_i(A)}} \; .
   \end{align}
   This is what we wanted to show.
\end{proof}
                  
\begin{proof}[\bf Proof of Lemma \ref{lemma:JustifyHL2}]
    Consider the case $T \leq N$, so that for $Q=\min(N,T)-R^0$ defined in Assumption~\ref{ass:EV}
    we have $Q=T-R^0$. If $N < T$, then we interchange the role of $N$ and $T$ in the following proof.\footnote{%
       We consider limits $N,T \rightarrow \infty$ with $N/T \rightarrow \kappa^2$. For $\kappa^2>1$
       we have $T<N$ holding asymptotically, while for $\kappa^2<1$ we have $T>N$ holding asymptotically and the role
       of $N$ and $T$ in the proof needs to be interchanged. For $\kappa^2=1$ there is a subtlety, because
       neither $T \leq N$ nor $N \leq T$ needs to hold asymptotically (the ordering of $N$ and $T$ can
       change arbitrarily often while $N$ and $T$ grow). We could rule out this
       subtlety by only considering asymptotic sequences that satisfy either always $T \leq N$ or always
       $N \leq T$, which would not diminish the practical implications of our results in any way. The
       proof can also be adjusted to jointly consider the cases
       $T \leq N$ and $N \leq T$ in the asymptotic, which is not complicated, but cumbersome.
       }
       
   Define
   \begin{align}
      C^\pm(\beta) &= B(\beta) + B'(\beta)
              \pm
            \left( \sqrt{ \frac 4 {a N}} M_{f^0} B^{(\rm rem,1)}(\beta) P_{f^0} \mp
                  \sqrt{ \frac {a N} 4} P_{f^0} \right)
          \nonumber \\  & \qquad  \qquad \qquad  \qquad \qquad
           \times \left(  \sqrt{ \frac 4 {a N}} M_{f^0} B^{(\rm rem,1)}(\beta) P_{f^0} \mp
                  \sqrt{ \frac {a N} 4} P_{f^0} \right)'
          \nonumber \\  & \qquad
              \pm \left(   \sqrt{ \frac 4 {a N}} M_{f^0}  e'  M_{\lambda^0}
             e  M_{f^0}  e'
        \lambda^0(\lambda^{0\prime}\lambda^0)^{-1}(f^{0\prime}f^0)^{-1}  f^{0\prime}
        \pm \sqrt{ \frac {a N} 4} P_{f^0} \right)
          \nonumber \\  & \qquad  \qquad
              \times \left(   \sqrt{ \frac 4 {a N}} M_{f^0}  e'  M_{\lambda^0}
             e  M_{f^0}  e'
        \lambda^0(\lambda^{0\prime}\lambda^0)^{-1}(f^{0\prime}f^0)^{-1}  f^{0\prime}
        \pm \sqrt{ \frac {a N} 4} P_{f^0} \right)'.
   \end{align}
   Since  $C^+(\beta)$ [respectively $C^-(\beta)$]
   is obtained by adding [respectively subtracting] a positive definite matrix
   from $B(\beta) + B'(\beta)$, we have
   \begin{align}
          \mu_r\left( C^-(\beta) \right) \leq
        \mu_r\left( B(\beta) + B'(\beta) \right)  \leq
         \mu_r\left( C^+(\beta) \right) .
   \end{align}
   The advantage of considering $C^\pm(\beta)$  instead
   of $B(\beta) + B'(\beta)$ directly is that there are no ``mixed terms''
   in $C^\pm(\beta)$, which start with $M_{f^0}$ and end with $P_{f^0}$,
   or vice versa, i.e. we can write
   $C^\pm(\beta) = C_1^\pm(\beta) + C_2^\pm$,
   where $C_1^\pm(\beta) = M_{f^0} C_1^\pm(\beta) M_{f^0}$
   and $C_2^\pm =  P_{f^0} C_2^\pm P_{f^0}$.
   Concretely, we have
   \begin{align}
      C_1^\pm(\beta) &= A(\beta)
         \pm \frac 4 {a N} M_{f^0} B^{(\rm rem,1)}(\beta) P_{f^0}
                     B^{(\rm rem,1) \prime}(\beta)   M_{f^0}
        \nonumber \\ & \quad
         \pm \frac 4 {a N} M_{f^0}  e'  M_{\lambda^0}
             e  M_{f^0}  e'
        \lambda^0(\lambda^{0\prime}\lambda^0)^{-1}(f^{0\prime}f^0)^{-1}
           (\lambda^{0\prime}\lambda^0)^{-1}  \lambda^{0\prime} e
           M_{f^0}  e'  M_{\lambda^0}
             e  M_{f^0}
             \nonumber \\ & \quad
               + M_{f^0}  \left[ (\beta-\beta^0) \cdot X -e \right]'  M_{\lambda^0}   e
     f^0 (f^{0\prime}f^0)^{-1}  (\lambda^{0\prime}\lambda^0)^{-1}\lambda^{0\prime}e M_{f^0}
             \nonumber \\ & \quad
               + M_{f^0}  e'  M_{\lambda^0}   \left[ (\beta-\beta^0) \cdot X \right]
     f^0 (f^{0\prime}f^0)^{-1}  (\lambda^{0\prime}\lambda^0)^{-1}\lambda^{0\prime}e M_{f^0}
             \nonumber \\ & \quad
               + M_{f^0}  e'  M_{\lambda^0}   e
     f^0 (f^{0\prime}f^0)^{-1}  (\lambda^{0\prime}\lambda^0)^{-1}\lambda^{0\prime}
     \left[ (\beta-\beta^0) \cdot X \right] M_{f^0}
      \nonumber \\ & \quad
            + \text{the last three lines transposed}
            + B^{(eeee)} + B^{(eeee) \prime},
     \nonumber \\
      C_2^\pm &=
            P_{f^0} B^{(\rm rem,2)} P_{f^0}
           + P_{f^0} B^{(\rm rem,2) \prime} P_{f^0}
           \pm \frac{a N} 2 P_{f^0} .
   \end{align}
   In the rest of the proof we always
   assume that  $N^{3/4} \left\| \beta -\beta^{0} \right\| \leq c$.
   We apply Lemma \ref{Lemma:EVbound1}
    to $C_1^\pm(\beta)$, with the $A$ in the lemma equal to the leading term
    $M_{f^0}  e'  M_{\lambda^0} e  M_{f^0}$, the $B$ in the lemma
    equal to the remainder of $C_1^\pm(\beta)$, and $q=q_{NT}$.
    Assumption~\ref{ass:EV} introduces $\rho_r$ and $w_r$ as the eigenvalues
    and corresponding eigenvectors of $M_{f^0}  e'  M_{\lambda^0} e  M_{f^0}$,
    where $r=1,\ldots,Q$ with $Q=\min(N,T)-R^0=T-R^0$.    
    If we can show that
    \begin{align}
      \sum_{r=q_{NT}}^{T-R^0} \frac {b_{NT}} {\rho_{R-R^0}-\rho_r}
        = o_P(1) ,
    \end{align}
    then Lemma \ref{Lemma:EVbound1} becomes applicable asymptotically, and
    for $r=1,\ldots,R-R^0$ we have wpa1
    \begin{align}
      \left| \mu_r\left( C_1^\pm(\beta) \right)
       -  \rho_r \right|
          \, &\leq \, \frac{(q_{NT}-1) \, b_{NT}}
                {1 - \sum_{s=q_{NT}}^{T-R^0} \frac {b_{NT}} {\rho_{r}-\rho_s}}
          \, \leq \, \frac{q_{NT} \, b_{NT}}
                {1 - \sum_{s=q_{NT}}^{T-R^0} \frac {b_{NT}} {\rho_{R-R^0}-\rho_s}} \; ,
    \end{align}
    where
    \begin{align}
       b_{NT} &= \max_{r,s}  \left| w_r'  
                 \left(  C_1^\pm(\beta) -
                        M_{f^0}  e'  M_{\lambda^0} e  M_{f^0} \right) w_s \right| .
    \end{align}
    We now check how the different terms in
    $C_1^\pm(\beta) -  M_{f^0}  e'  M_{\lambda^0} e  M_{f^0}$
    contribute to $b_{NT}$.
    Using the definition of $d_{NT}$ in equation \eqref{DefDNT} 
    we have
    \begin{align*}
      \max_{r,s}
       \left| w_r' M_{f^0}  e'  M_{\lambda^0} [(\beta-\beta^0) \cdot X]  M_{f^0} w_s
       \right| &\leq K \|e\| \|\beta-\beta^0\| \max_{k,r,s} \| v_r' X_k w_s \|
     \nonumber \\
          &\leq d_{NT} {\cal O}_P(N^{-1/4}) ,
     \nonumber \\
       \max_{r,s} \left| w_r' M_{f^0}  [(\beta-\beta^0) \cdot X]'  M_{\lambda^0} [(\beta-\beta^0) \cdot X]  M_{f^0} w_s
       \right| &\leq K^2 \left\| \beta -\beta^{0} \right\|^2
                   \max_{k,r} \| M_{\lambda^0} X_k w_r \|^2
         \nonumber \\
           &\leq K^2 N \left\| \beta -\beta^{0} \right\|^2
                   \max_{k,r,s} \| v_r' X_k w_s \|^2
         \nonumber \\
           &\leq d^2_{NT} {\cal O}_P(N^{-1/2}) ,
      \nonumber \\
        \max_{r,s}   \left|  w_r' \frac 4 {a N} M_{f^0} B^{(\rm rem,1)}(\beta) P_{f^0}
                     B^{(\rm rem,1) \prime}(\beta)   M_{f^0} w_s \right|
          &\leq \frac 4 {a N} \|B^{(\rm rem,1)}(\beta)\|^2  = {\cal O}_P(N^{-1/2})  ,
    \end{align*}
    \begin{align*}
         &\max_{r,s} \left| w_r' \frac 4 {a N} M_{f^0}  e'  M_{\lambda^0}
             e  M_{f^0}  e'
        \lambda^0(\lambda^{0\prime}\lambda^0)^{-1}(f^{0\prime}f^0)^{-1}
           (\lambda^{0\prime}\lambda^0)^{-1}  \lambda^{0\prime} e
           M_{f^0}  e'  M_{\lambda^0}
             e  M_{f^0} w_s \right|
       \nonumber \\
          &\qquad \leq   \frac 4 {a N} \|e\|^4
                        \left\| \lambda^0(\lambda^{0\prime}\lambda^0)^{-1}(f^{0\prime}f^0)^{-1}
           (\lambda^{0\prime}\lambda^0)^{-1}  \lambda^{0\prime} \right\|
                    \max_{r} \| w_r' e' P_{\lambda^0}\|^2
             \leq d_{NT}  {\cal O}_P(N^{-1})  ,
      \nonumber \\
        & \max_{r,s}
          \left| w_r' M_{f^0}  e'  M_{\lambda^0}   e
     f^0 (f^{0\prime}f^0)^{-1}  (\lambda^{0\prime}\lambda^0)^{-1}\lambda^{0\prime}
              e M_{f^0} w_s \right|
       \nonumber \\
          &\qquad \leq \|e\|
          \left\|  f^0 (f^{0\prime}f^0)^{-1}
                  (\lambda^{0\prime}\lambda^0)^{-1}\lambda^{0\prime} \right\|
           \max_s \| v_s' e P_{f^0} \|  \max_r \| w_r' e' P_{\lambda^0} \|
              \leq d^2_{NT} {\cal O}_P(N^{-1/2}) ,
      \nonumber \\
        & \max_{r,s}
           \left| w_r' M_{f^0}  \left[ (\beta-\beta^0) \cdot X \right]'  M_{\lambda^0}   e
     f^0 (f^{0\prime}f^0)^{-1}  (\lambda^{0\prime}\lambda^0)^{-1}\lambda^{0\prime}e M_{f^0} w_s
        \right|
     \nonumber \\
          &\qquad
          =\max_{r,s}
           \left| w_r'  \left[ (\beta-\beta^0) \cdot X \right]'
           \left( \sum_q v_q' v_q \right)   e
     f^0 (f^{0\prime}f^0)^{-1}  (\lambda^{0\prime}\lambda^0)^{-1}\lambda^{0\prime}e w_s
        \right|
     \nonumber \\
          &\qquad
        \leq K \|\beta-\beta^0\| N  \max_{r,s,k}  | v_s' X_k w_r |
            \max_r \| v_r' e P_{f^0} \|  \max_r \| w_r' e' P_{\lambda^0} \|
         \left\|  f^0 (f^{0\prime}f^0)^{-1}
                  (\lambda^{0\prime}\lambda^0)^{-1}\lambda^{0\prime} \right\|
         \nonumber \\
          &\qquad   \leq d^3_{NT} {\cal O}_P(N^{-3/4}),
     \nonumber \\
         &\max_{r,s}
           \left| w_r' M_{f^0}  e'  M_{\lambda^0}   \left[ (\beta-\beta^0) \cdot X \right]
     f^0 (f^{0\prime}f^0)^{-1}  (\lambda^{0\prime}\lambda^0)^{-1}\lambda^{0\prime}e M_{f^0}
          w_s \right|
     \nonumber \\ & \qquad
        \leq K   \|e\| \|\beta-\beta^0\|
            \left\|  f^0 (f^{0\prime}f^0)^{-1}
                  (\lambda^{0\prime}\lambda^0)^{-1}\lambda^{0\prime} \right\|
            \max_{r,k} \| v_r' X_k P_{f^0} \|
            \max_r \| w_r' e' P_{\lambda^0}   \|
         \nonumber \\
          &\qquad   \leq d^2_{NT} {\cal O}_P(N^{-1/2}),
             \nonumber \\ &
         \max_{r,s}
           \left| w_r'  M_{f^0}  e'  M_{\lambda^0}   e
     f^0 (f^{0\prime}f^0)^{-1}  (\lambda^{0\prime}\lambda^0)^{-1}\lambda^{0\prime}
     \left[ (\beta-\beta^0) \cdot X \right] M_{f^0} w_s \right|
     \nonumber \\ & \qquad
        \leq K   \|e\| \|\beta-\beta^0\|
            \left\|  f^0 (f^{0\prime}f^0)^{-1}
                  (\lambda^{0\prime}\lambda^0)^{-1}\lambda^{0\prime} \right\|
            \max_{r,k} \| v_r' e P_{f^0} \|
            \max_r \| w_r' X_k' P_{\lambda^0}   \|
         \nonumber \\
          &\qquad   \leq d^2_{NT} {\cal O}_P(N^{-1/2}).
 \end{align*}
 and analogously one can check that
 \begin{align}
    \max_{r,s}  \left| w_r'  B^{(eeee)} w_s \right|
       \leq d^2_{NT} {\cal O}_P(N^{-1})
           +d^3_{NT} {\cal O}_P(N^{-3/2}) .
 \end{align}
 All in all, we thus have
 \begin{align}
    b_{NT} &\leq
          {\cal O}_P(N^{-1/2})
          +  d_{NT} {\cal O}_P(N^{-1/4})
          + d^2_{NT} {\cal O}_P(N^{-1/2})
           +d^3_{NT} {\cal O}_P(N^{-3/4})
      \nonumber \\
          &\leq  d_{NT} {\cal O}_P(N^{-1/4}) \; ,
 \end{align}
 where in the last step we used that by assumption $d_{NT}\geq 1$ and
 $d_{NT} = o_P(N^{1/4})$. Therefore
  \begin{align}
      \sum_{r=q_{NT}}^{T-R^0} \frac {b_{NT}} {\rho_{R-R^0}-\rho_r}
      = q_{NT} d_{NT} {\cal O}_P(N^{-1/4})
        \frac 1 {q_{NT}}
        \leq \sum_{r=q_{NT}}^{T-R^0} \frac {1} {\rho_{R-R^0}-\rho_r}
      = o_P(1),
  \end{align}
  so that  Lemma \ref{Lemma:EVbound1} is indeed applicable asymptotically,
  and we find
    \begin{align}
      \left| \mu_r\left( C_1^\pm(\beta) \right)
       -  \rho_r \right|
          \, \leq \, \frac{q_{NT} \, b_{NT}}
                {1 - o_P(1)}
                    \, \leq \, q_{NT} \,  d_{NT} \, {\cal O}_P(N^{-1/4})
                    = o_P(1)  \; .
    \end{align}
    For $t=1,\ldots,R-R^0$ we thus have
   \begin{align}
      \mu_r \left(  C_1^\pm(\beta) \right)
       =  \rho_r
          + o_P(1)
       \geq \rho_{R-R^0}  + o_P(1)
       \geq \|C_2^\pm \| , \quad \text{wpa1},
   \end{align}
   where the last step follows because $\|C_2^\pm\|=a N/2 + {\cal O}_P(1)$
   and we assumed $\rho_{R-R^0} > a N$, wpa1.
   Since $C^\pm(\beta)$
   is block-diagonal with blocks $C_1^\pm(\beta)$
   and $C_2^\pm$ (in the basis defined by $f^0$),
   and $\mu_r \left(  C_1^\pm(\beta) \right)  \geq \|C_2^\pm \|$,
   it must be the case that wpa1 the largest $R-R^0$ eigenvalues of $C^\pm(\beta)$
   are those of $C_1^\pm(\beta)$. Thus,
    \begin{align}
      \left| \mu_r\left( C^\pm(\beta) \right)
       -  \rho_r \right|
                    = o_P(1)  \; ,
    \end{align}
    and also
    \begin{align}
      \left| \mu_r\left( B(\beta) + B'(\beta) \right)
       -  \rho_r \right|
                    = o_P(1)  \; ,
    \end{align}
    which holds uniformly over all $\beta$
    with $N^{3/4} \left\| \beta -\beta^{0} \right\| \leq c$.
    This concludes the proof.
\end{proof}

\begin{proof}[\bf Proof of Lemma~\ref{lemma:iidEV}]
    \# Part $(i)$.
Since $e$ has $iid {\cal N}(0,\sigma^2)$ entries, independent of $\lambda^0$ and $f^0$,
rotational invariance dictates that the distribution
of $v_r$ and $w_r$ is given by the Haar measure on the unit sphere
of dimension $N-R^0$ and $T-R^0$, respectively,
and the the lemma just provides a concrete representation of this. 
The bounds on 
$ \mathbbm{E}\left( \sqrt{N} \|  M_{\lambda^0} \widetilde v \|^{-1}  \right)^{\xi}$  and
$ \mathbbm{E}\left(  \sqrt{T} \| M_{f^0} \widetilde w \|^{-1}   \right)^{\xi}$ follow,
because the inverse chi-square distribution with dof $\nu$
possesses all moments smaller than $\nu/2$.

  \# Part $(ii)$.  
   Using part $(i)$ of the lemma we have
   \begin{align}
   w_r   \operatorname*{=}_d
       \| M_{f^0} \widetilde w \|^{-1}  M_{f^0} \widetilde w  
        =   \left\| \frac{ M_{f^0} \widetilde w } {\sqrt{T}} \right\|^{-1}
        \left(  \frac{\widetilde w} {\sqrt{T}}   -   \frac{P_{f^0} \widetilde w} {\sqrt{T}}  \right)  ,
   \end{align}
   where $\widetilde w$ be a $T$-vector with $iid {\cal N}(0,1)$ entries.
   It is also useful to define the time shift operator $L: \mathbbm{R}^T \rightarrow \mathbbm{R}^T$,
   which satisfies $(L w_r)_t = w_{r,t-1}$, and therefore $(L^\tau w_r)_t = w_{r,t-\tau}$. We then have   
   \begin{align}
       \sum_{t=\tau+1}^{T}  w_{r,t} w_{r,t-\tau}
       &= w_r' L^\tau w_r
     \nonumber \\  
       & \operatorname*{=}_d 
        \left\| \frac{ M_{f^0} \widetilde w } {\sqrt{T}} \right\|^{-2}
        \frac 1 T
      \left( \widetilde w' L^\tau \widetilde w 
              - \widetilde w' L^\tau P_{f^0} \widetilde w 
             - \widetilde w' P_{f^0} L^\tau \widetilde w 
             + \widetilde w' P_{f^0} L^\tau P_{f^0} \widetilde w  \right) .               
   \end{align}
   Given the distribution of $\widetilde w$ it is easy to show that
   $\left| \frac {\widetilde w' L^\tau \widetilde w} {\sqrt{T}} \right|$
   has arbitrary high bounded moments as $T$ becomes large, i.e.
   we have
   $\mathbbm{E} \left| \frac {\widetilde w' L^\tau \widetilde w} {\sqrt{T}} \right|^\xi = {\cal O}(1)$
   for any $\tau \geq 1$ and any $\xi>0$. Furthermore, using that $\| L \| =1$ we can bound
   \begin{align}
        \left| \widetilde w' L^\tau P_{f^0} \widetilde w \right|
        &\leq \| \widetilde w \| \| P_{f^0} \widetilde w\| 
        \nonumber \\
         \left| \widetilde w' P_{f^0} L^\tau \widetilde w \right|
        &\leq \| \widetilde w \| \| P_{f^0} \widetilde w\|
        \nonumber \\
         \left| \widetilde w' P_{f^0} L^\tau P_{f^0} \widetilde w \right|
         &\leq \| P_{f^0} \widetilde w\|^2 ,
   \end{align}
   where $\displaystyle \| \widetilde w \|^2 \operatorname*{=}_d \chi^2(T)$
   and $\displaystyle \| P_{f^0} \widetilde w\|^2 \operatorname*{=}_d \chi^2(R^0)$.
   Note that the rhs of the inequalities in the last display do not depend on $\tau$, i.e.
   the bounds are uniform over $\tau$.   
   The $\chi$-square distribution with $R^0$ degrees of freedom does not depend on $T$
   and has finite moments of all orders. Since $ \| \widetilde w \|^2$ is $\chi^2(T)$ distributed we find that
   $\frac 1 {\sqrt{T}} \| \widetilde w \|$ has arbitrarily high uniformly bounded moments
   as $T$ becomes large. Combining these results we obtain that all moments of
   $\left| \frac 1 {\sqrt{T}}
      \left( \widetilde w' L^\tau \widetilde w 
              - \widetilde w' L^\tau P_{f^0} \widetilde w 
             - \widetilde w' P_{f^0} L^\tau \widetilde w 
             + \widetilde w' P_{f^0} L^\tau P_{f^0} \widetilde w  \right) \right|$
    are uniformly bounded as $T$ becomes large.    
    Part $(i)$ of the lemma shows that the same is true for
   $ \left\| \frac{ M_{f^0} \widetilde w } {\sqrt{T}} \right\|^{-2}$.
   Using  Holder's inequality we thus find that for all $\xi>0$ we have
   \begin{align}
        \mathbbm{E}\left|  \sum_{t=\tau+1}^{T}  w_{r,t} w_{r,t-\tau}  \right|^\xi
        &=  \mathbbm{E}\left| \| M_{f^0} \widetilde w \|^{-2}  \widetilde w' M_{f^0} L^\tau M_{f^0} \widetilde w  \right|^\xi
          = {\cal O}(1/\sqrt{T}) ,
   \end{align}
   uniformly over $r$ and $\tau$.
   From this we obtain
  $\max_{r,\tau} \left| \sum_{t=\tau+1}^{T}  w_{r,t} w_{r,t-\tau} \right| = {\cal O}_P( T^{-1/2+\varepsilon} )$
  for any $\varepsilon>0$ (namely $\varepsilon=2/\xi$).
  This is the statement of the lemma for the special case where $r=s$.
  
  What is left to show is that
  $\max_{r \neq s} \max_{\tau} \left| \sum_{t=\tau+1}^{T}  w_{r,t} w_{s,t-\tau} \right| = {\cal O}_P( T^{-1/2+\varepsilon} )$, for $\varepsilon \in [0,1/12)$. 
   Let $\widetilde w^a$ and $\widetilde w^b$ be two $T$-vector with $iid {\cal N}(0,1)$ entries, independent of each other,
  and independent of $f^0$. Then we have for any $r,s = 1,\ldots,Q$ with $r \neq s$ that
  \begin{align}
   \left( \begin{array}{c} w_r \\ w_s \end{array} \right)
     \operatorname*{=}_d
       \left( \begin{array}{c}  \| M_{f^0} \widetilde w^a \|^{-1}  M_{f^0} \widetilde w^a
         \\ 
        \| M_{f^0} M_{w^a} \widetilde w^b \|^{-1}  M_{f^0} M_{w^a} \widetilde w^b  \end{array} \right) .
  \end{align}
  Note that this representation of the joint distribution accounts for the constraint $w_r' w_s = 0$,
  in addition to $\| w_r \| = \|w_s \| =1$ and the invariance under the orthogonal group $O(T-R^0)$.
  Using this representation of the joint distribution of $w_r$ and $w_s$ the proof is now analogous
  to the case $r=s$. The result can be shown for any $\varepsilon>0$.

  \# Part $(iii)$. This again follows since $e$ has $iid {\cal N}(0,\sigma^2)$ entries and from the resulting
  rotational invariance of $e$ wrt to orthogonal $O(N)$ and $O(T)$ rotations from the left and right, respectively.
\end{proof}

\section{Additional Monte Carlo Simulations}
\label{app:MC}
\subsection{``Empirical Monte Carlo''}

The static model in the empirical illustration reads
\begin{equation*}
Y_{it}= \sum_{k=1}^{8} \beta _{k}  X_{k,it}
+\alpha_{i}+ \gamma_i \, t + \delta_i \, t^{2} + \mu_t +\lambda _{i}^{\prime }f_{t}+e_{it}.
\end{equation*}
As described in the main text, 
estimates for $\beta$, $\lambda_i$ and  $f_t$ are obtained by applying
the LS estimation procedure with $R=4$ to
$\widetilde Y_{it}= \sum_{k=1}^{8} \beta _{k}  \widetilde X_{k,it}
  +\widetilde \lambda _{i}^{\prime } \widetilde f_{t} + \widetilde e_{it}$, 
where  
$\widetilde Y = M_{1_N} Y M_{(1_T, {\bf t}, {\bf t}^2)}$
and $\widetilde X_k = M_{1_N} X_k M_{(1_T, {\bf t}, {\bf t}^2)}$
are the outcome variable and regressors after projecting out
$\alpha_{i}$, $\gamma_i$, $\delta_i$ and $\mu_t$. We then construct
the bias corrected estimator $\widehat \beta_R^{\rm BC}$, as reported
in the $R=4$ column of Table~\ref{tab:EstimateBetaNoLag}.
We afterwards estimate $\alpha_{i}$, $\gamma_i$, $\delta_i$ and $\mu_t$
by applying least squares with outcome variable given by the residuals
$Y_{it} - \sum_{k=1}^{8} \widehat \beta^{\rm BC}_{R,k}  X_{k,it}
- \widehat \Lambda _{R,i}^{\prime } \widehat F_{R,t}$, obtaining
 $\widehat \alpha_{R,i}$, $\widehat \gamma_{R,i}$, $\widehat \delta_{R,i}$ and $\widehat \mu_{R,t}$.

For the simulation we generate $e_{it}$ according to the $MA(1)$ process
\begin{align*}
   e^s_{it} &= 0.1 (v_{it} +v_{i,t-1} )  ,
\end{align*}
where $v_{it} \, \sim \, iid \, t(5)$, i.e.~$v_{it}$ has a Student's t-distribution  with 5 degrees of freedom. The factor $0.1$ in the formula
for $e^s_{it}$ was chosen to reproduce standard deviations for 
$\widehat \beta^{\rm BC}_k$ in the simulation that are close to the estimated standard errors
in the actual application.

We set $R^0=4$ and generate the simulated outcome variable as
\begin{equation*}
Y^s_{it}= \sum_{k=1}^{8} \widehat \beta^{\rm BC}_{R^0,k}  X_{k,it}
+\widehat \alpha_{R^0,i}+ \widehat \gamma_{R^0,i} \, t + \widehat \delta_{R^0,i} \, t^{2} + \widehat \mu_{R^0,t} +\widehat  \Lambda _{R^0,i}^{\prime }
\widehat F_{R^0,t}+e^s_{it}.
\end{equation*}
The sample size is $N=48$ and $T=33$, as in the real data.
We generate 10.000 Monte Carlo samples in this way and for each
sample apply the 
same bias corrected estimator for $\beta$ that was
reported in Table~\ref{tab:EstimateBetaNoLag} for the real data.

Table~\ref{tab:MC-Empirical} reports the finite sample bias and standard deviation of
$\widehat \beta^{\rm BC}_{R}$ for $R \in \{0,\ldots,9\}$. 
We also report the empirical size of a size $5 \%$ $t$-test for whether 
each coefficient is equal to its true value. The standard error estimator
used for the $t$-test allows for heteroscedasticity and serial correlation.

We find that the
bias corrected estimates for $\beta_k$, $k=1,\ldots,8$, 
are essentially unbiased when $R \geq R^0$ factors are used in the estimation, but for $R<R^0$ the coefficient estimates are often biased.
For $\beta_k$, $k \geq 3$, there are only small changes in the
 standard deviation of the estimator between $R=4$ and $R=9$,
  but for $\beta_k$, $k=1,2$, we observe standard deviation inflation
of up to $25 \%$ between $R=4$ and $R=9$.

For $k \geq 5$ the empirical sizes of the $t$-test are quite accurate,
but for $k \leq 4$ the finite sample $t$-test overrejects
the null even for $R=R^0=4$. 

Given the relatively small sample size 
the difference between $R=9$ and $R^0=4$ is relatively large, and some
finite sample inefficiency and size distortions are not too surprising.

\subsection{Dynamic Model}

Here, we consider an AR(1) panel model with two factors ($R^0=2$)
and the following data generating process (DGP):
\begin{align}
   Y_{it} &= \beta^0 Y_{i,t-1} + \sum_{r=1}^2 \lambda_{ir} f_{tr} + e_{it} ,
   &
   f_{tr} = 0.5   \, f_{t-1,r} + \frac {\varepsilon_{tr}} {\sqrt{1-0.5^2}}   .
   \label{DGP-Dynamic}
\end{align}
The random variables
$\lambda_{ir}$, $\varepsilon_{tr}$ and $e_{it}$
are mutually independent; with $\lambda_{ir} \, \sim iid \, {\cal N}(1,1)$;  and
$\varepsilon_{tr}$ and $e_{it} \, \sim iid \, {\cal N}(0,1)$.
The AR(1) processes for $Y_{it}$ and $f_{tr}$ are initiated with 100 time periods before the actual estimation sample starts,
so that the initial conditions roughly correspond to the long-run static distribution.
We choose $\beta^0 \in \{0.2,0.5,0.8\}$, and
use $10,000$ repetitions in our simulation.
The true number of factors is chosen to be $R^0=2$.
For each draw of $Y$ and $X$
we compute the LS estimator $\widehat \beta_R$ according to equation \eqref{estimator}
for different values of $R$.

Table~\ref{tab:MC-Dynamic-1} reports bias and standard deviation of the estimator $\widehat \beta_R$
for $N=300$ and different combinations of $R$, $T$ and $\beta^0$.
Table~\ref{tab:MC-Dynamic-2} reports various quantiles of
the distribution of $\sqrt{NT}( \widehat \beta_R - \beta^0 )$
for $N=300$ and different combinations of $R$, $T$ and $\beta^0$.
Table~\ref{tab:MC-Dynamic-3} reports the size of a t-test with nominal size equal to $5 \%$ for $R \geq R^0$. 
We use the results in Bai~\cite*{Bai2009} and Moon and Weidner~\cite*{MoonWeidner2013}
to correct for the leading $1/N$ (not actually present in our DGP)
and $1/T$ (present in our DGP) biases in $\widehat \beta_R$ before calculating the t-test statistics, allowing for 
predetermined regressors and heteroscedsticity in both panel dimensions when estimating the bias and standard
deviation of $\widehat \beta_R$. 

\section{Comments Regarding Numerical Calculation of $\widehat \beta_R$}
\label{app:Numerics}

Different iteration schemes can be used to implement the LS estimator defined
 in \eqref{estimator} numerically:
 
 \begin{itemize}
    \item[(1)] Ahn, Lee and Schmidt \cite*{AhnLeeSchmidt2001} use
 an iteration scheme were the following steps are repeated until convergence:
 (a) for fixed $\widetilde \beta$ find $\widetilde F$ and $\widetilde \Lambda$
that minimize the LS objective function in \eqref{estimator} via principal component analysis (but $\widetilde \Lambda$
need not actually be computed);
 (b) for fixed $\widetilde F$ find $\widetilde \beta$ and $\widetilde \Lambda$
that minimize the LS objective function in \eqref{estimator} (but $\widetilde \Lambda$
need not actually be computed, because $\widetilde \beta$ can  be obtained by regressing
 $Y$ on $X_k M_{\widetilde F}$).
 
    \item[(2)]  Alternatively, Bai~\cite*{Bai2009} proposes the following iteration steps:
(a)  for fixed $\widetilde \beta$ find $\widetilde F$ and $\widetilde \Lambda$ 
that minimize the LS objective function in \eqref{estimator} via principal component analysis;
(b) for fixed $\widetilde F$ and $\widetilde \Lambda$ find $\widetilde \beta$
that minimizes the LS objective function in \eqref{estimator} by running a regression
of $(Y - \widetilde \Lambda   \widetilde F')$ on $X_k$. 

   \item[(3)] Another iteration scheme, which we have used in our implementation, and we have not found discussed
previously in the literature, is the following: (a)   for fixed $\widetilde \beta$ find $\widetilde F$ and
$\widetilde \Lambda$ that minimize the LS objective function in \eqref{estimator} via principal component analysis;
(b) for fixed $\widetilde F$ and $\widetilde \Lambda$ find $\widetilde \beta$ 
that minimizes the alternative objective function $\|M_{\widetilde \Lambda} (Y - \beta \cdot X) \allowbreak M_{\widetilde F}\|_{HS}^2$  by running a regression
of $Y$ on $M_{\widetilde \Lambda} X_k M_{\widetilde F}$.

 \end{itemize}
All three iteration schemes have the same step~(a), i.e. differ from each other only in step~(b). 
Each step of the iteration schemes~(1) and (2) minimizes the LS objective function, i.e. those schemes
guarantee that the sum of squared residuals is non-increasing in each step.
In contrast, step~(b) in scheme~(3) minimizes an alternative objective function, 
i.e. it is possible that the LS objective function in \eqref{estimator} is actually increasing during that step.
However, this step  can nevertheless be justified, namely one can show that
close to any (local) minimum the profile objective function ${\cal L}_{NT}^R(\beta)$
is well approximated by the alternative objective function
  $\frac 1 {NT} \|M_{\widetilde \Lambda} (Y - \beta \cdot X) M_{\widetilde F}\|^2_{HS}$,
  i.e. step~(b) in scheme~(3) is minimizing an approximation of ${\cal L}_{NT}^R(\beta)$.\footnote{%
  Step~(b) in scheme (1) and (2) can be equivalently described as minimizing
  the objective functions  
  $\frac 1 {NT} \| (Y - \beta \cdot X) M_{\widetilde F}\|^2_{HS}$
  and $\frac 1 {NT} \| (Y - \beta \cdot X - \widetilde \Lambda   \widetilde F')  \|^2_{HS}$, respectively,
  which are also approximations of ${\cal L}_{NT}^R(\beta)$. However, those approximations
  are less precise than the approximation in step~(b) of scheme (3). Namely, close to the minimizer
  $\widehat \beta_R$ of ${\cal L}_{NT}^R(\beta)$ we have
   ${\cal L}_{NT}^R(\beta) = \frac 1 {NT} \|M_{\widetilde \Lambda} (Y - \beta \cdot X) M_{\widetilde F}\|^2_{HS}
      + {\cal O}_P( \| \beta - \widehat \beta_R \|^3 )$,
    while the other two approximations have remainders of order 
    $\| \beta - \widehat \beta_R \|^2$.
  }
  
  Bai~\cite*{Bai2009} points out that the iteration scheme (2) is somewhat more robust towards 
  the choice of starting value for $\beta$, which was confirmed in our simulations exercises, both compared
  to scheme (1) and to scheme (3). However, once close to a (local) minimum of the LS objective function
  we found the convergence rate of  scheme (3) to be significantly faster than the conference
  rates of scheme (1) and (2). 
  Scheme (1) performed between scheme (2) and (3) in terms of both robustness and speed.  
  Each iteration scheme therefore has its relative advantages and disadvantages.
  We use scheme (3) for our final implementation, because the LS objective (and the profile
  objective function ${\cal L}_{NT}^R(\beta)$) can have multiple local minima, so that multiple optimization runs
  with different starting values are usually necessary anyways to achieve confidence that the global minimum
  was actually found. By using scheme (3) we minimize the time required for each optimization run, which enables
  us to try out more starting values within the same amount of total CPU time. Combining different iteration schemes
  (e.g. starting with scheme (2) and switching to scheme (3) once close to a minimum) might also be a good idea,
  which we have not explored, however.

\section{Verifying the Assumptions in Bai~(2009) for Example in Section~\ref{sec:AsymptoticSummary} }
\label{app:CheckBai}

Throughout this section we only consider the particular DGP in the example
of Section~\ref{sec:AsymptoticSummary} in the main text.
For this DGP it is easy to see that the OLS estimator
$\widehat \beta_0$ (the LS-estimator with $R=R^0=0$) is $\sqrt{NT}$-consistent, while the example shows that 
$\widehat \beta_1$ (the LS-estimator with $R=1$) is only $\sqrt{N}$-consistent. 
In the following we show that the regularity conditions imposed in Bai~\cite*{Bai2009}
are also satisfied for this DGP. This is interesting to verify, since then example
also shows that we need stronger Assumptions than those imposed in Bai~\cite*{Bai2009} 
in order to derive our results for $\widehat \beta_R$ for $R>R^0$.

\subsection*{Verifying Assumptions A, B, D, E}

\begin{itemize}

\item Since $R^0=0$ we find that
Assumption~A in Bai~\cite*{Bai2009} becomes $\frac 1 {NT} \sum_{i,t} X_{it}^2 >0$, which is satisfied. The
assumption would also be satisfied for $R^0>0$, since the component $\widetilde X$ makes the 
regressors $X$ a ``high-rank regressor''.

\item Bai's Assumption~B is trivially satisfied for $R^0=0$.

\item Assumption~D in Bai~\cite*{Bai2009} requires strict exogeneity of the regressors in the sense that 
$X$ and $e$ are independent, which is also satisfied.\footnote{
    One could also consider $\lambda_x$ and $f_x$ as random, but independent of $e$ and $\widetilde X$.
    In that case $X$ and $e$ are still strictly exogenous in the sense of mean-independence, i.e.
    we have $\mathbbm{E}( e|X) = 0$, but $e$ and $X$ are not fully independent.
    However, our Corollary~\ref{cor:LimitR0} in the main text
    (see also Moon and Weidner~\cite*{MoonWeidner2013}) shows that 
    the asymptotic distribution of $\widehat \beta_{R^0}$ 
    can be derived under the weaker exogeneity assumption
    $\mathbbm{E}( e_{it} X_{it} )=0$. Full independence of $e$ and $X$ is therefore only
    assumed for convenience in Bai~\cite*{Bai2009}, and his results on $\widehat \beta_{R^0}$  remain
    unchanged when only imposing $\mathbbm{E}( e|X) = 0$ instead. }

\item Finally, Assumption~E  in Bai~\cite*{Bai2009} becomes $\frac 1 {\sqrt{NT}} \sum_{it} X_{it} e_{it} 
   \rightarrow_d {\cal N}(0,D_Z)$,
where $D_Z = \lim_{N,T \rightarrow \infty} {\rm Var}\left[ \frac 1 {\sqrt{NT}} \sum_{it} X_{it} e_{it}  \right]$.
This is also satisfied, since $X_{it} e_{it}$ is independent across $i$ and over $t$, and has
bounded variance.

\end{itemize}

\subsection*{Verifying Assumption C}

A more difficult task is to verify
Assumption C in Bai~\cite*{Bai2009}, which contains regularity conditions for $e_{it}$.\footnote{
Essentially, Assumption C requires that $e_{it}$ is mean zero and weakly correlated across $i$ and $t$.
Thus, it plays the same role as our high-level 
Assumption~\ref{ass:SN}$(ii)$, which is easy to check since  
$\|e\| \leq \left\| \mathbbm{1}_{N} + c \, \frac{\lambda _{x}\lambda _{x}^{\prime }}{N} \right\|
  \| u \| \left\| \mathbbm{1}_{T}+c \, \frac{f_{x}f_{x}^{\prime }}{T}\right\| $ 
  and we have
  $\left\| \mathbbm{1}_{N} + c \, \frac{\lambda _{x}\lambda _{x}^{\prime }}{N} \right\| 
  \leq \left\| \mathbbm{1}_{N} \right\| + c \left\| \frac{\lambda _{x}\lambda _{x}^{\prime }}{N} \right\|  = {\cal O}(1)$,
  and 
  $ \left\| \mathbbm{1}_{T}+c \, \frac{f_{x}f_{x}^{\prime }}{T}\right\| 
  \leq  \left\| \mathbbm{1}_{T} \right\| +c \left\| \frac{f_{x}f_{x}^{\prime }}{T}\right\|  =  {\cal O}(1)$,
  and $\| u \| = {\cal O}_P( \sqrt{\min(N,T)})$ --- see Appendix~\ref{app:SpectralNorm} in the main text
  regarding the last statement.}
In our notation the assumption reads

\begin{itemize}
\item[(i)] $\mathbbm{E}\left( e_{it}\right) =0$ and $\mathbbm{E}\left( e_{it}^{8}\right) \leq M,$

\item[(ii)] Let $\mathbbm{E}\left( e_{it}e_{js}\right) =\sigma _{ij,ts}.$ Then, $\frac{1%
}{N}\sum_{i,j=1}^{N}\sup_{t,s}\left\vert \sigma _{ij,ts}\right\vert \leq M,$ 
$\frac{1}{T}\sum_{t,s=1}^{T}\sup_{i,j}\left\vert \sigma _{ij,ts}\right\vert
\leq M,$ and $\frac{1}{NT}\sum_{i,j=1}^{N}\sum_{t,s=1}^{T}\left\vert \sigma
_{ij,ts}\right\vert \leq M.$ Also, the largest eigenvalue of $\mathbbm{E}\left(
e_{i}e_{i}^{\prime }\right) $ is bounded uniformly in $i$ and $T.$

\item[(iii)] For every $\left( t,s\right) ,$ $\mathbbm{E}\left\vert \frac{1}{N^{1/2}}%
\sum_{i=1}^{N}\left[ e_{is}e_{it}-\mathbbm{E}\left( e_{is}e_{it}\right) \right]
\right\vert ^{4}\leq M.$

\item[(iv)] Moreover, 
\begin{eqnarray*}
\frac{1}{NT^{2}}\sum_{t,s,u,v}\sum_{i,j}\left\vert {\rm Cov}\left(
e_{it}e_{is},e_{ju}e_{jv}\right) \right\vert &\leq &M \\
\frac{1}{N^{2}T}\sum_{t,s}\sum_{i,j,k,l}\left\vert {\rm Cov}\left(
e_{it}e_{jt},e_{ks}e_{ls}\right) \right\vert &\leq &M.
\end{eqnarray*}
\end{itemize}

\bigskip

In the following  $M$ always denotes some global constant, whose precise value may change from equation
to equation.
Furthermore, we simply write $\lambda$ and $f$ instead of $\lambda_x$ and $f_x$. We use notation $1\{ \cdot \}$ to denote the indicator function.

\bigskip

We also define $v_1 = c\frac{\lambda }{\sqrt{N}}\frac{\lambda ^{\prime }u}{\sqrt{N}}$,
$v_2 = c\frac{uf}{\sqrt{T}}\frac{f^{\prime }}{\sqrt{T}}$,
$v_3 = \left( \frac{c^{2}}{\sqrt{NT}} \lambda ^{\prime }uf\right) \frac{\lambda }{\sqrt{N}}\frac{f^{\prime }}{%
\sqrt{T}} $,
and $v=v_1+ v_2 +v_3$.  We then have
\begin{eqnarray}
e &=&\left( I+c\frac{\lambda \lambda ^{\prime }}{N}\right) u\left( I+c\frac{%
ff^{\prime }}{T}\right)  
 = u+v .
 \label{eit} 
\end{eqnarray}%
Notice that 
\begin{eqnarray*}
v_{it} &=&\frac{\lambda _{i}}{\sqrt{N}}\left( \frac{c}{\sqrt{N}}%
\sum_{h=1}^{N}\lambda _{h}u_{ht}\right) +\left( \frac{c}{\sqrt{T}}\sum_{\tau
=1}^{T}u_{i\tau }f_{\tau }\right) \frac{f_{t}}{\sqrt{T}}+\left( \frac{c^{2}}{%
\sqrt{NT}}\lambda ^{\prime }uf\right) \frac{\lambda _{i}}{\sqrt{N}}\frac{%
f_{t}}{\sqrt{T}} \\
&=&\frac{\lambda _{i}}{\sqrt{N}}\tilde{\lambda}\left( \lambda ,u_{\cdot
t}\right) +\tilde{f}\left( f,u_{i\cdot }\right) \frac{f_{t}}{\sqrt{T}}+%
\tilde{v}\left( \lambda ,f,u\right) \frac{\lambda _{i}}{\sqrt{N}}\frac{f_{t}%
}{\sqrt{T}} \\
&=&v_{1,it}+v_{2,it}+v_{3,it},
\end{eqnarray*}
where we defined $\tilde{\lambda}\left( \lambda ,u_{\cdot
t}\right)$, $\tilde{f}\left( f,u_{i\cdot }\right)$ and $\tilde{v}\left( \lambda ,f,u\right)$ implicitly.
We also define $g_{1,ij,ts}=u_{it}u_{js}$, $g_{2,ij,ts}=u_{it}v_{js}$, $g_{3,ij,ts}=v_{it}u_{js}$, and
$g_{4,ij,ts}=v_{it}v_{js}$, so that
\begin{eqnarray*}
e_{it}e_{js} &=&g_{1,ij,ts}+g_{2,ij,ts}+g_{3,ij,ts}+g_{4,ij,ts} .
\end{eqnarray*}
In the following
we discuss part $(i)$, $(ii)$, $(iii)$ and $(iv)$ of Assumption C separately:
  
\begin{small}  
  
\subsubsection*{\underline{Part $(i)$}}

This is straightforward to check.

\subsubsection*{\underline{Part $(ii)$}}

Let $\sigma _{k,ij,ts}=\mathbbm{E}\left( g_{k,ij,ts}\right) .$

\begin{enumerate}
\item $\sigma _{1,ij,ts}:$ The desired result follows since 
\begin{equation*}
\sigma _{1,ij,ts}=1\left\{ i=j,t=s\right\} .
\end{equation*}

\item $\sigma _{2,ij,ts}:$ By definition, 
\begin{equation*}
\sigma _{2,ij,ts}=\mathbbm{E}\left( u_{it}v_{js}\right) =\mathbbm{E}\left( u_{it}v_{1,js}\right)
+\mathbbm{E}\left( u_{it}v_{2,js}\right) +\mathbbm{E}\left( u_{it}v_{3,js}\right) .
\end{equation*}%
Direct calculations show that%
\begin{eqnarray*}
\mathbbm{E}\left( u_{it}v_{1,js}\right) &=&\mathbbm{E}\left[ u_{it}\frac{\lambda _{j}}{\sqrt{N}}%
\left( \frac{c}{\sqrt{N}}\sum_{h=1}^{N}\lambda _{h}u_{hs}\right) \right] =%
\frac{c}{N}\sum_{h=1}^{N}\lambda _{j}\lambda _{h}\mathbbm{E}\left( u_{hs}u_{it}\right)
\\
&=&\frac{c}{N}\sum_{h=1}^{N}\lambda _{j}\lambda _{h}1\left\{ i=h,s=t\right\}
\\
&=&\frac{c}{N}\lambda _{i}\lambda _{j}1\left\{ t=s\right\} . \\
\mathbbm{E}\left( u_{it}v_{2,js}\right) &=&\mathbbm{E}\left[ u_{it}\left( \frac{c}{\sqrt{T}}%
\sum_{\tau =1}^{T}u_{j\tau }f_{\tau }\right) \frac{f_{s}}{\sqrt{T}}\right] =%
\frac{c}{T}\sum_{\tau =1}^{T}f_{\tau }f_{s}\mathbbm{E}\left( u_{it}u_{j\tau }\right) \\
&=&\frac{c}{T}\sum_{\tau =1}^{T}f_{\tau }f_{s}1\left\{ i=j,t=\tau \right\} \\
&=&\frac{c}{T}f_{t}f_{s}1\left\{ i=j\right\} . \\
\mathbbm{E}\left( u_{it}v_{3,js}\right) &=&\mathbbm{E}\left[ u_{it}\left( \frac{c^{2}}{\sqrt{NT}}%
\sum_{h=1}^{N}\sum_{\tau =1}^{T}\lambda _{h}f_{\tau }u_{h\tau }\right) \frac{%
\lambda _{j}}{\sqrt{N}}\frac{f_{s}}{\sqrt{T}}\right] =\frac{c^{2}}{NT}%
\sum_{h=1}^{N}\sum_{\tau =1}^{T}\lambda _{h}\lambda _{j}f_{\tau
}f_{s}\mathbbm{E}\left( u_{it}u_{h\tau }\right) \\
&=&\frac{c^{2}}{NT}\sum_{h=1}^{N}\sum_{\tau =1}^{T}\lambda _{h}\lambda
_{j}f_{\tau }f_{s}1\left\{ i=h,t=\tau \right\} \\
&=&\frac{c^{2}}{NT}\lambda _{i}\lambda _{j}f_{t}f_{s}.
\end{eqnarray*}%
Combining these, we have%
\begin{equation*}
\sigma _{2,ij,ts}=\frac{c}{N}\lambda _{i}\lambda _{j}1\left\{ t=s\right\} +%
\frac{c}{T}f_{t}f_{s}1\left\{ i=j\right\} +\frac{c^{2}}{NT}\lambda
_{i}\lambda _{j}f_{t}f_{s}.
\end{equation*}%
Since $\lambda _{i}$ and $f_{t}$ are bounded, we have the desired result,%
\begin{eqnarray*}
\frac{1}{N}\sum_{i,j=1}^{N}\sup_{t,s}\left\vert \frac{c}{N}\lambda
_{i}\lambda _{j}1\left\{ t=s\right\} +\frac{c}{T}f_{t}f_{s}1\left\{
i=j\right\} +\frac{c^{2}}{NT}\lambda _{i}\lambda _{j}f_{t}f_{s}\right\vert
&\leq &M, \\
\frac{1}{T}\sum_{t,s=1}^{T}\sup_{i,j}\left\vert \frac{c}{N}\lambda
_{i}\lambda _{j}1\left\{ t=s\right\} +\frac{c}{T}f_{t}f_{s}1\left\{
i=j\right\} +\frac{c^{2}}{NT}\lambda _{i}\lambda _{j}f_{t}f_{s}\right\vert
&\leq &M, \\
\frac{1}{NT}\sum_{i,j=1}^{N}\sum_{t,s=1}^{T}\left\vert \frac{c}{N}\lambda
_{i}\lambda _{j}1\left\{ t=s\right\} +\frac{c}{T}f_{t}f_{s}1\left\{
i=j\right\} +\frac{c^{2}}{NT}\lambda _{i}\lambda _{j}f_{t}f_{s}\right\vert
&\leq &M.
\end{eqnarray*}

\item $\sigma _{3,ij,ts}:$ The result for the term $\sigma
_{3,ij,ts} $ follows similarly to the case of $\sigma _{2,ij,ts}.$

\item $g_{4,ij,ts}:$ By definition, we have 
\begin{eqnarray*}
\mathbbm{E}\left( g_{4,ij,ts}\right) 
&=&\sum_{k=1}^{3}\sum_{l=1}^{3}\mathbbm{E}\left( v_{k,it}v_{l,js}\right) \\
&=&\mathbbm{E}\left( 
\begin{array}{c}
\left[ 
\begin{array}{c}
\frac{\lambda _{i}}{\sqrt{N}}\left( \frac{c}{\sqrt{N}}\sum_{k=1}^{N}\lambda
_{k}u_{kt}\right) +\left( \frac{c}{\sqrt{T}}\sum_{p=1}^{T}f_{p}u_{ip}\right) 
\frac{f_{t}}{\sqrt{T}} \\ 
+\frac{\lambda _{i}}{\sqrt{N}}\left( \frac{c^{2}}{\sqrt{NT}}%
\sum_{k=1}^{N}\sum_{p=1}^{T}\lambda _{k}f_{p}u_{kp}\right) \frac{f_{t}}{%
\sqrt{T}}%
\end{array}%
\right] \\ 
\times \left[ 
\begin{array}{c}
\frac{\lambda _{j}}{\sqrt{N}}\left( \frac{c}{\sqrt{N}}\sum_{l=1}^{N}\lambda
_{l}u_{ls}\right) +\left( \frac{c}{\sqrt{T}}\sum_{q=1}^{T}f_{q}u_{jq}\right) 
\frac{f_{s}}{\sqrt{T}} \\ 
+\frac{\lambda _{j}}{\sqrt{N}}\left( \frac{c^{2}}{\sqrt{NT}}%
\sum_{l=1}^{N}\sum_{q=1}^{T}\lambda _{q}f_{l}u_{lq}\right) \frac{f_{s}}{%
\sqrt{T}}%
\end{array}%
\right]%
\end{array}%
\right) .
\end{eqnarray*}%
Notice that 
\begin{eqnarray*}
\mathbbm{E}\left( v_{1,it}v_{1,js}\right) &=&\frac{c^{2}}{N^{2}}\sum_{k=1}^{N}%
\sum_{l=1}^{N}\lambda _{i}\lambda _{k}\lambda _{j}\lambda _{l}\mathbbm{E}\left(
u_{kt}u_{ls}\right) \\
&=&\frac{c^{2}}{N}\lambda _{i}\lambda _{j}\left( \frac{1}{N}%
\sum_{k=1}^{N}\lambda _{k}^{2}\right) 1\left\{ t=s\right\} .
\end{eqnarray*}

\begin{eqnarray*}
\mathbbm{E}\left( v_{1,it}v_{2,js}\right) &=&\mathbbm{E}\left[ \frac{\lambda _{i}}{\sqrt{N}}%
\left( \frac{c}{\sqrt{N}}\sum_{k=1}^{N}\lambda _{k}u_{kt}\right) \left( 
\frac{c}{\sqrt{T}}\sum_{q=1}^{T}f_{q}u_{jq}\right) \frac{f_{s}}{\sqrt{T}}%
\right] \\
&=&\frac{c^{2}}{NT}\left[ \sum_{k=1}^{N}\sum_{q=1}^{T}\lambda _{i}\lambda
_{k}f_{q}f_{s}\mathbbm{E}\left( u_{kt}u_{jq}\right) \right] \\
&=&\frac{c^{2}}{NT}\lambda _{i}\lambda _{j}f_{t}f_{s},
\end{eqnarray*}%
\begin{eqnarray*}
\mathbbm{E}\left( v_{1,it}v_{3,js}\right) 
&=&\mathbbm{E}\left[ \frac{\lambda _{i}}{\sqrt{N}}\left( \frac{c}{\sqrt{N}}%
\sum_{k=1}^{N}\lambda _{k}u_{kt}\right) \frac{\lambda _{j}}{\sqrt{N}}\left( 
\frac{c^{2}}{\sqrt{NT}}\sum_{l=1}^{N}\sum_{q=1}^{T}\lambda
_{l}f_{q}u_{lq}\right) \frac{f_{s}}{\sqrt{T}}\right] \\
&=&\frac{c^{3}}{N^{2}T}\mathbbm{E}\left[ \sum_{k=1}^{N}\sum_{l=1}^{N}\sum_{q=1}^{T}%
\lambda _{i}\lambda _{k}\lambda _{j}\lambda _{l}f_{q}f_{s}\mathbbm{E}\left(
u_{kt}u_{lq}\right) \right] \\
&=&\frac{c^{3}}{NT}\lambda _{i}\lambda _{j}f_{t}f_{s}\left( \frac{1}{N}%
\sum_{k=1}^{N}\lambda _{k}^{2}\right) ,
\end{eqnarray*}%
\begin{eqnarray*}
\mathbbm{E}\left( v_{2,it}v_{2,js}\right) &=&\mathbbm{E}\left( \left( \frac{c}{\sqrt{T}}%
\sum_{p=1}^{T}f_{p}u_{ip}\right) \frac{f_{t}}{\sqrt{T}}\left( \frac{c^{2}}{%
\sqrt{T}}\sum_{q=1}^{T}f_{q}u_{jq}\right) \frac{f_{s}}{\sqrt{T}}\right) \\
&=&\frac{c^{3}}{T^{2}}\sum_{p=1}^{T}\sum_{q=1}^{T}f_{p}f_{t}f_{q}f_{s}\mathbbm{E}%
\left( u_{ip}u_{jq}\right) \\
&=&\frac{c^{3}}{T}\left( \frac{1}{T}\sum_{p=1}^{T}f_{p}^{2}\right)
f_{t}f_{s}1\left\{ i=j\right\} ,
\end{eqnarray*}%
\begin{eqnarray*}
\mathbbm{E}\left( v_{2,it}v_{3,js}\right) 
&=&\mathbbm{E}\left[ \left( \frac{c}{\sqrt{T}}\sum_{p=1}^{T}f_{p}u_{ip}\right) \frac{%
f_{t}}{\sqrt{T}}\frac{\lambda _{j}}{\sqrt{N}}\left( \frac{c^{2}}{\sqrt{NT}}%
\sum_{l=1}^{N}\sum_{q=1}^{T}\lambda _{l}f_{q}u_{lq}\right) \frac{f_{s}}{%
\sqrt{T}}\right] \\
&=&\frac{c^{3}}{NT^{2}}\sum_{p=1}^{T}\sum_{l=1}^{N}\sum_{q=1}^{T}f_{p}f_{t}%
\lambda _{j}\lambda _{l}f_{q}f_{s}\mathbbm{E}\left( u_{ip}u_{lq}\right) \\
&=&\frac{c^{3}}{NT}\lambda _{i}\lambda _{j}f_{t}f_{s}\left( \frac{1}{T}%
\sum_{p=1}^{T}f_{p}^{2}\right) ,
\end{eqnarray*}%
and%
\begin{align*}
\mathbbm{E}\left( v_{3,it}v_{3,js}\right) 
&=\mathbbm{E}\left( \frac{\lambda _{i}}{\sqrt{N}}\left( \frac{c^{2}}{\sqrt{NT}}%
\sum_{k=1}^{N}\sum_{p=1}^{T}\lambda _{k}f_{p}u_{kp}\right) \frac{f_{t}}{%
\sqrt{T}}\frac{\lambda _{j}}{\sqrt{N}}\left( \frac{c^{2}}{\sqrt{NT}}%
\sum_{l=1}^{N}\sum_{q=1}^{T}\lambda _{q}f_{l}u_{lq}\right) \frac{f_{s}}{%
\sqrt{T}}\right) \\
&=\frac{c^{4}}{N^{2}T^{2}}\sum_{k=1}^{N}\sum_{l=1}^{N}\sum_{p=1}^{T}%
\sum_{q=1}^{T}\lambda _{i}\lambda _{k}\lambda _{j}\lambda
_{q}f_{p}f_{t}f_{l}f_{s}\mathbbm{E}\left( u_{kp}u_{lq}\right) \\
&=\frac{c^{4}}{NT}\lambda _{i}\lambda _{j}f_{t}f_{s}\left( \frac{1}{N}%
\sum_{k=1}^{N}\lambda _{k}f_{k}\right) \left( \frac{1}{T}\sum_{p=1}^{T}%
\lambda _{p}f_{p}\right) .
\end{align*}%
From these and using the boundedness of $\lambda _{i}$ and $f_{t}$ we 
obtain
\begin{equation*}
\frac{1}{N}\sum_{i,j=1}^{N}\sup_{t,s}\left\vert \sigma _{4,ij,ts}\right\vert
\leq M,\text{ }\frac{1}{T}\sum_{t,s=1}^{T}\sup_{i,j}\left\vert \sigma
_{4,ij,ts}\right\vert \leq M,\text{ }\frac{1}{NT}\sum_{i,j=1}^{N}%
\sum_{t,s=1}^{T}\left\vert \sigma _{4,ij,ts}\right\vert \leq M.
\end{equation*}

Combining these, we have the desired result.
\end{enumerate}

What is left to show is the bound on the largest eigenvalue of $\Omega_i = \mathbbm{E}\left(
e_{i}e_{i}^{\prime }\right) $, which is equivalent to the spectral norm of $\Omega_i$.
The spectral norm of a symmetric matrix is bounded by the infinity norm, i.e.
we have $\mu_1(\Omega_i) = \| \Omega_i \| \leq \| \Omega_i \|_{\infty} = \max_t \sum_s |\Omega_{i,ts}|$.
For the elements $\Omega_{i,ts}$ of the matrix $\Omega_i$ we have
 $\Omega_{i,ts} =   \mathbbm{E}\left(
e_{it }e_{is} \right) =  \sigma_{ii,ts}$. We thus have
\begin{align*}
    \mu_1( \Omega_i ) &\leq \max_t \sum_s |\sigma_{ii,ts}| = {\cal O}(1) ,
\end{align*}
where the last step follows by the above results on $\sigma _{ij,ts} = \sum_{k=1}^4 \sigma _{k,ij,ts}$.

\subsubsection*{\underline{Part $(iii)$}}

Write 
\begin{eqnarray*}
\mathbbm{E}\left[ \frac{1}{N^{1/2}}\sum_{i=1}^{N}\left[ e_{is}e_{it}-\mathbbm{E}\left(
e_{is}e_{it}\right) \right] \right] ^{4} &=&\mathbbm{E}\left\{ \sum_{k=1}^{4}\left[ 
\frac{1}{N^{1/2}}\sum_{i=1}^{N}\left( g_{k,ii,ts}-\sigma _{k,ii,ts}\right) %
\right] \right\} ^{4} \\
&\leq &M\left\{ \mathbbm{E}\left[ \frac{1}{N^{1/2}}\sum_{i=1}^{N}\left(
g_{k,ii,ts}-\sigma _{k,ii,ts}\right) \right] ^{4}\right\} .
\end{eqnarray*}

\begin{enumerate}
\item $g_{1,ii,ts}:$ Since $u_{it}\sim iidN\left( 0,1\right) $ across $i$
and over $t,$ it is straightforward to see that for all $t,s,$ 
\begin{equation*}
\mathbbm{E}\left[ \frac{1}{\sqrt{N}}\sum_{i=1}^{N}\left( g_{1,ii,ts}-\sigma
_{1,ii,ts}\right) \right] ^{4}\leq M.
\end{equation*}

\item $g_{2,ii,ts}:$ Next, notice that 
\begin{equation*}
\frac{1}{\sqrt{N}}\sum_{i=1}^{N}\left( g_{2,ii,ts}-\sigma _{2,ii,ts}\right)
=\sum_{k=1}^{3}\frac{1}{\sqrt{N}}\sum_{i=1}^{N}\left( u_{it}v_{k,is}-\mathbbm{E}\left(
u_{it}v_{k,is}\right) \right) .
\end{equation*}

Due to the boundedness of $\lambda _{i}$ and $f_{t}$ and iid normality of $%
u_{it},$ we have the following.

First, 
\begin{eqnarray*}
&&\mathbbm{E}\left( \frac{1}{\sqrt{N}}\sum_{i=1}^{N}\left( u_{it}v_{1,is}-\mathbbm{E}\left(
u_{it}v_{1,is}\right) \right) \right) ^{4} \\
&=&\mathbbm{E}\left[ \frac{1}{\sqrt{N}}\sum_{i=1}^{N}\left\{ u_{it}\frac{\lambda _{j}}{%
\sqrt{N}}\left( \frac{c}{\sqrt{N}}\sum_{h=1}^{N}\lambda _{h}u_{hs}\right) -%
\frac{c}{N}\lambda _{i}\lambda _{j}1\left\{ t=s\right\} \right\} \right] ^{4}
\\
&=&\lambda _{j}^{4}c^{4}\mathbbm{E}\left[ \frac{1}{N}\sum_{i=1}^{N}\left\{
u_{it}\left( \frac{1}{\sqrt{N}}\sum_{h=1}^{N}\lambda _{h}u_{hs}\right) -%
\frac{1}{\sqrt{N}}\lambda _{i}1\left\{ t=s\right\} \right\} \right] ^{4} \\
&\leq &M,
\end{eqnarray*}%
where the last equality follows since 
\begin{equation*}
\left( \frac{1}{N}\sum_{i=1}^{N}a_{i}\right) ^{4}\leq \frac{1}{N}%
\sum_{i=1}^{N}a_{i}^{4} ,
\end{equation*}%
and $\sup_{i,t,s}\mathbbm{E}\left( \left\{ u_{it}\left( \frac{1}{\sqrt{N}}%
\sum_{h=1}^{N}\lambda _{h}u_{hs}\right) -\frac{1}{\sqrt{N}}\lambda
_{i}1\left\{ t=s\right\} \right\} ^{4}\right) \leq M.$

Second, similarly to the first case, we have%
\begin{eqnarray*}
&&\mathbbm{E}\left( \frac{1}{\sqrt{N}}\sum_{i=1}^{N}\left( u_{it}v_{2,is}-\mathbbm{E}\left(
u_{it}v_{2,is}\right) \right) \right) ^{4} \\
&=&\mathbbm{E}\left[ \frac{1}{\sqrt{N}}\sum_{i=1}^{N}\left\{ u_{it}\left( \frac{c}{%
\sqrt{T}}\sum_{\tau =1}^{T}u_{i\tau }f_{\tau }\right) \frac{f_{s}}{\sqrt{T}}-%
\frac{cf_{s}f_{t}}{T}\right\} \right] ^{4} \\
&=&f_{s}^{4}c^{4}\mathbbm{E}\left[ \frac{1}{\sqrt{NT}}\sum_{i=1}^{N}\left\{
u_{it}\left( \frac{1}{\sqrt{T}}\sum_{\tau =1}^{T}u_{i\tau }f_{\tau }\right)
f_{s}-\frac{f_{t}}{\sqrt{T}}\right\} \right] ^{4} \\
&\leq &M.
\end{eqnarray*}

Third, 
\begin{eqnarray*}
&&\mathbbm{E}\left( \frac{1}{\sqrt{N}}\sum_{i=1}^{N}\left( u_{it}v_{3,is}-\mathbbm{E}\left(
u_{it}v_{3,is}\right) \right) \right) ^{4} \\
&=&\mathbbm{E}\left( \frac{1}{\sqrt{N}}\sum_{i=1}^{N}\left\{ u_{it}\left( \frac{c^{2}}{%
\sqrt{NT}}\sum_{h=1}^{N}\sum_{\tau =1}^{T}\lambda _{h}f_{\tau }u_{h\tau
}\right) \frac{\lambda _{i}}{\sqrt{N}}\frac{f_{s}}{\sqrt{T}}-\frac{c^{2}}{NT}%
\lambda _{i}^{2}f_{t}f_{s}\right\} \right) ^{4} \\
&=&f_{s}^{4}c^{8}\mathbbm{E}\left( \frac{1}{\sqrt{NT}}\sum_{i=1}^{N}\left\{
u_{it}\left( \frac{c^{2}}{\sqrt{NT}}\sum_{h=1}^{N}\sum_{\tau =1}^{T}\lambda
_{h}f_{\tau }u_{h\tau }\right) \frac{\lambda _{i}}{\sqrt{N}}-\frac{c^{2}}{N%
\sqrt{T}}\lambda _{i}^{2}f_{t}\right\} \right) ^{4} \\
&\leq &M.
\end{eqnarray*}

Combining these, we have
\begin{equation*}
\mathbbm{E}\left( \frac{1}{\sqrt{N}}\sum_{i=1}^{N}\left( g_{2,ii,ts}-\sigma
_{2,ii,ts}\right) \right) ^{4}<M.
\end{equation*}

\item $g_{3,ii,ts}:$ Similarly, we find
\begin{equation*}
\mathbbm{E}\left( \frac{1}{\sqrt{N}}\sum_{i=1}^{N}\left( g_{3,ii,ts}-\sigma
_{3,ii,ts}\right) \right) ^{4}<M ,
\end{equation*}%
because $g_{3,ii,ts}=g_{2,ii,st}.$

\item $g_{4,ii,ts}:$ Finally, 
\begin{equation*}
\frac{1}{\sqrt{N}}\sum_{i=1}^{N}\left( g_{4,ii,ts}-\sigma _{4,ii,ts}\right)
=\sum_{k=1}^{3}\sum_{l=1}^{3}\frac{1}{\sqrt{N}}\sum_{i=1}^{N}\left(
v_{k,it}v_{l,is}-\mathbbm{E}\left( v_{k,it}v_{l,is}\right) \right) .
\end{equation*}
Notice that 
\begin{eqnarray*}
&&\frac{1}{\sqrt{N}}\sum_{i=1}^{N}\left( v_{1,it}v_{1,is}-\mathbbm{E}\left(
v_{1,it}v_{1,is}\right) \right) \\
&=&\frac{1}{\sqrt{N}}\sum_{i=1}^{N}\left( \frac{\lambda _{i}}{\sqrt{N}}%
\left( \frac{c}{\sqrt{N}}\sum_{k=1}^{N}\lambda _{k}u_{kt}\right) \frac{%
\lambda _{i}}{\sqrt{N}}\left( \frac{c}{\sqrt{N}}\sum_{l=1}^{N}\lambda
_{l}u_{ls}\right) -\frac{c^{2}}{N}\lambda _{i}^{2}\left( \frac{1}{N}%
\sum_{k=1}^{N}\lambda _{k}^{2}\right) 1\left\{ t=s\right\} \right) \\
&=&\frac{c^{2}}{\sqrt{N}}\frac{1}{N}\sum_{i=1}^{N}\left\{ \lambda
_{i}^{2}\left\{ \left( \frac{1}{\sqrt{N}}\sum_{k=1}^{N}\lambda
_{k}u_{kt}\right) \left( \frac{1}{\sqrt{N}}\sum_{l=1}^{N}\lambda
_{l}u_{ls}\right) -\left( \frac{1}{N}\sum_{k=1}^{N}\lambda _{k}^{2}\right)
1\left\{ t=s\right\} \right\} \right\} .
\end{eqnarray*}%
Notice that $\sup_{i,t,s}\mathbbm{E}\left[ \lambda _{i}^{2}\left\{ \left( \frac{1}{%
\sqrt{N}}\sum_{k=1}^{N}\lambda _{k}u_{kt}\right) \left( \frac{1}{\sqrt{N}}%
\sum_{l=1}^{N}\lambda _{l}u_{ls}\right) -\left( \frac{1}{N}%
\sum_{k=1}^{N}\lambda _{k}^{2}\right) 1\left\{ t=s\right\} \right\} \right]
^{4}\leq M.$ Therefore, we have%
\begin{eqnarray*}
&&\mathbbm{E}\left[ \frac{1}{\sqrt{N}}\sum_{i=1}^{N}\left( v_{1,it}v_{1,is}-\mathbbm{E}\left(
v_{1,it}v_{1,is}\right) \right) \right] ^{4} \\
&\leq &\left( \frac{c^{8}}{N^{2}}\right) \frac{1}{N}\sum_{i=1}^{N}\mathbbm{E}\left[
\left\{ \lambda _{i}^{2}\left\{ \left( \frac{1}{\sqrt{N}}\sum_{k=1}^{N}%
\lambda _{k}u_{kt}\right) \left( \frac{1}{\sqrt{N}}\sum_{l=1}^{N}\lambda
_{l}u_{ls}\right) -\left( \frac{1}{N}\sum_{k=1}^{N}\lambda _{k}^{2}\right)
1\left\{ t=s\right\} \right\} \right\} ^{4}\right] \\
&\leq &M.
\end{eqnarray*}%
Similarly, we can show the rest of the cases.
\end{enumerate}

\subsubsection*{\underline{Part $(iv)$}}

Without loss of
generality, we set $N=T$ here.
We show that $\frac{1}{NT^{2}}%
\sum_{t,s,u,v}\sum_{i,j}\left\vert {\rm Cov}\left(
e_{it}e_{is},e_{ju}e_{jv}\right) \right\vert \leq M.$ The other case follows
by the same fashion because the DGP is symmetric between 
$i$ and $t.$ Notice that 
\begin{equation*}
{\rm Cov}\left( e_{it}e_{is},e_{ju}e_{jv}\right) ={\rm Cov}\left(
\sum_{k=1}^{4}g_{k,ii,ts},\sum_{k=1}^{4}g_{k,jj,uv}\right)
=\sum_{k=1}^{4}\sum_{l=1}^{4}{\rm Cov}\left( g_{k,ii,ts},g_{l,jj,uv}\right) .
\end{equation*}%
Among $\left\{ {\rm Cov}\left( g_{1,ii,ts},g_{1,jj,uv}\right) \right\} $ there are
six kinds, (a) the term of $\left( u,u\right) $ and $\left( u,u\right) $ (b)
the terms of $\left( u,u\right) $ and $\left( u,v\right) $ (c) the terms of $%
\left( u,u\right) $ and $\left( v,v\right) ,$ (d) the terms of $\left(
u,v\right) $ and $\left( u,v\right) ,$ (e) the terms of $\left( u,v\right) $
and $\left( v,v\right) $, and (f) the term of $\left( v,v\right) $ and $%
\left( v,v\right) .$

In what follows we use "two pairs among $\left\{
t_{1},t_{2},t_{3},t_{4}\right\} $" to denote the sum of the three terms like 
$1\left\{ t_{1}=t_{2}\right\} 1\left\{ t_{3}=t_{4}\right\} .$

The main step in establishing the required result, $\frac{1}{NT^{2}}%
\sum_{t,s,u,v}\sum_{i,j}\left\vert {\rm Cov}\left( g_{k,ii,ts},g_{l,jj,uv}\right)
\right\vert \leq M,$ is to find an upper bound of $\left\vert {\rm Cov}\left(
g_{k,ii,ts},g_{l,jj,uv}\right) \right\vert $ in the form of $\frac{1}{%
N^{a}T^{b}}1\left\{ \text{ some pairs of indices}\right\} ,$ so that the
power of $NT^{2}N^{a}T^{b}=N^{a+b+3},$ $a+b+3,$ is larger than or equal to
the number of outstanding summations.

In the following proofs, we use the following fact multiple times:%
\begin{align}
& \mathbbm{E}\left( u_{it}u_{is}u_{ju}u_{jv}\right) -\mathbbm{E}\left( u_{it}u_{is}\right)
\mathbbm{E}\left( u_{ju}u_{jv}\right)   \notag  \\
&\quad = 1\left\{ i\neq j\right\} 1\left\{ t=s\right\} 1\left\{ u=v\right\}
+1\left\{ i=j\right\} 1\left\{ \text{two pairs among }\left\{
t,s,u,v\right\} \right\}   \notag \\
&\qquad -1\left\{ t=s\right\} 1\left\{ u=v\right\}   \notag \\
&\quad = 1\left\{ i=j\right\} 1\left\{ \text{two pairs among }\left\{
t,s,u,v\right\} \right\} .  \label{useful.fact}
\end{align}

\begin{enumerate}
\item ${\rm Cov}\left( u_{it}u_{is},u_{ju}u_{jv}\right) :$ Notice that 
\begin{eqnarray*}
{\rm Cov}\left( g_{1,ii,ts},g_{2,jj,uv}\right)  
&=&{\rm Cov}\left( u_{it}u_{is},u_{ju}u_{jv}\right) =\mathbbm{E}\left(
u_{it}u_{is}u_{ju}u_{jv}\right) -\mathbbm{E}\left( u_{it}u_{is}\right) \mathbbm{E}\left(
u_{ju}u_{jv}\right)  \\
&=&1\left\{ i=j\right\} 1\left\{ \text{two pairs among }\left\{
t,s,u,v\right\} \right\} .
\end{eqnarray*}%
This implies that out of the six summations over indices $\left(
t,s,u,v,i,j\right) ,$ only three summations matter. Therefore, we have 
\begin{equation*}
\frac{1}{NT^{2}}\sum_{t,s,u,v}\sum_{i,j}\left\vert {\rm Cov}\left(
g_{1,ii,ts},g_{1,jj,uv}\right) \right\vert \leq M.
\end{equation*}

\item ${\rm Cov}\left( u_{it}u_{is},u_{ju}v_{jv}\right) :$ Notice that 
\begin{equation*}
{\rm Cov}\left( u_{it}u_{is},u_{ju}v_{jv}\right) =\sum_{k=1}^{3}{\rm Cov}\left(
u_{it}u_{is},u_{ju}v_{k,jv}\right) .
\end{equation*}

\begin{enumerate}
\item Notice that 
\begin{eqnarray*}
{\rm Cov}\left( u_{it}u_{is},u_{ju}v_{k,jv}\right)  
&=&\mathbbm{E}\left( u_{it}u_{is}u_{ju}v_{1,jv}\right) -\mathbbm{E}\left( u_{it}u_{is}\right)
\mathbbm{E}\left( u_{ju}v_{1,jv}\right)  \\
&=&\frac{c}{N}\left( \sum_{h=1}^{N}\lambda _{j}\lambda _{h}\left\{ \mathbbm{E}\left(
u_{it}u_{is}u_{ju}u_{hv}\right) -\mathbbm{E}\left( u_{it}u_{is}\right) \mathbbm{E}\left(
u_{ju}u_{hv}\right) \right\} \right)  \\
&=&\frac{c}{N}\lambda _{j}^{2}\left\{ \mathbbm{E}\left(
u_{it}u_{is}u_{ju}u_{jv}\right) -\mathbbm{E}\left( u_{it}u_{is}\right) \mathbbm{E}\left(
u_{ju}u_{jv}\right) \right\}  \\
&\leq &\frac{c}{N}\lambda _{j}^{2}1\left\{ i=j\right\} 1\left\{ \text{two
pairs among }\left\{ t,s,u,v\right\} \right\} .
\end{eqnarray*}%
Therefore, 
\begin{align*}
& \frac{1}{NT^{2}}\sum_{t,s,u,v}\sum_{i,j}\left\vert \mathbbm{E}\left(
u_{it}u_{is}u_{ju}v_{1,jv}\right) \right\vert  \\
&\qquad \qquad \leq \frac{M}{N^{2}T^{2}}\sum_{t,s,u,v}\sum_{i,j}1\left\{ i=j\right\}
1\left\{ \text{two pairs among }\left\{ t,s,u,v\right\} \right\} \leq M.
\end{align*}

\item Similarly, we have 
\begin{eqnarray*}
{\rm Cov}\left( u_{it}u_{is},u_{ju}v_{2,jv}\right)  
&=&\mathbbm{E}\left( u_{it}u_{is}u_{ju}v_{2,jv}\right) -\mathbbm{E}\left( u_{it}u_{is}\right)
\mathbbm{E}\left( u_{ju}v_{2,jv}\right)  \\
&=&\frac{c}{T}\left( \sum_{\tau =1}^{T}\left\{ \mathbbm{E}\left(
u_{it}u_{is}u_{ju}u_{j\tau }\right) -\mathbbm{E}\left( u_{it}u_{is}\right) \mathbbm{E}\left(
u_{ju}u_{j\tau }\right) \right\} f_{\tau }f_{v}\right)  \\
&=&\frac{c}{T}\sum_{\tau =1}^{T}f_{\tau }f_{v}1\left\{ i=j\right\} 1\left\{ 
\text{two pairs among }\left\{ t,s,u,\tau \right\} \right\} ,
\end{eqnarray*}%
which leads the desired result  
\begin{equation*}
\frac{1}{NT^{2}}\sum_{t,s,u,v}\sum_{i,j}\left\vert {\rm Cov}\left(
u_{it}u_{is},u_{ju}v_{2,jv}\right) \right\vert \leq M .
\end{equation*}

\item Notice that 
\begin{align*}
{\rm Cov}\left( u_{it}u_{is}u_{ju}v_{3,jv}\right)  
&=\mathbbm{E}\left( u_{it}u_{is}u_{ju}v_{3,jv}\right) -\mathbbm{E}\left( u_{it}u_{is}\right)
\mathbbm{E}\left( u_{ju}v_{3,jv}\right)  \\
&=\frac{c^{2}}{NT}\left( \sum_{h=1}^{N}\sum_{\tau =1}^{T}\lambda
_{j}\lambda _{h}f_{\tau }f_{v}\left[ \mathbbm{E}\left( u_{it}u_{is}u_{ju}u_{h\tau
}\right) -\mathbbm{E}\left( u_{it}u_{is}\right) \mathbbm{E}\left( u_{ju}u_{h\tau }\right) \right]
\right)  \\
&=\frac{c^{2}}{NT}\left( \sum_{\tau =1}^{T}\lambda _{j}^{2}f_{\tau }f_{v}%
\left[ \mathbbm{E}\left( u_{it}u_{is}u_{ju}u_{j\tau }\right) -\mathbbm{E}\left(
u_{it}u_{is}\right) \mathbbm{E}\left( u_{ju}u_{j\tau }\right) \right] \right)  \\
&=\frac{c^{2}}{NT}\left( \sum_{\tau =1}^{T}\lambda _{j}^{2}f_{\tau
}f_{v}1\left\{ i=j\right\} 1\left\{ \text{two pairs among }\left\{
t,s,u,\tau \right\} \right\} \right) .
\end{align*}%
Therefore, 
\begin{equation*}
\frac{1}{NT^{2}}\sum_{t,s,u,v}\sum_{i,j}\left\vert {\rm Cov}\left(
u_{it}u_{is}u_{ju}v_{1,jv}\right) \right\vert \leq M.
\end{equation*}%
Combining these, we have the desired result.
\end{enumerate}

\item ${\rm Cov}\left( u_{it}u_{is},v_{ju}v_{jv}\right) :$ Notice that 
\begin{align*}
&{\rm Cov}\left( u_{it}u_{is},v_{ju}v_{jv}\right)  \\
&=\mathbbm{E}\left( u_{it}u_{is}v_{ju}v_{jv}\right) -\mathbbm{E}\left( u_{it}u_{is}\right)
\mathbbm{E}\left( v_{ju}v_{jv}\right)  \\
&=\mathbbm{E}\left( 
\begin{array}{c}
u_{it}u_{is}\left( \frac{\lambda _{j}}{\sqrt{N}}\left( \frac{1}{\sqrt{N}}%
\sum_{k=1}^{N}\lambda _{k}u_{ku}\right) +\left( \frac{1}{\sqrt{T}}%
\sum_{p=1}^{T}f_{p}u_{jp}\right) \frac{f_{u}}{\sqrt{T}}+\tilde{v}\left(
\lambda ,f,u\right) \frac{\lambda _{j}}{\sqrt{N}}\frac{f_{u}}{\sqrt{T}}%
\right)  \\ 
\times \left( \frac{\lambda _{j}}{\sqrt{N}}\left( \frac{1}{\sqrt{N}}%
\sum_{l=1}^{N}\lambda _{l}u_{lv}\right) +\left( \frac{1}{\sqrt{T}}%
\sum_{q=1}^{T}f_{q}u_{jq}\right) \frac{f_{v}}{\sqrt{T}}+\tilde{v}\left(
\lambda ,f,u\right) \frac{\lambda _{j}}{\sqrt{N}}\frac{f_{v}}{\sqrt{T}}%
\right) 
\end{array}%
\right)  \\
& \quad -1\left\{ t=s\right\} \mathbbm{E}\left( 
\begin{array}{c}
\left( \frac{\lambda _{j}}{\sqrt{N}}\left( \frac{1}{\sqrt{N}}%
\sum_{k=1}^{N}\lambda _{k}u_{ku}\right) +\left( \frac{1}{\sqrt{T}}%
\sum_{p=1}^{T}f_{p}u_{jp}\right) \frac{f_{u}}{\sqrt{T}}+\tilde{v}\left(
\lambda ,f,u\right) \frac{\lambda _{j}}{\sqrt{N}}\frac{f_{u}}{\sqrt{T}}%
\right)  \\ 
\times \left( \frac{\lambda _{j}}{\sqrt{N}}\left( \frac{1}{\sqrt{N}}%
\sum_{l=1}^{N}\lambda _{l}u_{lv}\right) +\left( \frac{1}{\sqrt{T}}%
\sum_{q=1}^{T}f_{q}u_{jq}\right) \frac{f_{v}}{\sqrt{T}}+\tilde{v}\left(
\lambda ,f,u\right) \frac{\lambda _{j}}{\sqrt{N}}\frac{f_{v}}{\sqrt{T}}%
\right) 
\end{array}%
\right) .
\end{align*}%
Here there are 9 terms in the product. 
\begin{enumerate}
\item Notice that 
\begin{align*}
{\rm Cov}\left( u_{it}u_{is},v_{1,ju}v_{1,jv}\right)  
&=\mathbbm{E}\left( u_{it}u_{is}\frac{\lambda _{j}}{\sqrt{N}}\left( \frac{1}{\sqrt{N}}%
\sum_{k=1}^{N}\lambda _{k}u_{ku}\right) \frac{\lambda _{j}}{\sqrt{N}}\left( 
\frac{1}{\sqrt{N}}\sum_{l=1}^{N}\lambda _{l}u_{lv}\right) \right)  \\
& \quad -\mathbbm{E}\left( u_{it}u_{is}\right) \mathbbm{E}\left( \frac{\lambda _{j}}{\sqrt{N}}\left( 
\frac{1}{\sqrt{N}}\sum_{k=1}^{N}\lambda _{k}u_{ku}\right) \frac{\lambda _{j}%
}{\sqrt{N}}\left( \frac{1}{\sqrt{N}}\sum_{l=1}^{N}\lambda _{l}u_{lv}\right)
\right)  \\
&=\frac{1}{N^{2}}\sum_{k=1}^{N}\sum_{l=1}^{N}\lambda _{j}^{2}\lambda
_{k}\lambda _{l}\left( \mathbbm{E}\left( u_{it}u_{is}u_{ku}u_{lv}\right) -\mathbbm{E}\left(
u_{it}u_{is}\right) \mathbbm{E}\left( u_{ku}u_{lv}\right) \right)  \\
&=\frac{1}{N^{2}}\sum_{k=1}^{N}\lambda _{j}^{2}\lambda _{k}^{2}\left(
\mathbbm{E}\left( u_{it}u_{is}u_{ku}u_{kv}\right) -\mathbbm{E}\left( u_{it}u_{is}\right) \mathbbm{E}\left(
u_{ku}u_{kv}\right) \right)  \\
&=\frac{1}{N^{2}}\sum_{k=1}^{N}\lambda _{j}^{2}\lambda _{k}^{2}\left\{
1\left\{ i=k\right\} 1\left\{ \text{two pairs among }\left\{ t,s,u,v\right\}
\right\} \right\} .
\end{align*}%
Therefore, 
\begin{equation*}
\frac{1}{NT^{2}}\sum_{t,s,u,v}\sum_{i,j}\left\vert {\rm Cov}\left(
u_{it}u_{is},v_{1,ju}v_{1,jv}\right) \right\vert \leq M.
\end{equation*}

\item Notice that 
\begin{align*}
{\rm Cov}\left( u_{it}u_{is},v_{2,ju}v_{2,jv}\right)  
&= \mathbbm{E}\left( u_{it}u_{is}\left( \frac{1}{\sqrt{T}}\sum_{p=1}^{T}f_{p}u_{jp}%
\right) \frac{f_{u}}{\sqrt{T}}\left( \frac{1}{\sqrt{T}}%
\sum_{q=1}^{T}f_{q}u_{jq}\right) \frac{f_{v}}{\sqrt{T}}\right)  \\
& \quad -\mathbbm{E}\left( u_{it}u_{is}\right) \mathbbm{E}\left( \left( \frac{1}{\sqrt{T}}%
\sum_{p=1}^{T}f_{p}u_{jp}\right) \frac{f_{u}}{\sqrt{T}}\left( \frac{1}{\sqrt{%
T}}\sum_{q=1}^{T}f_{q}u_{jq}\right) \frac{f_{v}}{\sqrt{T}}\right)  \\
&= \frac{1}{T^{2}}\sum_{p=1}^{T}\sum_{q=1}^{T}f_{p}f_{q}f_{u}f_{v}\left[
\mathbbm{E}\left( u_{it}u_{is}u_{jp}u_{jq}\right) -\mathbbm{E}\left( u_{it}u_{is}\right) \mathbbm{E}\left(
u_{jp}u_{jq}\right) \right]  \\
&= \frac{1}{T^{2}}\sum_{p=1}^{T}\sum_{q=1}^{T}f_{p}f_{q}f_{u}f_{v}\left[
1\left\{ i=j\right\} 1\left\{ \text{two pairs among }\left\{ t,s,p,q\right\}
\right\} \right] .
\end{align*}%
Therefore, 
\begin{equation*}
\frac{1}{NT^{2}}\sum_{t,s,u,v}\sum_{i,j}\left\vert {\rm Cov}\left(
u_{it}u_{is},v_{2,ju}v_{2,jv}\right) \right\vert \leq M.
\end{equation*}

\item Notice that%
\begin{align*}
&{\rm Cov}\left( u_{it}u_{is},v_{3,ju}v_{3,jv}\right)  \\
&=\mathbbm{E}\left( u_{it}u_{is}\left( \frac{1}{\sqrt{NT}}\sum_{k=1}^{N}%
\sum_{p=1}^{T}\lambda _{k}f_{p}u_{kp}\right) \frac{\lambda _{j}f_{u}}{\sqrt{%
NT}}\left( \frac{1}{\sqrt{NT}}\sum_{l=1}^{N}\sum_{q=1}^{T}\lambda
_{l}f_{q}u_{lq}\right) \frac{f_{v}\lambda _{j}}{\sqrt{NT}}\right)  \\
&\quad -\mathbbm{E}\left( u_{it}u_{is}\right) \mathbbm{E}\left( \left( \frac{1}{\sqrt{NT}}%
\sum_{k=1}^{N}\sum_{p=1}^{T}\lambda _{k}f_{p}u_{kp}\right) \frac{\lambda
_{j}f_{u}}{\sqrt{NT}}\left( \frac{1}{\sqrt{NT}}\sum_{l=1}^{N}\sum_{q=1}^{T}%
\lambda _{l}f_{q}u_{lq}\right) \frac{f_{v}\lambda _{j}}{\sqrt{NT}}\right)  \\
&= \frac{1}{N^{2}T^{2}}\sum_{k=1}^{N}\sum_{p=1}^{T}\sum_{l=1}^{N}%
\sum_{q=1}^{T}\lambda _{k}\lambda _{j}^{2}\lambda _{l}f_{p}f_{u}f_{q}f_{v}%
\left[ \mathbbm{E}\left( u_{it}u_{is}u_{kp}u_{lq}\right) -\mathbbm{E}\left( u_{it}u_{is}\right)
\mathbbm{E}\left( u_{kp}u_{lq}\right) \right]  \\
&= \frac{1}{N^{2}T^{2}}\sum_{k=1}^{N}\sum_{p=1}^{T}\sum_{q=1}^{T}\lambda
_{k}^{2}\lambda _{j}^{2}\mathbbm{E}f_{p}f_{u}f_{q}f_{v}\left[ \mathbbm{E}\left(
u_{it}u_{is}u_{kp}u_{kq}\right) -\mathbbm{E}\left( u_{it}u_{is}\right) \mathbbm{E}\left(
u_{kp}u_{kq}\right) \right]  \\
&= \frac{1}{N^{2}T^{2}}\sum_{k=1}^{N}\sum_{p=1}^{T}\sum_{q=1}^{T}\lambda
_{k}^{2}\lambda _{j}^{2}\mathbbm{E}f_{p}f_{u}f_{q}f_{v}\left[ 1\left\{ i=k\right\}
1\left\{ \text{two pairs among }\left\{ t,s,p,q\right\} \right\} \right] .
\end{align*}%
Then, 
\begin{equation*}
\frac{1}{NT^{2}}\sum_{t,s,u,v}\sum_{i,j}\left\vert {\rm Cov}\left(
u_{it}u_{is},v_{2,ju}v_{2,jv}\right) \right\vert \leq M.
\end{equation*}

\item The desired result follows similarly since 
\begin{align*}
{\rm Cov}\left( u_{it}u_{is},v_{1,ju}v_{2,jv}\right)  
&=\mathbbm{E}\left( u_{it}u_{is}\frac{\lambda _{j}}{\sqrt{N}}\left( \frac{1}{\sqrt{N}}%
\sum_{k=1}^{N}\lambda _{k}u_{ku}\right) \left( \frac{1}{\sqrt{T}}%
\sum_{q=1}^{T}f_{q}u_{jq}\right) \frac{f_{v}}{\sqrt{T}}\right)  \\
& \quad -\mathbbm{E}\left( u_{it}u_{is}\right) \mathbbm{E}\left( \frac{\lambda _{j}}{\sqrt{N}}\left( 
\frac{1}{\sqrt{N}}\sum_{k=1}^{N}\lambda _{k}u_{ku}\right) \left( \frac{1}{%
\sqrt{T}}\sum_{q=1}^{T}f_{q}u_{jq}\right) \frac{f_{v}}{\sqrt{T}}\right)  \\
&= \frac{1}{NT}\sum_{k=1}^{N}\sum_{q=1}^{T}\lambda _{j}\lambda
_{k}f_{q}f_{v}\left\{ \mathbbm{E}\left( u_{it}u_{is}u_{ku}u_{jq}\right) -\mathbbm{E}\left(
u_{it}u_{is}\right) \mathbbm{E}\left( u_{ku}u_{jq}\right) \right\}  \\
&= \frac{1}{NT}\sum_{q=1}^{T}\lambda _{j}^{2}f_{q}f_{v}\left\{ \mathbbm{E}\left(
u_{it}u_{is}u_{ju}u_{jq}\right) -\mathbbm{E}\left( u_{it}u_{is}\right) \mathbbm{E}\left(
u_{ju}u_{jq}\right) \right\}  \\
&= \frac{1}{NT}\sum_{q=1}^{T}\lambda _{j}^{2}f_{q}f_{v}\left[ 1\left\{
i=j\right\} 1\left\{ \text{two pairs among }\left\{ t,s,u,q\right\} \right\} %
\right] .
\end{align*}

\item The desired result follows since%
\begin{align*}
& {\rm Cov}\left( u_{it}u_{is},v_{1,ju}v_{3,jv}\right)  \\
&=\mathbbm{E}\left( u_{it}u_{is}\frac{\lambda _{j}}{\sqrt{N}}\left( \frac{1}{\sqrt{N}}%
\sum_{k=1}^{N}\lambda _{k}u_{ku}\right) \left( \frac{1}{\sqrt{NT}}%
\sum_{l=1}^{N}\sum_{q=1}^{T}\lambda _{l}f_{q}u_{lq}\right) \frac{\lambda
_{j}f_{v}}{\sqrt{NT}}\right)  \\
& \quad -\mathbbm{E}\left( u_{it}u_{is}\right) \mathbbm{E}\left( \frac{\lambda _{j}}{\sqrt{N}}\left( 
\frac{1}{\sqrt{N}}\sum_{k=1}^{N}\lambda _{k}u_{ku}\right) \left( \frac{1}{%
\sqrt{NT}}\sum_{l=1}^{N}\sum_{q=1}^{T}\lambda _{l}f_{q}u_{lq}\right) \frac{%
\lambda _{j}f_{v}}{\sqrt{NT}}\right)  \\
&= \frac{1}{N^{2}T}\sum_{k=1}^{N}\sum_{l=1}^{N}\sum_{q=1}^{T}\mathbbm{E}\left( \lambda
_{j}^{2}\lambda _{k}\lambda _{l}\right) \mathbbm{E}\left( f_{q}f_{v}\right) \left\{
\mathbbm{E}\left( u_{it}u_{is}u_{ku}u_{lq}\right) -\mathbbm{E}\left( u_{it}u_{is}\right) \mathbbm{E}\left(
u_{ku}u_{lq}\right) \right\}  \\
&=\text{it should be }[\text{q=v and k=l]} \\
&=\frac{1}{N^{2}T}\sum_{k=1}^{N}\mathbbm{E}\left( \lambda _{j}^{2}\lambda
_{k}^{2}\right) \mathbbm{E}\left( f_{v}^{2}\right) \left\{ \mathbbm{E}\left(
u_{it}u_{is}u_{ku}u_{kv}\right) -\mathbbm{E}\left( u_{it}u_{is}\right) \mathbbm{E}\left(
u_{ku}u_{kv}\right) \right\}  \\
&=\text{it should be that }i=k \\
&=\frac{1}{N^{2}T}\mathbbm{E}\left( \lambda _{j}^{2}\lambda _{i}^{2}\right) \mathbbm{E}\left(
f_{v}^{2}\right) \left\{ \mathbbm{E}\left( u_{it}u_{is}u_{iu}u_{iv}\right) -\mathbbm{E}\left(
u_{it}u_{is}\right) \mathbbm{E}\left( u_{iu}u_{iv}\right) \right\}  \\
&\leq \frac{1}{N^{2}T}\mathbbm{E}\left( \lambda _{j}^{2}\lambda _{i}^{2}\right)
1\left\{ \text{two pairs among }\left\{ t,s,u,v\right\} \right\} .
\end{align*}

\item Similarly, the desired result follows since 
\begin{align*}
& {\rm Cov}\left( u_{it}u_{is},v_{2,ju}v_{3,jv}\right)  \\
&=\mathbbm{E}\left( u_{it}u_{is}\left( \frac{1}{\sqrt{T}}\sum_{p=1}^{T}f_{p}u_{jp}%
\right) \frac{f_{u}}{\sqrt{T}}\left( \frac{1}{\sqrt{NT}}\sum_{l=1}^{N}%
\sum_{q=1}^{T}\lambda _{l}f_{q}u_{lq}\right) \frac{\lambda _{j}f_{v}}{\sqrt{%
NT}}\right)  \\
& \quad -\mathbbm{E}\left( u_{it}u_{is}\right) \mathbbm{E}\left( \left( \frac{1}{\sqrt{T}}%
\sum_{p=1}^{T}f_{p}u_{jp}\right) \frac{f_{u}}{\sqrt{T}}\left( \frac{1}{\sqrt{%
NT}}\sum_{l=1}^{N}\sum_{q=1}^{T}\lambda _{l}f_{q}u_{lq}\right) \frac{\lambda
_{j}f_{v}}{\sqrt{NT}}\right)  \\
&=\frac{1}{NT^{2}}\sum_{l=1}^{N}\sum_{p=1}^{T}\sum_{q=1}^{T}\lambda
_{j}\lambda _{l}f_{p}f_{q}f_{u}f_{v}\left\{ \mathbbm{E}\left(
u_{it}u_{is}u_{jp}u_{lq}\right) -\mathbbm{E}\left( u_{it}u_{is}\right) \mathbbm{E}\left(
u_{jp}u_{lq}\right) \right\}  \\
&=\frac{1}{NT^{2}}\sum_{p=1}^{T}\sum_{q=1}^{T}\lambda
_{j}^{2}f_{p}f_{q}f_{u}f_{v}\left\{ \mathbbm{E}\left( u_{it}u_{is}u_{jp}u_{jq}\right)
-\mathbbm{E}\left( u_{it}u_{is}\right) \mathbbm{E}\left( u_{jp}u_{jq}\right) \right\}  \\
&=\frac{1}{NT^{2}}\sum_{p=1}^{T}\sum_{q=1}^{T}\lambda
_{j}^{2}f_{p}f_{q}f_{u}f_{v}\left[ 1\left\{ i=j\right\} 1\left\{ \text{two
pairs among }\left\{ t,s,u,q\right\} \right\} \right]  .
\end{align*}
\end{enumerate}

\item ${\rm Cov}\left( u_{it}v_{is},u_{ju}v_{jv}\right) :$ Notice that 
\begin{align*}
&{\rm Cov}\left( u_{it}v_{is},u_{ju}v_{jv}\right)  \\
&=\mathbbm{E}\left( u_{it}v_{is}u_{ju}v_{jv}\right) -\mathbbm{E}\left( u_{it}v_{is}\right)
\mathbbm{E}\left( u_{ju}v_{jv}\right)  \\
&=\mathbbm{E}\left( 
\begin{array}{c}
u_{it}\left( \frac{\lambda _{i}}{\sqrt{N}}\left( \frac{1}{\sqrt{N}}%
\sum_{k=1}^{N}\lambda _{k}u_{ks}\right) +\left( \frac{1}{\sqrt{T}}%
\sum_{p=1}^{T}f_{p}u_{ip}\right) \frac{f_{s}}{\sqrt{T}}+\tilde{v}\left(
\lambda ,f,u\right) \frac{\lambda _{i}}{\sqrt{N}}\frac{f_{s}}{\sqrt{T}}%
\right)  \\ 
\times u_{ju}\left( \frac{\lambda _{j}}{\sqrt{N}}\left( \frac{1}{\sqrt{N}}%
\sum_{l=1}^{N}\lambda _{l}u_{lv}\right) +\left( \frac{1}{\sqrt{T}}%
\sum_{q=1}^{T}f_{q}u_{jq}\right) \frac{f_{v}}{\sqrt{T}}+\tilde{v}\left(
\lambda ,f,u\right) \frac{\lambda _{j}}{\sqrt{N}}\frac{f_{v}}{\sqrt{T}}%
\right) 
\end{array}%
\right)  \\
& \quad -\mathbbm{E}\left( u_{it}v_{is}\right) \mathbbm{E}\left( u_{ju}v_{jv}\right) .
\end{align*}

\begin{enumerate}
\item The desired result follows since%
\begin{align*}
&{\rm Cov}\left( u_{it}v_{1,is},u_{ju}v_{1,jv}\right)  \\
&=\mathbbm{E}\left( u_{it}v_{1,is}u_{ju}v_{1,jv}\right) -\mathbbm{E}\left(
u_{it}v_{1,is}\right) \mathbbm{E}\left( u_{ju}v_{1,jv}\right)  \\
&=\mathbbm{E}\left( u_{it}u_{ju}\frac{\lambda _{i}}{\sqrt{N}}\left( \frac{1}{\sqrt{N}}%
\sum_{k=1}^{N}\lambda _{k}u_{ks}\right) \frac{\lambda _{j}}{\sqrt{N}}\left( 
\frac{1}{\sqrt{N}}\sum_{l=1}^{N}\lambda _{l}u_{lv}\right) \right)  \\
&\quad -\mathbbm{E}\left( u_{it}\frac{\lambda _{i}}{\sqrt{N}}\left( \frac{1}{\sqrt{N}}%
\sum_{k=1}^{N}\lambda _{k}u_{ks}\right) \right) \mathbbm{E}\left( u_{ju}\frac{\lambda
_{j}}{\sqrt{N}}\left( \frac{1}{\sqrt{N}}\sum_{l=1}^{N}\lambda
_{l}u_{lv}\right) \right)  \\
&=\frac{1}{N^{2}}\sum_{k=1}^{N}\sum_{l=1}^{N}\lambda _{i}\lambda
_{k}\lambda _{j}\lambda _{l}\left\{ \mathbbm{E}\left( u_{it}u_{ju}u_{ks}u_{lv}\right)
-\mathbbm{E}\left( u_{it}u_{ks}\right) \mathbbm{E}\left( u_{ju}u_{lv}\right) \right\}  .
\end{align*}%
So, 
\begin{align*}
& \frac{1}{NT^{2}}\sum_{t,s,u,v}\sum_{i,j}\left\vert {\rm Cov}\left(
u_{it}v_{1,is},u_{ju}v_{1,jv}\right) \right\vert  \\
&=\frac{1}{NT^{2}}\sum_{t,s,u,v}\sum_{i,j}\frac{1}{N^{2}}%
\sum_{k=1}^{N}\sum_{l=1}^{N}\left[ \lambda _{i}\lambda _{k}\lambda
_{j}\lambda _{l}\left\{ 
\begin{array}{c}
\mathbbm{E}\left( u_{it}u_{ju}u_{ks}u_{lv}\right)  \\ 
-\mathbbm{E}\left( u_{it}u_{ks}\right) \mathbbm{E}\left( u_{ju}u_{lv}\right) 
\end{array}%
\right\} \right]  \\
&\leq \frac{M}{N^{3}T^{2}}\sum_{t,s,u,v}\sum_{i,j,k,l}1\left\{ \text{two
pairs among }\left\{ i,j,k,l\right\} \right\} 1\left\{ \text{two pairs among 
}\left\{ t,s,u,v\right\} \right\}  \\
&\leq M.
\end{align*}

\item Also, we have%
\begin{align*}
&{\rm Cov}\left( u_{it}v_{2,is},u_{ju}v_{2,jv}\right)  \\
&=\mathbbm{E}\left( u_{it}v_{2,is}u_{ju}v_{2,jv}\right) -\mathbbm{E}\left(
u_{it}v_{2,is}\right) \mathbbm{E}\left( u_{ju}v_{2,jv}\right)  \\
&=\mathbbm{E}\left( u_{it}u_{ju}\left( \frac{1}{\sqrt{T}}\sum_{p=1}^{T}f_{p}u_{ip}%
\right) \frac{f_{s}}{\sqrt{T}}\left( \frac{1}{\sqrt{T}}%
\sum_{q=1}^{T}f_{q}u_{jq}\right) \frac{f_{v}}{\sqrt{T}}\right)  \\
&\quad -\mathbbm{E}\left( u_{it}\left( \frac{1}{\sqrt{T}}\sum_{p=1}^{T}f_{p}u_{ip}\right) 
\frac{f_{s}}{\sqrt{T}}\right) \mathbbm{E}\left( u_{ju}\left( \frac{1}{\sqrt{T}}%
\sum_{q=1}^{T}f_{q}u_{jq}\right) \frac{f_{v}}{\sqrt{T}}\right)  \\
&=\frac{1}{T^{2}}\sum_{p=1}^{T}\sum_{q=1}^{T}f_{p}f_{s}f_{q}f_{v}\left\{
\mathbbm{E}\left( u_{it}u_{ip}u_{jq}u_{ju}\right) -\mathbbm{E}\left( u_{it}u_{ip}\right) \mathbbm{E}\left(
u_{ju}u_{jq}\right) \right\}  \\
&=\frac{1}{T^{2}}\sum_{p=1}^{T}\sum_{q=1}^{T}f_{p}f_{s}f_{q}f_{v}\left\{
1\left\{ i=j\right\} 1\left\{ \text{two pairs among }\left\{ t,p,q,u\right\}
\right\} \right\} .
\end{align*}%
So, 
\begin{equation*}
\frac{1}{NT^{2}}\sum_{t,s,u,v}\sum_{i,j}\left\vert {\rm Cov}\left(
u_{it}v_{2,is},u_{ju}v_{2,jv}\right) \right\vert \leq M.
\end{equation*}

\item We can show the rest of the cases analogously.
\end{enumerate}

\item ${\rm Cov}\left( u_{it}v_{is},v_{ju}v_{jv}\right) :$ There are 4 kinds, (i)
\# of $v_{3,\cdot \cdot }$ = 0, (ii) \# of $v_{3,\cdot \cdot }$ = 1, \ (iii)
\# of $v_{3,\cdot \cdot }$ = 2, and (iv) \# of $v_{3,\cdot \cdot }$ = 4.

\begin{enumerate}
\item When \# of $v_{3,\cdot \cdot }$ = 0: For example, ${\rm Cov}\left(
u_{it}v_{1,is},v_{1,ju}v_{1,jv}\right) .$ The desired result follows since 
\begin{align*}
&\frac{1}{NT^{2}}\sum_{t,s,u,v}\sum_{i,j}\left\vert \mathbbm{E}\left(
u_{it}v_{1,is}v_{1,ju}v_{2,jv}\right) -\mathbbm{E}\left( u_{it}v_{1,is}\right) \mathbbm{E}\left(
v_{1,ju}v_{2,jv}\right) \right\vert  \\
&\leq \frac{M}{NT^{2}}\sum_{t,s,u,v}\sum_{i,j}\frac{1}{N^{2}T}\sum_{i^{\ast
},j^{\ast }}\sum_{v^{\ast }}\left( \mathbbm{E}\left( u_{it}u_{i^{\ast }s}u_{j^{\ast
}u}u_{jv^{\ast }}\right) -\mathbbm{E}\left( u_{it}u_{i^{\ast }s}\right) \mathbbm{E}\left(
u_{j^{\ast }u}u_{jv^{\ast }}\right) \right)  \\
&\leq \frac{M}{N^{3}T^{3}}\sum_{t,s,u,v,v^{\ast }}\sum_{i,j,i^{\ast
},j^{\ast }}1\left\{ \text{two pairs among }\left\{ i,i^{\ast },j,j^{\ast
}\right\} \right\} 1\left\{ \text{two pairs among }\left\{ t,s,u,v^{\ast
}\right\} \right\}  \\
&\leq M.
\end{align*}

\item When \# of $v_{3,\cdot \cdot }$ = 1: For example, ${\rm Cov}\left(
u_{it}v_{3,is},v_{1,ju}v_{2,jv}\right) .$\ The desired result follows since%
\begin{align*}
&\frac{1}{NT^{2}}\sum_{t,s,u,v}\sum_{i,j}\left\vert {\rm Cov}\left(
u_{it}v_{3,is},v_{1,ju}v_{2,jv}\right) \right\vert  \\
&=\frac{1}{NT^{2}}\sum_{t,s,u,v}\sum_{i,j}\left\vert \mathbbm{E}\left(
u_{it}v_{3,is}v_{1,ju}v_{2,jv}\right) -\mathbbm{E}\left( u_{it}v_{3,is}\right) \mathbbm{E}\left(
v_{1,ju}v_{2,jv}\right) \right\vert  \\
&\leq \frac{M}{NT^{2}}\frac{1}{N^{2}T^{2}}\sum_{t,s,u,v}\sum_{i,j}\sum_{i^{%
\ast },j^{\ast }}\sum_{s^{\ast },v^{\ast }}\left\{ \mathbbm{E}\left( u_{it}u_{i^{\ast
}s^{\ast }}u_{j^{\ast }u}u_{jv^{\ast }}\right) -\mathbbm{E}\left( u_{it}u_{i^{\ast
}s^{\ast }}\right) \mathbbm{E}\left( u_{j^{\ast }u}u_{jv^{\ast }}\right) \right\}  \\
&=\frac{M}{N^{3}T^{4}}\sum_{t,s,u,v}\sum_{i,j}\sum_{i^{\ast },j^{\ast
}}\sum_{s^{\ast },v^{\ast }}1\left\{ \text{two pairs among }\left\{
i,i^{\ast },j,j^{\ast }\right\} \right\} 1\left\{ \text{two pairs among }%
\left\{ t,s^{\ast },u,v^{\ast }\right\} \right\}  \\
&\leq M.
\end{align*}

\item When \# of $v_{3,\cdot \cdot }$ = 2: For example, ${\rm Cov}\left(
u_{it}v_{3,is},v_{3,ju}v_{2,jv}\right) .$ The desired result follows since%
\begin{align*}
&\frac{1}{NT^{2}}\sum_{t,s,u,v}\sum_{i,j}\left\vert {\rm Cov}\left(
u_{it}v_{3,is},v_{3,ju}v_{2,jv}\right) \right\vert  \\
&=\frac{1}{NT^{2}}\sum_{t,s,u,v}\sum_{i,j}\left\vert \mathbbm{E}\left(
u_{it}v_{3,is}v_{1,ju}v_{2,jv}\right) -\mathbbm{E}\left( u_{it}v_{3,is}\right) \mathbbm{E}\left(
v_{3,ju}v_{2,jv}\right) \right\vert  \\
&\leq \frac{M}{NT^{2}}\sum_{t,s,u,v}\sum_{i,j}\frac{1}{N^{2}T^{3}}%
\sum_{i^{\ast },j^{\ast }}\sum_{u^{\ast },s^{\ast },v^{\ast }}\left( \mathbbm{E}\left(
u_{it}u_{i^{\ast }s^{\ast }}u_{j^{\ast }u^{\ast }}u_{jv^{\ast }}\right)
-\mathbbm{E}\left( u_{it}u_{i^{\ast }s^{\ast }}\right) \mathbbm{E}\left( u_{j^{\ast }u^{\ast
}}u_{jv^{\ast }}\right) \right)  \\
&=\frac{M}{N^{3}T^{5}}\sum_{t,s,u,v}\sum_{i,j}\sum_{i^{\ast },j^{\ast
}}\sum_{u^{\ast },s^{\ast },v^{\ast }}1\left\{ \text{two pairs among }%
\left\{ i,i^{\ast },j,j^{\ast }\right\} \right\} 1\left\{ \text{two pairs
among }\left\{ t,s^{\ast },u^{\ast },v^{\ast }\right\} \right\}  \\
&\leq M.
\end{align*}

\item For the other cases, notice that an additional $v_{3,\cdot \cdot }$
term adds an extra summation. However, it also increases the order of the
denominator by one. Therefore, the required result follows analogously.
\end{enumerate}

\item ${\rm Cov}\left( v_{it}v_{is},v_{ju}v_{jv}\right) :$ 
Follows analogously to the previous cases.
\end{enumerate}

\end{small}

\bigskip
\bigskip
\bigskip
\bigskip
\bigskip
\bigskip
\bigskip
\bigskip

\begin{table}[tbh!]
\centering
\includegraphics[width=17cm]{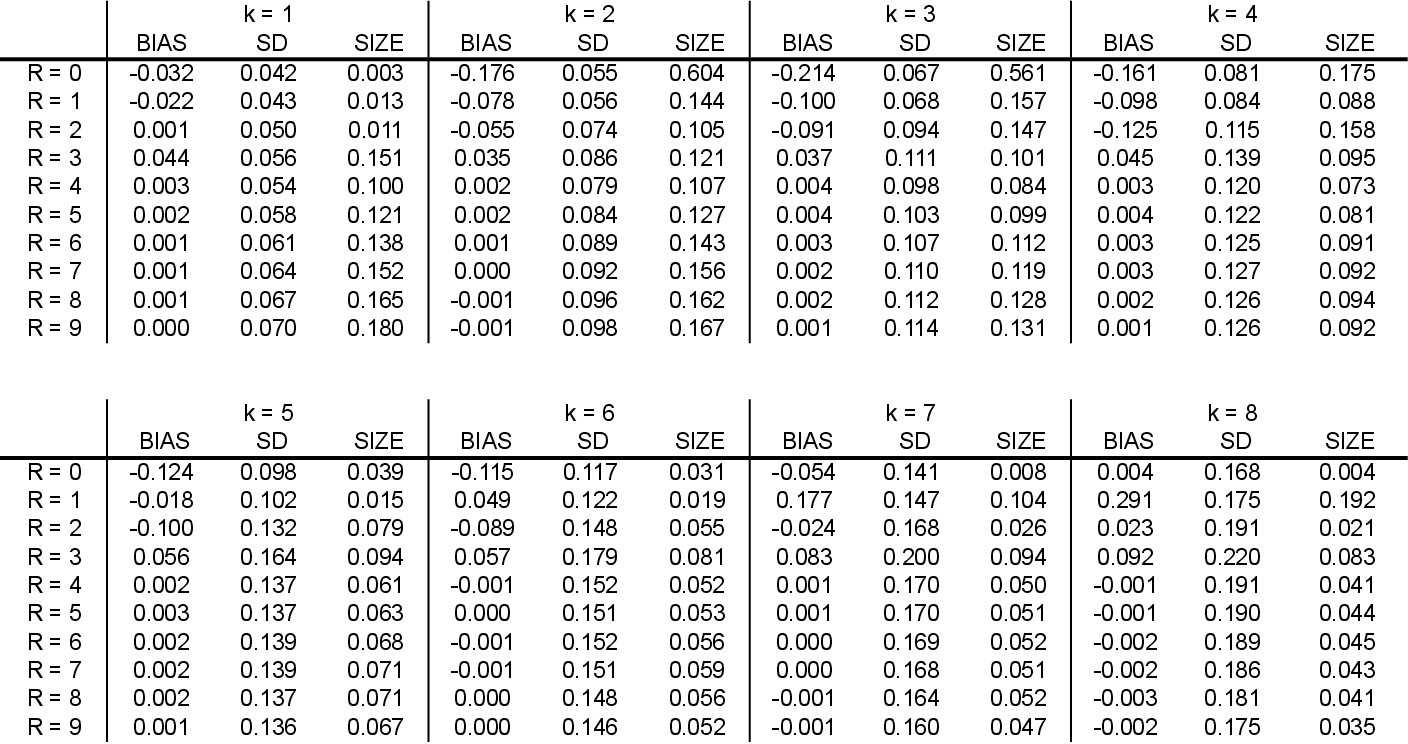}
\caption{\label{tab:MC-Empirical}\footnotesize
Results of the Empirical Monte Carlo Simulation. Bias and standard
deviation (SD) of $\widehat \beta^{\rm BC}_{R,k}$
are reported for the regressor $k=1,\ldots,8$
using $R=0,\ldots,9$ factors in the estimation procedure.
We also report the empirical size of a $5 \%$ nominal size $t$-test of the 
hypothesis $H_0: \beta_k = \beta^0_k$.  
Results are based on $10,000$ repetitions.}
\end{table}

\begin{table}[tbh!]
\centering
\includegraphics[width=14cm]{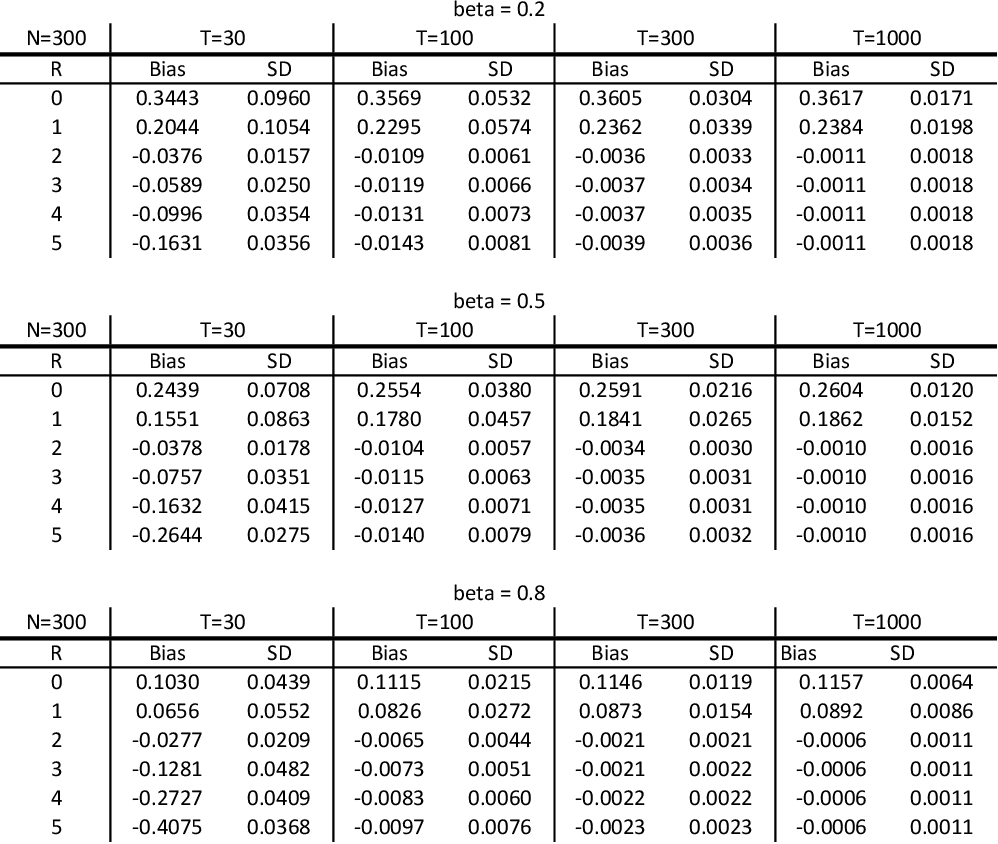}
\caption{\label{tab:MC-Dynamic-1}\footnotesize
For $N=300$ and different combinations of $T$ and true parameter $\beta^0$ we report
the bias and standard deviation of the estimator $\widehat \beta_R$, for $R=0,1,\ldots,5$, based on
simulations with $10,000$ repetition of design \eqref{DGP-Dynamic}, where the true number of
factors is $R^0=2$.}
\end{table}

\begin{table}[tbh!]
\centering
\includegraphics[width=14cm]{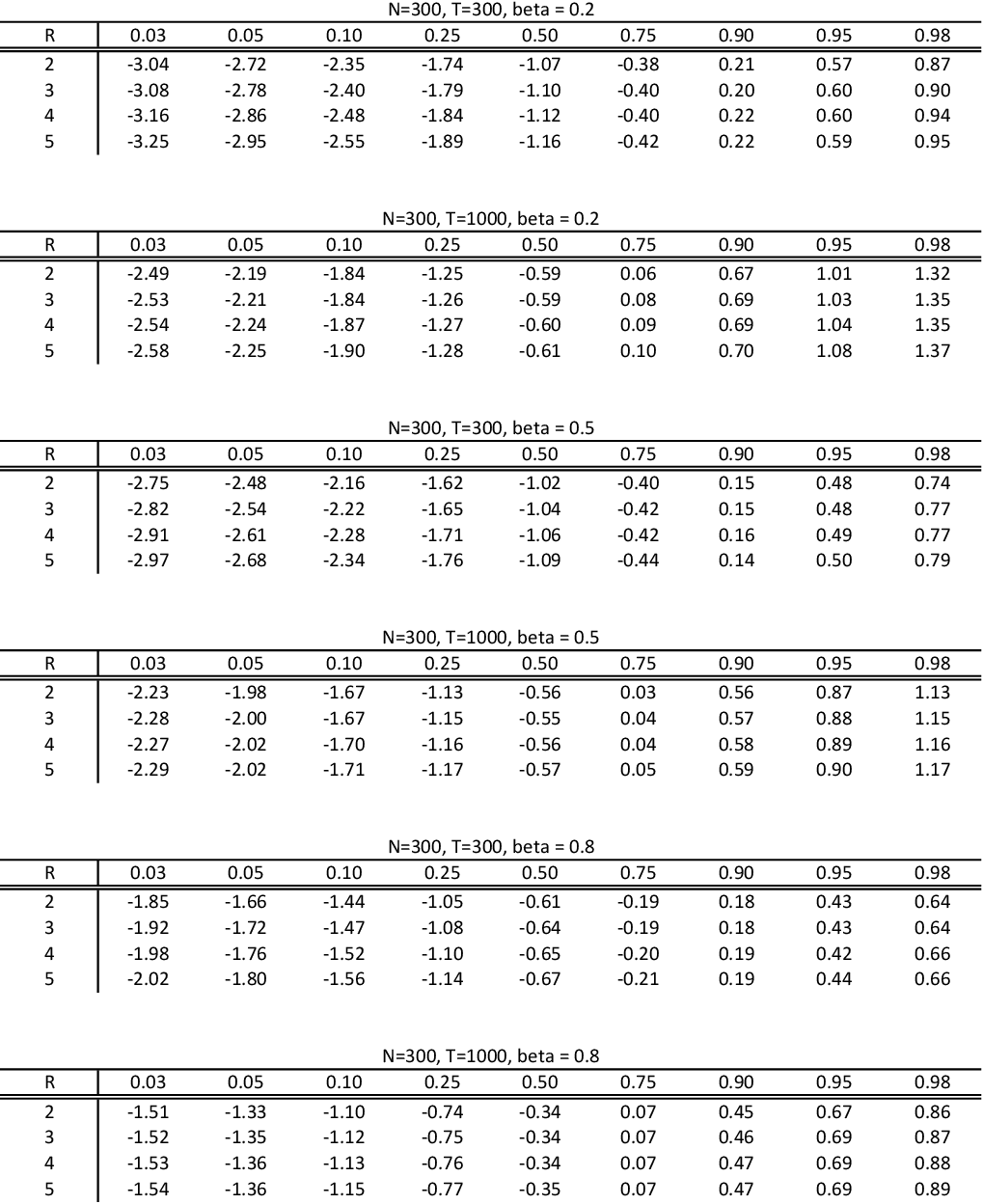}
\caption{\label{tab:MC-Dynamic-2}\footnotesize
For $N=300$ and different combinations of $T$ and true parameter $\beta^0$ we report certain
 quantiles of the distribution of $\sqrt{NT}( \widehat \beta_R - \beta^0 )$, for $R=2,3,4,5$,
 based on
simulations with $10,000$ repetition of design \eqref{DGP-Dynamic}, where the true number of
factors is $R^0=2$.
 }
\end{table}

\begin{table}[tbh!]
\centering
\includegraphics[width=14cm]{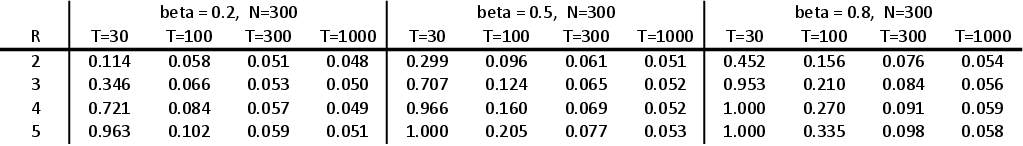}
\caption{\label{tab:MC-Dynamic-3}\footnotesize
The empirical size of a t-test with $5 \%$ nominal size is reported for
$N=300$ and different combinations of $T$, $R$ and true parameter $\beta^0$,
based on $10,000$ repetition  of design \eqref{DGP-Dynamic}.
A bias corrected estimator for $\beta$ is used to calculate the
test statistics, and we allow for predetermined regressors and heteroscedastic errors when estimating bias and standard
deviation. Results for $R=0,1$ are not reported since those have size=1 due to misspecification.}
\end{table}

\end{document}